\documentclass[12pt,oneside]{book}
\usepackage{RutgersThesis}

\usepackage{amsrefs} 
	
\usepackage{graphicx} 
\usepackage{lipsum} 

\usepackage{latexsym,hyperref,url}

\usepackage{tikz,color,lscape}

\usetikzlibrary{cd}
\usepackage{amsfonts, amsmath, amssymb, amsthm, color, enumitem, fancyhdr, relsize, tikz, graphicx, mathrsfs}
\usepackage[utf8]{inputenc}

\newtheorem{prop}{Proposition}[section]

\newtheorem{thm}[prop]{Theorem}
\newtheorem{lemma}[prop]{Lemma}

\theoremstyle{definition}
\newtheorem{remark}[prop]{Remark}
\newtheorem{example}[prop]{Example}

\newtheorem{defn}[prop]{Definition}

\newcommand{\R}{\mathbb{R}}
\newcommand{\C}{\mathbb{C}}
\newcommand{\N}{\mathbb{N}}

\newcommand{\Z}{\mathbb{Z}}
\newcommand{\Q}{\mathbb{Q}}

\newcommand{\tcw}{\textcolor{white}}

\newcommand{\lam}{\lambda}
\newcommand{\Lam}{\Lambda}

\renewcommand{\Im}{\operatorname{Im}}
\renewcommand{\Re}{\operatorname{Re}}

\renewcommand{\d}{\,d}

\newcommand{\sgn}{\operatorname{sgn}}

\newcommand{\diag}{\operatorname{diag}}
\newcommand{\ws}{\operatorname{ws}}
\newcommand{\bws}{\operatorname{b-ws}}
\newcommand{\diam}{\operatorname{diam}}
\newcommand{\dist}{\operatorname{dist}}

\newcommand{\bp}{\begin{pmatrix}}
\newcommand{\ep}{\end{pmatrix}}
\newcommand{\spn}{\operatorname{span}}

\newcommand{\sround}{\operatorname{s-rnd}}

\newcommand{\Tr}{\operatorname{Tr}}
\newcommand{\wlim}{\operatorname{w-lim}}
\newcommand{\slim}{\operatorname{s-lim}}

\newcommand{\dark}{}

\title{Constructing Nearby Commuting Matrices for Ogata's Theorem on Macroscopic Observables}
\department{Mathematics}
\author{David Herrera}
\director{Eric A. Carlen}
\date{October, 2024}
\committeenumber{4}

\begin{document}
	\frontmatter
	\maketitle
	\begin{abstract}
		Resolving a conjecture of von Neumann, Ogata's theorem in \cite{ogata2013approximating} showed the highly nontrivial result that arbitrarily many matrices corresponding to macroscopic observables with $N$ sites and a fixed site dimension $d$ are asymptotically nearby commuting observables as $N \to \infty$.

We develop a method to construct nearby commuting matrices for normalized highly reducible representations of $su(2)$ whose multiplicities of irreducible subrepresentations exhibit a certain monotonically decreasing behavior.

We then provide a constructive proof of Ogata's theorem for site dimension $d=2$ with explicit estimates for how close the nearby observables are. Moreover, motivated by the application to time-reversal symmetry explored in \cite{loring2016almost}, our construction has the property that   real macroscopic observables are asymptotically nearby real commuting observables.

This thesis also contains an introduction to the prerequisite matrix analysis for understanding the proof of the results gotten, a primer on the functional analysis of $C^\ast$-algebras frequently used in the literature of almost commuting operators, a detailed discussion of the history and main results in the problem of almost commuting matrices, a review of the basics of measurement of observables in finite dimensional quantum mechanical systems, and a discussion of the uncertainty principle for non-commuting bounded operators with some results obtained for observables with nearby commuting approximants.
	\end{abstract}

	\begin{acknowledgments}
	\begin{center}
\begin{singlespacing}
{\it For who sees anything different in you? What do you have that you did not receive? If then you received it, why do you boast as if you did not receive it?}\\
-- First Letter to the Corinthians 4:7 (ESV)

\end{singlespacing}
\end{center}
 
I am indebted to everyone whose help I received in any way while writing this dissertation, especially 
my thesis advisor, Professor Eric A. Carlen, and the rest of my thesis committee: Professors Ian Jauslin, Michael Kiessling, and Yoshiko Ogata. 

I would like to thank Professor Carlen for always providing helpful guidance throughout my stay here at Rutgers. 
His lectures and notes (\cite{carlen2016operator}, \cite{carlen2018positive}) introduced me to the interesting field of operator algebras and the almost/nearly commuting matrix problem.
His textbook draft (\cite{CarlenBook}) helped me think through what I would want to say about quantum mechanics.
His continued guidance during the writing of this thesis as well as providing context for the utility of the obtained results made this a reality.
Thank you for having me as your student. 

I would like to thank the rest of my Rutgers University professors, including the late Professor Brezis (may he rest in peace) and my oral exam committee: Professors Eric Carlen, Fioralba Cakoni, YanYan Li, and H{\'e}ctor Sussmann.
I learned a lot here and often think about discussions had and anecdotes shared. 

\vspace{0.1in}

I would like to thank Terry A. Loring for his feedback on the benefits and limitations of the construction presented in this thesis; Tatiana Shulman for helpful explanations of the cohomological aspects of her paper \cite{enders2019almost}; and Songhao Zhu, Matt Charnley, and Brian C. Hall for help with thesis formatting and citation.
I would also like to thank Ilya Kachkovsky for a helpful clarification about the constructiveness of his extension of Lin's theorem and Matthew Hastings for discussions concerning his paper \cite{hastings2009making} which helped me gain familiarity with techniques that later were integral to proving the main results of this thesis.

I would also like to thank the Mathematics Department administrative staff (especially Kathleen Guarino), building maintenance staff, the Rutgers Library, and my fellow graduate students 
C{\'e}sar Ram{\'i}rez Iba{\~n}ez, 
Xiaoxu Wu,
Weihong Xu,
Songhao Zhu,
Victoria Chayes, 
The Doahn Pham,
Nathan Mehlhop,
and others who in different ways have been supportive of my life, education, and teaching as a Rutgers graduate student. 

I also could not forget to thank the various contributors on the StackOverflow forums and various programming and LaTeX formatting help websites available online, including:\\ \url{http://detexify.kirelabs.org/classify.html} and \url{https://tikzcd.yichuanshen.de}

\vspace{0.1in}

I would like to thank my undergraduate professors and mentors: Mark McConnell for supporting me as a student and fellow educator, Elias Stein (may he rest in peace) for his inspiring and enlightening lectures and textbooks coauthored with Rami Shakarchi, 
Elliot Lieb for his engaging lectures and for gifting me a textbook through Professor Carlen, and 
others including: 
Marianne Korten, 
Adam Levine, 
Peter Ozsv{\'a}th,
Fabio Pusateri, 
Steven Sivek,
Suzanne Staggs, 
Zolt{\'a}n Szab{\'o},
Vlad Vicol, and
Xiaoheng Wang.

For encouraging me in my growth as a person and my pursuit of education and mathematics, I would like to specifically thank Shannon Osaka, Annie Lu, Stephen Timmel, Nathan Wei, Bill and Debbie Boyce, Manny, Lauren, and my high school teachers, especially
Christine Stromberg and Joan Manigrasso. 
Of the things I have accomplished and the mistakes I have learned from, I hope that I did not let you down.
I also would like to thank Arzu Holiday again for introducing me to Mrs. Stromberg. You forever changed my life.

I would like to thank my family for their love,  help, and support. I could not rightly summarize all the ways that you helped me during these many years: through homeschooling, middle/high school, undergrad, and graduate school. There have been rough patches and I could not have made it through them without you.

\vspace{0.05in}

This thesis contains material written by the author in the preprint \cite{herrera2022constructing}. This thesis' Front Matter and Chapters 1-4 is a greatly expanded form of material from that paper and the material of this thesis' Chapters 5-9 are essentially the main contents of that paper.

The research of \cite{herrera2022constructing} was partially supported by NSF grants DMS-2055282 and DMS-1764254.
I have been blessed by these (and the other NSF grants through my advisor) as well as the fellowships and teaching appointments I have received as an undergraduate at Princeton University and as a graduate student here at Rutgers University. The research I have completed and the education I have obtained would not have been possible without such support.
Thank you.

	\end{acknowledgments}
        
        \begin{dedication}
		\vspace{1.5 in}
            \vspace{1in}

\begin{center}
\begin{singlespacing}
{\it To those whose understanding (and, quite frankly, patience) made completing this a reality.\\
To the people who invested in, believed in, and/or prayed for me.\\ 
To the One who gave me what I have.

}
\end{singlespacing}
\end{center}
        \end{dedication}
 
	\tableofcontents

	\listoffigures

 

	\mainmatter
	
	\chapter{Introduction}
		\section{Summary of Dissertation Subject and its Context}
This dissertation is about almost commuting matrices: the history of the mathematical problem, the problem's connection to non-commuting operators in quantum mechanics, and proving an extension of Ogata's theorem that ties together these topics.
We solve a linear algebra problem related to the Uncertainty Principle for things you can measure, called observables, in quantum mechanics.

The observables of interest are those of a large (``macroscopic'') system made out of many identical small systems with only two possible values to be measured. The identical microscopic observables, for instance, could be many spin-1/2 particles whose spin when measured in a single direction produces only a measurement of spin up or spin down. Three macroscopic observables formed in this way are the total magnetization measured in the $x$, $y$, and $z$ directions.

In quantum mechanics, observables are represented by matrices. If the two observables are represented by square matrices $A, B$ that do not commute: $AB \neq BA$ then these observables are not compatible. Non-compatible observables can exhibit strange behavior upon measurement and give non-intuitive (or even no) answer when one asks simple questions about these observables such as ``What are the values of these observables at the current moment?''.

Although macroscopic observables in general are not compatible, the measure of how incompatible they are is very small and only gets smaller the more particles that there are. 
It was proven by Ogata that macroscopic observables can be approximated by commuting observables, however the result was based on a proof by contradiction which provides no information about how large the macroscopic system needs to be in order for the approximation to be of use, practically speaking.
One of our main results is to present a precise estimate of how close these incompatible observables are to compatible observables through an explicit construction of some nearby commuting observables.

This is interesting because compatible observables act more like things that you can measure in ``classical physics'' whereas, strictly speaking, a large quantum mechanical system always behaves according to the laws of quantum mechanics.
Large systems composed of many particles are typically well described by Newton's laws and standard electromagnetism: classical physics. However, some of the ways that systems behave that were not understood by a classical physics analysis led to the rise of quantum mechanics. 
It is a very important question how one connects the way that we perceive that the world operates (approximately classically) and the way that we understand that small systems behave (well explained by quantum mechanics) because large systems are made out of small systems.

Much work has been done over the decades to work out how one can view this connection between classical and quantum mechanics. The work in the current thesis contributes to this discussion.

\section{A Remark on Prerequisites}

This thesis is written in a way that is intended to be able to be read by the mathematically mature graduate student in physics or mathematics, not only those who are experts in both fields. I really mean this. 
The only formal prerequisites for understanding the entirety of this thesis is a familiarity with Hilbert spaces and some basic measure theory. For understanding the main results and proof in the thesis, only an understanding of finite Hilbert spaces is needed.

Much of the prerequisite material beyond these fundamentals for one or the other field is presented directly in order to help each type of intended reader understand the main ideas and examples as well as get a good impression for the types of arguments that go into reading the relevant literature. 

Chapter \ref{1.MathematicsPreliminaries} is intended to be a relatively comprehensive and self-contained primer on the mathematical prerequisites for understanding the almost commuting matrices literature, from the perspective of a mathematician but not assuming the matrix analysis or functional analysis background discussed. 
The goal is to allow the reader to accustom themselves with the relevant results, thinking, and methods found in the literature without requiring having seen a textbook introduction to the subject. 

The role of the survey of spectral theory and $C^\ast$-algebras in Section \ref{b.CastAlgs} is to help explain the necessary ideas to appreciate some of the arguments and results in Chapter \ref{2.AlmostCommutingMatrices} that use certain abstract $C^\ast$-algebras and to appreciate some of the results in Chapter \ref{3.MathematicalPhysicsOfAlmostCommutingObservables} expressed in the generality of observables represented by bounded or compact operators. 

In particular, in Chapter \ref{2.AlmostCommutingMatrices}, we present the main ideas concerning $C^\ast$-algebra lifting arguments for proving certain non-constructive results for almost commuting matrices, however the emphasis of our presentation is on constructive arguments producing asymptotic and numerical estimates which tend to be less abstract. An in-depth overview of almost commuting matrices without a review of the lifting method would be incomplete and a discussion of the lifting method without referencing $C^\ast$-algebra methods would be too cursory. 
All this said, Section \ref{b.CastAlgs} does not play a role the main results of this thesis.

The physics discussion in Chapter \ref{3.MathematicalPhysicsOfAlmostCommutingObservables} is primarily focused on presenting a mathematical framework for appreciating the relevant mathematical physics constructions and problems.
Quantum mechanical states, measurements, and observables are introduced in almost a platonic sense for the sake of later mathematical analysis, just as one may introduce the rules and intuition behind chess tactics for the sake of later game theoretic analysis without much comment on the impact that chess had on history or its modern cultural relevance. 
Chapter \ref{3.MathematicalPhysicsOfAlmostCommutingObservables}'s basic introduction is intended to explain the core constructions for the mathematician without much physics background. Later sections of the chapter discuss the uncertainty principle and uncertainty relations in more depth since this is one of the particularly interesting aspects of quantum mechanics that is relevant to the mathematical problem. Then once the connection to ``reality'' is explained, we focus on the mathematical aspects of discovering and proving inequalities.

With this in mind, very little space is devoted to some of the basic examples of quantum mechanics that are not directly relevant to the mathematical problem of almost commuting matrices. This includes the Schr\"odinger Equation and the experimental basis of quantum mechanics. Unfortunately, no discussion of topological insulators is made in this thesis, though there is a growing interest in applications of almost commuting matrices to this subject.

\section{Main Results Overview}

We now discuss our extension of Ogata's theorem that we prove in this thesis.

For $A \in M_d(\C)$, we define
\[T_N(A) = \frac1N\left(A \otimes I_d \otimes\cdots\otimes I_d+I_d \otimes A \otimes \cdots\otimes I_d+\cdots+I_d \otimes \cdots\otimes I_d \otimes A\right).\]
When $A$ is an observable, $T_N(A) \in M_{d^N}(\C)$ is a macroscopic observable.

Ogata (\cite{ogata2013approximating}) showed that if $A_1, \dots, A_m \in M_d(\C)$ are self-adjoint and $H_{i,N} = T_N(A_i)$ then there are commuting observables $Y_{i,N}$ that are nearby the almost commuting macroscopic observables $H_{i,N}$ in the sense that $\|H_{i,N}-Y_{i,N}\|\to 0$ as $N \to \infty$.

In Theorem \ref{mainthm} of this paper, we construct nearby commuting matrices for certain normalized direct sums of irreducible representations of $su(2)$. 
As a consequence of this, we provide a constructive proof of Ogata's theorem for $d = 2$ with an explicit constant and an asymptotic rate of decay of $N^{-1/7}$. More precisely, we prove:
\begin{thm}\label{OgataTheorem}
Let $\sigma_i$ be the norm $1/2$ Pauli spin matrices in Equation (\ref{pSpin}). There are commuting self-adjoint matrices $Y_{i,N} \in M_{2^N}(\C)$ such that
\begin{align}\|T_N(\sigma_1)-Y_{1,N}\|, \|T_N(\sigma_2)-Y_{2,N}\| &\leq 6.29 \,N^{-1/7}, \nonumber\\
\|T_N(\sigma_3)-Y_{3,N}\| &\leq 1.09\,N^{-3/7}
\nonumber
\end{align}
where $Y_{1,N}, iY_{2,N}$, and $Y_{3,N}$ are real.

Therefore, there is a linear map $Y_N:M_2(\C)\to M_{2^N}(\C)$ such that the $Y_{N}(A)$ commute for all $A\in M_2(\C)$,
\[Y_{N}(A^\ast)=Y_{N}(A)^\ast,\]
\[Y_{N}(A^T)=Y_{N}(A)^T,\]  and
\[\|T_N(A)-Y_{N}(A)\| \leq 17.92 \|A\|\,N^{-1/7}.\]

Consequently, $Y_{N}$ preserves the property of being self-adjoint, skew-adjoint, symmetric, antisymmetric, real, or imaginary.
\end{thm}
For $A$ self-adjoint, we thus obtain an explicit estimate for how close the commuting observables $Y_{N}(A)$ are to the macroscopic observables $T_N(A)$ as well as a construction of the $Y_{N}(A)$. In terms of how Ogata's theorem is presented in \cite{ogata2013approximating}, the  commuting observables $Y_{1,N}$, $Y_{2,N}$, and $Y_{3,N}$ are nearby the macroscopic observables of the $x$, $y$, and $z$ components of the total magnetization for a quantum spin system of $N$ sites of dimension $d=2$. 

The transpose symmetry of $Y_{N}$ due to our extension of Ogata's theorem may be of interest given the attention given to structured nearby commuting matrices in  \cite{loring2015k, loring2016almost}, which apply it to the theory of topological insulators.

The explicit estimates obtained, the additional structure of the matrices, and the simplification of Ogata's original argument for this case are some of the contributions of this construction.
However, due to the use of the Clebsch-Gordan change of basis and the large size of the matrices, it is unclear how useful the construction would be for generating or manipulating the constructed nearby commuting matrices.

\begin{remark}
As an example of the estimate from the theorem above, a three dimensional array of $N = (10^{10})^3$ particles gives a very small error compared to $\|A\|$. So, the estimate obtained is nontrivial for $N$ in the range of applications. See Remark \ref{useful} for more details. 
Our method can also provide an exponent of $-1/5$ by using \cite{kachkovskiy2016distance}, however the explicit constant is not given and $Y_N$ may not have the transpose symmetry. See Theorem \ref{OptimalResult}. 
\end{remark}

\section{Outline of Chapters}

In Chapter \ref{1.MathematicsPreliminaries}, we discuss the preliminary mathematical material for the subject of almost commuting matrices.
Section \ref{a.MatrixAnalysis} contains the Matrix Analysis that serves as a supplement to a standard linear algebra treatment in order to follow the main results of this thesis. 

Section \ref{b.CastAlgs} contains an introduction to $C^\ast$-algebras as is commonly seen in the field of almost commuting matrices. This section is not necessary for understanding the main results of this thesis. However, it presents some of the abstract context needed to fully appreciate the discussion about the almost commuting operators literature as discussed in Chapter \ref{2.AlmostCommutingMatrices}, Chapter \ref{3.MathematicalPhysicsOfAlmostCommutingObservables}, and the main papers that are discussed in Chapter \ref{2.AlmostCommutingMatrices}. 
Note that the topics in this section are discussed thematically and not in the order of logical progression. 

In Chapter \ref{2.AlmostCommutingMatrices}, we present a detailed exposition of the history of the almost/nearly commuting matrix problem, including long discussions about the main results and the methods employed.

Chapter \ref{3.MathematicalPhysicsOfAlmostCommutingObservables}
contains a primer on the basics of measurement of observables in finite dimensional quantum mechanical systems, a discussion of the uncertainty principle for non-commuting bounded operators with some results obtained for nearly commuting observables, and a review of the basics of macroscopic observables and Ogata's theorem.

In Chapter \ref{4.RepTheory}, we review the basic representation theory of $su(2)$.
We also develop representation theoretic estimates that will be used in later chapters. This is the first chapter that begins the proof of Theorem \ref{OgataTheorem}, which relies on framing the problem in terms of tensor representations of $su(2)$ so that almost commuting self-adjoint matrices can be constructed for the macroscopic observables associated to the Pauli matrices.

In Chapter \ref{5.GradualExchangeLemma}, we discuss weighted shift matrices and our version of Berg's gradual exchange lemma from \cite{berg1975approximation}. 
Berg's gradual exchange lemma  provides a way to perform a small perturbation of a direct sum of weighted shift operators to cause the orbits to interchange. This chapter includes an introduction to our weighted shift diagrams.

Note that because we discuss Berg's result and our extension of it in great detail in this and the next chapter, the exposition of the history and methods of almost commuting matrices in Chapter \ref{2.AlmostCommutingMatrices} only made brief mention of these contributions of Berg to the subject.

In Chapter \ref{6.BergThm}, we adapt Berg's construction from \cite{berg1975approximation} of a nearby normal matrix for an almost normal weighted shift matrix. Our adaptation of Berg's result is aimed at obtaining an optimal estimate in terms of  $\|\,[S^\ast, S]\,\|$ with the additional structure that when the almost normal matrix $S$ is real, the nearby normal constructed will be real as well.

In Chapter \ref{7.GradualExchangeProcess}, a method is developed to obtain almost invariant projections of direct sums of weighted shift matrices that can be used to make almost reducing subspaces. 
This method and the construction of nearby commuting matrices using it are referred to as the gradual exchange process. 
Suppose that $S$ is a direct sum of weighted shift matrices and $A$ is a direct sum of diagonal matrices. Under some conditions on $A$ and $S$, we construct nearby commuting matrices $A'$ and $S'$ using the gradual exchange process. Several illustrations are included to illustrate the algorithm. 

In Chapter \ref{8.MainTheorem}, we prove Theorem \ref{mainthm}, a constructive result with estimates concerning nearby commuting matrices, and Theorem \ref{OgataTheorem}.


\section{Outline of Argument}
\label{Outline of Argument}

We now outline our approach to proving the extension of Ogata's theorem in this thesis. 
Although the result that we prove using this method is for $d = 2$, we only assume this in the discussion below when necessary.

It is sufficient to prove Ogata's theorem for self-adjoint $A_1, \dots, A_k$ being a $\C$-basis for $M_d(\C)$.
In particular, constructing nearby commuting matrices is only an interesting problem for $k \leq d^2$ due to the following reduction. Suppose that the $A_i$ are linearly independent and that we can find nearby commuting matrices $Y_{i, N}$ for $T_N(A_i)$.
If we have a matrix $A \in M_d(\C)$
that can be expanded as $A = \sum_{i=1}^k c_i A_i$
then define 
\begin{align}\label{Ydef}
Y_{N}(A) = Y_{N}\left(\sum_{i=1}^k c_i A_i\right)= \sum_{i=1}^k c_i Y_{i,N}.
\end{align}

We then see that for any $A, B\in M_d(\C)$ in the span of the $A_i$, it holds that $Y_{N}(A)$ and $Y_{N}(B)$ commute. If the constructed $Y_{i,N}$ are self-adjoint, then  $Y_{N}(A)$ is self-adjoint whenever $A$ is. 
Moreover, because all norms on finite dimensional spaces are equivalent, there is a constant $C$ only depending on the $A_i$ such that 
\begin{align}\label{Yestimate}
\|T_N(A) - Y_{N}(A)\| \leq \max_{1\leq i \leq k}\|T_N(A_i)-Y_{i,N}\|\sum_{i=1}^k |c_i| \leq \left(C \max_{1\leq i \leq k}\|T_N(A_i)-Y_{i,N}\|\right)\|A\|
\end{align}
converges to zero uniformly as $N \to \infty$ for $\|A\|$ bounded. Because $T_N(I_d) = I_{d^N},$ if $A_i$ for $i = i_0$ is a multiple of the identity, then we need only focus on constructing nearby commuting matrices for the other $A_i$ and can ignore $i=i_0$ in $\sum_i |c_i|$.

We now specialize to the case $d = 2$. 
We choose the specifically useful basis $A_i$ of $M_2(\C)$ given by $\sigma_1, \sigma_2, \sigma_3, \frac12 I$, where we use the following convention for the Pauli spin matrices:
\[
\sigma_1= \frac12\bp 0 & 1\\1&0  \ep,\;\; \sigma_2 = \frac12\bp 0 & i\\ -i &0 \ep,\;\; \sigma_3=\frac12\bp -1&0\\0&1 \ep.
\]

For any $A \in M_2(\C), $  write $A = \sum_i c_i A_i$. 
Using Equation (\ref{M_2norm}), we have \begin{align}\label{pauliNorm}\|A\| =\frac12\sqrt{|c_1|^2+|c_2|^2+|c_3|^2} + \frac{|c_4|}2 .\end{align}
So, by the Cauchy-Schwartz inequality,
\begin{align}
\nonumber
\sum_i|c_i| \leq \sqrt{3}\sqrt{|c_1|^2 + |c_2|^2 + |c_3|^2} + |c_4| \leq 2\sqrt{3}\|A\| .\end{align}
This inequality is sharp exactly when $|c_1| = |c_2| = |c_3|$ and $c_4 = 0$. This gives $C = 2\sqrt{3}$ in Equation (\ref{Yestimate}). In our proof of Theorem \ref{OgataTheorem}, we will have that $\|T_N(A_i) - Y_{i,N}\|$ for $i = 1$, $2$ is much larger than this expression for $i=3$, so we will obtain a value of $C$ close to $2\sqrt2$.

Because we chose $A_4 = \frac12I_2$,  we only need to construct nearby commuting matrices for $A_i$ being $\sigma_1, \sigma_2, \sigma_3$, as stated above. (We would not include $|c_4|$ in Equation (\ref{Yestimate}) in this case.)
So, we then focus on constructing nearby commuting matrices for $T_N$ applied to $\sigma_3=A_3$ and $\sigma_+ = A_1 + iA_2$.

The key perspective used to construct nearby commuting matrices for $T_N(\sigma_3)$ and $T_N(\sigma_+)$ is to use the representation theory of $su(2)$ discussed in Chapter \ref{4.RepTheory}. For any irreducible representation $S^\lam$ of $su(2)$, we have $S^\lam(\sigma_3)$ and $S^\lam(\sigma_+)$ given explicitly as a diagonal and a weighted shift matrix, up to a unitary change of basis. The distribution of the multiplicities of the irreducible subrepresentations $S^\lam$ in the tensor representation $(S^{1/2})^{\otimes N}$ is discussed in Lemma \ref{1/2mult}. Because $T_N = \frac{1}{N}(S^{1/2})^{\otimes N}$, this simultaneously gives $T_N(\sigma_3)$ as a direct sum of diagonal matrices and $T_N(\sigma_+)$ as a direct sum of weighted shift matrices, up to a unitary change of basis. 

Chapter \ref{5.GradualExchangeLemma} and Chapter \ref{6.BergThm} discuss the needed results for weighted shift matrices in preparation for the gradual exchange process, which is the purpose of Chapter \ref{7.GradualExchangeProcess}. This construction is more general than the context of the proof of Ogata's theorem. Suppose that $A_r \in M_{n_r}(\C)$ are diagonal and $S_r \in M_{n_r}(\C)$ are weighted shift matrices, where the eigenvalues of the diagonal matrices $A_r$ have a certain nested structure.
The gradual exchange process lemma (Lemma \ref{gep}) provides a construction of nearby commuting matrices for $A = \bigoplus_r A_r$ and  $S = \bigoplus_r S_r$.  The next two paragraphs go into some more detail about the results used in this lemma.

Lemma \ref{gep} is built up through Lemma \ref{proto-gep} and Lemma \ref{proto-gep2}, which construct almost invariant subspaces that are localized with respect to the spectrum of $A$ and are almost invariant under $S$ in a particular way. Lemma \ref{proto-gep} is proved by building a braided pattern of exchanges using Berg's gradual exchange lemma (Lemma \ref{GELws}) for the direct sum of two weighted shift matrices.  Lemma \ref{proto-gep2} generalizes Lemma \ref{proto-gep} by handling the case that not all the diagonal matrices $A_r$ have the same size.

The subspaces constructed in Lemma \ref{proto-gep2} are used in the proof of Lemma \ref{gep} to construct nearby commuting matrices $A'$ and $S'$.
Berg's construction of a nearby normal matrix for an almost normal weighted shift matrix (the focus of Chapter \ref{6.BergThm}) is used in this last step to construct $S''$ from $S'$. 
For this last step, it is used that the matrices $A'$ and $S'$ constructed are actually a direct sum of diagonal matrices and a direct sum of weighted shift matrices, though with a different basis than $A$ and $S$ are expressed as a direct sum and with a different block structure.

Chapter \ref{8.MainTheorem} is focused on completing the construction of nearby commuting matrices for $\frac{1}{N}S(\sigma_3)$ and  $\frac{1}{N}S(\sigma_+)$ for various reducible representations $S$ of $su(2)$.
Using various estimates for the entries of $S^\lam(\sigma_+)$ gotten in Lemma \ref{d-ineq}, Lemma \ref{gep} is directly applied  to obtain in Lemma \ref{Snearby}. 
Given certain estimates for the irreducible representations making up $S$, this lemma provides a construction of commuting matrices $A'$ self-adjoint and $S''$ normal nearby $\frac1NS(\sigma_3)$ and $\frac1NS(\sigma_+)$. This then provides commuting self-adjoint $\Re(S''), \Im(S''), A'$ nearby $\frac1NS(\sigma_1)$, $\frac1NS(\sigma_2)$, $\frac1NS(\sigma_3)$.

Work done in Example \ref{Ex1} is collected into Lemma \ref{Ex2Lemma} which is then  optimized and extended to cover trivial cases as Lemma \ref{bigLstepLemma}.
This lemma provides nearby commuting matrices when $S = S^{\lam_1}\oplus \cdots\oplus S^{\lam_m}$ has an optimized fixed spacing between the $\lam_i$. 
By breaking up more natural reducible representations  into direct sums of representations of this form, one obtains the main theorem (Theorem \ref{mainthm}). From that we obtain our extension of Ogata's Theorem for $d=2$ (Theorem \ref{OgataTheorem}).

\chapter{Mathematical Preliminaries}
\label{1.MathematicsPreliminaries}

Much of the preliminary material of Chapter \ref{1.MathematicsPreliminaries} is standard but presented for the person wanting to get a good impression about the subset of matrix and $C^\ast$-algebra functional analysis relevant to the almost/nearly commuting matrix problem and the contents of this thesis.

One can see the references
\cite{strung2021introduction,
carlen2016operator, 
carlen2018positive,
bratteli1979operator,
rudin1991functional,
conway2007course,
blackadar2006operator,
davidson1996c}
for Functional Analysis and
\cite{lancaster1985theory,
johnson1985matrix, 
hiai2014introduction}
for Matrix Theory.

\section{Matrix Analysis Review}
\label{a.MatrixAnalysis}
Let $M_{m,n}(\C)$ denote the $m \times n$ matrices with complex entries and $M_d(\C) = M_{d,d}(\C)$. For $z \in \C$, $\overline{z}$ will denote its complex conjugate.
If $A \in M_d(\C)$, let $A^T$ denote the transpose of $A$, $\overline{A}$ denote the complex conjugate of $A$ consisting of entries $\overline{A}_{ij}=\overline{A_{ij}}$, and $A^\ast$ denote the adjoint (alias conjugate transpose) of $A$.

We say that $A$ is self-adjoint (alias Hermitian) if $A^\ast = A$. If  $A$ self-adjoint has all non-negative eigenvalues we write $A \geq 0$. 
We say $U \in M_d(\C)$ is unitary if $U^\ast = U^{-1}$. 
Let $R(A)$ denote the range of $A$. Let $I=I_d$ denote the identity matrix in $M_d(\C)$. If $z \in \C$, then when we write $A - z$ we mean $A - zI$.

If $S, T \in M_d(\C)$, the commutator of $S$ and $T$ is $[S,T]= ST-TS$. The so-called self-commutator of $S$ refers to the commutator $[S^\ast ,S]$ of $S^\ast$ and $S$.
The spectrum $\sigma(S)$ of $S$ is the set of eigenvalues of $S$. Another way of expressing that $A \geq 0$ ($A$ is ``positive'') is saying that $A$ is self-adjoint with $\sigma(A) \subset [0, \infty)$.
The (unnormalized) trace of a matrix $A\in M_d(\C)$ is denoted by $\Tr[A]$. Because $\Tr[AB] = \Tr[BA]$ for $A, B \in M_d(\C)$, if $U \in M_d(\C)$ is unitary then $\Tr[U^\ast AU] = \Tr[A]$ and $\Tr[AB] = \Tr[(U^\ast A U)(U^\ast B U)]$.

\subsection{Projections}
We will say that $F \in M_d(\C)$ is a projection if it satisfies $F^\ast = F$ and $F^2 = F$. Often such a matrix is called an ``orthogonal projection'' but since all projections that we will be considering will be orthogonal we simply refer to $F$ as a projection as is common for self-adjoint idempotents in the context of $C^\ast$-algebras. Everything that we say in this subsection for projections in $M_d(\C)$ holds also for projections on an infinite dimensional Hilbert space.

If $F_1, \dots, F_k$ is a collection of projections we will refer to them as ``orthogonal projections'' if each $F_i$ is a projection and $F_iF_j = 0$ for $i \neq j$ which is equivalent to $R(F_i)$ and $R(F_j)$ being orthogonal subspaces in $\C^d$. This is sometimes referred to  as ``mutually orthogonality''. If the projections $F_1, \dots, F_k$ satisfy $F_1+\cdots + F_k = I$ then we say that the projections $F_1, \dots, F_k$ are complete. A collection of complete projections are orthogonal.

If $F$ is a projection then it orthogonally projects onto its range $R(F)$. Note that $F=0$ orthogonally projects onto $\{0\}$ and $I$ orthogonally projects onto $\C^d$. The matrix $1-F=I-F$ is the projection that projections onto the orthogonal complement of $R(F)$. If $T \in M_d(\C)$ then $R(F)$ is an invariant subspace of $T$ if and only if $(1-F)TF = 0$. $T$ commutes with $F$ if and only if both $R(F)$ and $R(1-F)$ are invariant subspaces of $T$. Note that because $F = F^\ast$, $[T,F]=0$ if and only if $[T^\ast,F] = -[T,F]^\ast = 0$ if and only if $R(F)$ and $R(1-F)$ are reducing subspaces for $T$ because they are invariant subspaces for both $T$ and $T^\ast$. In the infinite dimensional setting, these properties hold if $T$ is bounded.

If $F_1, F_2$ are projections then $F_1 \leq F_2$ means that the self-adjoint matrix $F_2-F_1$ is positive which holds if and only if $R(F_1) \subset R(F_2)$. Then $F_1 \leq F_2$ if and only if $F_2F_1 = F_1$. This is equivalent to $F_1F_2 = F_1$. For instance, if $F_1F_2 = F_1$ then $F_1F_2$ is self-adjoint so 
\[F_1F_2 =(F_1F_2)^\ast= F_2F_1.\]
Consequently, $F_1 \leq F_2$ imples that $F_1, F_2$ commute.

By $F_1 \perp F_2$ we mean that the ranges of $F_1$ and $F_2$ are orthogonal. Then $F_1$ and $F_2$ are orthogonal if and only if $F_1F_2 = 0$ if and only if $F_2F_1 = 0$ if and only if $F_1 \leq 1-F_2$ if and only if $F_2 \leq 1-F_1$. 

If $F_1$ and $F_2$ satisfy $F_1 \leq F_2$ or $F_1 \perp F_2$ then $F_1$ and $F_2$ commute. Generally, projections do not commute however if $F_1, F_2$ are commuting projections then $F_1F_2$ is a projection because $(F_1F_2)^\ast = F_2F_1 = F_1F_2$ and $(F_1F_2)^2 = F_1^2F_2^2= F_1F_2$.

Also, $F_1 \leq F_2$ if and only if $(1-F_2)F_1 = 0$ if and only if $F_1 \perp 1-F_2$ if and only if $R(F_1)$ is a subset of $R(F_2)$.
Note that if $F_1 \perp F_2$ then $F=F_1+F_2$ is a projection since the sum of self-adjoint matrices is self-adjoint and because
\[F^2=(F_1+F_2)^2 = F_1^2 +  F_1F_2 + F_2F_1 + F_2^2 = F.\] 
Note also that if $F_1 \leq F_2$ then $F=F_2 - F_1$ is a projection by a similar calculation and $R(F)$ is the subspace $R(F_2)\ominus R(F_1)$, the orthogonal complement of $R(F_1)$ viewed as a subspace of $R(F_2)$. In particular, if $F_1 \leq F_2$ then $F_1$ and $F=F_2 - F_1$ are orthogonal projections such that $F_1 + F = F_2$.

\subsection{Normal Matrices}

If $S \in M_d(\C)$, we define its real and imaginary parts as $\Re(S) = \frac12(S+S^\ast)$ and $\Im(S) = \frac{1}{2i}(S-S^\ast)$. So, $\Re(S), \Im(S)$ are self-adjoint and we have the decomposition $S = \Re(S) + i \Im(S)$. Note that $\Re(S), \Im(S)$ may each have real and/or imaginary entries.

We say that $N \in M_d(\C)$ is normal if $N^\ast N = N N^\ast$. $N$ is normal if and only if $[N^\ast, N] = 0$ if and only if $[\Re(N), \Im(N)] = 0$. So, the condition that $N$ is normal is equivalent to the condition that the two self-adjoint matrices $\Re(N)$, $\Im(N)$ commute.

If $N$ is normal, then for $\Omega \subset \C$, the spectral projection $E_{\Omega}(N)$ is the projection onto the span of the eigenvectors of $N$ with eigenvalues in $\Omega$. By definition, if $\Omega \cap \sigma(N) = \emptyset$ then $E_{\Omega}(N) = 0$.
The spectral theorem states that if $N$ is normal then \[N = \sum_{\lambda \in \sigma(N)} \lambda E_{\{\lambda\}}(N).\]
Note that also $E_{\C}(N) = I$.
Converse to the spectral theorem, if $F_1, \dots, F_k$ are orthogonal projections such that $F_1 + \cdots + F_k = I$ then for distinct $\lambda_j \in \C$, one has that $N = \sum_{j=1}^k \lambda_j F_j$ is normal with spectrum $\sigma(N) = \{\lambda_j\}_j$ and $E_{\{\lambda_j\}}(N) = F_j$.

All diagonal matrices are normal. If $D$ is diagonal and $U$ is unitary then $U^\ast DU$ is normal. The spectral theorem is often stated in the equivalent form that all normal matrices have a decomposition of this form, which is necessarily non-unique.

As special examples, a normal matrix $A$ is self-adjoint if and only if $\sigma(A) \subset \R$, $F$ normal is a projection if and only if $\sigma(F) \subset \{0, 1\}$, and $U$ normal is unitary if and only if $\sigma(U)$ is a subset of the unit circle.

The spectral theorem directly implies the familiar fact that the trace of a normal matrix is the sum of its eigenvalues, counted with multiplicity. In particular, if $F$ is a projection then $\Tr[F]$ is the rank of $F$ which is the dimension of $R(F)$. 
If $F$ is a projection then $\Tr[FAF] = \Tr[F^2 A] = \Tr[FA].$ If $\rho \geq 0$ satisfies $\Tr[\rho] = 1$ then the eigenvalues of $\rho$ belong to $[0, 1]$. If $\rho \geq 0$ satisfies $\Tr[\rho] = 0$ then $\sigma(\rho)=\{0\}$ so $\rho = 0$.

\subsection{Norms of Matrices}

\begin{defn}
A semi-inner product $\langle - , - \rangle$ on a $\C$-vector space $V$ is an $\R$-bilinear form that 
\begin{enumerate}[label=(\roman*)]
\item is conjugate-linear in the first argument: 
\[\langle c_1v_1+c_2v_2 , w \rangle=\overline{c_1}\langle v_1, w \rangle+\overline{c_2}\langle v_2 , w \rangle\] for $v_1, v_2, w \in V$ and $c_1,c_2\in \C$,
\item is $\C$-linear in the second argument:
\[\langle v, c_1w_1+c_2w_2 \rangle= c_1\langle v, w_1 \rangle+c_2\langle v , w_2 \rangle\] for $v, w_1, w_2 \in V$ and $c_1,c_2\in \C$,
\item and satisfies the non-negativity condition: 
\[\langle v , v \rangle \geq 0\] 
for all $v\in V$. 
\end{enumerate}
If $\langle v , v \rangle= 0$ only is true when $v = 0$, then $\langle - , - \rangle$ is an inner product.
\end{defn}
Note that we use the convention that $\langle - , - \rangle$ is conjugate-linear in the first argument and linear in the second argument. Note also that the conjugation property $\langle w , v \rangle = \overline{\langle v , w \rangle}$ can be seen by applying the stated properties to $\langle v +w, v+w\rangle, \langle v +iw, v+iw\rangle \in \R$.

The Cauchy-Schwartz inequality
\begin{equation}\label{Cauchy-Schwartz}
|\langle v , w \rangle| \leq \sqrt{\langle v , v \rangle\langle w , w \rangle}
\end{equation}
holds for any semi-inner product.

For the vectors $v=(v_1, \dots, v_d)^T, w=(w_1, \dots, w_d)^T$  in $\C^d$, their (standard, Euclidean) inner product is $\langle v,w \rangle = \sum_{j=1}^d \overline{v_j} w_j$. Whenever the notation $\langle - , - \rangle$ is used, it will either refer to this standard inner product on $\C^d$ or an understood inner product on an infinite dimensional Hilbert space. All other semi-inner products will be denoted using a subscript, such as the Hilbert-Schmidt inner product $\langle - , - \rangle_{HS}$ on the vector space $M_d(\C)$.

The adjoint of $A$ is the unique matrix that satisfies $\langle Av,w \rangle  = \langle v,A^\ast w \rangle$ (or equivalently $\langle w,Av \rangle  = \langle A^\ast w, v \rangle$) for all vectors $v, w \in \C^d$.
\begin{defn}
For $v \in \C^d$, its (Euclidean) norm is $\|v\| = \sqrt{|v_1|^2 + \cdots + |v_d|^2}$.
If $S \in M_d(\C)$, then its operator norm is defined as
\[\|S\| = \max_{v \in \C^d: \|v\| = 1} \|Sv\|= \max_{v \in \C^d: \|v\| \leq 1} \|Sv\|.\]
\end{defn}
This norm is defined so that the following norm inequality for $v, Sv \in \C^d$ holds:
\[\|Sv\| \leq \|S\|\|v\|.\]
The operator norm of $S$ also admits the inner product formulation:
\[\|S\|= \max_{v, w \in \C^d: \|v\| = \|w\| = 1} |\langle Sv, w \rangle|\]
from which it follows that $\|S^\ast\| = \|S\|$. 

Note that because the operator norm is the matrix norm that we will be using almost exclusively, any reference to ``the norm'' of a matrix or operator will be a reference to its operator norm.
Although we will not make much use of other matrix norms, one other important norm is the (unnormalized) Hilbert-Schmidt norm:
\[\|A\|_{HS}^2 = \sum_{i,j=1}^d|a_{i,j}|^2 = \Tr[A^\ast A].\] 
The Hilbert-Schmidt norm is induced by the inner product 
\[ \langle A, B\rangle_{HS} = \Tr[A^\ast B]\]
on $M_d(\C)$. The existence of such an inner product provides different approaches to solving matrix norm problems for the Hilbert-Schmidt norm than the operator norm. 

The Hilbert-Schmidt norm is sometimes also referred to as the Frobenius norm, however sometimes these names are used also for the normalized version of the norm: $\sqrt{\frac1{d}\sum_{i,j}|a_{i,j}|^2}$, normalized so that the identity has norm $1$.

Because all norms on the finite dimensional vector space $M_d(\C)$ are equivalent, any matrix norm inequality can be converted into an inequality in terms of another norm. The equivalence constants for the Hilbert-Schmidt and operator norms depend on the dimension $d$ as follows:
\begin{equation}\label{equiv} \|A\| \leq \|A\|_{HS} \leq \sqrt{d}\|A\|.
\end{equation}
In general for a non-diagonal matrix, there is not a simple closed-form expression for the operator norm in terms of the entries of the matrix, unlike the Hilbert-Schmidt norm. However, the operator norm satisfies many properties which help when calculating and estimating the operator norm. 

The operator norm is submultiplicative: 
$\|ST\| \leq \|S\|\|T\|$. 
By the triangle inequality and $\|S^\ast\| = \|S\|$, one sees that $\|\Re(S)\|, \|\Im(S)\| \leq \|S\|$ and $\|S\| \leq \|\Re(S)\| + \|\Im(S)\|$. 
The operator norm has the property that if $S$ is diagonal with diagonal entries $S_{jj}$ then $\|S\| = \max_j |S_{jj}|$. Consequently, $\|I\| = 1$.

The operator norm is also unitarily-invariant: $\|USV\|= \|S\|$ for any $S \in M_d(\C)$ and any unitaries $U, V \in M_d(\C)$. 
Consequently, permuting the rows or columns of a matrix does not change its operator norm. 

So, if each row and each column of a matrix $S$ has at most one non-zero entry then $\|S\|$ is the maximum of the absolute values of its entries. 
Written in terms of the entries, suppose that there is a permutation $\tau$ of the set $\{1, \dots, d\}$ and that $S$ is a matrix whose entries are identically zero except possibly $S_{j, \tau(j)}$ for $j = 1, \dots, d$. Then $\|S\| = \max_{j} |S_{j, \tau(j)}|$.

\begin{remark}
Using the fact that the operator norm and the Hilbert-Schmidt norm are unitarily invariant, one can show the inequalities in (\ref{equiv}). One should think of the inequalities in (\ref{equiv}) as essentially the same as the equivalence between the $\ell^\infty(\N)$ and  $\ell^2(\N)$ norms of a sequence supported on a set containing $d$ elements (i.e. supported on a set of measure $d$ in the counting measure). 

To see this, we use the singular value decomposition $U\Sigma V^\ast$ of $A$ where $U, V$ are unitary and $\Sigma$ is diagonal with non-negative entries $\sigma_j$, called the singular values of $A$. We then see that the inequalities between $\|\Sigma\| = \max_j\sigma_j$ and $\|\Sigma\|_{HS} = \sqrt{\sum_j \sigma_j^2}$ are exactly those of the corresponding $\ell^\infty$ and $\ell^2$ norms of the entries of $\Sigma=\diag(\sigma_1, \dots, \sigma_d)$:
\[\max_j \sigma_j \leq \sqrt{\sum_j \sigma_j^2} \leq \sqrt{d}\max_j \sigma_j.\]
\end{remark}

We now explore some properties of the operator norm which greatly justifies its alternative name: the spectral norm.
If $N$ is a normal matrix then 
\begin{equation}\label{spectral norm}
\|N\| = \max_{\lambda \in \sigma(N)} |\lambda|
\end{equation} 
because there is a unitary $U$ so that $U^{\ast}NU$ is diagonal with the entries being the elements of the spectrum of $N$ repeated with multiplicity.
In particular, for any non-zero projection $F$, $\|F\| = 1$ and for any unitary $U$, $\|U\| = 1$.

The operator norm satisfies the following $C^\ast$-identity: 
\begin{equation}\label{C ast identity}
\|S^\ast S\| = \|S\|^2.
\end{equation}
We can see this as follows. Obtaining the inequality 
\[\|S^\ast S\|\leq \|S^\ast\| \|S\| = \|S\|^2\]
is straightforward. The opposite direction follows from the inner product formulation,
\[\|S^\ast S\| \geq \max_{\|v\| = 1} |\langle S^\ast S v, v \rangle| = \max_{\|v\| = 1} \langle S v, Sv \rangle = \|S\|^2.\]
So, we obtain (\ref{C ast identity}).
Because $S^\ast S$ is self-adjoint with non-negative eigenvalues, $\|S\|$ is the square root of the largest eigenvalue of $S^\ast S$ by (\ref{spectral norm}).

With this in mind, we can obtain a representation of the operator norm for a few small-dimensional examples. For instance, consider $C\in M_2(\C)$ self-adjoint. Then its eigenvalues are some real numbers $\alpha < \beta$ and $\|C\| = \max(|\alpha|, |\beta|)$. If $C$ had trace zero then $\beta = -\alpha$ and hence $\det(C) = -\alpha^2$. Otherwise, $C-\Tr[C]/2$ has trace zero with eigenvalues $\pm\frac12(\alpha-\beta)$ so 
\[\|C\| = \frac12|\alpha-\beta|+\frac12|\alpha+\beta|=\sqrt{-\det(C-\Tr[C]/2)}+\frac12|\Tr[C]\,|.\]
So, if $C = \bp a & b \\ \overline{b} & c\ep$ for $a, c \in \R$ and $b \in \C$ then $C-\frac12\Tr[C] = \bp \frac12(a-c) & b \\ \overline{b} & \frac12(c-a)\ep$, hence
\begin{equation}\label{M_2norm}
\|C\| =  \frac12|a+c|+\sqrt{\frac14(a-c)^2  + |b|^2}.
\end{equation}

If $C = \bp 0 & a & b\\0&0&c\\0&0&0\ep$ is a strictly upper triangular matrix in $M_3(\C)$ then
\[C^\ast C = \bp 0&0&0\\0&|a|^2&\overline{a}b\\0&\overline{b}a&|b|^2+|c|^2\ep.\]
So, using the formula for the norm of a $2\times 2$ self-adjoint matrix, we obtain
\begin{align}\label{strictlyUpperNorm}
\|C\|=\sqrt{\frac12(|a|^2+|b|^2+|c|^2)+\sqrt{\frac14(|a|^2+|b|^2+|c|^2)^2-|ac|^2}}.
\end{align}
This is rather complicated but a simple inequality that one can obtain is:
\[\|C\| \leq \sqrt{|a|^2+|b|^2+|c|^2} \leq \sqrt{3}\max(|a|, |b|, |c|).\]
Note that the first inequality above is simply $\|C\| \leq \|C\|_{HS}$.

If we wanted a less precise bound in terms of the maximum of the entries of $C$, we could instead use the following decomposition of $C$ into the sum of matrices with only one non-zero entry in each column and row:
\[\|C\| = \left\|\bp 0 & a & 0\\0&0&c\\0&0&0\ep+\bp 0 & 0 & b\\0&0&0\\0&0&0\ep\right\| \leq \max(|a|, |c|)+|b| \leq 2\max(|a|, |b|, |c|).\]

The matrix norm $\max_{i,j}|C_{i,j}|$ is an easy-to-compute norm but is not unitarily invariant and does not satisfy some of the other nice algebraic properties that the operator norm satisfies.
In general, if $C \in M_d(\C)$ with $k$ non-zero diagonals then $\|C\| \leq k\max_{i,j}|C_{i,j}|$.  
Making use of the fact that $|C_{i,j}| = |\langle e_i,C e_j \rangle| \leq \|C\|$, we obtain
\[\max_{i,j}|C_{i,j}| \leq \|C\| \leq d\max_{i,j}|C_{i,j}|.\]

If $C$ is diagonal, then the lower inequality $\max_{i,j}|C_{i,j}| \leq \|C\|$ is sharp. 
If $C$ is the matrix whose entries are all $1$'s, then the upper inequality is sharp since $C$ is self-adjoint with $\sigma(C) = \{0, d\}$. 
If we choose a continuous path of unitaries $U(t)$ that goes from $U(0)=I$ to $U(1)$ which diagonalizes the matrix consisting of all $1$'s, we see that the operator norm of  $U(t)^\ast C U(t)$ is unchanged but the entries change in such a way that $\max_{i,j}|C_{i,j}|$ can be anything between $\|C\|$ and $\frac1d\|C\|$.

So, when $C$ is a generic non-sparse large matrix, we see that the maximum absolute value of its entries might not provide a useful estimate for the operator norm of $C$ due to the dependence on the dimension.

\subsection{Direct Sums}

\begin{defn}
The direct sum of $u \in \C^{d_1}$ and $v \in \C^{d_2}$ is the block column vector 
\[u \oplus v =\bp u \\ v \ep\in \C^{d_1+d_2} = \C^{d_1}\oplus\C^{d_2}.\]
\end{defn}
The direct sum of vectors satisfies various straightforward properties, including 
\[\langle u_1\oplus v_1, u_2\oplus v_2\rangle_{d_1+d_2}=\langle u_1, u_2\rangle_{d_1}+\langle v_1, v_2\rangle_{d_2}\]
and
\[u\oplus v = (u\oplus 0)+(0\oplus v)\]
is the unique decomposition into orthogonal vectors in $\C^{d_1}\oplus 0$, $0\oplus\C^{d_2}$.
If $u_1, \dots, u_{d_1} \in \C^{d_1}$ and $v_1, \dots, v_{d_2} \in \C^{d_2}$ are orthonormal bases then $u_i \oplus 0, 0\oplus v_j$ form an orthonormal basis of the direct sum $\C^{d_1}\oplus\C^{d_2}$. In particular,
\[e_1\oplus 0, \dots, e_{d_1}\oplus 0, 0\oplus e_1, \dots, 0\oplus e_{d_2}\]
are the standard basis vectors of $\C^{d_1+d_2}$ simply relabeled.

\begin{defn}
If $S \in M_{d_1}(\C)$ and $T \in M_{d_2}(\C)$ then their direct sum $S \oplus T \in M_{d_1+d_2}(\C)$ is defined to be the block matrix $\bp S & 0 \\ 0 & T\ep$ with diagonal submatrices $S$ and $T$. 
\end{defn}
This is a so-called exterior direct sum. The direct sum of square matrices is defined so that $(S \oplus T)(u \oplus v) = (Su) \oplus (Tv)$.
Notice that every vector in $\C^{d_1+d_2}$ can be expressed as a direct sum while a matrix in $M_{d_1+d_2}(\C) $ is not the direct sum of two matrices in $M_{d_1}(\C)$ and $M_{d_2}(\C)$ unless it has the required block structure.

If $S_1, S_2 \in M_{d_1}(\C)$ and $T_1, T_2 \in M_{d_2}(\C)$, the direct sum satisfies the properties 
\[c(S_1\oplus T_1)+(S_2\oplus T_2) = (cS_1+S_2)\oplus (cT_1+T_2),\] 
\[(S_1\oplus T_1)(S_2\oplus T_2) = (S_1S_2)\oplus (T_1T_2),\] and $(S\oplus T)^\ast = (S^\ast\oplus T^\ast)$. 
If $I_{d_1}$, $I_{d_2}$, $I_{d_1+d_2}$ are the respective identities in $M_{d_1}(\C)$, $M_{d_2}(\C)$, $M_{d_1+d_2}(\C)$ then \[I_{d_1}\oplus I_{d_2} = I_{d_1+d_2}.\] Because of these facts, the property of being normal, being self-adjoint, being a projection, being unitary, and being invertible are each preserved under direct sums. In particular,
\[\sigma(S\oplus T) = \sigma(S) \cup \sigma(T).\]

Let $F_1, \dots, F_k$ be orthogonal projections in $M_d(\C)$ such that $\sum_j F_j = I$. Then the orthogonal subspaces $R(F_j)$ span $\C^d$. If $v \in \C^d$ then $v = \sum_j v_j$ where $v_j = F_j v$ is a decomposition of $v$ into orthogonal vectors. In this way, $\C^d$ is the internal direct sum of the subspaces $R(F_j)$ and we can think of $v$ as being the direct sum $v_1\oplus\cdots\oplus v_k$. 

If $S$ is any matrix that commutes with each $F_j$ then each $R(F_j)$ is an invariant subspace of $S$ and of $S^\ast$. So, we can think of the restriction $S|_{R(F_j)}$ of $S$ to each subspace $R(F_j)$ on its own terms. If $S_j$ is this restriction then $S$ can be identified with the direct sum $\bigoplus_j S_j$. 
This is a so-called internal direct sum. The internal direct sum of vectors $v = \bigoplus_j F_j v$ and the internal direct sum of operators $S = \bigoplus_j S|_{R(F_j)}$ are compatible in the sense that 
\[Sv = \left(\bigoplus_j S|_{R(F_j)}\right)\left(\bigoplus_j F_j v\right)= \bigoplus_j SF_jv = \bigoplus_j F_jSv.\]
Note that if $S = \bigoplus_j S_j$  then $\|S\| = \max_j \|S_j\|$. 

It is also a fact that if there is a permutation $\tau$ of $\{1, \dots, k\}$ such that $S$ maps $R(F_j)$ into $R(F_{\tau(j)})$  for each $j$ then $S$ is given by a (generally non-diagonal) block matrix with submatrices $S_j$ induced by the restriction of $S$ from $R(F_j)$ into $R(F_{\tau(j)})$ and $\|S\| = \max \|S_j\|$. We will use this estimate to bound the norm of perturbations that have this type of block structure.

Consider the following explicit example for $k=3$. Let $F_1, F_2, F_3$ be a collection of complete projections in $M_d(\C)$ with $d_j = \dim R(F_j)$. If $S$ maps $R(F_1)$  into $R(F_2)$, $R(F_2)$ into $R(F_3)$, and $R(F_3)$ into $R(F_1)$ then up to a unitary change of basis, we can write $S$ as a block matrix with norm
\[\left\|\bp 0 & 0 & S_{13}\\ S_{21}&0&0\\ 0&S_{32}&0\ep\right\| = \max(\|S_{13}\|, \|S_{32}\|, \|S_{21}\|)\]
for $S_{21} \in M_{d_2,d_1}(\C)$, $S_{32} \in M_{d_3,d_2}(\C)$, $S_{13} \in M_{d_1,d_3}(\C)$.

\subsection{Tensor Products}

Tensor products of vector spaces (which induce tensor products of linear operators) can be defined and studied abstractly. However, we will choose a specific convention which is easy to state in terms of matrices. This 
concrete definition of the tensor product of matrices is sometimes referred to as the kronecker product. 
\begin{defn}
Let $S\in M_{m, n}(\C), T \in M_{r, s}(\C)$ with the $i,j$th entry of $S$ denoted by $S_{ij}$. We define the tensor product of $S$ and $T$ to be the block matrix
\[S \otimes T = \bp 
S_{11}T & \cdots & S_{1n}T \\ 
\vdots  & \ddots & \vdots \\
S_{m1}T & \cdots & S_{mn}T 
\ep \in M_{mr, ns}(\C).\]
Because a column vector in $\C^d$ is a $d \times 1$ matrix, we can use this same definition to define the tensor product of $u \in \C^{d_1}$ and $v \in \C^{d_2}$ as
\[u \otimes v = \bp 
u_1v \\ 
\vdots \\
u_{d_1}v 
\ep = (u_1v_1, \dots u_1v_{d_2}, u_2v_1, \dots, u_2v_{d_2}, \dots, u_{d_1}v_1, \dots, u_{d_1}v_{d_2})^T \in \C^{d_1d_2}.\]
\end{defn}
For example, if $A = \diag(a_1, \dots, a_{d})$ is diagonal then 
\[I_d \otimes A = \bp 
A &       &  \\ 
 & \ddots &  \\
 &        & A 
\ep 
=\diag(a_1, \dots, a_d, a_1, \dots, a_d, \dots, a_1, \dots, a_d) \in M_{d^2}(\C)\] 
and 
\[A \otimes I_d = \bp 
a_1I_d &       &  \\ 
 & \ddots &  \\
 &        & a_dI_d 
\ep 
= \diag(a_1, \dots, a_1, a_2, \dots, a_2, \dots, a_d, \dots, a_d)\in M_{d^2}(\C).\]
In particular, $I_{d_1}\otimes I_{d_2} = I_{d_1d_2}$.
The tensor product satisfies the property that for 
$A \in M_{\ell, m}(\C)$, $B \in M_{p, q}(\C)$, $C \in M_{m,n}(\C)$, $D \in M_{q, r}(\C)$, \[(A \otimes B)(C \otimes D)=(AC) \otimes (BD).\] This implies that if $S, T \in M_d(\C)$ and $u, v \in \C^d$ then \[(S\otimes T)(u \otimes v) = (Su)\otimes (Tv).\]
If $S \in M_{d_1}(\C), T\in M_{d_2}(\C)$ then $(S\otimes T)^\ast = S^\ast \otimes T^\ast$. So if $S, T$ are self-adjoint, unitary, or normal then so is $S\otimes T$. 
Consequently, if $S \in M_{d_1}(\C), T\in M_{d_2}(\C)$ are normal with diagonalizations $S = U^\ast D_S U$, $T = V^\ast D_T V$ then
\[S\otimes T = (U\otimes V)^\ast (D_S \otimes D_T) (U \otimes V)\]
is a diagonalization of $S\otimes T$ from which we see the property 
\[\sigma(S\otimes T) = \sigma(S)\cdot\sigma(T).\]

It is a fact that if $u, w\in \C^{d_1}$ and $v, x\in \C^{d_2}$ then
\begin{equation*}
\langle u\otimes v, w \otimes x \rangle_{d_1d_2}=\langle u, w  \rangle_{d_1}\langle v, x \rangle_{d_2}.
\end{equation*}
This implies that if $u, w$ are orthogonal or $v, x$ are orthogonal then $u \otimes v, w \otimes x$ are orthogonal. It also implies that $\|u \otimes v\|_{d_1d_2}=\|u\|_{d_1}\|v\|_{d_2}$ so if $u$ and $v$ are unit vectors then $u \otimes v$ is a unit vector.

An important property of the trace is that it respects direct sums and tensor products, which can directly be seen from the block matrix form:
\begin{prop}
Suppose that $S \in M_{d_1}(\C)$, $T \in M_{d_2}(\C)$ and $\Tr_d$ is the trace on $M_d(\C)$. Then
\[\Tr_{d_1+d_2}[S\oplus T] = \Tr_{d_1}[S]+\Tr_{d_2}[T],\]
\[\Tr_{d_1d_2}[S\otimes T] = \Tr_{d_1}[S]\Tr_{d_2}[T].\]
\end{prop}

A pure tensor is a vector in $\C^{d_1d_2}$ that can be expressed as the tensor product of two vectors in $\C^{d_1}$ and $\C^{d_2}$. For example 
\[e_1\otimes e_1 + e_2\otimes e_2 = (1, 0, 0, 1)^T \in \C^4\]
can be shown to not be a pure tensor.
Despite examples like this, the pure tensors span the entire space. In particular, if $u_1, \dots, u_{d_1} \in \C^{d_1}$ and $v_1, \dots, v_{d_2} \in \C^{d_2}$ are orthonormal bases then $u_i \otimes v_j$ form an orthonormal basis of $\C^{d_1d_2}$.

So, we define $\C^{d_1}\otimes \C^{d_2}$ to be $\C^{d_1d_2}$ with the understanding that $\C^{d_1}\otimes \C^{d_2}$ is the span of pure tensors $u \otimes v$, since the set of all pure tensors is not a vector space.
Also, unlike the case of the direct sum, the space $M_{m, n}(\C)\otimes M_{r, s}(\C)$ defined to be the span of tensor products of matrices in $M_{m, n}(\C)$ and $M_{r, s}(\C)$ equals the full space $M_{mr, ns}(\C)$. 

\begin{remark}
There is another convention for defining the tensor product of matrices which is done by simply interchanging the roles of $S$ and $T$ in the definition of $S\otimes T$ above. The issue of there being different equally valid conventions is important when discussing how to order the basis vectors of the tensor product.

There is not a natural way of ordering the pure tensor basis vectors $u_i \otimes v_j$ since this corresponds to ordering the points on the square lattice of points $(i,j)\in \{1, \dots, d_1\}\times \{1, \dots, d_2\} \subset \Z^2$.  
For instance, the pure tensors $e_i \otimes e_j$ are simply the standard basis vectors of $\C^{d_1d_2}$ but they are not just written differently than usual but they can also be potentially ordered differently if one is not careful.

Ordering the vectors as $u_i\otimes v_j$ by using a ``dictionary order'' is a convention that is equivalent to ordering the points $(i,j)$ by rows in the index lattice.
This is the convention that we used when defining the kronecker product above because
\begin{align*}
e_1\otimes e_1, \;\;\; e_1\otimes e_2,\;\;\; & \dots,\;\;\;  e_1\otimes e_{d_2},\\
e_2\otimes e_1,\;\;\;  e_2\otimes e_2,\;\;\; & \dots,\;\;\;  e_2\otimes e_{d_2},\\
   & \dots  \\
e_{d_1}\otimes e_1,\;\;  e_{d_1}\otimes e_2,\;\; & \dots,\;\;\;  e_{d_1}\otimes e_{d_2},
\end{align*}
is the same list of vectors (in the same order) as
\begin{align*}
e_1,\;\;\;\;\;\;\;\;\;\;\;\;\;\;\; &e_2,\;\;\;\;\;\;\;\;\;\;\;\;\;\;\;  \dots,\;\;\;  e_{d_2},\\
e_{d_2+1},\;\;\;\;\;\;\;\;\;\; &e_{d_2+2},\;\;\;\;\;\;\;\;\;\;  \dots,\;\;\; e_{2d_2},\\
   &\;\;\;\;\;\;\;\;\;\;\;\;\;\;\;\;\;\;\;\; \dots  \\
e_{(d_1-1)d_2+1},\;\; &e_{(d_1-1)d_2+2},\;\;  \dots,\;\; e_{d_1d_2}.
\end{align*}
However, ordering the standard basis vectors by columns is just as natural as the convention that we are using. There also are situations where ordering by diagonals is preferred.
\end{remark}

\begin{remark}
Despite these nice algebraic facts about tensor products of matrices, actually working with concrete examples can be difficult due to the way that a generic vector / matrix does not decompose into the sum of pure tensors in a visually obvious way as the direct sum does. Another complicating factor is that the size of the tensor product of matrices can be quite large. This can easily make by-hand calculations (and even some numerical calculations) time-consuming and sometimes less than enlightening.
\end{remark}

Before moving on, we state a definition which will play a direct role in later chapters of this thesis.
\begin{defn}
For $A \in M_{d_1}(\C)$ and $B \in M_{d_2}(\C)$, the kronecker sum of $A$ and $B$ is defined to be the matrix $A \otimes I_{d_2} + I_{d_1} \otimes B$.
\end{defn}
If $A, B$ are normal with orthonormal eigenbases $u_i, v_j$ and eigenvalues $\lambda_i, \mu_j$, respectively, then the $u_i\otimes v_j$ form an orthonormal eigenbasis of $\C^{d_1d_2}$ for the kronecker sum of $A$ and $B$ with eigenvalues $\lambda_i + \mu_j$. Note that grouping these eigenvectors by eigenvalue will generally produce a complicated ordering of the basis vectors.

\subsection{Commuting Matrices and Joint Spectrum}

We say that two matrices $A, B \in M_d(\C)$ commute when the commutator $[A,B]=AB-BA$ equals the zero matrix. If $A$ is a block diagonal matrix of the form $A = a_1I_{d_1} \oplus \cdots \oplus a_{k}I_{d_k}$ with $a_i$ distinct then $B$ commutes with $A$ if and only if $B$ has the block structure $B = B_1 \oplus \cdots \oplus B_k$ with $B_j \in M_{d_j}(\C)$. In this case, $B$ is normal if and only if each of the $B_j$ is normal. 

Consequently, the spectral theorem can be extended to the form that if $N_1, \dots, N_m$ are commuting normal matrices then the matrices $N_j$ can be simultaneously diagonalized with a unitary change of basis. Namely, if the normal $N_j$ commute then there exists a unitary $U$ such that all the $U^\ast N_j U$ are diagonal.

Another way to see this can be derived from a result that holds more generally but for matrices is rather easy to prove: 
\begin{thm}
For $N\in M_d(\C)$ normal, $S$ commutes with $N$ if and only if $S$ commutes with all spectral projections of $N$. 

Consequently, $S$ commutes with $N$ if and only if $S$ commutes with $N^\ast$, which is the Putnam-Fuglede theorem.
\end{thm}
\begin{proof}
The direction assuming that $S$ commutes with the spectral projections of $N$ is a direct consequence of the spectral theorem. 
The converse can be seen by noting that because $\sigma(N)$ is finite, for any $\Omega \subset \C$ there is a polynomial $p(z)$ such that $p(z) = 1$ for any $z \in \Omega \cap \sigma(N)$ and $p(z) = 0$ for any $z \in  \sigma(N) \setminus \Omega$. Then one can verify that $p(N) = E_{\Omega}(N)$. 
So, if $S$ commutes with $N$, it commutes with the spectral projections of $N$ also.
\end{proof}
This theorem provides a way to obtain block diagonal representations of normal matrices.
In particular, if $N$ and $S$ commute then
\[S = \left(\sum_{\lambda \in \sigma(N)}E_{\{\lambda\}}(N)^2\right)S=\sum_{\lambda \in \sigma(N)}E_{\{\lambda\}}(N)SE_{\{\lambda\}}(N)\]
is an internal direct sum decomposition of $S$ which can be given a block structure by a unitary change of basis. 

Another consequence of this result is that if $N_1, \dots, N_m$ are commuting normal matrices then the spectral projections $E_{\Omega_j}(N_k)$ commute. 
So, we obtain a refinement of the projections $E_{\lambda}(N_k)$ by defining $F_1, \dots, F_r$ to be the non-zero projections $\prod_{j=1}^mE_{\lambda_j}(N_j)$ for $(\lambda_1, \dots, \lambda_m) \in \sigma(N_1)\times \cdots\times \sigma(N_m)$. 
We then see that the $F_k$ are orthogonal projections that sum to $I$ and hence there exist $\lambda_k^j \in \C$ such that $N_j = \sum_{k=1}^r \lambda_k^j F_k$. 
So, choosing an orthonormal basis $\beta_k$ of $R(F_k)$, we have that $\beta = \beta_1 \cup \cdots \cup \beta_r$ is an orthonormal basis of eigenvectors for each of the $N_j$. 
Moreover, putting the vectors in $\beta$ into the columns of a matrix $U$, we see that the $U^\ast N_j U$ are diagonal. 
\begin{defn}
For each eigenvector of $\beta$, let $\lambda_1, \dots, \lambda_m$ be the eigenvalues of $\mathscr N = (N_1, \dots, N_m)$ corresponding to this eigenvector. The collection $\sigma(\mathscr N)$ of all points $\lambda = (\lambda_1, \dots, \lambda_m) \in \C^m$ is called joint spectrum of the commuting normal matrices $N_1, \dots, N_m$.
\end{defn}
If the $N_j$ are self-adjoint then the joint spectrum is a subset of $\R^m$. The joint spectrum of $\mathscr{N}$ belongs to the Cartesian product $\sigma(N_1) \times \cdots\times \sigma(N_m)$ but in general can be any subset of this set that satisfies the property that projecting the joint spectrum onto the $j$th axis of $\C^m$ gives the spectrum of $\sigma(N_j)$.

\begin{defn}
Let $\mathscr N = (N_1, \dots, N_m)$ be a collection of commuting normal matrices in $M_d(\C)$ and $\lambda=(\lambda_1, \dots, \lambda_m)\in\sigma(\mathscr N)$. We define the joint spectral projection
\[E_{\lambda}(\mathscr N)= E_{\lambda_1}(N_1)\cdots E_{\lambda_m}(N_m).\]
\end{defn}
The joint spectral projections are orthogonal for distinct elements of the joint spectrum and they satisfy:
\[\sum_{\lambda \in \sigma(\mathscr N)}E_\lambda(\mathscr N) = I.\]

\begin{example}
Note that it is not the case that the $N_j$ can always simultaneously satisfy a block diagonal structure where each block is a unique multiple of the identity. For instance, consider $N_1=\diag(1, 1, 0)$ and $N_2=\diag(1, 0, 0)$. There is not a way to group the eigenvalues of these matrices so that both matrices would both have the same block-diagonal structure with each block being a unique multiple of the identity.

The joint spectrum of $N_1$ and $N_2$ is $\{(1,1), (1,0), (0,0)\} \subset \{0,1\}\times \{0,1\}$. In this case the projection
$E_{\lambda_1}(N_1)E_{\lambda_2}(N_2)$ for $(\lambda_1, \lambda_2) \in \{0,1\}^2$ is zero when $(\lambda_1, \lambda_2)=(0,1)$ and $1$-dimensional otherwise.

If the matrices could be simultaneously written as block multiples of the identity where each of the multiples is unique, then the joint spectrum would contain at most one point on each horizontal and vertical line in $\C^2$.
\end{example}

Consider the case of the joint spectrum for two commuting self-adjoint matrices $A, B$. This is a subset of $\R^2$. The points of the joint spectrum belonging to the strip $\{z: a\leq\Re(z)\leq b\}$ correspond to the spectrum of $A$ in the the interval $[a,b]$ in the sense that 
\[E_{[a,b]}(A) = \sum_{\lambda \in \sigma(B)}E_{[a,b]}(A)E_{\{\lambda\}}(B),\] 
where the eigenvalues $\lambda$ of $B$ provide a decomposition of the subspace $R(E_{[a,b]}(A))$ into the orthogonal subspaces $R(E_{[a,b]}(A)E_{\{\lambda\}}(B))$ which are invariant under $B$.

For two commuting self-adjoint matrices $A$ and $B$, the joint spectrum of $A$ and $B$ as a subset of $\R^2$ is naturally identified as the spectrum of the normal matrix $N=A+iB$ as a subset of $\C$. 
So, the spectrum of $A$ is gotten by projecting the spectrum of $N$ onto the real axis and the spectrum of $B$ is gotten by projecting the spectrum of $N$ onto the imaginary axis.
A consequence of this is that if there are no elements of the spectrum of $N$ in the strip $\{z: a\leq\Re(a)\leq b\}$ then $A$ has no eigenvalues in $[a,b]$. The analogous property holds for $B$. 

If $N$ has spectrum that is approximately equal to the unit circle then $A$ and $B$ both have spectrum approximately equal to the interval $[-1,1]$. If $N$ has spectrum that is approximately equal to the boundary of the square $\{0, 1\}\times [0,1] \cup [0,1]\times\{0,1\}$ then both $\sigma(A)$ and $\sigma(B)$ are approximately equal to the interval $[0,1]$.

If $N$ has spectrum that is approximately equal to a square grid $\{-n, -n+1, \dots, n\}\times [-n,n] \cup [-n,n]\times\{-n, -n+1, \dots, n\}$ for some positive integer $n$ then $\sigma(A)$ and $\sigma(B)$ are approximately equal to $[-n,n]$. Note that in this latter example, if we consider the spectrum of $N$ restricted to the strip: $\{z: 0 < \Re(z) < 1 \}$ by considering the normal matrix $NE_{(0,1)}(A)$ then $NE_{(0,1)}(A)$ has spectrum approximately equal to $\left((0,1)\times\{-n, -n+1, \dots, n\}\right)\cup \{0\}$ and hence the spectrum of $BE_{(0,1)}(A)$ is approximately equal to $\{-n, -n+1, \dots, n\}$.

\newpage
\section{$C^\ast$-Algebra Functional Analysis Primer}
\label{b.CastAlgs}

The term ``$C^\ast$'' is pronounced the same as ``sea star''.
Although we will not make use of the following more general setting in many of the later chapters of this thesis, a brief discussion of the main facts about $C^\ast$-algebras will help the reader appreciate some of the results concerning almost/nearly commuting matrices which are presented in terms of $C^\ast$-algebras or whose proofs rely on the use of abstract $C^\ast$-algebras.

We will now define concrete and abstract $C^\ast$-algebras and discuss the relationship between them.
First, all the definitions and properties that we discussed for the operator norm on $M_d(\C)$ have appropriate extensions to bounded operators $B(\mathcal H)$ on an infinite dimensional Hilbert space $\mathcal H$, which are the linear transformations on $\mathcal H$ which have finite operator norm. For the purposes of this discussion, we will assume that the reader is familiar with the basic facts of Hilbert space.
If $\mathcal H$ is $d$-dimensional then $B(\mathcal H)\cong M_d(\C)$.
With the operator norm, $B(\mathcal H)$ is a complete normed $\C$-vector space (i.e. a Banach space).
\begin{defn}
A concrete $C^\ast$-algebras is a closed $\ast$-subalgebra of some $B(\mathcal H)$. That is, a concrete $C^\ast$-algebra $\mathcal A$ is a closed subspace of some $B(\mathcal H)$ that is closed under taking sums, products, and adjoints.
\end{defn}

We now define an abstract $C^\ast$-algebra.
\begin{defn}
A Banach algebra is a complete normed vector space (i.e. a Banach space) in which one can multiply elements of $\mathcal A$ and for which the norm is submultiplicative.
We will say that $\mathcal A$ is an abstract $C^\ast$-algebra if it is a Banach algebra with a conjugate-linear involution $\ast: \mathcal A \to \mathcal A$ that satisfies the $C^\ast$-identity of (\ref{C ast identity}). 
\end{defn}
Note that the $C^\ast$-identity and submultiplicity imply that $\|S^\ast\| = \|S\|$ and hence $\|S\|^2=\|SS^\ast\|=\|S^\ast\|^2=\|S^\ast S\|$, just as these properties implied these identities for the operator norm of matrices.

\subsection{Algebraic and Metric Structure}

We will always assume that $C^\ast$-algebras contain a multiplicative unit. So, if $\mathcal A$ is a (unital) concrete $C^\ast$-algebra then its multiplicative unit is a projection $P$. Moreover, if we restrict each element of $\mathcal A$ to the range $R(P)$ of $P$ then $\mathcal A$ can be viewed as a concrete $C^\ast$-algebra of $B(R(P))$ whose unit is the same as that of $\mathcal A$. 

A $C^\ast$-subalgebra $\mathcal B$ of $\mathcal A$ is a subspace that contains the multiplicative unit of $\mathcal A$ and that is a $C^\ast$-algebra with adjoint, norm, and algebraic operations inherited from $\mathcal A$. The requirement that a $C^\ast$-subalgebra contains the unit of the larger $C^\ast$-algebra assures that an element of $\mathcal B$ is invertible in $\mathcal B$ if and only if it is invertible in $\mathcal A$ with the same inverse. 

As a generalization of the set of eigenvalues of a matrix, the spectrum $\sigma(A)$ of an operator $A$ in $B(\mathcal H)$ consists of all $z \in \C$ such that $A - zI$ is not invertible. 
Likewise if $\mathcal A$ is a (unital) $C^\ast$-algebra with unit $1_{\mathcal A}$ and $A \in \mathcal A$ then its spectrum $\sigma(A)$ is the set of $z \in\C$ such that $A - z1_{\mathcal A}$ is not invertible. In any $C^\ast$-algebra the following continues to hold for any normal operator $N$:
\begin{equation}
\label{Cast spectral norm}
\|N\| = \max_{z \in \sigma(N)} |z|.
\end{equation}
This is the $C^\ast$-algebra version of (\ref{Cast spectral norm})  which states that the operator norm of $N$ equals its so-called spectral radius of a normal element.
As in the matrix case, this formula may or may not hold when $N$ is not normal. For instance, if $N$ is a non-zero nilpotent matrix then $\|N\| \neq 0$ but $\sigma(N) = \{0\}$.

One can use the formula $(1_{\mathcal A}-C)^{-1} = \sum_{k \geq 0} C^k$ which holds for $\|C\| < 1$ to show that the set of invertible elements is open and that the spectrum is contained in the closed disk in $\C$ of radius $\|A\|$ centered at $0$. It is a non-trivial fact that the spectrum of any element $A$ is non-empty. These facts imply that the spectrum is always a non-empty compact subset of $\C$. 
Although the spectrum $\sigma(A)$ of $A\in B(\mathcal H)$ is a generalization of the set of eigenvalues of $A$, $A$ might not have any eigenvalues (even if $A$ is self-adjoint) while  $\sigma(A)$ is always non-empty. 

Due to the statements made earlier, if $A \in \mathcal B$ and $\mathcal B$ is a $C^\ast$-subalgebra of $\mathcal A$ then the spectrum of $A$ as an element of $\mathcal B$ and the spectrum as an element of $\mathcal A$ are the same. 
As we describe later, we can view $\mathcal B$ as a $C^\ast$-subalgebra of some $B(\mathcal H)$ so that definition of spectrum for $C^\ast$-algebras coincides with that of an operator in $B(\mathcal H)$.

\begin{defn}
A $\ast$-homomorphism between two (unital) $C^\ast$-algebras $\mathcal A, \mathcal B$ is a linear map $\varphi: \mathcal A \to \mathcal B$ that satisfies $\varphi(AB) = \varphi(A)\varphi(B)$, $\varphi(1_{\mathcal A}) = 1_{\mathcal B}$, and $\varphi(A^\ast) = \varphi(A)^\ast$ for all $A, B \in \mathcal A$.
\end{defn}
A $\ast$-homomorphism is a linear map, is a homomorphism of unital (and typically non-commutative) rings, and ``commutes'' with the adjoint: $\varphi\circ \ast_A = \ast_B \circ \varphi$. 
These various algebraic properties ensure that $\ast$-homomorphisms respect the various algebraic structures of $C^\ast$-algebras and that one can study $C^\ast$-algebras using some of the familiar arguments available from abstract algebra. 
Many of these ``standard'' proofs carried over from abstract algebra and category theory are often casually referred to as ``abstract nonsense'' due to their ``symbol pushing'' nature. Regardless of  how it is conceptualized, this abstract scaffolding for  $C^\ast$-algebra constructions is well-studied and ready to be applied. 

Some simple consequences of the fact that $\ast$-homomorphisms respect some of the algebraic properties of their domain and codomain $C^\ast$-algebras include the fact that if an element $A$ of $\mathcal A$ is invertible, self-adjoint, unitary, normal, or has some of many other properties that can be expressed in algebraic terms or in terms of the adjoint, then its image $\varphi(A)$ will also retain these properties.

For instance, a consequence of the fact that a $\ast$-homomorphism $\varphi$ maps invertible elements to invertible elements is that the spectrum of the image $\varphi(A)$ of an element $A \in \mathcal A$ can only get smaller: $\sigma(\varphi(A))\subset \sigma(A)$. 

An important algebraic construction is that of the quotient $C^\ast$-algebra. 
\begin{defn}
We say that a subset $S$ of a $C^\ast$-algebra is self-adjoint if $S$ is closed under taking the adjoint of its elements. If $\mathcal I$ is a closed self-adjoint two-sided ideal of a $C^\ast$-algebra $\mathcal A$ then we can form the quotient $C^\ast$-algebra $\mathcal A / \mathcal I$ consisting of translates of $\mathcal I$ by elements of $\mathcal A$ which are often called cosets: 
\[[A] = A + \mathcal I = \{A + C: C \in \mathcal I\}\]
with the naturally defined operations:
\[(A + \mathcal I) + (B + \mathcal I) = (A+B) + \mathcal I, \;\; c(A + \mathcal I) = (cA)+ \mathcal I,\]
\[(A + \mathcal I)(B + \mathcal I) = (AB) + \mathcal I, \;\;   (A + \mathcal I)^\ast = A^\ast + \mathcal I.\]
The norm on the quotient space is the norm defined in the same way as the quotient of normed vector spaces:
\[ \|A + \mathcal I\| = \inf_{C \in \mathcal I} \|A-C\|.\]
We say that $\mathcal I$ is a proper ideal if $\mathcal I \neq 0, \mathcal A$.
\end{defn}

The assumptions we made on $\mathcal I$ are exactly those naturally imposed to ensure that the quotient $\mathcal A / \mathcal I$ is a $C^\ast$-algebra. Natural consequences of these definitions include that $1_{\mathcal A}+\mathcal I = 1_{\mathcal A / \mathcal I}$ and that the quotient map $\pi: \mathcal A \to \mathcal A / \mathcal I$ defined by $\pi(A) = A+\mathcal I$ is a surjective $\ast$-homomorphism with kernel $\mathcal I$.

Note that because we will make frequent use of the commutator notation: $[A,B] = AB-BA$, we will refrain from using the notation $[A]$ when referring to the elements of the quotient space but instead write $A + \mathcal I$ or $\pi(A)$.
\begin{defn}
For an infinite dimensional Hilbert space $\mathcal H$, an operator $A \in B(\mathcal H)$ is compact if it is the limit of finite rank operators. We will only discuss compact operators when $\mathcal H$ is infinite dimensional since otherwise ever operators on $\mathcal H$ has finite rank.
Let $K(\mathcal H)$ denote the subspace of compact operators in $B(\mathcal H)$. 
It can be shown that $K(\mathcal H)$ is a proper closed self-adjoint two-sided ideal of $B(\mathcal H)$.

An element $\lambda$ of the spectrum of $A$ belongs to the essential spectrum $\sigma_e(A)$ of $A$ when there does not exist a compact operator $K$ so that $\lambda \not\in\sigma(A+K)$.
In other words, $\sigma_e(A)$ is the subset of the spectrum $\sigma(A)$ that cannot be removed by adding a compact operator to $A$.
\end{defn}
\begin{example}
For instance, no compact operator is invertible so $\sigma_e(0)=\sigma(0)=\{0\}$. Thus, the identity operator is not compact and $\sigma_e(I)=\{1\}$.
A consequence of this is that $K(\mathcal H)$ is a proper subset of $B(\mathcal H)$.

Because one may be interested in studying properties of an operator modulo compact perturbations of any size norm, a natural formalism to enlist is to define the quotient $C^\ast$-algebra: $C(\mathcal H)=B(\mathcal H)/K(\mathcal H)$. Morally speaking, this allows us to use the $C^\ast$-algebraic properties from $B(\mathcal H)$ whilst treating compact operators as zero.

The $C^\ast$-algebra $C(\mathcal H)$ is referred to as the Calkin algebra. With this definition, one can systematically use the results and methods of quotient $C^\ast$-algebras to study properties of operators that are unchanged under a perturbation by compact operators.
The central example of this is the essential spectrum.

The essential spectrum $\sigma_e(A)$ of $A \in B(\mathcal H)$ is exactly the spectrum of $\pi(A) \in C(\mathcal H)$. So, since
we know that for any $A \in B(\mathcal H)$, $\sigma(\pi(A)) \subset \sigma(A)$ we automatically obtain the basic fact $\sigma_e(A) \subset \sigma(A)$. 
\end{example}
We now discuss why $\ast$-homomorphisms also respect the metric (and hence topological) structure of $C^\ast$-algebras. We first need to detail the connection between the algebraic structure and the norm.
By the $C^\ast$-identity, for any $A \in \mathcal A$:
\[\|A\| = \|A^\ast A\|^{1/2}.\]
So, the norm of elements of $\mathcal A$ are determined by the norm restricted to $\mathcal A_{s.a.}$, the $\R$-linear subspace of $\mathcal A$ consisting of self-adjoint elements.
By (\ref{Cast spectral norm}), 
\[\|A^\ast A\| = \max_{\lambda \in \sigma(A^\ast A)} \lambda.\]
We then see that the norm of $A$ is determined by the set $\sigma(A^\ast A)$ which has a characterization in terms of the algebraic operations of the $C^\ast$-algebra.
So, based on this algebraic characterization of the norm of $A$, we see that any $\ast$-homomorphism $\varphi$ is a contraction with operator norm at most $1$ since
\[\|\varphi(A)\|^2 = \max_{\lambda \in \sigma(\varphi(A)^\ast \varphi(A))} \lambda = \max_{\lambda \in \sigma(\varphi(A^\ast A))} \lambda \leq \max_{\lambda \in \sigma(A^\ast A)} \lambda =\|A\|^2.\]

Clearly, if the $\ast$-homomorphism $\varphi$ is not injective then it cannot be an isometry because it has a non-zero kernel. However, if $\varphi$ is in fact injective then one can show that the spectrum is unchanged after applying $\varphi$. This implies that any injective $\ast$-homomorphism (and hence any $\ast$-isomorphism) of $C^\ast$-algebras is automatically an isometry and a $\ast$-isomorphism onto its image. This means that in terms of the metric, $\ast$-homomorphisms are rather rigid.

\vspace{0.1in}

It is immediate that concrete $C^\ast$-algebras satisfy the axioms of an abstract $C^\ast$-algebra.
By the Gelfand-Naimark-Segal (GNS) construction, any such abstractly defined $C^\ast$-algebra can be $\ast$-isomorphically embedded in $B(\mathcal H)$ for some Hilbert space $\mathcal H$.
This means that when we identify concrete and abstract $C^\ast$-algebras using the GNS construction then various algebraic and metric properties will be preserved under this identification.

The GNS construction has the property that $\mathcal H$ can be chosen to be separable if $\mathcal A$ is separable and $\mathcal H$ finite dimensional if $\mathcal A$ is finite dimensional. 
A consequence of this is that any finite dimensional $C^\ast$-algebra $\mathcal A$ can be $\ast$-isomorphically identified with the $C^\ast$-algebra of block matrices $M_{d_1}(\C) \oplus \cdots \oplus M_{d_k}(\C)$ for some $d_1 \leq \dots \leq d_k$.

If there is a $\ast$-homomorphism $\varphi$ between $C^\ast$-algebras then immediately it is the case that the image of $\varphi$ is a $C^\ast$-subalgebra of the codomain $C^\ast$-algebra and the kernel of $\varphi$ is a closed self-adjoint two-sided ideal of the domain of $\varphi$.

We already saw examples of concrete $C^\ast$-algebras as matrix algebras, direct sums of matrix algebras, the space of bounded operators $B(\mathcal H)$, and direct sums of spaces $B(\mathcal H)$. 

\begin{remark}
It is common for elements of a(n abstract) $C^\ast$-algebra to be denoted by lower case  variables (e.g. $a, p, n, x, y, \dots$) and elements of $B(\mathcal H)$ to be denoted by capital letter variables (e.g. $A, P, N, T, U, V, \dots$) and lower case variables used for elements of $\mathcal H$ (e.g. $u, v, w, \dots$). Because elements of any  $C^\ast$-algebra $\mathcal A$ can be viewed as operators in $B(\mathcal H)$, there is not necessarily a conflict between the convention of simply using capital letter variables for elements of any $C^\ast$-algebra. 

In this chapter we will use capital letter variables to denote elements of $C^\ast$-algebras. 
This will become particularly relevant in later subsections when we discuss how embedding $\mathcal A$ into some $B(\mathcal H)$ allows us to discuss some important features of certain types of $C^\ast$-algebras. 

With the same mindset, in \cite{kachkovskiy2016distance}'s work on Lin's theorem this identification of elements of a $C^\ast$-algebra $\mathcal A$ as an element of some $B(\mathcal H)$ is commented upon and used in constructions with the foreknowledge that the end-results will lie in $\mathcal A$ even if some of the components may not belong to $\mathcal A$.
\end{remark}

\subsection{Continuous Functional Calculus}

One of the most fundamental examples of an abstract $C^\ast$-algebra is the space of complex-valued continuous functions $C(X)$ on a compact Hausdorff topological space $X$ with the norm 
\[\|f\| = \sup_{x\in X}|f(x)|\]
and the adjoint of $f$ being $\overline{f}(x)=\overline{f(x)}$.

Note that $z_0 \in \sigma(f)$ if and only if $f(x)-z_0$ is invertible as a continuous function if and only if $f(x)-z_0$ is bounded away from zero.
Since $X$ is compact, $f(X) \subset \C$ is compact.
So, the spectrum of $f$ is its range.

One of the reasons that this is such a fundamental example is the classification of abelian $C^\ast$-algebras:
\begin{thm}\label{abCast}
Every (unital) abelian $C^\ast$-algebra $\mathcal A$ is $\ast$-isomorphic to some $C(X)$, where $X$ is a compact Hausdorff topological space that is unique up to homeomorphism.
\end{thm}
We will sketch the existence of the isomorphism of $\mathcal A$ to some $C(X)$ but before that we give the motivating example.
\begin{example}
Suppose that $\mathcal A = C(X)$. It is well known that the dual space of all continuous linear functionals of $C(X)$ is given by (complex-valued regular Borel, \cite[Theorem C.18]{conway2007course}) measures on $X$.
The measures that single out particular points of $X$ are the dirac masses: $\delta_x$. So, if $x \in X$, the evaluation map $\phi_x(f) = f(x)$ is a continuous linear functional of $C(X)$ that corresponds to a point of $X$.

However, to distinguish it from other elements of the dual space and to make use of the $C^\ast$-algebraic structure of $C(X)$, we note that the evaluation map is also multiplicative: 
\[\phi_x(fg) = \phi_x(f)\phi_x(g).\] 
This property is simply not held by other continuous linear functionals on $C(X)$ because if a non-zero measure $\mu$ is not supported at a single point then it is possible to find two non-zero functions $f, g \in C(X)$ supported on disjoint compact neighborhoods so that
\[\left(\int f\, d\mu\right)\left(\int g\, d\mu\right)\neq 0 = \int fg\, d\mu.\]
Also, the multiplicativity implies that the measure must be non-negative. This helps us identify the only multiplicative elements of $C(X)^\ast$ as the evaluation maps.
The evaluation maps also already respect the adjoint. So, it turns out that the maps $\phi_x$ are the only $\ast$-homomorphisms of $C(X)$ into $\C$. 

Note that $\{\phi_x(f): x \in X\}$ is the range of $f$. This shows that every element of the spectrum of $f\in C(X)$ can be gotten as the image of some $\phi_x$ which acts on $f$. Likewise, if it were true that $\mathcal A$ were $\ast$-isomorphic to some $C(X)$ then we can identify these $\ast$-homomorphisms $\phi_x$ of $\mathcal A$ into $\C$ with the element $x$ of $X$.  
\end{example}
We now proceed with the sketch of the proof.
\begin{proof}
A character $\phi$ of an abelian $C^\ast$-algebras $\mathcal A$ is a $\ast$-homomorphism from $\mathcal A$ into the one-dimensional abelian $C^\ast$-algebra $\C$. Since $\phi(1_{\mathcal A}) = 1$, for any $A \in \mathcal A$, $\phi(A) \in \C$ is invertible if $A$ is invertible. We then see that $\phi(A) \in \sigma(A)$.

It can be shown that for any non-zero element $A \in \mathcal A$ and $\lambda \in \sigma(A)$, there is a character $\phi$ which witnesses that $A-\lambda1$ is not invertible by satisfying $\phi(A-\lambda1) = 0$. This is done by using the standard ring-theoretic argumentation of finding a  maximal ideal $\mathcal I$ of $\mathcal A$  containing $A-\lambda1$ and showing that the projection $\phi: \mathcal A \to \mathcal A / \mathcal I$ is a character with the desired property.
This shows that every element of the spectrum of $A$ can be gotten as the image of some character of $\mathcal A$.

Let $X$ be the set of characters on $\mathcal A$. This is a subset of the dual $\mathcal A'$ of $\mathcal A$ of all continuous linear functionals on $\mathcal A$. So, we can endow $X$ with the weak-$\ast$ topology, meaning that we define the open sets of $X$ to be generated by the topological subbasis of preimages of open subsets of $\C$ by the evaluation maps $X \ni \phi \mapsto \phi(A) \in \C$.

With this topology, $X$ is Hausdorff and compact. Moreover, the evaluation functions above are all continuous maps into $\C$ which identifies $\mathcal A$ as a subset of $C(X)$. One then shows that this identification of $\mathcal A$ as the set of continuous functions on its characters $X$ is a $\ast$-isomorphism of $C^\ast$-algebras.

For uniqueness, see Corollary 1 in Section 18.2.1 of \cite{kadets2018course}.
\end{proof}

\begin{defn}
A very useful example of a $C^\ast$-sublgebra is the (unital) $C^\ast$-algebra generated by a subset $S$ of some $C^\ast$-algebra $\mathcal A$. We shall denote this by $\mathcal A_S$. 
\end{defn}
As is common with the ``[insert mathematical term] generated by'' definition, there are two equivalent ways to define the $C^\ast$-algebra generated by $S$. 
One definition is to characterize it as the closure of the set of all possible elements of $\mathcal A$ that can be gotten by taking sums, products, and adjoints of elements of $\operatorname{span}(S) \cup \{1_{\mathcal A}\}$. An equivalent definition is to define it as the smallest $C^\ast$-algebra containing $S$. 

If $N$ is a normal element of a $C^\ast$-algebra $\mathcal A$ then the $C^\ast$-algebra generated by $N$ is abelian. The spectral mapping theorem for polynomials states that if $p(z, \overline{z})$ is a polynomial in the two variables $z$ and $\overline{z}$ then
\begin{equation}\label{polyNorm}
\sigma(p(N, N^\ast)) = \{p(z, \overline{z}) : z \in \sigma(N)\}.
\end{equation}
Note that for the constant polynomial $p = 1$, $p(N, N^\ast) = 1_{\mathcal A}$.
 This then provides a norm on the subset $\mathcal P(N)$ of $\mathcal A$ formed by polynomials of $N$ and $N^\ast$ by:
\[\|p(N, N^\ast)\| = \max_{z \in \sigma(N)}|p(z, \overline{z})|.\]
Let $P(\sigma(N))$ denote the subset of $C(\sigma(N))$ of consisting of polynomials in $z$ and $\overline{z}$.

So, if we define $\Phi: P(\sigma(N)) \to \mathcal P(N)$ by
$\Phi: p \mapsto p(N,N^\ast)$ then we obtain a map that is an isometry that respects the algebraic structure of $\mathcal P(N)$. 
The $C^\ast$-algebra generated by $N$ is the closure of $\mathcal P(N)$. Since $\Phi$ defined above is an isometry (and hence uniformly continuous), it extends uniquely as a map from $\overline{P(\sigma(N))} = C(\sigma(N))$ to $\overline{P(N)} = \mathcal A_N$.

The main consequence of this is that for any function $f$ that is defined and continuous on $\sigma(N)$, we can define $f(N)$ as the limit of a sequence of polynomials $p_k(N,N^\ast)$ of $N$ and $N^\ast$, where $p_k(z,\overline{z})$ converges to $f$ uniformly on $\sigma(N)$. This is referred to as the (continuous) functional calculus for normal elements of a $C^\ast$-algebra. With this definition the spectral mapping theorem holds:
\begin{equation}\label{contCalcNorm}
\sigma(f(N)) = f(\sigma(N)), \;\; \|f(N)\| = \max_{z \in \sigma(N)}|f(z)|.
\end{equation}

The continuous functional calculus respects the standard algebra of continuous functions. For instance, 
\[(f+g)(N) = f(N)+g(N), \;\; (cf)(N) = c(f(N)),\]
\[(fg)(N) = f(N)g(N), \;\; (f\circ g)(N)=f(g(N))).\]
It respects the adjoint operation: if $f$ is a continuous function on $\sigma(N)$ and $\overline{f}(z) = \overline{f(z)}$ then $f(N)^\ast = \overline{f}(N)$. 
It respects uniform convergence: if $f_k$ is a sequence of continuous functions that converge uniformly on $\sigma(N)$ to some continuous function $f$ then $f_k(N) \to f(N)$.
It also is respected by $\ast$-homomorphisms:
\[\varphi(f(N)) = f(\varphi(N))\]
by taking the limit of the identity
$\varphi(p_k(N, N^\ast)) = p_k(\varphi(N), \varphi(N)^\ast)$ if $p_k$ converges to $f$ uniformly on $\sigma(N)$.

\begin{example}
The simplest example of this is a normal matrix $N\in M_d(\C)$ that has $d$ eigenvalues $\lambda_1, \dots, \lambda_d$ (possibly with repetitions). If $N$ is diagonalized as $N=U^\ast D U$ with $D=\diag(\lambda_j)$ and $U$ unitary then one can simply define a continuous function $f$ of $N$ by 
\[f(N) = U^\ast \bp f(\lambda_1) &&\\ & \ddots&\\&&f(\lambda_d)\ep U.\]

However, one can easily see that a polynomial $p(D)$ computed by matrix algebra is equal to $\diag(p(\lambda_j))$ which is simply the diagonal matrix gotten by applying $p$ to the diagonal entries of $D$. By continuity, this provides an equality between $f(D)$ defined using the continuous functional calculus and entry-wise.
Consequently, the continuous functional calculus definition of $f(N)$ agrees with diagonalization definition above.

This means that as far as functions of the matrix goes, the normal matrix $N$ with spectrum $\{\lambda_1, \dots, \lambda_d\}$ can be identified with the identity function $g(z)=z$ in $C(\{\lambda_1, \dots, \lambda_d\})$ and $f(N)$ can be identified with the function $f(g(z))=f(z)$ on $\{\lambda_1, \dots, \lambda_d\}$. 
\end{example}

The map $\Phi$ that we defined above is an isometric $\ast$-homomorphism of $C(\sigma(N))$ onto $\mathcal A_{N} \subset \mathcal A$. This allows us to define continuous functions of normal operators.
The definition of functions of operators is a fundamental and an incredibly useful tool for dealing with elements of $C^\ast$-algebras. One of the features of this is that we can estimate the norm of the operator $f(N)$ by using only information about what the spectrum of $N$ is and estimates for how large the function $f$ is. 

For instance, the continuous functional calculus can be used to show that if $A \geq 0$ is a non-invertible element of the $C^\ast$-algebra $\mathcal A$ and $\varphi: \mathcal A \to \mathcal B$ is an injective $\ast$-homomorphism of $C^\ast$-algebras, then $\varphi(A)$ is not invertible.
\begin{proof}
For $\varepsilon > 0$, let $f_\varepsilon(x)$ be the continuous piecewise-linear real-valued function on $\R$ defined by
\[f_\varepsilon(x) = \left\{\begin{array}{cl}
1, & x \leq 0\\
\frac1\varepsilon(\varepsilon-x), & 0<x<\varepsilon\\
0, & x \geq \varepsilon
\end{array}\right..\]
Because $A$ is not invertible, $0$ belongs to the spectrum of $A$. So, $f_\varepsilon(A) \in \mathcal A$ satisfies
\[\|f_\varepsilon(A)\| = \max_{z \in \sigma(A)}|f_\varepsilon(z)| = f_\varepsilon(0) = 1.\]
In particular, $f_\varepsilon(A) \neq 0$. Therefore, $\varphi(f_\varepsilon(A)) = f_\varepsilon(\varphi(A))$ is not zero because $\varphi$ is injective. So,
\[\max_{z \in \sigma(\varphi(A))}|f_\varepsilon(z)|=\|f_\varepsilon(\varphi(A))\| > 0.\]
We then see that the spectrum of $\varphi(A) \geq 0$ has non-empty intersection with the support $(-\infty,\varepsilon]$ of $f_\varepsilon$. Because $\varepsilon > 0$ was arbitrary and the spectrum $\sigma(A)\subset [0, \infty)$ is closed, we deduce that $0$ belongs to the spectrum of $\varphi(A)$. So, $\varphi(A)$ is not invertible. 
\end{proof}
This result provides the missing piece of the proof for the result that we stated previously: any injective $\ast$-homomorphism preserves the spectrum of self-adjoint elements and hence is an isometry.

\begin{remark}
Note that a simpler construction of the continuous functional calculus can be gotten by using Theorem \ref{abCast}. This is done by noting that if 
$N \in \mathcal A$ is normal then the $C^\ast$-algebra $\mathcal A_{N}$ generated by $N$ is abelian. Then $N$ is identified by some $\ast$-homomorphism $\Phi: \mathcal A_N \to C(X)$ with a continuous function on $g \in C(X)$ whose range is $\sigma(N)\subset \C$. For $f$ continuous on $\sigma(N)$, we can define $f(N) = \Phi^{-1}(f\circ g) \in \mathcal A_N.$

One immediately sees that $\sigma(f(N)) = f(\sigma(N))$. Consequently, (\ref{contCalcNorm}) holds. One quickly can confirm that this definition of 
$f(N)$ agrees with the polynomial functional calculus and by continuity the continuous functional calculus as defined previously.
\end{remark}

Although this equivalent definition of the continuous functional calculus is ``cleaner'', it is more abstract and does not directly make use of passage from one type of functions (polynomials) to a wider class of functions (continuous functions) which will motivate our discussion later about a non-continuous functional calculus.
However, these reasons are why this approach to proving the spectral theorem provides ease in defining and proving other results such as the following multivariable spectral theorem.

\begin{defn}
Let $A_1, \dots, A_m \in \mathcal A$ be commuting normal elements of a $C^\ast$-algebra $\mathcal A$. The $C^\ast$-algebra $\mathcal A_S$ generated by $S=\{A_1, \dots, A_m\}$ 
is abelian so there is a compact Hausdorff topological space $X$ and a $\ast$-isomorphism $\Phi: \mathcal A_S \to C(X)$ such that $g_j = \Phi(A_j) \in C(X)$ and $\sigma(A_j) = \sigma(g_j)$ is the range of $g_j$.

We say that the joint spectrum $\sigma(A_1, \dots, A_m)$ of $A_1, \dots, A_m$ is the subset of $\C^m$ consisting of the range of $\vec g = (g_1, \dots, g_m): X \to \C^m$.
If $f$ is a continuous function on the joint spectrum of $A_1, \dots, A_m$ then we define $f(A_1, \dots, A_m) = \Phi^{-1}(f\circ \vec g)$. 
\end{defn}
\begin{remark}
It follows immediately from the definition that
\begin{equation}\label{jointSpBound}
\sigma(A_1, \dots, A_m) \subset \sigma(A_1) \times \cdots\times \sigma(A_m)
\end{equation}
and
\begin{equation}\label{jointNorm}
\|f(A_1, \dots, A_m)\| = \max_{(\lambda_1, \dots, \lambda_m)\in\sigma(A_1, \dots, A_m)} |f(\lambda_1, \dots, \lambda_m)|.
\end{equation}
Note that this provides an immediate justification for (\ref{polyNorm}) since if $N$ is normal then the joint spectrum of $N, N^\ast$ is $\{(z,\overline{z}): z \in \sigma(N)\}$.
\end{remark}

\begin{remark}
Because $\Phi$ is multiplicative, it is easy to see that if a function decomposes as a product of functions of its arguments: $f(z_1, \dots, z_m) = f_1(z_1)\cdots f_m(z_m)$ then $f(A_1, \dots, A_m)$ as defined above equals $f_1(A_1)\cdots f_m(A_m).$ Then because such functions are dense in $C(\sigma(A_1)\times\cdots\times\sigma(A_m))$ by the Stone-Weierstrass theorem, we see that this definition using $\Phi$ agrees with the natural definition based on building the multivariable function from functions of a single variable.
\end{remark}

\subsection{Strong and Weak Operator Topologies}

For a normal operator $N$ in $B(\mathcal H)$ and a rather general type of set $\Omega \subset \C$, one can construct the spectral projection $E_\Omega(N)$ of $N$ corresponding to $\Omega$. See a standard treatment of the spectral theorem such as \cite{rudin1991functional} for more about the construction. This is an example of a (discontinuous) characteristic function $\chi_\Omega$ of $N$: $E_N(\Omega) = \chi_{\Omega}(N)$ is an operator in $B(\mathcal H)$.

If $\mathcal A$ is a concrete $C^\ast$-subalgebra of $B(\mathcal H)$, we can construct the continuous functions of any normal element $N$ of $\mathcal A$ using the functional calculus for $N$ as an element of $\mathcal A$ or as an element of $B(\mathcal H)$. These will both produce the same operator. However, in general the spectral projections and other discontinuous functions of $N$ that can be defined as elements of $B(\mathcal H)$ using more general functional calculus methods will not belong to $\mathcal A$.

To capture these more general functions of $N$, we will use a more general sense of convergence of operators. We illustrate this need with the following example. 
\begin{example}\label{pwConvEx}
Suppose that we consider the compact self-adjoint operator $A$ defined by $Ae_0=0$ and $Ae_k = \frac1ke_k$ for $k \geq 1$ on the infinite dimensional Hilbert space $\mathcal H$ with orthonormal basis $e_0, e_1, \dots$. The spectrum of $A$ is $\{\lambda_k\}_{k\geq0}=\{0\} \cup \{1/k: k \geq 1\}$. 

Let $(a,b)$ be an interval with $a > 0$. Then the spectral projection $E_{(a,b)}(A)$ of $A$ onto $(a,b)$ belongs to the $C^\ast$-algebra generated by $A$ because $\chi_{(a,b)}$ is a continuous function on the spectrum of $A$. For example, there is a continuous function on $\R$ that interpolates $\chi_{(a,b)}$ on $\sigma(A)$.

Likewise, for any eigenvalue of $A$ of the form $1/k$ the spectral projection onto the eigenspace $E_{\{1/k\}}(A)$ also equals a continuous function of $A$. 
However, the spectral projection of $A$ onto $0$ does not belong to the $C^\ast$-algebra $B(\mathcal H)_A$ generated by $A$. The reason is that if $f(x)$ is a continuous function on $\sigma(A)$ then $f(A)e_k = f(\lambda_k)e_k$. This characterizes all continuous functions on $\sigma(A)$ and hence all elements of $B(\mathcal H)_A$. Since $\chi_{\{0\}}(A)e_k = \delta_{0,k}e_k$ is not included in this characterization, we see that $\chi_{\{0\}}(A)\not\in B(\mathcal H)_A$.

Let $f_k$ be zero on $[1/k, \infty)$, $f_k(0)=1$, and linear on $[0, 1/k]$. Then $f_k$ is continuous on $\sigma(A)$ and converging pointwise on $\sigma(A)$ to  $\chi_{\{0\}}$ as $k \to \infty$. As already said, $f_k(A)$ does not converge to $\chi_{\{0\}}(A)$ because for each $k$, $f_k(A)$ belongs to the closed subspace $B(\mathcal H)_A$ which excludes $\chi_{\{0\}}(A)$. 

If we more carefully inspect the convergence, we calculate 
\[\left(f_k(A)-\chi_{\{0\}}(A)\right)e_j = f_k(1/j)e_j.\] 
For any given $k$, taking $j$ large shows that
\[ \|f_k(A)-\chi_{\{0\}}(A)\| \geq 1.\]
So, for every $k$, $f_k(A)$ is at least a distance of $1$ from the spectral projection $E_{\{0\}}(A)$.

Notice that this argument is the same argument verbatim for proving that $f_k$ does not converge to $\chi_{\{0\}}$ uniformly on $\sigma(A)$. This is not too surprising because the operator norm, by definition, is a supremum norm on the unit ball in $\mathcal H$. 
\end{example}

To obtain discontinuous operators of $A$ which are the pointwise limits of continuous functions, we will define a different notion of convergence of operators. In order to stress that one is speaking about the standard notion of convergence of operators and not the different notions that we will discuss next, one may say that $A_k \to A$ (i.e. $\|A_k-A\| \to 0$) is the convergence of $A_k$ to $A$ ``in norm''.
\begin{defn}
Suppose that $A_k$ is a sequence of operators in $B(\mathcal H)$. We say that $A_k$ converges strongly to $A$ if for every $v \in \mathcal H$, $A_k v \to Av$. That is, the operators $A_k$ converge pointwise to $A$ as functions from $\mathcal H$ to $\mathcal H$. We also can write 
$\slim_k A_k = A$ or $A_k \xrightarrow{s} A$. 
\end{defn}
It is often useful to speak of this type of convergence in topological terms, so we introduce seminorms $\phi_{v}(C) = \|Cv\|$ so that $A_k$ converges to $A$ strongly if and only if $\phi_v(A_k - A) \to 0$ for every $v\in \mathcal H$. These seminorms induce a topology onto $B(\mathcal H)$ that is called the strong operator topology. This makes $B(\mathcal H)$ a locally convex topological vector space (with which comes many benefits such as the applicability of the Han-Banach theorem). 
A subbasis of the topology at the origin is given by the sets $\{C: \phi_v(C) <\varepsilon\}$. So, if $G$ is a set containing some operator $A_0 \in B(\mathcal H)$ and $G$ is strongly open then there are $v_1, \dots, v_n \in \mathcal H$ and some $\varepsilon > 0$ so that
\[\{C: \phi_{v_1}(A-C), \dots, \phi_{v_n}(A-C) < \varepsilon\} \subset G.\]

When $\mathcal H$ is infinite dimensional, the strong operator topology for $B(\mathcal H)$ does not have a countable topological basis at each point so one cannot use convergence of sequences to characterize strongly closed sets. To properly make arguments concerning the strong operator topology, one should use the open sets induced by the seminorms or make use of some generalization of sequences that is appropriate for such general topological vector spaces, such as net or ultrafilter convergence. That said, there is nothing incorrect about using sequences of operators that converge in this topology. It is that the topology is not determined by convergence of sequences.

Another important notion of convergence of operators is the convergence of operators with respect to the weak operator topology. This is a generalization of saying that a sequence of matrices $A_k$ converges to a matrix $A$ if the entries of the $A_k$ converge to those of $A$. This notion of convergence is particularly nice because it reduces convergence in a Hilbert space into convergence of complex numbers.
\begin{defn}
We say that $A_k$ converges weakly to $A$ if for every $u, v \in \mathcal H$,  $\langle u, 
A_kv \rangle \to \langle u, Av\rangle$.
We can express this type of convergence as $\wlim_k A_k = A$ or $A_k \overset{w}{\rightharpoonup} A$. 
\end{defn}
We define the functionals $\phi_{u,v}(C) = \langle u, Cv\rangle$ with associated seminorms $|\phi_{u,v}|$. 
Essentially everything that we said previously for the strong operator topology holds with replacing the seminorms $\phi_v$ with the $|\phi_{u,v}|$.

It is important to note that if $u, v\in \mathcal H$ are unit vectors then 
\[\phi_{u,v}(C) \leq \phi_v(C) \leq \sup_{\|w\|=1}\|Cw\|.\]
So,
\[\{C: \|C\| < \varepsilon\} \subset \{C: \phi_v(C) < \varepsilon\} \subset \{C: \phi_{u,v}(C) < \varepsilon\}.\]
This means that the open ball of radius $\varepsilon$ centered at the origin is contained in the strongly open neighborhood of the origin of the form $\{C: \phi_v(C) < \varepsilon\}$ any unit vector $v$ and this strongly open neighborhood belongs to the weakly open neighborhood of the origin of the form $\{C: \phi_{u,v}(C) < \varepsilon\}$ for any unit vector $u$.
Note that what can make this slightly confusing is the contravariance: the direction of the inequalities and subsets are in opposite directions.

This implies that any set that is open with respect to the strong operator topology or the weak operator topology is automatically open with respect to the standard topology. When $\mathcal H$ is finite dimensional, one can show that all these topologies are identical.
However, when $\mathcal H$ is infinite dimensional, it is not possible to guarantee that $\|C\|$ is small if $\phi_v(C)$ is small or $\phi_{u,v}(C)$ are small for finitely many vectors $u, v$. A consequence of this is that the unit ball is the intersection of infinitely many strong-open neighborhoods and is the intersection of infinitely many weak-open neighborhoods but is not open in the strong operator or the weak operator topologies.

So, the norm topology includes more open sets than the strong operator topology. Likewise, the strong operator topology contains more open sets than the weak operator topology.
However, it can be shown that any set that is convex will be weakly closed if and only if it is strongly closed. This implies that there is no difference in speaking about the weak closure or the strong closure of a subspace of $B(\mathcal H)$. In particular, the weak closure of a concrete $C^\ast$-subalgebra $\mathcal A$ of $B(\mathcal H)$ is the same as the strong closure of $\mathcal A$.
\begin{defn}
A (unital) $C^\ast$-subalgebra of $B(\mathcal H)$ that is weakly closed is referred to as a von Neumann algebra. 
\end{defn}
This type of $C^\ast$-algebra will then have the property of being closed under weakly convergent sequences. The von Neumann double commutant theorem states that if $S$ is a subset of $B(\mathcal H)$ that contains the identity then the von Neumann algebra generated by $S$ is equal to the space of all operators in $B(\mathcal H)$ that commute with all operators that commute with every element of $S$. This provides an algebraic characterization of the weak closure of any $C^\ast$-subalgebra of $B(\mathcal H)$.

\begin{remark}
One might be slightly concerned that this definition of a von Neumann algebra may not extend to abstract $C^\ast$-algebras since it could conceivably depend on the way that $\mathcal A$ is embedded in some $B(\mathcal H)$. However, there turns out to be an internal characterization of von Neumann algebras in terms of duality. We will not discuss this further because this is a huge subject.
\end{remark}

We return to the example that motivated the introduction of alternative topologies on $B(\mathcal H)$: the operator $A$ defined by $A e_0 = 0, Ae_j = \frac1je_j$ and $f_k(A)$ not converging to $\chi_{\{0\}}(A)$ in norm. Let $v \in \mathcal H$ and write $v = \sum_{j \geq 0} c_j e_j$ with $\sum_j |c_j|^2 < \infty$. Then
\[\phi_v\left(f_k(A)-E_{\{0\}}(A)\right)^2 = \|\sum_{j \geq 1} c_j f_k(1/j)e_j \|^2 = \sum_{j \geq 1} |c_j|^2|f_k(1/j)|^2.\]
The positive numbers $f_k(1/j)$ are bounded by $1$ and converging to zero as $k \to \infty$. So, by the dominated convergence theorem with the finite measure $\mu = \sum_{j} |c_j|^2 \delta_{1/j}$, we deduce that $\phi_v(f_k(A)-E_{\{0\}}(A)) \to 0$ as $k \to \infty$. This shows that $f_k(A)$ converges strongly to $E_{\{0\}}(A)$.

Notice that the rate that $\phi_v\left(f_k(A)-E_{\{0\}}(A)\right)$ converges to zero depends on $v$ based on how well the measure $\mu$ is supported around $0$. This is why we cannot get norm convergence.

In general, if $N$ is an element of a von Neumann algebra $\mathcal A$ then not only do the continuous functions of $N$ belong to $\mathcal A$ but also any function that is the pointwise limit  of a sequence of uniformly bounded continuous functions on $\sigma(N)$. We will explain this in more detail later when discussing weak convergence.
This discontinuous functional calculus then provides the result that we can construct the spectral projections of $N$ for most reasonable sets $S$ such as singleton sets, line segments, circles, disks, etc. 

\subsection{Baire Functional Calculus}

Recall the standard argument for the existence and uniqueness of the adjoint of an operator $T$ in $B(\mathcal H)$: For each $u \in \mathcal H$, define the linear map $\ell_u: v \mapsto \langle u, Tv\rangle$. By the Riesz representation theorem, there is a unique vector which we denote $T^\ast u$ such that
$\ell_u(v)=\langle  T^\ast u, v\rangle$. Standard arguments then can be used to show that $T^\ast$ is also a bounded linear operator and that the adjoint satisfies the standard properties.

This sort of argument for $\{A_ke_j\}_k$ using the weak compactness of the unit ball in $\mathcal H$ can be used to prove the following result guaranteeing the existence of the weak limit of a bounded sequence of operators:
\begin{prop}\label{weakConv}
Suppose that $A_k$ is a sequence of operators in $B(\mathcal H)$ that are uniformly bounded $\|A_k\| \leq M$ and satisfy the property that $\langle u, A_k v \rangle$ converges for each $u, v \in \mathcal H$. Then there is a unique operator $A \in B(\mathcal H)$ so that
$A_k \to A$ weakly and $\|A\| \leq \limsup_k \|A_k\| \leq M$.
\end{prop}
Now we return to the reason that we introduced the notion of weak convergence. 
For $u, v \in \mathcal H$ and a normal operator $N$ in $B(\mathcal H)$, the map that sends a continuous function $g$ on $\sigma(N)$ to $\langle u, g(N)v \rangle$ is a continuous linear functional of the Banach space $C(\sigma(N))$.  
Because the dual space of $C(\sigma(N))$ is the space of (complex-valued regular Borel) measures on the compact set $\sigma(N)$, there exists a measure $\mu_{u,v}$ on $\sigma(N)$ such that
\begin{equation}\label{spMeasureInt}
\langle u, g(N)v \rangle = \int_{\sigma(N)} g\, d\mu_{u,v}
\end{equation} 
for every $g \in C(\sigma(N))$.
The measure $\mu_{u,v}$ is known as the spectral measure for $N$. Because
\[\left|\int_{\sigma(N)} g\, d\mu_{u,v}\right| =|\langle u, g(N)v \rangle| \leq \|u\|\|v\|\max_{z\in\sigma(N)}|g(z)|\]
for any $g\in C(\sigma(N))$, we see that 
\begin{equation*}
\|\mu_{u,v}\|\leq \|u\|\|v\|.
\end{equation*}

Suppose that $N$ is a normal operator and $f_k$ is some sequence of uniformly bounded continuous functions on $\sigma(N)$ and $f$ is some (not necessarily continuous) function on $\sigma(N)$ that is the pointwise limit of the $f_k$. The functions $f_k$ and $f$ are then measureable with respect to $\mu_{u,v}$. By the dominated convergence theorem, $\int_{\sigma(N)} f_k\, d\mu_{u,v}$ converges to $\int_{\sigma(N)} f\, d\mu_{u,v}$. 

Therefore, we see that $\langle u, f_k(N)v \rangle$ converges for each $u, v \in \mathcal H$. The operators $f_k(N)$ are also uniformly bounded. Therefore, we conclude by Proposition \ref{weakConv} that there is a unique operator which we call $f(N)$ such that $f_k(N)$ converges weakly to $f(N)$ and $\|f(N)\| \leq \limsup_k\|f_k\|$. 
Consequently, (\ref{spMeasureInt}) holds for $f(N)$.

If $f$ is a bounded function that is the pointwise limit of continuous functions $f_k$ on $\sigma(N)$ then we can always assume that $|f_k| \leq \|f\|$ so that we obtain 
\begin{equation}\label{discontNormIneq}
\|f(N)\| \leq \|f\|=: \sup_{z \in \sigma(N)}|f(z)|.
\end{equation} 
As discussed below, this can be a strict inequality.

We can iterate this process to obtain $f(N)$ when $f$ is the pointwise limit of uniformly bounded functions which are themselves the pointwise limit of uniformly bounded  continuous functions and so on. We then obtain the so-called Baire functional calculus for any function that belongs to the space $B(\sigma(N))$ of functions gotten by iterating this pointwise convergence of uniformly bounded functions starting with $C(\sigma(N))$. With the norm $\|f\| = \sup_{z\in \sigma(N)} |f(z)|$ and conjugation as the adjoint, the space $B(\sigma(N))$ is a $C^\ast$-algebra containing $C(\sigma(N))$.

It follows by standard arguments that the Baire functional calculus extends the continuous functional calculus and that it satisfies some of the natural properties of the continuous functional calculus including that it is a $\ast$-linear map from the $C^\ast$-algebra $B(\sigma(N))$ into the von Neumann algebra generated by $N$ that respects positivity.
Depending on the operator $N$, this map might not be an isometry so the spectrum of $f(N)$ belongs to the closure of the range of $f$ but might not equal it. 
We will not go into what conditions are needed for weak and strong convergence to respect multiplication of operators. 

The Baire functional calculus provides the spectral projections of $N$ for any set whose characteristic function belongs to $B(\sigma(N))$. This includes all the common sets such as points, lines, circles, and rectangles. The range of the spectral projection onto $\Omega$ contains the eigenvectors whose eigenvalues are in $\Omega$. One can show that the spectral projection onto a set which contains a relatively open subset of $\sigma(N)$ will be non-zero.
    
We now can prove:
\begin{thm}\label{finApprox}
Let $N \in B(\mathcal H)$ be a normal operator. Then for any $\varepsilon>0$, there exists a normal operator $N_{\varepsilon}$ belonging to the von Neumann algebra generated by $N$ such that $N_\varepsilon$ has finite spectrum and $\|N-N_\varepsilon\| \leq \varepsilon$.

If $N$ is self-adjoint, then $N_{\varepsilon}$ can be chosen to be self-adjoint as well.
\end{thm}
\begin{proof}
For half-open disjoint rectangles $\Omega_j$ with diameter $\varepsilon$ whose union contains $\sigma(N)$, one can obtain the orthogonal spectral projections $E_{\Omega_j}(N)$ which add to $E_{\C}(N)=I$. Let $z_j\in \Omega_j$. Then
\begin{align*}
\|N - \sum_j z_jE_{\Omega_j}(N)\| &= \|\sum_j E_{\Omega_j}(N)N - \sum_j z_jE_{\Omega_j}(N)\| \\
&= \| \sum_jE_{\Omega_j}(N)(N-z_j)E_{\Omega_j}(N)\| =\max_j \|(N-z_j)E_{\Omega_j}(N)\|,
\end{align*}
because the projections $E_{\Omega_j}(N)$ are orthogonal. By (\ref{discontNormIneq}),
\[\|(N-z_j)E_{\Omega_j}(N)\|\leq \max_{z\in\sigma(N)}|(z-z_j)\chi_{\Omega_j}(z)|\leq \sup_{z\in\Omega_j}|z-z_j|\leq \varepsilon.\]
So, we conclude.

If $N$ is self-adjoint then $E_{\Omega_j}(N) = 0$ if $\Omega_j \cap \R = \emptyset$. So, we may suppose that each of the sets $\Omega_j$ intersects $\R$ and choose $z_j \in \Omega_j \cap \R$.
\end{proof}
\begin{remark}
In the proof we glossed over the technicality of 
$(N-z_j)g(N) = f(N)$, where $g(z)=\chi_{\Omega_j}(z)$ and $f(z)=(z-z_j)\chi_{\Omega_j}(z)$. We stated earlier that we would not speak about the technicalities of multiplication and convergence of operators in the strong or weak operator norms, but we make this single exception.

Let $g_k$ be a sequence of uniformly continuous functions that converge to $g$ pointwise. Then $f_k(z) = (z-z_j)g_k(z)$ is a sequence of uniformly continuous functions that converge to $f$ pointwise. So, $g_k(N) \overset{w}{\rightharpoonup} g(N)$ and $(N-z_j)g_k(N) \overset{w}{\rightharpoonup} f(N)$. The technical question is whether $(N-z_j)g_k(N)$ also converges to $(N-z_j)g(N)$ weakly.

Let $u, v \in \mathcal H$. Then
\begin{align*} 
\langle u, (N-z_j)g_k(N)v&\rangle = \langle (N-z_j)^\ast u, g_k(N)v\rangle \\
&\to \langle (N-z_j)^\ast u, g(N)v\rangle = \langle u, (N-z_j)g(N)v\rangle.
\end{align*}
This implies that $(N-z_j)g(N) = f(N)$ as desired.

The reason that this worked is because we were multiplying the sequence $g_j(z)$ by a fixed function $h(z)$ to obtain $f(z)$. So, in the inner product expression that we needed to converge, we could use the adjoint to change the fixed vector $u$ into the fixed vector $h(N)^\ast u$. This argument does not work if we were considering the product $g_k(z)h_k(z)$.
\end{remark}

\begin{example}
It is not always true that the spectral projection of a single point is non-empty. For instance, if $A$ is defined as $A e_k = \frac1k e_k$ on the span of $e_1, e_2, \dots$ then the eigenvalues of $A$ are $1/k$. Because the spectrum is closed, $0$ also belongs to the spectrum but it is not an eigenvalue. One can then show that $E_{\{0\}}(A) = 0$.

If one determines that an element of the spectrum of $N$ is an eigenvalue then the projection onto that eigenspace is given by the characteristic function supported on the eigenvalue applied to $N$. Likewise, if $\chi_{\{\lambda\}}(N)=0$ then $\lambda$ is not an eigenvalue of $N$.
\end{example}

\begin{example}\label{joint evalues}
It is not true in general that a normal operator on an infinite dimensional Hilbert space has any eigenvalues. 
For instance, consider the Hilbert space of $\mathcal H=L^2([0,1])$ and the multiplication operator $(Af)(x) = xf(x)$. This operator is self-adjoint.
If $\lambda \in \R$ and $f\in \mathcal H$ is supported in a set $S$ then 
\[\|Af-\lambda f\|^2 = \int_S |x-\lambda|^2|f(x)|^2dx.\]
The only way that this can equal zero is if $f$ is supported in the set $\{\lambda\}$ which implies that $f = 0$ as an element of $\mathcal H$.

Therefore, $A$ does not have any eigenvectors, however if $f$ is supported in an interval $[\lambda-\varepsilon,\lambda+\varepsilon]$ then
\[\|Af-\lambda f\| \leq \varepsilon\|f\|.\]
So, every element of $[0,1]$ is an approximate eigenvalue in this sense. 

This estimate also shows that if $(A-\lambda1)^{-1}$ existed as a linear operator then its operator norm is at least $\varepsilon^{-1}$. Consequently, $A-\lambda1$ cannot have an inverse in $B(\mathcal H)$ and hence the spectrum of $A$ contains $[0,1]$. If $\lambda \in \C$ is not in $[0,1]$ then 
\[(A-\lambda1)^{-1}f(x) = (x-\lambda)^{-1}f(x)\]
and the operator $(A-\lambda1)^{-1}$ has norm $\left(\min_{x \in [0.1]}|x-\lambda|\right)^{-1}$.
So, the spectrum of $A$ is $[0,1]$ but without any element of the spectrum being an eigenvalue.

A standard argument shows that if $N\in B(\mathcal H)$ is normal then $N$ is invertible if and only if $\|N v\|=\|N^\ast v\|$ is bounded away from zero for every unit vector $v \in H$.
Likewise, $N$ is not invertible if and only if there is a sequence of unit vectors $v_k\in \mathcal H$ such that $Nv_k \to 0$.
This means that any $\lambda \in \sigma(N)$ is an approximate eigenvalue in the sense that there is a sequence $v_k$ of unit vectors so that $\|Nv_k - \lambda v_k\| \to 0$.

This implies that if $(\lambda_1, \lambda_2)$ belongs to the joint spectrum $\sigma(A,B)$ of two self-adjoint operators $A, B \in B(\mathcal H)$ then there is a sequence of unit vectors $v_k \in \mathcal H$ such that $Av_k - \lambda_1v_k \to 0$ and $Bv_k - \lambda_2v_k \to 0.$ Consequently, the joint spectrum of two commuting self-adjoint operators consists of their approximate joint eigenvalues.

\end{example}

Because of these sorts of counter-examples, it is no longer true that the spectral mapping theorem holds in exactly the way that it does for continuous functions. If $f$ is the pointwise limit of uniformly bounded continuous functions on $\sigma(N)$ then we obtain
\[\|f(N)\| \leq \max_{z \in \sigma(N)}|f(N)|.\]
It is possible for this to be a strict-inequality such as in the case of $f=\chi_{\{0\}}$ when $0$ belongs to the spectrum of $N$ but is not an eigenvalue. 


\subsection{Polar Decomposition}

Now that we have discussed a lot of theory related to $C^\ast$-algebras, we will discuss the motivating constructions that appear frequently in the literature that is of interest. The first of which is the polar decomposition.

Consider the polar decomposition $A = UP$ of an element $A$ of a $C^\ast$-algebra $\mathcal A$, where $U$ is unitary and $P$ is positive.  This is a generalization of the polar form $z = r e^{i \theta}$ of a complex number $z$, where $r  = |z|\geq 0$ and $|e^{i\theta}|=1$. If $z \neq 0$ then the phase $e^{i\theta}$ (but not the real number $\theta$) is unique and given by $e^{i\theta} = zr^{-1}$. The value $z = 0$ does not have a unique polar decomposition.

For $A \in \mathcal A$, define $P = \sqrt{A^\ast A}$ using the functional calculus applied to the self-adjoint element $A^\ast A$. Even though $P$ might not come from the continuous functional calculus applied to $A$ since $A$ might not be normal, we denote $\sqrt{A^\ast A}$ as $|A|$.
As far as uniqueness is concerned, if $A = UP$ with $U$ unitary then it is necessarily the case that $A^\ast A = P^2$ so one can show that $P$ is unique. The non-uniqueness of the polar decomposition can only come from $U$ not being unique.

If $A$ is invertible then $P$ is invertible. We can then define $U = AP^{-1}$ so that $A = UP$. It then is straightforward to see that $U$ is unitary. 
The polar decomposition in this case is unique. 
The example of $A = 0$ provides a simple example of a normal and non-invertible operator whose polar decomposition is not unique. When $A$ is invertible but close to not being invertible, a small perturbation of $A$ can cause a large perturbation in $U$ (See Section 12 of \cite{bhatia1997and}).

If $N$ is normal and not invertible, then $P$ can be defined to be $|N|$ as before but it is not invertible. Observe that the expression $NP^{-1}$ is the function $f(z) = z/|z|$ of $N$. Because $\sigma(N)$ contains $0$, $f$ is not well-defined and might not be able to be made continuous on $\sigma(N)$. If we instead define 
\begin{equation}\label{polarFunction}
f(z) = \left\{\begin{array}{ll} z/|z|, & z \neq 0 \\ 1, & z = 0 \end{array}\right.,
\end{equation}
then $U=f(N)$ is a normal operator that belongs to the von Neumann algebra generated by $N$. 
Then since $\sigma(f(N))$ belongs to the unit circle, $U$ is unitary. Moreover, one can show that $UP = f(N)|N|$ and since $|z|f(z) = z$, we deduce that $UP = N$.

If $N$ has $0$ as an eigenvalue so that $\chi_{\{0\}}(N) \neq 0$ then the polar decomposition is not unique because we could have chosen any value for $f(0)$ that belongs on the unit circle. 
As the simple example of $N = 0$ shows, if the kernel of $N$ has a dimension of at least two then the unitary in the polar decomposition does not even need to be a function of $N$. 

When describing the polar decomposition, we chose the convention of $A = UP$. We could have chosen the alternative convention of $A = QV$ where $Q = \sqrt{AA^\ast} = |A^\ast|$ and $V$ is unitary. If $A$ is invertible than this alternative polar decomposition is also unique with the relationship being that if $A = U|A|$ then $A = |A^\ast|U$.

If $N$ is normal, then $N^\ast N = NN^\ast$ so $|N| = |N^\ast|$. We also can choose the polar factor $U$ to be a function of $N$ so that it commutes with $|N|$. Then we see  that both forms of the polar decomposition can coincide: $N = U|N| = |N|U$. This is one of many examples where a normal operator is a generalization of a complex number. We summarize these results as:
\begin{prop}
If $N \in B(\mathcal H)$ is normal, then $N$ is a product of a commuting unitary $U$ and the positive operator $|N|$, with $|N|$ belonging to the $C^\ast$-algebra $B(\mathcal H)_N$ generated by $N$ and $U$ able to be chosen to belong to the von Neumann algebra generated by $N$. If $N$ is invertible then $U=N|N|^{-1}$ belongs to $B(\mathcal H)_N$.
\end{prop}

Moreover, a polar decomposition exists for any matrix $A$ in $M_d(\C)$ even if $A$ is not invertible and not normal. However, it is not necessarily unique. A proof of this is an easy consequence of the singular value decomposition, however we discuss the following geometrical construction.

Consider the calculation
\[ \|A v\|^2 = \langle Av, Av \rangle = \langle v, A^\ast Av \rangle = \langle \sqrt{A^\ast A}v, \sqrt{A^\ast A}v \rangle = \|\, |A|v\, \|^2.\]
So, if we define the function $U[\,|A|v\,] = Av$ for $v\in \C^d$, this is a well-defined function from the range of $|A|$ into the range of $A$ so that $U|A| = A$. From the definition, we can also see that this is a linear map with range $R(A)$.
If $A$ is invertible then so is $|A|$ and hence the range of $A$ and the range of $|A|$ both equal $\C^d$. So, $U$ is actually an isometry on $\C^d$ and is hence unitary. 

If $A$ is not invertible, then the identity $\|Av\| = \||A|v\|$ implies that $A$ and $|A|$ have the same kernel, which implies that their ranges have the same dimension by the rank-nullity theorem. This shows that $U$ is an isometry between two subspaces of $\C^d$ of the same dimension and can be extended to a unitary on all of $\C^d$ by defining it to be a non-unique isometry of the orthogonal complement of the ranges of $|A|$ and of $A$. So, when $A$ is not invertible, $U$ has a non-unique extension to a unitary linear operator that satisfies $U|A| = A$.

This construction of the polar decomposition for matrices provides some useful insight into a condition for an operator $A$ in $B(\mathcal H)$ to not have a polar decomposition. It is clear from the arguments above that $U$ is completely determined on the range of $|A|$ and there is no issue with defining it as an isometry from $R(|A|)$ onto $R(A)$. Also, $U$ extends as an isometry onto the closure $\overline{R(|A|)}$ onto the closure $\overline{R(A)}$.

However, if the orthogonal complements $\mathcal H \ominus \overline{R(|A|)}$ and $\mathcal H \ominus \overline{R(A)}$ do not have the same dimension then it is not possible to extend $U$ as a unitary to all of $\mathcal H$. This is not an issue if $\mathcal H$ is finite dimensional, however if $\mathcal H$ is infinite dimensional then it is possible that these orthogonal complements have different dimensions. 

\begin{example}
This is the case for the unilateral shift $S$ defined by $Se_k= e_{k+1}$ on the Hilbert space $\mathcal H$ spanned by $e_1, e_2, \dots$. We can see this by noting that $\mathcal H \ominus R(S)$ has dimension $1$ but $R(|S|) = \mathcal H$ since $S^\ast S = 1$. This also provides an example of a non-normal operator $S$ where $|S|$ is invertible but $S$ is not.  

It can be shown also that the spectrum of $S$ is the entire unit disk and that since $S$ is an isometry no element of the open unit disk is an approximate eigenvalue.
It also can be shown that there is no invertible operator $A$ in $B(\mathcal H)$ such that $\|S-A\| < 1$.
\end{example}

\begin{remark}\label{puncture}
A useful application of the polar decomposition is that it provides a direct way of obtaining an invertible perturbation of an operator. If $A = UP$ and $\varepsilon > 0$ then $A_\varepsilon = U (P + \varepsilon 1)$ is invertible because the spectrum of $P + \varepsilon 1$ belongs to $[\varepsilon, \infty)$. Also, $A_\varepsilon$ is approximately equal to $A$:
\[\|A_\varepsilon - A\| = \varepsilon\|U\| = \varepsilon.\]
This also provides a way of taking a normal operator $N$ and finding a nearby element of the von Neumann algebra generated by $N$ that is normal and invertible. 
\end{remark}

\subsection{Special Types of $C^\ast$-Algebras}

There are many types of $C^\ast$-algebras and often additional assumptions are made which makes them simultaneously more manageable and also less pathological, where a pathology is loosely defined to be a property that makes whatever you are trying to show more difficult (or even impossible) than certain idealistic examples.

\begin{example}
For instance, consider the $C^\ast$-algebra $C([0,1])$. This $C^\ast$-algebra is abelian so every element is normal. $C([0,1])$ is separable because polynomials with rational coefficients are dense.

The spectrum of a function in $C([0,1])$ is its range.
Therefore, the spectrum of every element in $C([0,1])$ is a connected set. In particular, a self-adjoint element of $C([0,1])$ is a continuous real-valued function on $[0,1]$ so its spectrum is a closed interval. 
The only elements of $C([0,1])$ that have discrete spectrum are the constant functions which are the multiples of the identity. 

Invertible elements of $C([0,1])$ are dense. This can be seen because every complex-valued continuous function can be approximated by a smooth function that omits $0$. 
However, invertible elements are not dense in the closed $\R$-subspace consisting of the self-adjoint elements of $C([0,1])$. For instance, if $f(0)=-1$ and $f(1) = 1$ then any real-valued function $g$ that satisfies $\|f-g\| \leq 1$ is not invertible. However, if $g = f+i\varepsilon$ then $g$ is a nearby invertible element that is not self-adjoint.
\end{example}

\begin{defn}
A (unital) $C^\ast$-algebra $\mathcal A$ has stable rank $1$ if invertible elements are dense in $\mathcal A$.
\end{defn}
\begin{defn}
A (unital) $C^\ast$-algebra $\mathcal A$ has real rank zero if invertible elements are dense in the space $\mathcal A_{s.a.}$ of self-adjoint elements of $\mathcal A$.

Having real rank zero is equivalent to the property that every self-adjoint element can be arbitrarily approximated by a self-adjoint element with discrete spectrum.
\end{defn}

\begin{example}
We showed above that the $C^\ast$-algebra $C([0,1])$ has stable rank 1 and does not have real rank zero. Although self-adjoint elements can be approximated by invertible elements, it is not always possible to do this when requiring that the approximating invertible elements be self-adjoint.

$C(\{1, \dots, n\})$ has stable rank 1 and real rank zero since the spectrum of every element is discrete. This $C^\ast$-algebra is isomorphic to the $C^\ast$-algebra of diagonal matrices in $M_d(\C)$.

The $C^\ast$-algebra $M_d(\C)$ has stable rank 1 and real rank zero and is not abelian when $d \geq 2$. 

The $C^\ast$-algebra $B(\mathcal H)$ does not have stable rank 1 if $\mathcal H$ is infinite dimensional because there are operators, such as the unilateral shift, that are not nearby any invertible operator. 
\end{example}

\begin{example}
Any von Neumann algebra $\mathcal A$ has real rank zero, which includes the spaces $B(\mathcal H)$. 
This is a straightforward application of the polar decomposition. If $A$ is self-adjoint then $A = U|A|$, where $U, |A|$ commute. Now, defining $U\in \mathcal A$ by (\ref{polarFunction}), we see that $U$ is also self-adjoint. So, $A_\varepsilon = U(|A|+\varepsilon 1)$ is an invertible self-adjoint element such that $\|A-A_\varepsilon\| \leq \varepsilon$.

We can also see this result by using Theorem \ref{finApprox}, since any normal element $N \in \mathcal A$ can be approximated by a normal element $N_\varepsilon$ with finite spectrum. Choosing $\eta>0$ small enough, we then see that $N_\varepsilon+\eta1$ is a nearby invertible normal operator. If $N$ is self-adjoint then $N_\varepsilon$ can be chosen to be self-adjoint as well.
\end{example}

\begin{example}
Because $C([0,1])$ is separable, it can be embedded into $B(\mathcal H)$ with $\mathcal H$ separable. Because $B(\mathcal H)$ has real rank zero and $C([0,1])$ does not, this illustrates the fact that whether elements of a $C^\ast$-algebra can be approximated by other elements with certain properties may depend on the $C^\ast$-algebra that the approximating elements are drawn from.
The space $C([0,1])$ can also be embedded in the von Neumann algebra $L^\infty([0,1])$ which has real rank zero and stable rank 1.
\end{example}

We saw examples illustrating the fact that having real rank zero does not imply having stable rank 1 (and vice-versa). Moreover, it is not even the case that a normal element of a $C^\ast$-algebra of real-rank zero can always be approximated by an invertible element.
For instance, \cite{loring1995normal} uses a sequence of unilateral weighted shift operators and Berg's gradual exchange technique to construct a normal element of a $C^\ast$-algebra of real rank zero that is not nearby any invertible operator (and hence not nearby any normal operator with discrete spectrum). \cite{hadwin1997normal} provides a simplified and expanded form of this result.

\chapter{Almost Commuting Matrices}
\label{2.AlmostCommutingMatrices}

\section{Introduction}

The almost/nearly commuting matrix problem is an approximation problem expressible as: ``Can matrices whose commutators are approximately equal to the zero matrix be approximated by matrices whose commutators are exactly equal to zero?''  A short-hand way of expressing this is to ask when are almost commuting matrices nearly commuting. One is also interested in additional questions about the structure of the nearby commuting matrices and how far they are from the original matrices. See \cite{higham1989matrix} for a list of various matrix approximation problems.

We now define the following terminology to make the concepts of being almost and nearly commuting precise.
\begin{defn}
Following \cite{hastings2010almost}, we say that matrices $A_1, \dots, A_k\in M_d(\C)$ are $\delta$-almost commuting if $\|[A_i,A_j]\| \leq \delta$ for each $i$ and $j$, where $\|-\|$ is the operator norm. 
We say that $A_1, \dots, A_k$ are $\varepsilon$-nearly commuting if there are commuting matrices $A_i'$ such that $\|A_i'-A_i\| \leq \varepsilon$ for each $i$.
\end{defn}
Expressing the almost/nearly commuting matrix problem in more detail, one can ask what conditions on the $A_i$ are necessary so that for any $\varepsilon > 0$ there is a $\delta > 0$ so that if the $A_i$ are $\delta$-almost commuting then they are $\varepsilon$-nearly commuting.

There are various versions and generalizations of this problem, including the problem of almost commuting matrices expressed in terms of different matrix norms or almost commuting operators on an infinite dimensional space. In this discussion we will focus entirely almost commuting operators in the operator norm. We also are interested primarily in the case where the $A_i$ are self-adjoint and the $A_i'$ are also self-adjoint due to the applications to observables in quantum mechanics. The spectral theory of self-adjoint and normal matrices also provides useful tools for addressing this problem, which makes the case of $A_i, A_i'$ self-adjoint more manageable.

There is interest in exploring the dependence of $\varepsilon$ on $\delta$ and on the size of the matrices $d$.
For those interested in approximation problems of bounded operators on infinite dimensional Hilbert spaces, a dimension-independent result can be used to obtain results about compact operators (\cite{davidson2001local}).
We will discuss the usefulness of numerical estimates for our application to quantum mechanics in Chapter \ref{3.MathematicalPhysicsOfAlmostCommutingObservables}.

Note that the numerical estimates can be framed in terms of determining a function $\varepsilon(\delta, d)$ such that if $A_i\in M_d(\C)$ are $\delta$-almost commuting then the $A_i$ are $\varepsilon=\varepsilon(\delta, d)$-nearly commuting. This can be expressed as an inequality of the form:
\[\max_i\|A_i' - A_i\| \leq \varepsilon\left(\max_{i,j} \|[A_i, A_j]\|, d\right).\]
This framework will be used to discuss some of the estimates obtained.

We will now briefly survey the early work on this problem.
Rosenthal in 1969 (\cite{rosenthal1969almost}) wrote a paper raising awareness of the problem of almost/nearly commuting matrix for two self-adjoint matrices in the Hilbert-Schmidt norm. Halmos (\cite{halmos1976some}) in 1976 included the almost/nearly commuting operator problem for two operators in his list of open problems about Hilbert space operators.
Only partial results were know at the time. 
It was known that nearby commuting matrices did exist (\cite{bastian1974subnormal, luxemburg1970almost}), unlike in the infinite dimensional case since Berg and Olsen in \cite{berg1981note} provided an example of two almost commuting self-adjoint operators for which there are no nearby commuting self-adjoint operators due to a Fredholm index obstruction. The early known positive results were essentially compactness results which gave no explicit information about how $\varepsilon$ depends on $\delta$ or $d$.

We will now begin a discussion about several aspects of the dependence of $\varepsilon$ on $\delta$ and on $d$. Besides being interesting on its own terms, it also is important for its application to compact operators (see \cite{davidson2001local}) and also to non-commuting observables as discussed in the next chapter.

\section{Inequality Scaling}\label{Inequality Scaling}

Suppose that $A, B$ are $\delta$-almost commuting and $\epsilon$-nearly commuting, with nearby commuting matrices $A'$ and $B'$. An essential fact is that the norms $\|[A,B]\|$ and $\|A'-A\|, \|B'-B\|$ scale differently when replacing $A, B, A', B'$ with $cA, cB, cA', cB'$ for $c > 0$: the former scales quadratically and the latter scales linearly. 
This limits the type of dependence that $\epsilon$ can have on $\delta$ if we are permitted to use any value of $c$ based on standard scaling arguments. For instance, we cannot have a theorem such as
\begin{equation}\label{ineqScaled}
\varepsilon(\delta,d) \leq Const.\delta^{\alpha}
\end{equation}
for any $\alpha \neq 1/2$ (where the constant $Const.$ may depend on $d$). This is because scaling the matrices in
\[\|A'-A\|, \|B'-B\| \leq Const.\|AB-BA\|^\alpha\]
by $c$ gives the inequality
\[\|A'-A\|, \|B'-B\| \leq c^{2\alpha -1}Const.\|AB-BA\|^\alpha.\]
Taking $c \to \infty$ if $\alpha < 1/2$ or $c \to 0^+$ if $\alpha > 1/2$ would violate this inequality because given any non-commuting matrices $A, B$ we would have the existence of commuting matrices $A', B'$ arbitrarily close to $A, B$, respectively. 

\begin{remark}
One way to modify the original inequality so that it is resistant to such a scaling argument is to simply require that we cannot scale $A, B$ by imposing some restrictions on these matrices. If we required a norm bound such as $\|A\|, \|B\| \leq 1$, then we can only scale the inequality with $c \to 0^+$ and hence we could plausibly have a version of (\ref{ineqScaled}) for any $0 < \alpha \leq 1/2$. 
A way to have an inequality with any given $\alpha>0$ without imposing a restriction on the norms of $A, B$ is to introduce terms into the inequality so that both sides scale in the same way such as:
\begin{equation}\label{scalingIndep}
\|A'-A\|, \|B'-B\| \leq Const. \max(\|A\|, \|B\|)^{1-2\alpha}\|AB-BA\|^\alpha.
\end{equation}

Remark 1.5 of \cite{kachkovskiy2016distance} provides very simple self-adjoint operators $A_c, B_c$ with norm $1$ and $\|[A_c, B_c]\|\to 0$ such that the minimal distance to commuting self-adjoint matrices converges to zero at the same asymptotic rate as $\|[A_c,B_c]\|^{1/2}\to 0$. This shows that we cannot prove a result of the form (\ref{scalingIndep}) with $\alpha > 1/2$.

Note also that if $\beta < \alpha$ then because $\|[A,B]\| \leq 2\max(\|A\|, \|B\|)$, we see that
\[\|[A,B]\|^\alpha \leq 2^{\alpha-\beta}\max(\|A\|, \|B\|)^{\alpha-\beta}\|[A,B]\|^\beta.\]
With this in mind, the inequality (\ref{scalingIndep}) with $\alpha = 1/2$ is not only homogeneous so that it is equivalent to (\ref{ineqScaled}) but is also the best inequality of this form that can be plausibly true. 
\end{remark}

\begin{remark}
The second remark that we wish to make is that the converse of the almost/nearly commuting matrix problem is true with essentially no restriction on the matrices. 

Suppose that $A_i$ are any matrices and $A_i'$ are any commuting matrices so that the distance $\varepsilon=\max_i\|A_i'-A_i\|$ is minimized. Suppose further that $\|A_i\| \leq M$ for some positive constant $M$. Because
\begin{align*}
[A_i, A_j] &= [A_i-A_i', A_j]+[A_i', A_j-A_j']+[A_i', A_j'],
\end{align*}
we see that
\[\|[A_i, A_j]\| \leq 2\|A_i-A_i'\|\|A_j\|+2\|A_j-A_j'\|\|A_i'\| \leq 2\varepsilon M + 2 \varepsilon(M+\varepsilon) = 4M\varepsilon + 2\varepsilon^2.\]
Note further that the zero matrices are commuting approximants, so $\varepsilon \leq M$. Therefore, 
\[\|[A_i, A_j]\| \leq 6M\varepsilon.\] 
This shows that $A_i, A_j$ are $6M\varepsilon$-almost commuting. In simpler terms, being nearly commuting and bounded implies that the matrices are almost commuting. The boundedness assumption is necessary for scaling reasons.
We can restate this inequality as
\[\min_{A_i' \mbox{ commuting}}\left(\max_{i}\|A_i'-A_i\|\right) \geq \frac{\|[A_i, A_j]\|}{6\max_i\|A_i\|},\]
which is scaling-invariant. 

Just as in Remark 1.5 on (1.2) from \cite{kachkovskiy2016distance}, our inequality provides a lower bound for how close the $A_i$ are to nearby commuting matrices based on the size of the commutators $[A_i, A_j]$ and this inequality is in fact asymptotically sharp by a simple scaling argument.
The form of this simple converse of the almost/nearly commuting matrix inequality holds for any number of matrices and it shows that the norm of the commutators is an obstruction to the matrices being nearby commuting matrices. This readily applies to operators as well.

The almost/nearly commuting matrix problem asks if there are any other obstructions and what are the best estimates relating how small the commutator needs to be in order for the matrices to be nearby commuting matrices.
\end{remark}

\section{``Easy'' Dimensional-Dependent Examples}

Everything that we have discussed here applies within the universe of $d \times d$ matrices for a fixed value of $d$. 
In general there is no useful formula for the relationship between the operator norm of a matrix and the operator norms of its proper submatrices except when the matrix has a nice block structure or other special properties.  
This makes it difficult to describe how the optimal estimate for the almost/nearly commuting matrix problem depends on $d$.

For instance, if matrices $X, Y\in M_d(\C)$ have nearby commuting matrices $X', Y'$ with $\max(\|X'-X\|, \|Y'-Y\|)$ minimized then it is trivial to show that for the $(d+1)\times(d+1)$ block matrices
\[A = \bp X & 0_{d,1}\\ 0_{1,d} & 0\ep, \;\; B = \bp Y & 0_{d,1}\\ 0_{1,d} & 0\ep\]
that $\|[A,B]\| = \|[X,Y]\|$ and the minimal distance that $A, B$ are to commuting matrices is at least that of $X,Y$. However, even the simple question of whether there are closer commuting matrices is not simply answered. Likewise, for a general $(d+1)\times(d+1)$ matrix it does not appear obvious that this can be easily reduced to a problem of smaller matrices. The question of whether the optimal estimate becomes worse for larger matrices is a problem asked in \cite{pearcy1979almost} and appears to be still unsolved.

All this said, we can obtain a definite estimate for the very simple example of two almost commuting $2 \times 2$ matrices. We first reformulate the problem in terms of almost normal matrices which changes the problem of finding two matrices with a given property into a problem of finding a single matrix with a related property. 
\begin{defn}
A matrix $S\in M_d(\C)$ is $\delta$-almost normal if $\|[S^\ast, S]\| \leq \delta$ and $S$ is $\varepsilon$-nearly normal if there is a normal matrix $N$ so that $\|N-S\| \leq \varepsilon.$ 
\end{defn}
If $A = \Re(S), A' = \Re(N)$ and $B=\Im(S), B' = \Im(N)$ then $A, A', B, B'$ are self-adjoint and
\[ \|A'-A\|, \|B'-B\| \leq \|N-S\| \leq \|A'-A\|+\|B'-B\|,\]
\[ [A, B] = [\Re(S), \Im(S)] = \frac1{4i}[S + S^\ast, S-S^\ast]= \frac1{2i}[S^\ast, S],\]
\[ [A',B'] = 0.\]
Likewise, if $A, B$ are self-adjoint with $A',B'$ commuting self-adjoint matrices then for $S = A+iB$ and $N = A'+iB'$, $N$ is normal with the same identities and inequalities above still true.

This equivalence of the almost/nearly normal matrix problem and the two almost/nearly commuting self-adjoint matrices problem is a well-known and very useful way of proving results of almost/nearly commuting self-adjoint matrices. For instance, in 1977, Phillips (\cite{phillips1977nearest}) found the nearest normal to a so-called binormal operator, which has a representation as a block upper-triangular $2\times 2$ block operator matrix. Applying this to matrices in $M_2(\C)$, Phillips showed that if \[S = \bp a & b \\ 0&c\ep\] where $a, b, c \in \C$ ($b \neq 0$) and $a-b = u|a-b|$ is a polar decomposition of $a-b\in\C$ then the nearest normal to $S$ is 
\[N=\bp a & \frac12b \\ \frac12 u^2\overline{b}&c\ep\] and by (\ref{M_2norm}): 
\[\|N-S\| = \frac12|b| \leq \frac12\left(|b|^4 + |a-c|^2|b|^2\right)^{1/4}= \frac12\|[S^\ast,S]\|^{1/2}\]
is an equality when $a=c$.

By Shur's theorem, any matrix is unitarily equivalent to an upper-triangular matrix of this form. So, Phillips showed that for any $S\in M_2(\C)$, the nearest normal matrix $N$ satisfies \[\|N-S\| \leq \frac12\|[S^\ast,S]\|^{1/2}.\]

If one merely chooses $N=\bp a & 0 \\ 0&c\ep$ then $\|N-S\| \leq \|[S^\ast,S]\|^{1/2}.$
By a remark made in the paper, this approximation method of simply discarding the strictly-upper triangular part generalizes so that if $S \in M_d(\C)$ then there is a normal matrix so that
\[\|N-S\| \leq (d-1)\|[S^\ast,S]\|^{1/2}.\]
This shows that almost normal matrices are nearby normal matrices in a way that depends on the matrix size. Consequently, the same asymptotic estimate holds for almost commuting self-adjoint matrices.

Earlier in 1962, Henrici (\cite{henrici1962bounds}) used a similar method to obtain the inequality
\[\|N-S\|_{HS} \leq \left(\frac{d^3-d}{12}\right)^{1/4}\|[S^\ast,S]\|_{HS}^{1/2},\]
with a characterization for when this is an equality. As noted in that paper, we can convert this Hilbert-Schmidt inequality into an inequality for the operator norm to obtain
\[\|N-S\| \leq \left(\frac{d(d^3-d)}{12}\right)^{1/4}\|[S^\ast,S]\|^{1/2},\]
which is asymptotically the same as Phillip's inequality as $d \to \infty$. 

Note that we cannot try to take advantage of the conversion between the Hilbert-Schmidt and operator norms since a matrix in $M_d(\C)$ for which Henrici's inequality is an equality has $N$ being a multiple of the identity and $\|N-S\|$ asymptotically decaying to zero as $d \to \infty$ at the same rate as $d\|[S^\ast, S]\|^{1/2}$.

By the equivalence, this shows that two almost commuting self-adjoint matrices $A$, $B$ are nearly commuting in a dimensional dependent way. For instance, if $A,B\in M_2(\C)$ then
\[\|A'-A\|, \|B'-B\| \leq \frac1{2^{1/2}}\|[A,B]\|^{1/2}\]
and if $A,B\in M_d(\C)$ then
\[\|A'-A\|, \|B'-B\| \leq \frac{d}{3^{1/4}}\|[A,B]\|^{1/2}.\]
However, if $\|A\|, \|B\|=1$ and $\|[A,B]\|$ is not much smaller than the dimensional-dependent value of $d^{-2}$, this estimate is not much better than simply choosing a trivial pair of matrices such as $A'=A, B'=0$ or $A'=B'=0$. 

\section{Small Dimensional Examples}

We will discuss a method that involves carefully constructing $B'$ based on knowing the eigenvalues of $A$. However, we need to mention an important constraint to choosing $B'$ based on the distribution of the eigenvalues of $A$.
Using a unitary change of basis, we can assume that $A$ is diagonal. To illustrate why choosing $A'=A$ and only perturbing $B$ to a matrix that commutes with $A'$ will not be sufficient, we will assume that $A'$ has distinct eigenvalues. Then for $[A',B']=0$, we require that $B'$ is diagonal. 

To find the optimal matrix $B'$ to minimize $\|B'-B\|$, we would need to find the closest diagonal matrix to $B$. Because the operator norm is not induced by an inner product on $M_d(\C)$, there is not an obvious nearest diagonal matrix. We can obtain a lower bound for $\|B'-B\|$ based on the off-diagonal entries of $B$:
\[\|B'-B\| \geq \max_{i\neq j}|\langle e_i, (B'-B)e_j\rangle| = \max_{i\neq j}|B_{i,j}|.\]

\begin{example}
For instance, consider $A =\bp 0&0&0\\0&\delta&0\\0&0&2\ep$ and $B = \bp 0&1&0\\0&0&0\\0&0&0\ep$. It is the case that for any diagonal matrix $B' \in M_3(\C)$, $\|B'-B\| \geq 1$ although $\|[A,B]\| = \delta$.
Because the operator norm is not uniformly convex, a nearest diagonal matrix may not be unique. For instance, we could choose $B'$ to be $\diag(0,0,r)$ for any $|r| \leq 1$ so that $\|B'-B\| = 1$. 

So, not only can the nearest matrix $B'$ that commutes with $A$ be much farther from $B$ than the size of $\|[A,B]\|$ but the nearest commuting matrix is not unique. The reason that the choice $A'=A, B'=\diag(B)$ is poor is that $A$ has close eigenvalues so $B$ can have a small commutator with $A$ without being almost a diagonal matrix in the same way that any matrix commutes with the identity matrix regardless of whether the matrix is diagonal or not.
We see that when constructing nearby commuting matrices, we will in general need to perturb all matrices involved.
\end{example}
\begin{example}
Consider the following example in $M_3(\C)$:
\[A = \bp d_1 & 0 & 0 \\ 0 & d_2 & 0 \\ 0&0&d_3\ep, \;\; B = \bp \ast & a & b \\ 0 & \ast & c \\ 0&0&\ast\ep,\]
where $d_1 \leq d_2 \leq d_3$ and $a,b,c \in \C$.
Suppose that $\|[A,B]\| = \delta$ (which we will think of as small).
The diagonal part of $B$ commutes with $A$ so we see that
\[[A, B] =  \bp 0 & a(d_1-d_2) & b(d_1-d_3) \\ 0 & 0 & c(d_2-d_3) \\ 0 & 0 & 0\ep.\]

Then $a(d_1-d_2)$, $b(d_1-d_3)$, and $c(d_2-d_3)$ have absolute values that are at most $\delta$. The exact conditions on these three terms so that $\|[A,B]\|$ is small are somewhat complicated however if we allow a dimensional dependence, it is sufficient to require that $|a(d_1-d_2)|$, $|b(d_1-d_3)|$, and $|c(d_2-d_3)|$ are at most $\delta/2$. 

We then see that for either of the terms $|a(d_1-d_2)|$, $|b(d_1-d_3)|$, $|c(d_2-d_3)|$ at least one of the following must be true: the corresponding entry of $B$ is small or the difference between the corresponding diagonal entries of $A$ is small. So, if $d_1 \approx d_2$ and $d_1, d_2$ are far from $d_3$ then both $b, c$ are bounded by some constant multiple of $\delta$ and $a \leq \frac{\delta}{|d_1-d_2|}$ is not required to be small but not allowed to be too large.

From this, we see that perturbing $B$ to make $b, c$ equal to zero will make two of the entries of $[A,B]$ equal to zero without introducing a large perturbation to $B$. However, we cannot necessarily replace $a$ with zero while guaranteeing that the perturbation to $B$ will be small. Here we instead perturb $A$ by merging the eigenvalues $d_1$ and $d_2$ into a single eigenvalue: $\frac12(d_1+d_2)$. 
We then obtain the commuting matrices
\[A' = \bp \frac12(d_1+d_2) & 0 & 0 \\ 0 & \frac12(d_1+d_2) & 0 \\ 0&0&d_3\ep, \;\; B' = \bp \ast & a & 0 \\ 0 & \ast & 0 \\ 0&0&\ast\ep\]
with $\|A'-A\| \leq \frac12|d_1-d_2|$. However it is not true that $\|B'-B\|$ is bounded by $\max(|b|, |c|)$, but by a multiple of it. This is where the dimensional dependence appears. More precisely,  
\[\|B'-B\| = \sqrt{(B'-B)^\ast(B'-B)}= \sqrt{|b|^2+|c|^2} \leq \sqrt2\max(|b|, |c|).\]
We now use a value $\Delta > 0$ as a threshold for the eigenvalues of $A$ being close so that $d_2-d_1 \leq \Delta$ but $d_3-d_1 \geq d_3-d_2 > \Delta$. So, using $|b(d_1-d_3)|, |c(d_2-d_3)| \leq \delta$, we obtain
\[\|A'-A\| \leq \frac12\Delta\]
and
\[\|B'-B\| \leq \sqrt{2}\frac{\delta}{\Delta}.\]

There are also other cases depending on which eigenvalues of $A$ are close or not.
In general, this method of constructing a nearby commuting matrix will require choosing which of the eigenvalues of $A$ to merge and which of the entries of $B$ to make identical to zero.
If we instead deem the latter two entries of $A$ to be close then we would perform a similar method to construct $A'$ and $B'$ resulting in the same estimate. 

If we deem all the entries of $A$ to be close (within a distance $\Delta$) then we would leave $B$ unchanged and replace $A$ with $d_0I$, where $d_0 = \frac12(d_3+d_1)$ is the midpoint of $[d_1, d_3]$. This provides the estimate
\[\|A'-A\| \leq \frac12\Delta, \;\; \|B'-B\| = 0.\]

We now discuss the case where all the eigenvalues of $A$ are far.
Using the formula (\ref{strictlyUpperNorm}) for the operator norm of a strictly upper triangular matrix, we obtain
\[\|\bp 0&a&b\\0&0&c\\0&0&0\ep\|^2 = \frac12(|a|^2+|b|^2+|c|^2)+\sqrt{\frac14(|a|^2+|b|^2+|c|^2)^2-|ac|^2}\]
and
\begin{align*}
\|[A,B]\|^2 &= \frac12\left(|a(d_1-d_2)|^2+|b(d_1-d_3)|^2+|c(d_2-d_3)|^2\right)\\
&+\sqrt{\frac14\left(|a(d_1-d_2)|^2+|b(d_1-d_3)|^2+|c(d_2-d_3)|^2\right)^2-|ac(d_1-d_2)(d_2-d_3)|^2}.
\end{align*}
Due to the presence of the negative term, the inequalities $d_2-d_1, d_3-d_2 > \Delta$ do not clearly lead to $\|B'-B\| \leq \frac1\Delta\|[A,B]\|$.
One way to estimate $\|B'-B\|$ in terms of $\|[A,B]\|$ is to use Hadamard product estimates, but we will simply use the equivalence of $\|-\|$ and $\|-\|_{HS}$ for $2\times 2$ submatrices:
\begin{align*}
\|B'-B\| &\leq \sqrt{|a^2|+|b^2|+|c^2|} \\
&\leq \frac1\Delta\sqrt{|a(d_1-d_2)|^2+|b(d_1-d_3)|^2+|c(d_2-d_3)|^2} \leq \frac{\sqrt{2}}\Delta\|[A,B]\|
\end{align*}
so
\[\|A'-A\|=0, \;\; \|B'-B\| \leq \frac{\sqrt{2}}\Delta\|[A,B]\|.\]

In each of the cases, we obtained
\[\|A'-A| \leq \frac12\Delta, \;\; \|B'-B\| \leq \sqrt{2}\frac{\delta}{\Delta}.\]
Now, choosing $\Delta = 2^{3/4}\delta^{1/2}$ we obtain
\[\|A'-A|, \|B'-B\| \leq 2^{-1/4}\|[A,B]\|^{1/2}.\]
\end{example}

\section{Eigenvalue Grouping Method}

The main property being used in the prior section is that how much $B$ maps an eigenvector of $A$ into a different eigenspace is controlled by how far the eigenvalues are apart and how small $\|[A,B]\|$ is.
For matrices with more eigenvalues, this method runs into trouble because estimating $\|B'-B\|$ is not simple because we should not just discard many terms using only the fact that each of these entries has a small absolute value.

By carefully estimating the norm of $B'-B$ and choosing the grouping of the eigenvalues of $A$, Pearcy and Shields (\cite{pearcy1979almost}) in 1978 proved
\begin{thm} (\cite{pearcy1979almost}) For $A, B \in M_d(\C)$ with $A$ self-adjoint, there are commuting matrices $A', B'$ such that 
\[\|A'-A\|, \|B'-B\| \leq \left(\frac{d-1}{2}\right)^{1/2}\|[A,B]\|^{1/2}.\]
Moreover, $A'$ is self-adjoint and if $B$ is self-adjoint then $B'$ is also.
\end{thm}
Written for the purpose of illustration, the author's Proposition 9.1 of \cite{herrera2020hastings} contains a simplified version of this method that produces an estimate similar to that obtained by Phillips with the observation made that the constant does not depend on the dimension $d$ but on the number of distinct eigenvalues of $A$ (which is of course at most $d$).

\begin{remark}\label{degenerate}
A consequence of the Pearcy and Shields' result is that if $\|[A,B]\|$ is much smaller than $d^{-1}$ then there are commuting matrices that are nearby, regardless of whether $B$ is self-adjoint. We can understand how this represents a degenerate case as follows. If we assume that $\|A\|, \|B\|=1$ then the commutator $\|[A,B]\|$ will then be typically be much smaller than the average spacing of the eigenvalues of $A$. This causes many of the off-diagonal entries of $B$ to be rather small so $B$ is well approximated by a diagonal matrix.
See Section 9 of \cite{herrera2020hastings} for a discussion of how almost commuting matrices behave in this degenerate case which is contrary to the counter-examples that we will discuss later. 
\end{remark}

At this moment, we will present a straightforward proof of Pearcy and Shields' inequality (but at the expense of a larger numerical constant). We present it here because we are not aware of it being noted anywhere else. First, we need the lemma:
\begin{lemma} \label{tridiag lemma}
(Lemma 3.4 of \cite{berg1991almost}) Let $A, B\in B(\mathcal H)$ be operators with $A$ self-adjoint. Choose a set $\{a_j\}$ of increasing real numbers so that the intervals $I_j = [a_{j}, a_{j+1})$ have length at least $\Delta > 0$ and cover $\sigma(A)$. If $E_j = E_{I_j}(A)$ then with respect to these projections, we can write $A$ as the direct sum of $A_j = AE_j$ and $B$ as a block operator matrix with entries $B_{i,j} = E_iBE_j$.

Let $\tilde{B}$ denote the tridiagonal part of $B$, gotten by replacing the blocks $B_{i,j}$ of $B$ with $0$ for $|i-j| \geq 2$. Then
\[\|\tilde{B}-B\| \leq \frac{12}\Delta \|[A,B]\|.\]
\end{lemma}
This lemma is proved by first choosing an explicit function $f \in L^1(\R)$ determined by $\hat{f}(\xi) = 1/\xi$ outside a neighborhood of $0$. One considers auxilliary operator 
\[C = (B-\tilde{B})-\int_\R e^{-itA}[A,B-\tilde{B}]e^{itA}f(t)dt\]
which can be shown to be a block diagonal operator that is nearby $B-\tilde{B}$. This then provides a bound for the norm of $B-\tilde{B}$ since its diagonal block entries are zero.

This lemma greatly simplifies the calculation of $\|\tilde B - B\|$ by Pearcy and Shields.
So, we present a simple proof of:
\begin{prop}
Let $A, B\in B(\mathcal H)$ be operators with $A$ self-adjoint and $\sigma(A)$ containing at most $n$ elements. Then there are commuting operators $A', B'$ such that $A'$ is self-adjoint with $\sigma(A')$ containing at most $n$ elements and
\[\|A'-A\|, \|B'-B\| \leq (6n)^{1/2}\|[A,B]\|^{1/2}.\]
Moreover, there are spectral projections $E_i$ of $A$ such that $B' = \sum_{|i-j| \leq 1} E_i B E_j$. Therefore, if $B$ is self-adjoint then so is $B'$.
\end{prop}
\begin{proof}
Let $\delta = \|[A,B]\|$. Choose $\Delta > 0$ and $a_j\in \R$ so that $I_j = [a_j, a_{j+1})$ form consecutive intervals whose union contains $\sigma(A)$. Let $\tilde{B}$ be as in the above lemma. 

We suppose that $\sigma(A)$ contains $n$ values. Merge consecutive intervals $I_j$ for which $I_j \cap \sigma(A) \neq \emptyset$ into larger intervals $I_i'$. Because $\sigma(A)$ contains at most $n$ elements, we see that each interval $I_i'$ is gotten by merging at most $n$ intervals $I_j$. So, each $I_i'$ has length $|I_i'|$ at most $n\Delta$. Also, the $I_i'$ cover the spectrum of $A$.

Define $E_i' = E_{I_i'}(A)$. Note that each $E_i'$ commutes with $\tilde B$ because $I_i' = I_{j_1} \cup \cdots \cup I_{j_2}$ where $(I_{j_1-1} \cup I_{j_2+1}) \cap \sigma(A) = \emptyset$ so $E_{j_1-1} = E_{j_2+1} = 0$. 

So, define $a_i'$ to be the midpoint of each $I_i'$. Setting $A' = \sum_i a_i' E_i'$ and $B' = \tilde B$, we see that
\[ \|A'-A\| \leq \frac12\max_i |I_i'| \leq \frac12 n\Delta\]
by our choice of $I_i'$ and
\[\|B'-B\| \leq \frac{12}{\Delta}\delta\]
by Lemma \ref{tridiag lemma}.
Choosing $\Delta = \sqrt{24\delta/n}$, we obtain the result.
\end{proof}

The main issue with extending this method to remove the dimensional dependence is that it relies on knowing how the eigenvalues are spaced so as to optimally group them. If we were more careful in the estimate as Pearcy and Shields did, we could obtain a sharper estimate in the proof above but we simply cannot remove the dependence on $n$. 

Note that perturbing $A$ to stretch out the gaps between the eigenvalues of $A$ does not help since that would increase $\|[A,B]\|$. In, fact this result cannot be improved without further assumptions on $B$ as the first example of the next section shows.

\section{Dimensional-Dependent Counter-Examples}
\label{Dimensional-Dependent Counter-Examples}

Consider the following example of Choi (\cite{choi1988almost}). Let
\[A = \bp 
d_1 &     &        &  \\
    & d_2 &        &  \\
    &     & \ddots &  \\ 
    &     &        & d_n 
\ep, \;\; 
B = \bp
0    &        &          &\\
b_1  & 0      &          &\\
     & \ddots & \ddots   &\\
     &        &  b_{n-1} & 0
    \ep\]    
where $d_i$ are certain real numbers that are evenly spaced from approximately $-1$ to $1$ with $d_{i+1}-d_i = 1/n$ and $b_i \geq 0$ satisfy $b_i^2+d_i^2=1$. Choi showed:
\begin{thm} (\cite{choi1988almost})
$A, B$
satisfy the property that $\|[A,B]\| \leq 2/n$ but for any commuting matrices $A', B'$: $\|A'-A\|, \|B'-B\| \geq 1-1/n$. 
\end{thm}
Note that $A$ is self-adjoint but $A', B'$ are far away even with neither of them required to be self-adjoint.

The proof of this result is based on the following reasoning. Observe that if $r, s, t \in \C$ then since $J = \bp r & s \\ t & -r\ep$ has trace zero, it is the case that the two eigenvalues of $J$ are symmetric about the origin. This observation generalizes to the block matrix $J = \bp R & S \\ T & -R\ep$ where $R, S, T \in M_d(\C)$ if $R$ and $S$ commute. In particular, the signature of $J$, the difference between the number of its positive and negative eigenvalues, is equal to zero.

If one defines $J = \bp A+I/n & B \\ B^\ast & -A-I_n\ep$, one sees that $J$ is self-adjoint and approximately unitary. One can show that the non-zero signature of $J$ is unchanged under small perturbations of $A$ and $B$. So analyzing how large of a perturbation of $A$, $B$ would cause a change in the signature of $J$ to zero provides the estimate.

In Choi's example, $n\|[A,B]\|$ is bounded below and $A, B$ are not nearby any commuting matrices. So, this shows that the Pearcy and Shields's estimate is asymptotically sharp when $B$ is not assumed to be self-adjoint.
Moreover, even though $B$ is not self-adjoint, we see that if we calculate
\[ B^\ast B = \bp 
b_1^2 &     &        &  \\
    & \ddots &        &  \\
    &       & b_{n-1}^2 &  \\ 
    &     &          & 0 
\ep, \;\; 
BB^\ast = 
\bp 
0 &     &        &  \\
    & b_1^2 &        &  \\
    &     & \ddots &  \\ 
    &     &        & b_{n-1}^2 
\ep,\]
then
\[\|[B^\ast,B]\| = \max\left(|b_1|^2, \max_{1\leq i\leq n-2} |b_{i+1}^2-b_i^2|, |b_n|^2\right)\] 
and
\[\max_i |b_{i+1}^2-b_i^2|= \max_i (d_{i+1}^2-d_i^2) = \max_i(d_{i+1}-d_i)(d_{i+1}+d_i) \leq 2/n.\]
Consequently, the weighted shift matrix $B$ from Choi's example is almost normal because the entries $b_i\in \R$ have squares that change slowly, starting and ending near $0$. 

Prior to Choi's paper, in 1983 Voiculescu showed
\begin{thm}\label{Voiculescu} (\cite{voiculescu1981remarks}) The following unitaries
\[U_n = \bp
0    &        &          & 1\\
1  & 0      &          &\\
     & \ddots & \ddots   &\\
     &        &  1 & 0
    \ep, \;\; 
V_n = \bp 
1 &     &        &  \\
  & e^{2\pi i/n} &        &  \\
  &     & \ddots &  \\ 
  &     &        & e^{2\pi (n-1)i/n} 
\ep\]
in $M_n(\C)$ are almost commuting satisfying $\|[U_n, V_n]\| = 2\pi/n$ but are not nearby commuting unitary matrices. 
\end{thm}
Voiculescu's proof reduced this to a result of Halmos concerning the non-existence of certain finite dimensional projections $F$ that almost commute with the unilateral shift $S$. Halmos's proof essentially relies on the fact that $SF$ and $FS$ are partial isometries which have different ranks if the range of $F$ contains $e_1$. This shows that $\|SF-FS\|$ cannot be small due to this rank obstruction.

Davidson (\cite{davidson1985almost}) in 1985 provided two sequences of matrices $A_n, B_n$ with $A_n$ self-adjoint and $B_n$ normal that are not nearby commuting matrices $A', B'$ with $A'$ self-adjoint. This method used matrices similar to Choi's example: a diagonal matrix and a weighted shift matrix, however with two modifications. The first is that the analogue of $B$ was defined so that the entries slowly increase linearly from $0$ to $1$, remain constant for a long stretch, then decrease linearly back to zero. Then this matrix was replaced with a nearby normal matrix using a theorem of Berg which we will discuss in more depth in Chapter \ref{6.BergThm}.

These counter-examples by Voiculescu and Davidson provide counter-examples to the problem for $k\geq 3$ almost commuting self-adjoint matrices. 
If we define $A_{1,n} = \Re U_n, A_{2,n} = \Im U_n, A_{3,n} = \Re V_n, A_{4,n} = \Im V_n$ then Voiculescu's result shows that in general four almost commuting self-adjoint matrices $A_{i,n}$ may not be (simultaneously) nearly commuting. This is because if $A'_{i,n}$ were nearby commuting self-adjoint matrices then $U'_n = A_{1,n}'+iA_{2,n}', V'_n = A_{3,n}'+iA_{4,n}'$ are commuting normal matrices close to $U_n, V_n$. Then we can easily perturb these commuting almost unitaries to commuting unitaries.

Davidson's result from above also implies that there exist three almost commuting self-adjoint matrices that are not nearly commuting since $A_1 = A$, $A_{2} = \Re B$, $A_{3} = \Im B$ are three almost commuting self-adjoint matrices that are not nearly commuting. This is also true for Choi's example.
Earlier in 1981, Voiculescu (\cite{voiculescu1981remarks}) proved this result by investigating some of the properties of the $C^\ast$-algebra of the Heisenberg group.

\section{Liftings}
\begin{defn}\label{seqSpace}
Let $\mathcal A$ be the following $C^\ast$-algebra of bounded sequences matrices in matrix algebras $M_{d_n}(\C)$: $\{(T_n): T_n \in M_{d_n}(\C), \sup_n\|T_n\|<\infty\}$ and $\mathcal I$ the closed ideal of all sequences in $\mathcal A$ for which $\|T_n\| \to 0$. The norm on $\mathcal A$ is
\[\|(T_n)\|_{\mathcal A} = \sup_n\|T_n\|.\]
\end{defn}
Note that in the literature one often sees the notation of $\mathcal A = \prod_n M_{d_n}(\C)$, the infinite direct product of matrix $C^\ast$-algebras, and $\mathcal I = \bigoplus_n M_{d_n}(\C)$, the infinite direct sum of matrix $C^\ast$-algebras.

Voiculescu at the end of \cite{voiculescu1981remarks} conjectures that an almost commuting matrix problem for certain types of matrices might be prohibited from having a solution if there is a cohomological obstruction to a certain ``lifting'' problem which we now define.
\begin{defn}
For a topological space $X$, we say that a $\ast$-homomorphism $\varphi: C(X) \to \mathcal A / \mathcal I$ lifts to a $\ast$-homomorphism $\tilde\varphi:C(X)\to \mathcal A$  if $\varphi = \pi\circ\tilde\varphi$, where $\pi: \mathcal A \to \mathcal A / \mathcal I$ is the projection onto the quotient $C^\ast$-algebra.
\end{defn}
\begin{remark}
This is represented symbolically by the following diagram:
\[\begin{tikzcd}                                                & \mathcal A \arrow[d, "\pi"] \\
C(X) \arrow[r, "\varphi"] \arrow[ru, "\tilde \varphi", dashed] & \mathcal A / \mathcal I             
\end{tikzcd}\]
where one interprets the solid arrows as given maps and the dotted arrow as the map whose existence we are discussing. 

When we write such diagrams, one often says that they ``commute'' if composing the functions corresponding to any sequence arrows from a given starting space to a given final space leads to the same map regardless of the path of arrows taken through the diagram. In particular, the condition $\varphi = \pi\circ\tilde\varphi$ is equivalent to the above diagram commuting. 

The use of the word ``lifting'' to describe the existence of this map $\tilde\varphi$ refers to the perspective that $\pi$ maps elements of $\mathcal A$ ``down'' into the quotient $\mathcal A / \mathcal I$, so undoing this process is viewed as ``lifting'' the elements of the range of $\varphi$ in $\mathcal A / \mathcal I$ ``up''.
\end{remark}
The relevance of this algebraic problem to almost/nearly commuting matrices is why several of the later papers are concerned with this lifting property of $C(X)$. 
\begin{defn}(\cite{enders2019almost}) Recall that for any sequence of dimensions $(d_n)_n \in \N$, we can construct the $C^\ast$-algebras $\mathcal A$ and $\mathcal I$ as in Definition \ref{seqSpace}. 

We say that $C(X)$ is matricially semiprojective if any $\ast$-homomorphism of $C(X)$ into $\mathcal A/\mathcal I$ lifts to a $\ast$-homomorphism of $C(X)$ into $\mathcal A$, regardless of the choice of $(d_n)_n$.
\end{defn}
Note that as \cite{enders2019almost} remarks, this property sometimes goes by different names.

\begin{example}
For an almost/nearly commuting matrix problem with matrices satisfying certain relations, the space $X$ loosely speaking consists of all possible joint spectra of the matrices of the almost commuting matrices if they were actually commuting.

For two almost commuting unitaries, the corresponding topological space $X$ is the 2-torus $T^2$, the product of two unit circles. For $m$ almost commuting self-adjoint matrices that with norm at most one, $X = [-1,1]^m$. For three almost commuting self-adjoint matrices $A, B, C$ that approximately satisfy $A^2+B^2+C^2=I$ then $X$ is the (two dimensional) unit sphere $S^2$ in $\R^3$.
\end{example}

\begin{example}
For instance, suppose that $A_n, B_n, C_n \in M_{d_n}(\C)$ are three almost commuting self-adjoint matrices with $\|A_n\|$, $\|B_n\|$, $\|C_n\| \leq 1$ and 
\[\|[A_n, B_n]\|, \|[A_n, C_n]\|, \|[B_n, C_n]\| \to 0\]
as $n \to \infty$.
Then $(A_n), (B_n), (C_n) \in \mathcal A$. Since the norm of each of the commutators converges to zero, we see that $([A_n, B_n])$, $([A_n, C_n])$, $([B_n, C_n])$ belong to $\mathcal I$, so their image under $\pi$ is zero. 
Therefore, $\pi((A_n)), \pi((B_n)), \pi((C_n))$ are commuting self-adjoint elements of $\mathcal A / \mathcal I$ since $\pi$, as a $\ast$-homomorphism, maps commutators to commutators: 
\[\pi([a, b])=\pi(ab-ba)=\pi(a)\pi(b)-\pi(b)\pi(a)=[\pi(a), \pi(b)],\]
for $a, b \in \mathcal A$.

For commuting self-adjoint contractions $A_n', B_n', C_n'$, the condition 
\[\|A_n'-A_n\|, \|B_n'-B_n\|, \|C_n'-C_n\| \to 0\] is equivalent to $\pi((A_n'))=\pi((A_n))$, $\pi((B_n'))=\pi((B_n))$, $\pi((C_n'))=\pi((C_n))$. This means that $(A_n), (B_n), (C_n)$ are nearly commuting if and only if their images in the quotient $\mathcal A/\mathcal I$ can be lifted to commuting elements of $\mathcal A$.  

A different perspective can be framed in terms of $\ast$-homomorphisms out of the $C^\ast$-algebra $C([-1,1]^3)$.
The sequences of almost commuting matrices $A_n', B_n', C_n'$ induce a $\ast$-homomorphism $\varphi: C([-1,1]^3) \to \mathcal A / \mathcal I$ as follows. Let $g_x(x,y,z) = x$, $g_y(x,y,z) = y$, $g_z(x,y,z) = z$, $g_1(x,y,z)=1$. Then define $\varphi$ on the algebra generated by these functions by 
\begin{equation}\label{vpDef}
\varphi(g_x) = \pi((A_n)),\; \varphi(g_y) = \pi((B_n)),\; \varphi(g_z) = \pi((C_n)),\; \varphi(g_1) = \pi((I_{d_n})).
\end{equation}
Note that the algebra generated by $g_x, g_y, g_z, g_1$ is dense in $C([-1,1]^3)$ due to the Stone-Weierstrass theorem. Moreover, this algebra is the polynomials in $g_x, g_y, g_z, g_1$ which have a basis of monomials so it is easy to see that $\varphi$ is well-defined on this algebra. 

If $f(x,y,z)$ is a function belonging to the algebra generated by $g_x$, $g_y$, $g_z$, $g_1$ then the spectral theorem guarantees that 
\[\|f\left(\pi((A_n)), \pi((B_n)), \pi((C_n))\right)\| \leq \max_{(\lambda_1, \lambda_2, \lambda_3)\in [-1,1]^3} |f(\lambda_1, \lambda_2, \lambda_3)| = \|f\|_{C(X)}\]
by (\ref{jointNorm}) and (\ref{jointSpBound}).
So, $\varphi$ extends to a map on $C(X)$ by continuity. Then $\varphi$ is a $\ast$-homomorphism from $C([-1,1]^3)$ into $\mathcal A / \mathcal I$. 

Now, suppose that $A_n$, $B_n$, $C_n$ are asymptotically nearby commuting matrices $A_n'$, $B_n'$, $C_n'$. Because $(A_n')$, $(B_n')$, $(C_n')$ commute, we can define \begin{equation}\label{vpLiftDef}
\tilde\varphi(g_x) = (A_n'),\; \tilde\varphi(g_y) = (B_n'),\; \tilde\varphi(g_z) = (C_n'),\; \tilde\varphi(g_1) = (I_{d_n})
\end{equation}
and its extension to $C(X)$.
Then since $\pi((A_n'))=\pi((A_n))$, $\pi((B_n'))=\pi((B_n))$, $\pi((C_n'))=\pi((C_n))$, we see that $\tilde\varphi$ is a lift of $\varphi$.

Conversely, if $\tilde\varphi:C([-1,1]^2)\to \mathcal A$ is a lift of $\varphi$ then $(A_n'):=\tilde\varphi(g_x)$, $(B_n'):=\tilde\varphi(g_y)$, $(C_n'):=\tilde\varphi(g_z)$ are self-adjoint commuting elements of $\mathcal A$ which $\pi$ maps to $\varphi(g_x)=\pi((A_n))$, $\varphi(g_y)=\pi((B_n))$, $\varphi(g_z)=\pi((C_n))$, respectively. This implies that $A_n'$, $B_n'$, $C_n'$ are nearby commuting matrices for $A_n$, $B_n$, $C_n$.

What this tells us is that the existence of $\ast$-homomorphisms from $C([-1,1]^3)$ to $\mathcal A/\mathcal I$ that cannot be lifted to $\mathcal A$ is equivalent to there being three almost commuting self-adjoint matrices that cannot be approximated by commuting self-adjoint matrices. This equivalence between the approximation problem and the lifting problem played an influential part in the later work on this problem which we will discuss later.
\end{example}
We now present an example to illustrate the case of matrices satisfying some relations.
\begin{example}
Suppose that $A_n, B_n, C_n \in M_{d_n}(\C)$ are three almost commuting self-adjoint matrices with $\|A_n\|$, $\|B_n\|$, $\|C_n\| \leq 1$,
\[\|[A_n, B_n]\|, \|[A_n, C_n]\|, \|[B_n, C_n]\| \to 0\]
as $n \to \infty$, and the additional property that 
\[\|A_n^2+B_n^2+C_n^2-I_{d_n}\|\to 0.\]
As before, we obtain a $\ast$-homomorphism $\varphi:C(S^2)\to \mathcal A/\mathcal I$ satisfying (\ref{vpDef}) by viewing $g_x$, $g_y$, $g_z$, $g_1$ as functions on the sphere $S^2$.

The only plausible issue with such a definition is that it may not be well-defined because monomials in $g_x$, $g_y$, $g_z$, $g_1$ are not linearly independent on $S^2$ because they satisfy $g_x^2+g_y^2+g_z^2=1$. One basis for the algebra consists of monomials of the form $g_x^{k_x}g_y^{k_y}g_z^{k_z}g_1^{k_1}$, where $k_1 \in \{0, 1\}$ and $k_y$, $k_z$, $k_1\in \N_0$. However, because 
\[\varphi((A_n))^2+\varphi((B_n))^2+\varphi((C_n))^2-\varphi((I_{d_n}))=0\] 
we see that $\varphi$ is indeed well-defined. Then because
\[\|f\left(\pi((A_n)), \pi((B_n)), \pi((C_n))\right)\| \leq \max_{(\lambda_1, \lambda_2, \lambda_3)\in S^2} |f(\lambda_1, \lambda_2, \lambda_3)| = \|f\|_{C(S^2)}\]
we obtain the $\ast$-homomorphism of $C(S^2)$ into $\mathcal A / \mathcal I$.
As before, the existence of nearby commuting matrices that satisfy the relation then implies that $\varphi$ can be lifted.
\end{example}

Expanding on this lifting equivalence discussed by Voiculescu, Loring (\cite{loring1988k}) showed that Voiculescu's and Davidson's counter-examples could be viewed in terms of the $K$-theory of the torus and of the sphere. He showed the relevance of the non-zero second cohomology for these examples to not permit there to be nearby commuting normal matrices. 
As noted at the end of \cite{loring1988k}, this work was done independently of Choi's work on the signature obstruction but both used the same type of obstruction.

In 1989, Exel and Loring (\cite{exel1989almost}) developed a winding number obstruction for almost commuting unitary matrices which shows that there are not any commuting matrices nearby Voiculescu's unitaries. This winding number obstruction and its equivalence to the $K$-theory and other equivalent obstructions have also been studied (\cite{exel1991invariants, exel1993soft, loring2014quantitative}).

For a more in depth algebraic treatment of the lifting method, see \cite{loring1997lifting}.
We will return the this lifting reformulation of almost commuting matrices when discussing the Enders-Shulman theorem in Section \ref{The Enders-Shulman Theorem}.

\section{Davidson's Projection Reformulation}

Davidson in \cite{davidson1985almost} produced two equivalent formulations of the almost/nearly commuting matrix problem for two self-adjoint matrices. 
The second reformulation is an infinite dimensional version of the first reformulation. 
The first reformulation begins with applying a result similar to Lemma \ref{tridiag lemma} to two almost commuting self-adjoint matrices $A, B$ to reduce to the case that $A$ is block-diagonal with diagonal blocks being multiples of the identity and $B$ is block tridiagonal with respect to this structure. 

Davidson then considers restricting $A$ and $B$ to $E_I=E_{I}(A)$ for some intervals $I$: 
\[A_I = AE_I(A), \;\; B_I = E_I(A) B E_I(A)\]
for which $B_I$ is block tridiagonal with many subblocks.
The matrices $A_I$ and $B_I$ have the same block structure, however we will only ask that there is a projection $F_I\leq E_I$ such that $F_I$ contains the smallest eigenvalue eigenspace of $A_I$, is orthogonal to the largest eigenvalue eigenspace of $A_I$, and almost commutes with $B_I$.

The subspace $F_I$ then provides a way to cut the spectrum of $A$ inside $E_I$ without causing a large increase in $\|[A,B]\|$. Suppose that $G_j = (E_{I_j}-F_{I_j})+F_{I_{j+1}}$. We know that $B_j= E_{I_j}BE_{I_j}$ almost commutes with $F_{I_j}$, so it also almost commutes with its complement $E_{I_j}-F_{I_j}$. Likewise, $B_{j+1}$ almost commutes with $F_{I_{j+1}}$. 

Now, if we consider the image of $G_j$ under $B$, we see that there are two components: one coming from each of the summands. The only difference between $B$ applied to these summands and the respective compressions $B_j, B_{j+1}$ is that $B$ also maps the image of one summand partially into the image of the other. However, because both summands are together in $G_j$, we see that $R(G_j)$ is almost an invariant subspace for $B$: $(1-G_j)BG_j \approx 0$. Because $B$ is self-adjoint, this implies that $B$ almost commutes with $G_j$.

In general, suppose that $I_j$ are consecutive intervals so that each $B_{I_j}$ contains $L$ blocks and the projections can be chosen so that $\|[F_I, B_I]\| \leq \delta_F(L)$, where $\delta_F(L)\to 0$ as $L\to \infty$ and that $L\to \infty, |I_j|\to 0$ as $\|[A,B]\|\to 0$.

Then we can define the projections 
\[G_j:\;\; F_{I_1},\; (E_{I_1}-F_{I_1})+F_2,\; \dots,\; (E_{I_{m-1}}-F_{m-1})+F_m,\; E_{I_m}-F_m.\] 
These projections are orthogonal, satisfying 
\[G_j \leq E_{I_j}+E_{_{j+1}},\;\; \sum_j G_j = I,\;\; \|[B,G_j]\|\leq 2\delta_F.\] 
If $a_j$ is the midpoint of $I_j$, define 
\[A' = \sum_j a_j G_{j}, \;\; B' = \sum_j G_jBG_j.\]
Then $A', B'$ are commuting self-adjoint matrices with
\[\|A'-A\| \leq \frac12\max_j|I_j|, \;\; \|B'-B\| \leq 2\delta_F(L).\]
So,  $\|A'-A\|, \|B'-B\|\to 0$ as $\|[A,B]\|\to 0$.
We will make use of a similar type of projection construction in Chapter \ref{7.GradualExchangeProcess}.

Davidson in Theorem 5.2 of \cite{davidson1985almost} used his projection reformulation of the problem to show:
\begin{thm} For $k \in \N$, there exists a function $\varepsilon(-, k)$ with the properties that $\varepsilon(\delta, k)\to0$ as $\delta \to 0$ and if $A \in M_d(\R)$ is a diagonal matrix with increasing diagonal entries, $B \in M_d(\C)$ is a self-adjoint matrix that is $k$-banded (meaning that its entries satisfy $B_{i,j}=0$ if $|i-j|> k$), and $A,B$ are $\delta$-almost commuting then they are $\varepsilon(\delta, k)$-nearly commuting.
\end{thm}
This in particular shows that any counter-example to the problem of the almost/nearly commuting problem for two self-adjoint matrices would have had to be different than the type of counter-examples we discussed earlier. 
This is because the counter-examples all had  one of the matrices being diagonal and the other being a weighted shift matrix which is a type of tridiagonal (i.e. $1$-banded) matrix. 
The only exception is that of Voiculescu's unitaries, however a careful inspection of the proof shows that the contradiction in the proof-by-contradiction is derived by looking at the matrices on a strict subset of the spectrum of the diagonal unitary.

\section{Spectral Surgery}

Davidson also proved the following dilation result: 
\begin{thm} \label{dilationThm}
(Theorem 4.4 of \cite{davidson1985almost})
If $A,B \in M_d(\C)$ are self-adjoint then there exist self-adjoint commuting matrices $C, D \in M_d(\C)$, $A_1, B_1 \in M_{2d}(\C)$ with $\|C\| \leq \|A\|$ and $\|D\| \leq \|B\|$ such that
\[\|A\oplus C-A_1\|, \|B\oplus D - B_1\| \leq 25 \|[A,B]\|^{1/2}.\]
\end{thm}
Although this does not prove that $A, B$ are nearly commuting, it shows they can be embedded into self-adjoint block matrices $A\oplus C, B\oplus D$ that are nearly commuting.

This type of result then formed the basis of a proof of numerical estimates for the Brown-Douglas-Fillmore (BDF) theorem in \cite{berg1991almost}. This 1991 Berg and Davidson result roughly concerns showing that if $S$ is an operator that is essentially normal ($[S^\ast, S]$ is compact) then, subject to the vanishing of certain index obstructions, one can decompose $S$ as a sum $S = N + K$, where $N$ is normal, $K$ is compact, and $\|K\| \leq Const. \|[S^\ast, S]\|^{1/2}$.
The proof given in \cite{berg1991almost} made use of this sort of dilation technique to perform what we will refer to as spectral surgery. 

A normal operator has a spectrum in $\C$ as well as an essential spectrum. We will use the term ``spectral surgery'' to refer to performing punctures, cuts, and deformations of a spectrum of an operator $S$ by performing certain perturbations of $S$ until the spectrum is of a certain amenable form. Note that we may require that cuts and punctures  have a certain minimal size. 
For instance, a cut to the spectrum of $S$ to ``remove'' $\R$ may be done by simply making $\sigma(S)$ have an empty intersection with $\R$ or it may additionally require moving $\sigma(S)$ away from $\R$ by a fixed small distance. Exactly what is needed depends on the context. 

These perturbations of $S$ need to be done without causing certain estimates to be too large. Once the normal operator is perturbed so that the spectrum is of a certain desirable form, other arguments are used to obtain the desired result.
Because Berg and Davidson's numerical BDF theorem concerns compact operators, it turns out that dilation results similar to Theorem \ref{dilationThm} are enough to obtain the result. 

In 1990, Szarek (\cite{szarek1990almost}) improved the dimensional dependence of the almost/nearly commuting self-adjoint matrices problem to $\varepsilon = Const.n^{1/13}\delta^{2/13}$ if both matrices are self-adjoint. He used the projection reformulation of Davidson to reduce the problem to constructing a certain projection. However, he was not able to fully remove the dimensional dependence.  
In \cite{szarek1990almost}, Szarek states that the key consequence of his result is that the problem of two almost commuting self-adjoint matrices is ``completely different'' than the (explicit) counter-examples that existed at the time due to $A, B$ being nearly commuting if $\sqrt{d}\|[A,B]\| \to 0$.

In 1995, Huaxin Lin (\cite{lin1996almost}) showed that two almost commuting self-adjoint matrices are nearby commuting self-adjoint matrices. So, the following result has come to be known as Lin's Theorem:
\begin{thm}\label{Lin'sThm}
There is a function $\varepsilon=\varepsilon(\delta)$ with $\varepsilon(\delta)\to 0$ as $\delta \to 0^+$ so that if $A, B \in M_d(\C)$ are self-adjoint with $\|A\|, \|B\| \leq 1$ then there are commuting self-adjoint matrices $A', B' \in M_d(\C)$ so that
\[\|A'-A\|, \|B'-B\| \leq \varepsilon(\|[A,B]\|).\]
\end{thm}
\begin{remark}
Lin's argument can be summarized as following. Suppose that this result were not true. This would mean that there are matrices $T_n \in M_{d_n}(\C)$ that are almost normal: $\|[T_n^\ast, T_n]\|\to 0$ without being nearly normal. So, there is some $\varepsilon > 0$ so that for any sequence $T_n'$ of normal matrices: $\|T_n' - T_n\| \geq \varepsilon$. We can then package this sequence forming a counter-example into an element $z = (T_n)$ of $\mathcal A$ as defined in Definition \ref{seqSpace}.

Now, one notes the equivalence of Lin's theorem to being able to lift a normal element $z$ in $\mathcal A / \mathcal I$ to a normal element $\tilde z$ of $\mathcal A$, which consists of normal matrices. Lin used various $C^\ast$-algebraic methods to perform spectral surgery on $z$ to obtain a normal element $z' \in \mathcal A/\mathcal I$ with discrete spectrum satisfying $\|z'-z\| < \varepsilon$. One can show that $z'$ can be readily lifted to a normal element $\tilde z'$ of $\mathcal A$. 

If we express $\tilde z' = (T_n')$ then the matrices $T_n'$ are normal so for some $n$ large, it is the case that $\|T_n'-T_n\| < \varepsilon$, which contradicts the assumption that we made that Lin's theorem was false.
\end{remark}

After this result was made, Friis and R{\o}rdam (\cite{friis1996almost}) in 1996 provided a simplified proof of Lin's theorem. The structure has the same outline as we discussed above but the alternate ``$C^\ast$-algebraic methods'' employed by Friis and R{\o}rdam were much simpler.
They also discuss how their method extends to $C^\ast$-algebras with an approximation property of certain types of elements by invertible elements. Consequently, they obtain a version of Lin's theorem for almost commuting self-adjoint elements of a $C^\ast$-algebra $\mathcal B$ with stable rank 1. The assumption that $\mathcal B$ have stable rank 1 (or a more general approximation property) allows one to perform the punctures as part of the spectral surgery using the polar decomposition as discussed in Remark \ref{puncture}.

This simplified method of proving  Lin's theorem has been used to prove various versions of Lin's theorem which we will discuss later in this chapter.

\section{Explicit Estimates}

The proof of Lin's theorem provided by Lin and by Friis and R{\o}rdam were nonconstructive and did not provide explicit control of $\varepsilon = \varepsilon(\delta)$. 
In this section we will focus on the efforts to provide constructions and asymptotic estimates.

Extending Lin's theorem in this direction has garnered interest in recent years (\cite{hastings2009making, hastings2011making, filonov2011relation, kachkovskiy2016distance, herrera2020hastings, li2022vector}).
Hastings from 2008-2011 (\cite{hastings2009making, hastings2011making}) provided argumentation whose goal was showing that there is a function $E(t)$ that increases slower than any positive power of $t$ as $t\to\infty$ so that $\varepsilon(\delta) \leq E(1/\delta)\delta^{1/5}$. However, there were several versions posted to arXiv.org after the original paper was published with the aim of resolving issues with the proof. 
 
Later in 2020, using some very helpful suggestions by Hastings, the author (\cite{herrera2020hastings}) presented a clear exposition of Hastings' approach including details for various claims and resolutions for gaps in the arguments of \cite{hastings2011making}. 
Hastings' argument uses Davidson's reformulation of Lin's theorem in a way similar to Szarek in addition to Lieb-Robinson estimates and bootstrapping Lin's theorem to obtain the asymptotic estimate. The proof is not constructive and does not provide numerical bounds.

Hastings' original published paper (\cite{hastings2009making}) included a constructive proof of Lin's theorem for almost commuting self-adjoint matrices assuming that $A$ is diagonal and $B$ is tridiagonal. The estimate obtained is $\varepsilon(\delta) = E(1/\delta)\delta^{1/2}$ where $E(t)$ is an explicit function that grows slower than any positive power of $t$. This is very close to the optimal estimate discussed in Section \ref{Inequality Scaling}.

The author's work in \cite{herrera2020hastings} also includes a conceptual discussion of several aspects of Hastings, Szarek's, and other prior approaches toward a constructive proof of Lin's theorem, with exception to \cite{kachkovskiy2016distance}. \cite{herrera2020hastings} also includes an outline for how to use Davidson's projection perspective to show that if $H$ and $N$ are almost commuting matrices with $H$ self-adjoint and $N$ normal with spectrum belonging to a nice one dimensional set then $H, N$ are nearly commuting.

Earlier in 2015, Kachkovskiy and Safarov (\cite{kachkovskiy2016distance}) proved that one can choose $\varepsilon(\delta) = Const.\delta^{1/2}$ for almost normal operators of a $C^\ast$-algebra of real-rank zero given that translates of the operator by multiples of the identity are nearby invertible elements. This proof is constructive for matrices and also constructive in general except for invoking the real rank zero and approximation by invertible elements properties.

Kachkovskiy and Safarov's argument uses the same type of spectral surgery as Friis and R{\o}rdam, however this is done without embedding a sequence $(T_n)$ into an abstract $C^\ast$-algebra and applying the quotient map $\pi$ to obtain a normal element. Instead, they use a generalization of Davidson's dilation theorem (Theorem \ref{dilationThm}) to embed an almost normal operator $T$ into a $2\times 2$ block operator matrix that is nearly normal.
This then provides certain normal operators $T_1$ and $N$ so that \[\|T \oplus N - T_1\| \leq Const.\|[T^\ast, T]\|^{1/2}.\]
The proof then proceeds to perform spectral surgery on $T_1$ so as to make it have a discrete spectrum with a certain separation between the points of $\sigma(T_1)$. 

In the Friis and R{\o}rdam proof, it was required to approximate the normal element $z$ in $\mathcal A / \mathcal I$ by another element in this same $C^\ast$-algebra but with a discrete spectrum. Notice that although $z$ can be approximated by a normal element with a discrete spectrum in the von Neumann algebra generated by $z$, one is not guaranteed that this provides an element belonging to the appropriate $C^\ast$-algebra $\mathcal A / \mathcal I$. 

The methods used involve using the polar decomposition to pop holes in the spectrum, the continuous functional calculus to deform the spectrum to a square netting shape, then using unitaries in $\mathcal A$ to cut the spectrum of $z$ along one-dimensional sets. Then the continuous functional calculus is used to contract the spectrum into a discrete set.
All these modifications of $z$ can be done in the $C^\ast$-algebra $\mathcal A/ \mathcal I$ to obtain the desired normal $z'$.

Likewise, Kachkovskiy and Safarov's argument needed to perform spectral surgery on $T_1$ while maintaining the fact that the upper-left corner block is still approximately equal to $T$ and that this corner block is approximately normal. If $T_1$ has a discrete spectrum with a certain spacing between the elements of $\sigma(T_1)$, then the upper-left corner block can be perturbed to a normal operator which approximates $T$.

To do this, one formulates how close $T_1$ is to being block diagonal by maintaining the commutator of $T$ with the block matrix $P = I \oplus 0$ small throughout the spectral surgery operation. 
So, one can use the continuous functional calculus for continuous deformations of the spectrum if the deformations are smooth enough so as to not cause the commutator with $P$ to become too large.

However one cannot simply use the polar decomposition to pop holes of a non-trivial size in the spectrum of $T_1$ nor can we simply cut the spectrum of $T_1$ when it is one dimensional as is done in the Friis-R{\o}rdam proof. The reason is that the method of performing these non-continuous spectral alterations of $T_1$ introduce perturbations of $T_1$ that may have large commutators with $P$. 

The resolution to this issue is to use the fact that at any stage of this process we maintained that $\|[T,P]\|$ is small. Hence, there is a diagonal operator $\diag_P(T) = PTP + (1-P)T(1-P)$ nearby $T$. Since $\diag_P(T)$ commutes with $P$, we can obtain the polar decomposition of this operator which will commute with $P$ as well. 
Each of the popping holes and cutting is then essentially done by grafting in the unitary-part of the polar decomposition of $\diag_P(T)$ into $T$.

There are also two other modifications that complicate the proof but improve the result. The first and most prevalent throughout the paper is that the $C^\ast$-algebra of real rank zero is not assumed to be a von Neumann algebra so it is necessary to assume a certain local approximation property by invertible operators and maintain this throughout the spectral surgery. Additionally, the proof carefully makes sure that the many local transplants performed as part of  the spectral surgery can be done without the final estimates depending on the number of punctures and cuts made. This ensures that the optimal asymptotic exponent of $1/2$ is obtained in
\[\|N-T\|\leq C_{KS}\|[T^\ast, T]\|^{1/2}.\]

\section{Almost Representations}

Having discussed almost normal matrices, we now discuss some of the general theory concerning when a fixed number of almost commuting matrices satisfying some relations are nearly commuting.

Suppose that we have almost commuting self-adjoint matrices $A_1, \dots, A_m$ approximately satisfying some constraints
\[\|A_j\| \leq M_j, \;\; M_j \in (0, \infty],\]
\[\|p(A_1, \dots, A_m)\|\approx 0, \;\; p \in \mathscr P,\] 
where $\mathscr P$ is a collection of functions $p$ of $m$ self-adjoint matrices that satisfy 
\[p(\diag_i(a_i^1), \dots, \diag_i(a_i^m))) = \diag_i(p(a_i^1, \dots, a_i^m)),\]
\[p(U^\ast A_1U, \dots, U^\ast A_mU)=U^\ast p(A_1, \dots, A_m)U\]
for $U$ unitary, and are continuous in the operator norm on bounded sets, independently of the size of the matrices.
This applies for $p$ being a polynomial in $A_1, \dots, A_n$ but also for example $p(A, B) = |A|+AB-B^2$.

With this presentation, the question of whether an almost-nearly commuting matrix problem has a positive answer can be expressed in terms of the matricial semiprojectivity of $C(X)$ for some compact subset $X$ of $\R^m$. 
The space $X$ is the smallest set that contains the joint-spectrum of any commuting matrices $A_1', \dots, A_m'$ that satisfy these relations. So, 
\[X = \{x \in \R^m: p(x)=0, p \in P; |x_j| \leq M_j\}=: \mathscr P \cap M\]
endowed with the topology inherited from $\R^m$. Because the functions $p$ are continuous, $X$ is closed (but perhaps not compact, depending on the functions $p$). If all the $M_j$ are finite then $X\subset \R^m$ is compact.

\begin{defn}
Following the terminology of \cite{hastings2010almost}, we say that the self-adjoint matrices $A_1, \dots, A_m$ $\delta$-almost represent $X = \mathscr P \cap M$ if $\|[A_i, A_j]\| \leq \delta$, $\|A_j\| \leq M_j$, and $\|p(A_1, \dots, A_m)\| \leq \delta$ for all $p \in \mathscr P$.
\end{defn}
\begin{example}
One of the simplest examples of this is $X= S^1$, the unit circle in $\C$. Since $S^1 = \{z \in \C: \overline{z}z=1\}$, the self-adjoint matrices $A, B$ can be said to $\delta$-represent the unit circle if 
\[\|[A,B]\|\leq \delta, \;\;\; \|A\| \leq M_A, \|B\| \leq M_B, \;\;\;  \|(A+iB)^\ast(A+iB)-I\| \leq \delta.\]
This can also be reformulated more simply for $T = A+iB$ being almost unitary: 
\[\|[T^\ast, T]\|\leq 2\delta, \;\;\;  \|\Re(T)\| \leq M_A, \|\Im(T)\|\leq M_B, \;\;\; \|T^\ast T - I\| \leq \delta.\]
It turns out that the last inequality is all that is needed since if we perturb $T$ to a matrix $T'$ satisfying $(T')^\ast T' = I$ then $T'$ is automatically normal so $[A',B']=0$ and $T'$ automatically has norm $1$. The norm restriction through $M$, which is not needed in this example, is only required in order to guarantee $T'$'s existence.

We can solve this approximation problem by using the singular value decomposition 
$T = U \Sigma V^\ast$ and $T^\ast T = V \Sigma^2 V^\ast$ so
\[\|T^\ast T - I\|=\|\Sigma^2-I\| = \max_i |\sigma_i^2-1|.\]
So, if we set $T' = UV^\ast$ then since $\sigma_i+1 \geq 1$:
\[\|T'-T\| = \|I-\Sigma\| =\max_i |1-\sigma_i| \leq \max_i |(\sigma_i+1)(\sigma_i-1)|=\|T^\ast T - I\|.\]
Hence, we see that almost unitary matrices are nearby unitary matrices and that almost representations of $S^1$ are nearby actual representations.  
\end{example}
\begin{remark}
Note that we could have instead required $\||A+iB|-I\| \leq \delta$. 
Attempting to solve this alternative form of the almost unitary matrix problem using the same method as we did above will cause one to notice the fact that changing the functions in $\mathscr P$ may produce different estimates for how close the nearby commuting matrices are even if $X=\mathscr P \cap M$ is unchanged.
\end{remark}

Much research has been done related to the stability of relations for elements of $C^\ast$-algebras. For instance, 
\cite{akemann1977ideal,
loring1993c,
loring1996stable,
loring1989noncommutative,
loring2013lifting}. The almost representation of a compact set is deeply related to this problem. 

For instance, consider the relation $(V_i-1)(V_j-1)=0$ for  finitely many unitaries $V_i, V_j$. Part of what Loring showed in the 1989 paper \cite{loring1989noncommutative} was that this relation is stable in the sense that if $\max_{i,j}\|(V_i-1)(V_j-1)\|$ is small then there exist nearby unitaries $V_i, V_j$ such that $(V_i-1)(V_j-1)=0.$  
 
Consider the topological space $X = S^1 \vee \cdots \vee S^1$ of finitely many circles connected together at a single point. This can be embedded in $\C^n$ by viewing each copy of $S^1$ as the unit circle centered at $-1$ in one of the complex axes: $\{(0, \dots, 0, z-1, 0, \dots, 0): |z|=1\}$ so that the copies of $S^1$ are orthogonal and all intersecting at the origin. 
Note that $\|[V_i, V_j]\| \leq 2\|(V_i-1)(V_j-1)\|$. So, Loring proved that $C(X)$ for $X = S^1 \vee \cdots \vee S^1$ is matricially semiprojective. As noted by \cite{enders2019almost} (resp. \cite{friis1996almost}), Loring in \cite{loring1989noncommutative} (resp. in \cite{loring1996stable}) also showed that any $1$-dimensional CW-complex is matricially semiprojective.
 
As discussed previously, the $2$-torus and $2$-sphere were shown to not have the property that almost representations are nearby actual representations. 
Lin showed that the rectangle $[-1,1]^2$ (or equivalently the disk in $\C$) does.
Loring in 1996 (\cite{loring1998matrices}) showed that three self-adjoint matrices forming an almost representation of $S^2$ are nearby an actual representation when a certain obstruction vanishes.

Eilers, Loring, and Pedersen proved in 1996 (\cite{eilers1999morphisms}) that representations of the two dimensional non-orientable 2-manifold $X=\R P^2$ embedded in $\R^4$ as 
\[X = \{(z,w)\in \C^2: w^2 = (1-|z|)z\}\]
are stable, as well as some other non-orientable $2$-manifolds gotten by gluing together multiple points of the boundary of the unit disk. They also showed that any two almost commuting unitaries whose Exel-Loring winding number obstruction vanishes are nearby commuting unitaries.
The method of proof of these results was based on applying new results concerning commuative diagrams of $C^\ast$-algebras to the Friis and R{\o}rdam's proof of Lin's theorem.

Around the same time, Gong and Lin (\cite{gong1998almost}) extended Lin's almost multiplicative morphisms approach for $C(X)$, where $X$ is a $2$-dimensional compact metric space. This implied Lin's theorem as well as the result from \cite{eilers1999morphisms} for almost commuting unitaries for $C^\ast$-algebras that have real rank zero, stable rank 1, and some other conditions which hold for matrix algebras.

In 2009, Osborne (\cite{osborne2009almost}) derived numerical estimates for almost commuting unitary matrices when both matrices contain a gap in their spectrum. This was done by computing how the matrix logarithm which transforms each of these unitary matrices into a self-adjoint matrix affects the norms and commutators.

Hastings and Loring (\cite{hastings2010almost}) in 2010 explored different geometries for almost representations. They applied explicit transformations and Hastings' result in \cite{hastings2011making} to obtain some asymptotic estimates for how close almost representations of the rectangle, disk, annulus, cylinder, and sphere were to actual representations. 

The result for almost representations of the rectangle $[-1,1]^2$ is just Lin's theorem. An almost representation of the annulus is given by an almost normal matrix $T$ with $\|T\| \leq 1$ and $\|T^{-1}\| \leq 2$ and an almost representation of the cylinder is given by a self-adjoint matrix $H$ and a unitary matrix $U$ that almost commute.

Explicit geometric and algebraic transformations were used to translate between the almost and actual representations of these different spaces. An important requirement is the commutators of the appropriate resulting matrices have small norm so that $\delta_1$-almost representations are transformed into $\delta_2$-almost representations where $\delta_2 \to 0$ as $\delta_1 \to 0$.
For instance, the transformation between the annulus and the cylinder is gotten by using the polar decomposition which corresponds to the geometric transformation of converting an annulus into a cylinder using polar coordinates.

The more interesting transformation concerns transforming between the cylinder and the sphere. Each of the cylinder's boundary circles could be squeezed to a point to obtain an almost (resp. actual) representation of the sphere from an almost (resp. actual) representation of the cylinder. The reverse transformation required the vanishing of the obstruction as shown in \cite{loring1998matrices}.
They, however, were unable to find a transformation that could transform an almost representation of the torus into an almost representation of the cylinder due to the required commutator estimates not being true.

\begin{remark}\label{degenerate2}
In line with what is said in Remark \ref{degenerate}, one can see by Propositions 5.2 and 5.4 of \cite{hastings2010almost} (resp. Lemma 3.4 of \cite{exel1991invariants}) that when the matrices $A, B, C \in M_{d}(\C)$ almost represent the torus (resp. sphere) have commutators that have norm $o(1/d)$ as $d \to \infty$ then the obstruction to $A, B, C$ being nearly commuting vanishes.

More generally, Corollary 9.4 of \cite{herrera2020hastings} shows that the almost commuting matrix problem for three self-adjoint matrices matrices $A, B, C \in M_{d}(\C)$ is degenerate if the commutator is $o(1/d)$ in the sense that there are nearby commuting matrices whereas in general three almost commuting self-adjoint matrices are not nearly commuting.
\end{remark}

\section{The Enders-Shulman Theorem}
\label{The Enders-Shulman Theorem}

As a culmination of the problem of what almost representations are nearby actual representations, Enders and Shulman proved the following:
\begin{thm} (\cite{enders2019almost}) \label{ES}
Let $X$ be a compact metric space with finite covering dimension $\dim X$. Let $H^2(X,\Q)$ be the second \v Cech cohomology group of $X$ with rational coefficients. We will refer to $H^m(X,\Q)$ as the ``rational cohomology'' of $X$.

Then $C(X)$ is matricially semiprojective if and only if $\dim X \leq 2$ and $H^2(X, \Q) = 0$.
\end{thm}
An interesting theorem which served as part of the proof that $\dim X \leq 2$ for the forward-direction is
\begin{thm} (\cite{enders2019almost}) Let $X$ be a compact metric space with finite covering dimension $\dim X$. 

Then $C(X)$ is matricially semiprojective if and only if $C(Y)$ is matricially semiprojective for all closed subsets $Y$ of $X$.
\end{thm}
\begin{example}
To illustrate how this is natural, suppose for the sake of illustration that we did not know that $X=[-1,1]^3$ is not matricially semiprojective but we did know that $S^2$ is not. Let $A, B, C$ be an almost representation of the sphere which  almost represent $[-1,1]^3$ with $A^2+B^2+C^2 \approx I$. 
Suppose that there are always nearby commuting self-adjoint matrices $A', B', C'$ with joint spectrum in $[-1,1]^3$. Since $\|A'-A\|, \|B'-B\|, \|C'-C\|$ are small, $(A')^2+(B')^2+(C')^2 \approx I$. 

Consequently, the joint spectrum of $A', B', C'$ is nearby the unit sphere so we can perturb the eigenvalues of these matrices so that the matrices still commute and that they actually are a representation of the sphere. 
This contradicts the fact that we know that not all almost representations of the sphere are nearby actual representations of the sphere. So, $[-1,1]^3$ is not matricially semiprojective. 

Note that this is the same sort of argument that we discussed earlier in this section to show that Voiculescu's unitaries not being nearby commuting unitaries (Theorem \ref{Voiculescu}) implies that four almost commuting self-adjoint matrices are not nearly commuting in general.  
\end{example}

The result of Enders and Shulman answers the question of whether any almost representation is nearby an actual representation for all the geometric examples we have seen thus far. For reference, we include a few results concerning the rational cohomology which cover all the cases of interest for us.

The following definition and properties 
are from Appendix E of \cite{bredon1993topology}.
\begin{defn}
A subspace $X \subset \R^n$ is said to be a
Euclidean Neighborhood Retract (ENR) if it is the retract of some open neighborhood of $X$.
\end{defn}
Being an ENR is an intrinsic property which is equivalent to $X$ being locally compact and locally contractible.
Any manifold and any finite CW complex is an ENR.
If $X$ is an ENR then the \v Cech and singular cohomologies are isomorphic. 

If $X$ is a topological space embedded in $\R^n$ then its rational singular cohomology vanishes if the rational singular cohomologies of its connected components vanish (Theorem V.8.4 of \cite{bredon1993topology}). It is well known that the rational singular cohomology is invariant under homotopy equivalences.
If $X$ is a compact connected $m$-manifold then $H^m(X,\Q) = \Q$ if $M$ is orientable and $H^m(X,\Q)=0$ if $M$ is not orientable (Theorem VI.7.14 of \cite{bredon1993topology}).
If $X$ is a compact proper subset of an orientable 2-manifold then $H^2(X,\Q)=0$ (Theorem VI.8.5 of \cite{bredon1993topology}).
\begin{example}
These general facts are enough to determine the rational cohomology for all the examples that we have considered.
For instance, Osborne (\cite{osborne2009almost}) showed that an almost representation of 
\[X=\{e^{i\theta}: |\theta| \leq \theta_0\}\times \{e^{i\theta}: |\theta| \leq \theta_1\} \subset S^1 \times S^1 = T^2\]
for $\theta_0, \theta_1 < \pi$ is nearby an actual representation, where the estimate depends on how close $\theta_0, \theta_1$ are to $\pi$. Noting that $X$ is a compact proper subset of the 2-torus, we can independently deduce that $\dim X \leq 2$ and $H^2(X, \Q)=0$ so $C(X)$ is matricially semiprojective by Theorem \ref{ES}.
\end{example}
\begin{example}
The author in \cite{herrera2020hastings} used Davidson's reformulation of Lin's theorem to show that if $N$ normal and $H$ self-adjoint are almost commuting with $\sigma(N)$ belonging to a nice $1$-dimensional compact set $S\subset \C$ then $N, H$ are nearly commuting. 
Because $X=S\times[-1,1]$ is homotopic to $S$, which is $1$-dimensional, we have $H^2(X,\Q)=H^2(S, \Q)=0$ and $\dim X = 2$. So representations of $X$ are matricially stable by Theorem \ref{ES}. 

However, the proof in \cite{herrera2020hastings} provides a way to obtain numerical estimates based on the geometry of $S$. One can see from the discussion presented how the estimate worsens as $S$ becomes ``less one dimensional'' which is not something captured in results that do not provide explicit estimates.

Using this approach, an asymptotic estimate is gotten for two almost commuting unitaries where only one of the matrices has a spectral gap. Again, the non-constructive version of this result follows simply by knowing the topological properties of 
$X=\{e^{i\theta}: |\theta| \leq \theta_0\}\times S^1$. Another way to obtain an asymptotic estimate is to observe that the space $X$ can be embedded in $\R^3$ and can be transformed into the cylinder.
\end{example}


\section{Etc.}

Relevant for our work in this thesis, Hastings and Loring (\cite{hastings2010almost}) explored almost representations of the sphere derived from the irreducible spin representations of $su(2)$, which is essentially Choi's example from Section \ref{Dimensional-Dependent Counter-Examples}. This is discussed briefly in Example \ref{repEx}. These examples play a central role in this thesis.

Motivated by problems in physics, various authors have explored structured almost commuting matrices (\cite{hastings2010almost, hastings2011topological, loring2016almost, loring2013almost, loring2014almost, loring2014quantitative, loring2015k}). Many of the arguments made in these papers are similar to the arguments that we have discussed, with appropriate modifications made to make use of the structure of the given almost commuting matrices. 

For instance, 
\cite{loring2016almost} showed that a real almost normal matrix is nearby a real normal matrix by investigating the lifting of certain $C^\ast$-algebras which in addition to including the adjoint $\ast$ they also include the trace (since conjugation is the transpose of the adjoint). They then showed that two almost commuting real self-adjoint (alias ``real symmetric'') are nearby commuting real self-adjoint matrices.
\cite{hastings2010almost} and other papers explored index obstructions for these structured almost commuting matrices. 

We will discuss Ogata's theorem for almost commuting macroscopic observables more in Chapter \ref{3.MathematicalPhysicsOfAlmostCommutingObservables} after we define the matrices for which this result applies. This result provides a non-trivial example of more than two almost commuting matrices that are nearly commuting. Moreover, these matrices have physical meaning in quantum mechanics. Ogata's proof of her theorem in \cite{ogata2013approximating} uses thermodynamical properties of the given matrices together with some of the algebraic arguments from Lin's original paper (rather than the simplified approach of Friis and R{\o}rdam). 

\chapter{Mathematical Physics of Almost Commuting Observables}
\label{3.MathematicalPhysicsOfAlmostCommutingObservables}

One can see the references
\cite{pade2018quantum, pade2018quantum2, ludyk2018quantum}
for gentle introductions to Quantum Mechanics for those without a deep physics background,
\cite{hall2013quantum, busch1985note}
for treatments involving more advanced mathematical formalisms, and
\cite{hayashi2017group, woit2017quantum}
for treatments with an emphasis on applications of representation theory.

We will include in the first four sections of this chapter a brief review of the basics of a matrix formulation of quantum mechanical states and measurement of observables. We do this to motivate the construction of macroscopic observables and to discuss the physical significance of there being nearby commuting observables for them. Latter sections intertwine mathematical and physics results about the uncertainty principle, uncertainty relations, almost commuting observables, and macroscopic observables. 

By the end of this chapter, we will have discussed all the relevant background for Ogata's theorem and Theorem \ref{OgataTheorem}. The remaining chapters are devoted to proving Theorem \ref{OgataTheorem}.

\section{Basics of Quantum States and Measurement}

We begin with the notion of a state. 
A state is a representation of everything that can be known about a system. An observable, loosely speaking, is something that one can measure (i.e. it is able to be observed). This typically includes position, momentum, energy (through the Hamiltonian operator), spin, photon polarization, etc. Because the systems we are particularly interested in are finite dimensional, we will typically work with observables that have finitely many possible measured values.

The way that this matrix formulation will work is that the states and observables will be represented by vectors and matrices, respectively. 
\begin{defn} A pure state is represented by a unit vector $\psi$ in $\C^d$ which is called the wavefunction.  
We do not distinguish between the states of two wavefunctions $\psi_1, \psi_2$ if there is a phase $\omega = e^{i\theta}$ such that $\psi_2 = \omega \psi_1$. 

More generally, a pure state is represented by a unit vector in an infinite dimensional Hilbert space $\mathcal H$.
(In the infinite dimensional setting of $\psi$ belonging to the separable Hilbert space $L^2(\R)$, the terminology of ``wavefunction'' is more natural.)

An observable is represented by a self-adjoint matrix in $M_d(\C)$. More generally, an observable is represented by a bounded (or unbounded) self-adjoint operator on an infinite dimensional Hilbert space. 

We may conflate a state with a unit vector $\psi$ representing it and we may conflate an observable with the self-adjoint operator $A$ representing it.
\end{defn}

The possible observed values of an observable, irrespective of what the state is, are the eigenvalues of the associated matrix. That is, the set of all possible measurements of an observable $A$ is the spectrum of $A$. If the state $\psi$ is an eigenvector of $A$ with eigenvalue $\lambda$ then whenever $A$ is observed the value $\lambda$ will be measured and the state is unchanged after measurement. We refer to such a vector $\psi$ as an eigenstate of $A$ and since the value $\lambda$ will always be observed, we say that $A$ has a definite value in the state $\psi$.

If $\psi$ is not an eigenvector of $A$ then consider the eigendecomposition 
\[\psi = c_1 \psi_1 + \cdots + c_k\psi_k\] 
where the $\psi_j$ are orthonormal eigenvectors of $A$ with distinct eigenvalues $\lambda_j$ and $c_j \neq 0$. This linear combination of unit eigenvectors is referred to as the state being a superposition of eigenstates. The vectors $\psi_j$ are the normalized projections of $\psi$ onto the eigenspaces of $A$.

Because $\|\psi\|=1$, the terms $|c_j|^2$ add to $1$. When $A$ is observed, there is a value $j=j_0$ such that the value $\lambda_{j_0}$ is measured and the state  becomes $\psi_{j_0}$ after the measurement. The measured value of $\lambda$ is random and has value $\lambda_{j_0}$ with probability $|c_{j_0}|^2$. Note that this is a well-defined notion of probability because $\sum_j |c_j|^2=1$.

\begin{example}
The change of the state from $\psi$ to one of the vectors $\psi_j$ is referred to as the ``collapse of the wavefunction.'' We illustrate how this description is fitting (even for finite dimensional systems).

Suppose that we view $\psi$ as being identified with the graph in $\{1, \dots, k\}\times\C\subset \R^3$ of a function on $\{1, \dots, k\}$ into $\C\cong \R^2$ with value $c_j$ at $j$. 
The eigenstates $\psi_j$ are represented as a peaked function supported on $\{j\}$ with magnitude $1$. Note that rotating the entire wavefunction $\psi$ about the argument axis $\R \supset \{1, \dots, k\}$ is equivalent to multiplying the wavefunction by a phase and hence does not change the state. However one cannot rotate the individual components of a wavefunction without changing the state since the ``shape'' of the wave has changed.

With this perspective, after the measurement, the state $\psi$ collapses horizontally from a graph in $\{1, \dots, k\}\times \C$ to being zero everywhere except at a single value of $j$ and with absolute value $1$. So, the graph condenses from a wave in $\{1, \dots, k\}\times \C$ whose amplitudes have squares summing to $1$ to a wave in $\{1, \dots, j_0-1, j_0+1, \dots, k\}\times \{0\} \cup \{j_0\}\times \C$ with a single amplitude of $1$.
\end{example}

\begin{defn}
Because the measurement of $A$ in the state $\psi$ is random, we can speak of the expected value of this measurement defined as:
\[\langle A \rangle_\psi =\sum_{j=1}^k\lambda_j|c_j|^2.\] 
\end{defn}
If we extend $\psi_1, \dots, \psi_k$ to an orthonormal eigenbasis of $A$ and use these as the columns of the unitary matrix $U$ then $U^\ast A U$ is diagonal with the first $k$ entries being $\lambda_1, \dots, \lambda_k$. We then see that \[\langle A \rangle_\psi =\langle c_1\psi_1 \dots+ c_k\psi_k, c_1\lambda_1\psi_1 \dots+ c_k\lambda_k\psi_k \rangle = \langle \psi, A\psi \rangle.\] 

Notice that if $\psi$ is an eigenvector of $A$ with eigenvalue $\lambda$ then both of the above expressions show that $\langle A \rangle_\psi = \lambda$. As discussed, in general $\psi$ is the superposition of orthogonal eigenstates of $A$ with varying measured values $\lambda_j$ and weights $|c_j|^2$. 
\begin{defn}
The standard deviation $\Delta_\psi A\geq 0$ and variance $(\Delta_\psi A)^2$ are defined by 
\[\left(\Delta_\psi A\right)^2 = \langle (A-\langle A\rangle_\psi)^2\rangle_\psi \geq 0.\] 
\end{defn}
The variance is the expected value of the square of the distance between the measurement and the expected value of the measurement. One can verify that
\[\left(\Delta_\psi A\right)^2  = \langle A^2\rangle_\psi - \langle A\rangle_\psi^2.\]
If $\psi$ is an eigenstate of $A$ then the measurement will be the same as the expected value so $\Delta_\psi A = 0$ as can be checked using the definition above.

The variance can also be seen as a measure of how localized the eigendecomposition of the wavefunction is, where eigenstates are viewed as far away if their eigenvalues are far away.

\section{Density Matrices}

If $P_\psi$ is the rank $1$ projection which projects onto the span of a pure state $\psi$, we see that 
\[\langle \psi, A\psi \rangle = \langle P_\psi\psi, AP_\psi\psi \rangle = \langle \psi, P_\psi AP_\psi\psi \rangle.\]
If $e_1 = (1, 0, \dots, 0)^T \in \C^d$ and $V$ is a unitary whose first column is $\psi$, then 
\[\langle \psi, P_\psi AP_\psi\psi \rangle = \langle Ve_1, P_\psi AP_\psi Ve_1 \rangle = \langle e_1, (V^\ast P_\psi AP_\psi V)e_1 \rangle\]
is the $(1,1)$ entry of the matrix $V^\ast P_\psi AP_\psi V$.
If $v_j$ are the columns of $V$ then the columns of $P_\psi V$ are $P_\psi v_j$. Consequently, $P_\psi V$ has as its first column $\psi$ and the zero vector as its other columns. Likewise, the columns, except possibly the first, of $(V^\ast P_\psi A)P_\psi V$ are zero. So, 
\[\langle e_1, (V^\ast P_\psi AP_\psi V)e_1 \rangle = \Tr[V^\ast P_\psi AP_\psi V] = \Tr[P_\psi A].  \]

From these calculations, we see that
\begin{equation}\label{expectation}
\langle A \rangle_\psi = \langle \psi, A \psi\rangle = \Tr[P_\psi A].
\end{equation}
Further, if $U$ is unitary such that $UAU^\ast = \diag(a_i)$ is diagonal, then $f(A) = U^\ast \diag(f(a_i))U$ for $f:\R\to\C$. So, 
\[\sum_{j=1}^kf(\lambda_j)|c_j|^2 = \Tr[P_\psi f(A)].\] 
This is the expected value of the random variable $f(\lambda_j)$. If $f$ is real-valued then $f(A)$ is an observable and hence $\Tr[P_\psi f(A)]$ is the expected value of $f(A)$.

In particular, if $\chi_\Omega$ is the function that equals $1$ on $\Omega \subset \C$ and zero elsewhere then $\chi_{\Omega}(A) = E_{\Omega}(A)$ and hence the probability of observing a value in $\Omega$ for a measurement of $A$ in the state $\psi$ is
\[\sum_{j: \lambda_j \in \Omega}|c_j|^2 = \langle \chi_\Omega(A) \rangle_\psi = \Tr[P_\psi E_\Omega(A)].\]
This is the familiar fact from probability theory that the expected value of the indicator function of an event is the probability of that event.

\begin{defn}
We say that $\rho \in M_d(\C)$ is a density matrix if $\rho \geq 0$ and $\Tr[\rho]=1$.
\end{defn}
We already showed that the expected value of $A$ in the pure state $\psi$ can be expressed as $\Tr[\rho A]$ where the density matrix is the rank $1$ projection $\rho = P_\psi$.

\begin{example}
Let $\psi_1, \dots, \psi_r$ be any collection of states and $p_1, \dots, p_r \geq 0$ be some non-negative numbers with sum $\sum_j p_j=1$.

Suppose that one constructs a mixture of quantum and classical probabilities by choosing the quantum state $\psi_j$ with probability $p_j$ then measuring the observable $A$. The expected value of this measurement is expressible as
\[\sum_{j=1}^r p_j\langle A\rangle_{\psi_j} = \Tr\left[\left(\sum_{j=1}^r p_j P_{\psi_j}\right) A\right] = \Tr[\rho A],\]
where $\rho = \sum_{j=1}^r p_j P_{\psi_j}$ is a density matrix. 
Thus, this scenario provides the expected value of an observable in terms of a density matrix. 
\end{example}
\begin{defn} We now identify the density matrix $\rho$ as the general definition of a state of a system with
\[\langle A\rangle_\rho = \Tr[\rho A], \;\; \Delta_\rho A = \sqrt{\langle (A-\langle A\rangle_\rho)^2\rangle_\rho}.\]

If $\rho = P_\psi$ for some unit vector $\psi$, then we say that $\rho$ is a pure state. Otherwise, $\rho$ is a mixed state.
\end{defn}
\begin{remark} 
As discussed above, this definition of a state in terms of a density matrix includes the case of a classical ensemble of a quantum system with different states.
This definition also appears in other contexts, including when one wants to describe the state of a subsystem where the larger system has a state that is represented by a density matrix.
In the infinite dimensional setting a state can be represented by a density operator which is a positive self-adjoint compact (trace-class) operator $\rho \geq 0$ such that $\Tr[\rho]=1$, where $\Tr[\rho]$ is the sum of the eigenvalues of $\rho$, counted with multiplicity. See \cite{hall2013quantum} for the definition so that this produces a well-defined way of measuring the expected value of an observable.
\end{remark}

We make a few observations about properties of density matrices. First, this definition of a state for $\rho = P_{\psi}$ removes the phase-invariance ambiguity of the pure state due to thinking of a state as a vector in $\C^d$. 

The following result shows that there is no ambiguity in thinking about a state in terms of the density $\rho$ or in terms of the expected values of measurements of observables:
\begin{prop} A state represented by the density matrix $\rho \in M_d(\C)$ is uniquely determined by the expectation values of $\Tr[\rho A]$ for all observables $A \in M_d(\C)$. 
\end{prop}
\begin{proof}
Let $\rho_1, \rho_2$ be two density matrices which produce the same expected values for any observable. Then since $A = \rho_1-\rho_2$ is an observable, we calculate
\[0=\Tr[A\rho_1]-\Tr[A\rho_2] =\Tr[A(\rho_1-\rho_2)]=\Tr[(\rho_1-\rho_2)^2].\]
However, because $(\rho_1-\rho_2)^2 \geq 0$, we deduce that $\rho_1-\rho_2=0$ as desired.
\end{proof}

We now state some elementary bounds for the expected value.
If $A \geq 0$ then it is a fact that $\rho^{1/2}A\rho^{1/2} \geq 0$ so 
\[\Tr[\rho A] = \Tr[\rho^{1/2}A\rho^{1/2}] \geq 0.\] 
This implies that if $\sigma(A)$ belongs to an interval $[a,b]$ then 
\[\Tr[\rho A] \in [a,b]\] as well. This is a natural property that we should expect: If a measurement has possible values in the convex set $[a,b]$ then the expected value of that measurement should also be in that set. 
An immediate consequence of this property is that
\[|\langle A\rangle_\rho| \leq \|A\|.\]
We also have a bound for the standard deviation:
\[\Delta_\rho A = \sqrt{\langle A^2\rangle_\rho - \langle A\rangle_\rho^2} \leq \sqrt{\langle A^2\rangle_\rho}\leq \|A\|.\]

\section{Composite Systems}

We will consider the situation of states and observables on composite systems. Suppose that a composite system consists of one subsystem whose (pure) states are in $\C^{d_1}$ and another subsystem whose (pure) states are in $\C^{d_2}$. The states for the composite system are then in $\C^{d_1}\otimes \C^{d_2} = \C^{d_1d_2}$. 
The simplest of the states of the composite system are the so-called product states, which are tensor products of states of the subsystems considered:
\begin{defn} A state $\psi \in \C^{d_1d_2}$ is a (pure) product state if there are (pure) states $\phi \in \C^{d_1}$, $\varphi \in\C^{d_2}$ such that $\psi = \phi \otimes \varphi$. A state on the composite system that is not a product state is referred to as an entangled state. 
\end{defn}
Because the tensor product of orthonormal bases on each space provides an orthonormal basis for the tensor product space, entangled states are linear combinations of product states.

\begin{defn}
If $A \in M_{d_1}(\C)$ is an observable on the first subsystem then the associated observable on the composite system is $A \otimes I_{d_2}$. Likewise, for an observable $B \in M_{d_2}(\C)$ on the second subsystem, the associated composite system observable is $I_{d_1}\otimes B$. 
\end{defn}
We now discuss why this is a natural definition that captures some of the same properties that one might expect of measurements of isolated subsystems based on our experience with classical mechanics.

First consider measuring the observable $A \otimes I_{d_2}$ in the product state $\phi \otimes \varphi$. Let $\phi_1, \dots, \phi_k \in \C^{d_1}$ be eigenstates of $A$ with distinct eigenvalues $\lambda_j$ such that $\phi = c_1\phi_1 + \cdots + c_k\phi_k$. Extend $\phi_1, \dots, \phi_k$ to an orthonormal basis of $\C^{d_1}$ by adjoining the vectors $\phi_{k+1}, \dots, \phi_{d_1}$. Extend $\varphi$ to an orthonormal basis of $\C^{d_2}$ by adjoining the vectors $v_2, \dots, v_{d_2}$ and define $v_1 = \varphi$. Then $\phi_i \otimes v_j$ form an eigenbasis of $\C^{d_1d_2}$ for $A \otimes I_{d_2}$ and $\phi \otimes \varphi = c_1\phi_1\otimes\varphi + \cdots + c_k\phi_k\otimes\varphi$. 

Upon measurement of $A\otimes I_{d_2}$, we obtain a value $\lambda_j$ with probability $|c_j|^2$ with new state $\phi_j \otimes \varphi$.
So, the effect of measuring $A \otimes I_{d_2}$ in the product state $\phi \otimes \varphi$ is equivalent to measuring $A$ in the state $\phi$ (but with tensoring the obtained state by $\varphi$).
This indicates that measuring $A$ on the first subsystem does not change the state with respect to the second subsystem. Moreover, the measurements and their probabilities are the same as if we were viewing the first subsystem as an independent system. 

The analogous statement holds for the second subsystem. These observations and the fact that $A \otimes I_{d_2}$, $I_{d_1}\otimes B$ commute indicates that performing measurements within each of these subsystems are independent of each other.

Product states also exhibit a property analogous to probabilistic independence with respect the measurement of the commuting observables $A \otimes I_{d_2}$ and $I_{d_1} \otimes B$. Write $\phi = c_1\phi_1 + \cdots + 
c_k\phi_k$ and $\varphi = d_1\varphi_1 + \cdots + d_r \varphi_r$, where the $\phi_i$ are eigenvectors of $A$ with distinct eigenvalues $\lambda_i$ and the $\varphi_j$ are eigenvectors of $B$ with distinct eigenvalues $\mu_j$.

Consider the probability of a particular sequence of measurements where we first measure $A \otimes I_{d_2}$ then measure $I_{d_1} \otimes B$. First, we measure value $\lambda_{i_0}$ with probability $|c_{i_0}|^2$ and the state becomes $\phi_{i_0} \otimes \varphi$ then we measure $\mu_{j_0}$ with probability $|d_{j_0}|^2$ and the state becomes $\phi_{i_0} \otimes \varphi_{j_0}$. We then see that the probability of measuring $\lambda_{i_0}$ on the first measurement then $\mu_{j_0}$ on the second measurement is the product $|c_{i_0}|^2|d_{j_0}|^2$ and also $\sum_i|c_i|^2|d_{j_0}|^2 = |d_{j_0}|^2$. This shows that the probability of measuring the value of $B$ on the composite system is independent of the result of the measurement of $A$ when the pure state is a product state.

Moreover, we see that the probabilities of each measurement and the resulting state do not depend on the order that we we measure the first subsystem and the second subsystem; we could measure $B$ then $A$ and we would obtain the same results with the same probabilities.

For entangled states, the scenario is much messier. 
This is where some ``quantum effects'' can be seen based on the definition of the state and how states change upon measurement.
For instance, consider the entangled state $\psi = \frac{1}{\sqrt{2}}(e_1\otimes e_1 + e_2\otimes e_2) \in \C^2\otimes \C^2$ and observables $A = \diag(\lambda_1, \lambda_2)$ and $B = \diag(\mu_1, \mu_2)$ each on $\C^2$ with distinct eigenvalues. Then the result of measuring $A \otimes I_{2}$ is $\lambda_{i_0}$ with probability $1/2$ and the resulting state is $e_{i_0} \otimes e_{i_0}$. Therefore, the result of measuring $I_{2}\otimes B$ afterward is $\mu_{i_0}$ with probability $1$ because $e_{i_0} \otimes e_{i_0}$ is an eigenstate of $I_{2}\otimes B$. 

So, we see that the results of measuring the values of $A \otimes I_{2}$ then $I_{2} \otimes B$ are probabilistically dependent despite these observables commuting. The reason is that after the measurement the state changed from an entangled state to a product state so the probability distribution of the measurements of $I_{d_1} \otimes B$ changed after the measurement of $A \otimes I_{2}$. This even can happen if $A = B$.

If we think of adding the measurement of both observables on the separate subsystems, we would consider $A \otimes I_{d_2} + I_{d_1}\otimes B$, which is the kronecker sum of $A$ and $B$.

\begin{remark}
We can also consider product states formed from density matrices. If $\rho_1 \in M_{d_1}(\C)$, $\rho_2 \in M_{d_2}(\C)$ are density matrices then $\rho=\rho_1\otimes\rho_2$ is a density matrix in $M_{d_1d_2}(\C)$. It also has the same sort of independence property that we discussed above:
\[\Tr[(A\otimes I_{d_2})(\rho_1\otimes\rho_2)] = \Tr[A\rho_1\otimes \rho_2] =\Tr[A\rho_1]\Tr[\rho_2]=\Tr[A\rho_1]\]
that the expectation of an observable $A$ in the state $\rho_1$ is the same as the expectation of the observable $A\otimes I_{d_2}$ in the product state $\rho_1\otimes \rho_2$.

The concepts of purification and the partial trace provide ways of connecting density matrices on a subsystem and pure states on a composite system. See any of \cite{heinosaari2012mathematical, bengtsson2006geometry, CarlenBook} for more about this.
\end{remark}

\section{Repeated and Joint Measurement}

Suppose that we have $m$ observables $A_1, \dots, A_m \in M_d(\C)$. If a state $\psi$ is an eigenvector of each of these matrices then the measurement of any of these observables will not change the state and will return a definite value.

Now suppose that $A_1$, $\dots$, $A_m$ commute. We can then simultaneously diagonalize these matrices which means that there is an orthonormal basis $\psi_j$ of $\C^d$ so that $A_i \psi_j = \lambda_{i,j}\psi_j$, where $\lambda_{i,j} \in \R$. This means that there is a basis of states for which all the observables can have definite values. 
The joint spectrum 
\[\sigma(A_1, \dots, A_m)=\{(\lambda_{1,j}, \dots, \lambda_{m,j}): j=1,\dots,d\} \subset \sigma(A_1)\times\cdots\sigma(A_m)\subset \R^d\] 
consists of the definite values of these observables for each of the eigenstates $\psi_1$, $\dots$, $\psi_d$.

If the original state $\psi$ is not an eigenstate of all the observables then it is not proper to speak of the value of the observables for which $\psi$ is not an eigenvector. 
However, measuring the observable $A_1$ then measuring $A_2$, \dots, then measuring $A_m$ will result in the state being an eigenvector of all these matrices. This is true because after measuring the observable $A_i$, the state becomes an eigenvector of $A_i$. A subsequent measurement by $A_{i+1}$ does not change this fact because this projects the state onto an eigenspace of $A_{i+1}$ which commutes with the eigenspaces of $A_i$, ensuring that the resulting chain of projections belongs to the joint spectral projection of $A_i$ and $A_{i+1}$.

So, upon a finite number of measurements, we can reduce to the setting where the values of these observables are definite because the observables commute. If the observables do not commute then this is not necessarily true. Consider the observables:
\[A_1 = \bp 1&0\\0& -1 \ep,\;\; A_2 = \bp 0&1\\1& 0 \ep.\]
These matrices do not commute (they anticommute: $A_1A_2 = -A_2A_1$). 

Consider the state $\psi = e_1$. Consider the sequence of measurements $A_1$, $A_2$, $A_1$, $A_2$, $\dots$. 
After the measurement of $A_1$, the state will remain $e_1$. After the measurement of $A_2$, the state will become $\frac{1}{\sqrt2}(e_1+e_2)$ or $\frac{1}{\sqrt2}(e_1-e_2)$ with equal probability. 
After the measurement of $A_1$ the state will become $e_1$ or $e_2$ with equal probability, regardless of which of those two eigenvectors of $A_2$ the former state was. 
After the measurement of $A_2$, the state will become $\frac{1}{\sqrt2}(e_1+e_2)$ or $\frac{1}{\sqrt2}(e_1-e_2)$ with equal probability, regardless of which of those two eigenvectors of $A_1$ the former state was. 
The pattern continues. 
The diagram: 
\[\begin{tikzcd}
e_1 \arrow[loop, distance=2em, in=125, out=55]  & \frac{1}{\sqrt2}(e_1+e_2) \arrow[l] \arrow[ld] & e_1 \arrow[rd, Rightarrow] \arrow[r, Rightarrow] & \frac{1}{\sqrt2}(e_1+e_2) \arrow[Rightarrow, loop, distance=2em, in=125, out=55]  \\
e_2 \arrow[loop, distance=2em, in=305, out=235] & \frac{1}{\sqrt2}(e_1-e_2) \arrow[l] \arrow[lu] & e_2 \arrow[ru, Rightarrow] \arrow[r, Rightarrow] & \frac{1}{\sqrt2}(e_1-e_2) \arrow[Rightarrow, loop, distance=2em, in=305, out=235]
\end{tikzcd}\]
illustrates how the state can change upon different measurements, where a single arrow reflects how the state can change after measuring $A_1$ and a double arrow reflects how the state can change after measuring $A_2$.

In this example we see that it simply is not possible to have a state that is unchanged after measuring $A_1$ or $A_2$ regardless of what the original state was because $A_1$ and $A_2$ share no eigenvectors. It is possible to simultaneously know the values of $A_1$ and $A_2$ even if these observables do not commute if they share some common eigenvectors.

If the observables $A_1, \dots, A_m$ commute, there is not a limitation to knowing what the values of these observables are at the same time. So, given any pure state $\psi$, we can imagine there being a process that simultaneously measures all the observables simultaneously and performs the change of state in a way that is identical to the sequential measurement discussed above. 

\begin{example}\label{jointMeasurement}
Let $\mathscr A = (A_1, \dots, A_m)$ be a collection of commuting observables and let $E_\lambda(\mathscr A)$ be the joint spectral projection of $\mathscr A$ for $\lambda \in \sigma(\mathscr A)$.  The ranges of these spectral projections are orthogonal subspaces, so given the pure state $\psi$, we can decompose $\psi = c_1\psi_1+\cdots+c_k\psi_k$ where $c_j \neq 0$ and each $\psi_j$ is a unit vector in the range of a distinct joint spectral projection $E_{\lambda^j}(\mathscr A)$. 

We then can say that a joint measurement of $\mathscr A$ by these projections will produce the new state $\psi_j$ with probability $|c_j|^2$ and the measured values of $\mathscr A$ will be $\lambda^j \in \R^m$.

Note that $\psi_j$ is the vector gotten by normalizing the projection $E_{\lambda^j}(\mathscr A)\psi$ of $\psi$ onto the $\lambda^j$-joint eigenspace of $A_1$, $\dots$, $A_m$. In general, the projections $E_{\lambda^j}(\mathscr A)$ do not project onto an eigenspace of any of the matrices $A_j$, but intersections of those eigenspaces. 
\end{example}
\begin{example}\label{jointMeasurementRefinement}
Let $F_1, \dots, F_r$ to be some orthogonal $1$-dimensional projections that are contained in the joint spectral projections of commuting observables $A_1, \dots, A_m$. Write $F_j \leq E_{\lambda^j}(\mathscr A)$, where $\lambda^j$ for $j=1, \dots, r$ are elements of the joint spectrum that are non-unique if any of the joint spectral eigenspaces have dimension greater than $1$.

This refinement of the joint spectral projections will allow us to obtain a different way of jointly measuring $\mathscr A$.
For the pure state $\psi$, we can decompose it as $\psi = \sum_{j=1}^r c_j \psi_j$, where $\psi_j$ is a vector spanning the range of $F_j$ and $c_j \in \C$. We then can define a measurement of $\mathscr A$ by these projections to produce a new state $\psi_j$ with probability $|c_j|^2$ and the measured values of $\mathscr A$ will be $\lambda^j$.

Note that with this definition of a measurement, we may have some $c_j$ equal to zero so that $\psi_j$ will never become the new state. Likewise, there may be different states $\psi_{j_1}, \dots, \psi_{j_s}$ in the same joint eigenspace so that the new state is not uniquely determined by the values measured by $A_1$, $\dots$, $A_m$. However, if $F_{j_1}+\cdots+ F_{j_s} = E_{\lambda}(\mathscr A)$ then 
\[\|E_{\lambda}(\mathscr A)\psi\|^2 = \|F_{j_1}\psi\|^2+\cdots+\|F_{j_s}\psi\|^2=|c_{j_1}|^2+\cdots+|c_{j_s}|^2\]
so the probability of measuring the values $\lambda \in \sigma(\mathscr A)$ is the same when measuring $\mathscr A$ using the joint spectral projections or when using this refinement of the joint spectral projections.
\end{example}

The purpose of the prior examples is to show that there are many non-unique ways to jointly measure some commuting observables. In fact, the last example can be used for a single observable ($m=1$) so that there can be different notions of measuring even a single observable. However, in all these cases the observables were commuting. We will briefly discuss joint measurement of non-commuting observables in Section \ref{Uncertainty Principle}.

\section{Perturbation of Observables}

Suppose that $A, A' \in M_d(\C)$ are self-adjoint. We will think of $A'$ as a perturbation of $A$ of norm $\|A'-A\|$ and we will investigate in what ways the eigenvalues and spectral projections are perturbed. This mathematical analysis will provide information about how the possible measurements of the associated observables are affected and how the representation of a state in terms of an eigenbasis changes.

Perhaps the simplest result of this flavor is that the change in the expected value of an observable under any state is bounded by the norm of the perturbation: 
\begin{equation}\label{expPerturbation}
|\langle A'\rangle_\rho - \langle A\rangle_\rho| = |\langle A'-A\rangle_\rho| \leq \|A'-A\|.
\end{equation}

The perturbation of the eigenvalues of $A, A'$ is similarly well-behaved, as we discuss now. Because the eigenvalues of a matrix are the roots of its characteristic polynomial, they do exhibit some continuity with respect to perturbations. However, it is the fact that $A, A'$ are self-adjoint that makes the eigenvalue perturbation well-behaved because the perturbation of non-normal matrices do not behave as well when the size of the perturbation is measured by the operator norm.

\begin{example}
Consider the almost normal weighted shift matrix $B\in M_n(\C)$ from the beginning of Section \ref{Dimensional-Dependent Counter-Examples}. This matrix satisfies $\|[B^\ast ,B]\| \leq 2/n$. We also know that $\sigma(B) = \{0\}$ because $B$ is strictly upper triangular. 

By \cite{kachkovskiy2016distance}'s result for Lin's theorem, there is a normal matrix $N$ such that
\[ \|N-B\| \leq Const./n^{1/2}.\]
However, what we know is that $\|N\| \approx \|B\|\approx 1$ for $n$ large equals the largest absolute value of the eigenvalues of $N$. This means that the change in the eigenvalues of the non-normal matrix $B$ will in general depend on the dimension.
\end{example}

Specifically for our scenario of the perturbation of self-adjoint matrices there is Weyl's inequality (Theorem 4.3.1 of \cite{johnson1985matrix} using the formulation from \cite{bhatia1983perturbation}):
\begin{thm}
Let $A, A' \in M_d(\C)$ be self-adjoint. If $\lambda_i$ are the eigenvalues of $A$ and $\lambda_i'$ are the eigenvalues of $A'$, both counted with multiplicity, then there is a permutation $\tau$ of the set $\{1, \dots, d\}$ such that 
\[\max_i|\lambda'_{\tau(i)}-\lambda_i| \leq \|A'-A\|.\]
\end{thm}
This tells us that if there is a small difference between two observables in the operator norm then the possible values to be measured will be similar. Here are two simple examples that illustrate the type of change in spectral projections that can happen with perturbations.

\begin{example}
Let $A = \diag(0, 0.1, 0.2)$, $A' = A + \varepsilon I$. Then $\|A'-A\|= \varepsilon$, $\sigma(A')=\sigma(A)+\varepsilon$, and the eigenspaces themselves are shifted in terms of their eigenvalue labeling but not in any other manner. 

Notice however that $E_{[0, 0.1]}(A)$ projects onto the span of $e_1, e_2$ and $E_{[0, 0.1]}(A')$ projects onto the span of $e_1$ for $\varepsilon \in (0,0.1)$. 
Since these two projections have a different ranks, we see that $\|E_{[0, 0.1]}(A') - E_{[0, 0.1]}(A)\| = 1$. Regardless, we have this relationship:  $E_{[0, 0.1]}(A') \leq E_{[0, 0.1]}(A)$. 

Also, $E_{[0.05, 0.15]}(A)$ projects onto the span of $e_2$ and $E_{[0.05, 0.15]}(A')$ projects onto the span of $e_1$ for $\varepsilon \in (0.05,.15)$. 
So, $E_{[0.05, 0.15]}(A)$ and $E_{[0.05, 0.15]}(A')$ are orthogonal, but $E_{[0.05, 0.15]}(A')\leq E_{[0.05, 0.15+\varepsilon]}(A)$.

There are also perturbations of $A$ by other diagonal matrices that can move any of the eigenvalues independently, which can have the effect of merging some of the eigenspaces or breaking eigenspaces into orthogonal subspaces.
\end{example}
\begin{example}
Let $A = \diag(a,b)$ and $U_\theta = \bp 
\cos\theta &-\sin\theta\\ 
\sin\theta & \cos\theta
\ep$  for $\theta \in (-\pi, \pi]$ and instead define $A' = U^\ast AU$. If $a=b$ then $A' = A$ and the spectral projections of $A$ are unchanged since $A = aI$. If $a \neq b$ then for $\theta$ far from $\{0, \pm \pi\}$, the eigenspaces of $A'$ are significantly different than those of $A$. If $a \approx b$ then despite this, the perturbation is small in norm:
\[A'-A = (a-b)\bp \cos^2\theta-1 &  \cos\theta\sin\theta\\
\cos\theta\sin\theta  & \sin^2\theta\ep\]
so $\|A'-A\| = |a-b|\sin^2\theta$. Now, if we inspect how $E_{\{a\}}(A')$ is different from $E_{\{a\}}(A)$, we see that that the original spectral projection can considerably rotate toward the range of $E_{\{b\}}(A)$. If $|a-b|$ is however not small then the rotation of the eigenspaces must be small.
\end{example}

The Davis-Khan theorem tells us how the spectral projections can change under a perturbation measured with the operator norm:
\begin{thm} (\cite{bhatia1997and}) Suppose that $A, A'$ are self-adjoint operators in  $B(\mathcal H)$. Suppose that $\Omega_0 = [a_0, b_0] \subset \Omega = [a, b]$. 

Let $\epsilon = \dist(\R\setminus \Omega, \Omega_0) = \min(a_0-a, b-b_0)$ be the distance from $\Omega_0$ to the complement of $\Omega$.
Then 
\[\|(1-E_{\Omega}(A'))E_{\Omega_0}(A)\| \leq \frac{\|A'-A\|}{\epsilon}.\]
\end{thm}
\begin{remark}
Consider the product of projections $(1-E_{\Omega}(A))E_{\Omega_0}(A')$. If this equals zero then $E_{\Omega_0}(A') \leq E_{\Omega}(A)$. If the product instead has a small norm, 
then the range of $E_{\Omega_0}(A')$ is ``almost'' a subset of the range of $E_{\Omega}(A)$. 
So, the Davis-Khan theorem implies that if the norm of the perturbation is much smaller than the distance from $\Omega_0$ to the complement of $\Omega$ in $\R$ then the spectral projection of $A'$ on $\Omega_0$ is almost a subspace of the spectral projection of $A$ on $\Omega$. 

If the distance from $\Omega_0$ to $\R \setminus \Omega$ is smaller than the norm of the perturbation, then we will have poor control of these spectral projections. Our examples above illustrate some of the type of behavior that occur. 
\end{remark}

\begin{example}\label{invPropWidth}
Consider observables $A', A$. If $\psi$ is a normalized eigenvector of $A'$ with eigenvalue $\lambda$, then
\begin{align}
\|\psi-&E_{[\lambda-\epsilon, \lambda+\epsilon]}(A)\psi\| = \|(1-E_{[\lambda-\epsilon, \lambda+\epsilon]}(A))E_{\{\lambda\}}(A')\psi\| \nonumber\\
&\leq \|(1-E_{[\lambda-\epsilon, \lambda+\epsilon]}(A))E_{\{\lambda\}}(A')\| \leq \frac{\|A'-A\|}{\epsilon}.\label{perturbedDecomposition}
\end{align}
If we express the state $\psi$ in an eigenbasis of $A$, this inequality then provides a bound for how localized $\psi$ is with respect to this basis. 
More specifically, let us write $\psi = \sum_{\lambda \in \sigma(A)} c_\lambda \psi_\lambda$, where $\psi_\lambda$ is an eigenstate of $A$ with measured value $\lambda$. Then 
\[E_{[\lambda-\epsilon, \lambda+\epsilon]}(A)\psi = \sum_{\lambda \in \sigma(A)\cap[\lambda-\epsilon, \lambda+\epsilon]}c_\lambda \psi_\lambda\]
so (\ref{perturbedDecomposition}) is equivalent to
\[\sum_{\lambda \in \sigma(A)\setminus[\lambda-\epsilon, \lambda+\epsilon]}|c_\lambda|^2 \leq \left(\frac{\|A'-A\|}{\epsilon}\right)^2.\]
This provides a bound for the probability of measuring a value of $A$ that is not within $\pm \epsilon$ of the definite value $\lambda$ of $A'$. The bound for the probability that we do not measure $A$ with a value in this interval of radius $\epsilon$ is inversely proportional to $\epsilon^2$ and proportional to the square of the norm of the perturbation.

Another way to view this inequality is to imagine decomposing the eigenvector $\psi$ of $A'$ with respect to the eigenprojections of $A$. Changing the picture from the decomposition with respect to $A'$ to the decomposition with respect to $A$, this inequality show that the wavefunction  expanding horizontally. This ``inverse collapse'' satisfies the property that the majority of the mass of the wavefunction is localized if the perturbation is small.
\end{example}

From (\ref{expPerturbation}), already saw that the change of the expected value is small: 
\[|\langle A\rangle_\psi-\lambda| = |\langle A'\rangle_\psi - \langle A\rangle_\psi|\leq \|A'-A\|.\]
We can also estimate the standard deviation of $A$ with respect to this eigenvector $\psi$ of $A'$. However, we first estimate the change for a general state represented by a density matrix $\rho$. Before that. we make the following important observation:
\begin{defn}\label{rhoNorm}
Let $\rho \in B(\mathcal H)$ be a density operator. Then for any observables $A, B$, we define the semi-inner product
\[ \langle A, B\rangle_\rho = \Tr[A^\ast B \rho]\]
and seminorm
\[ \|A\|_\rho = \sqrt{\langle A, A\rangle_\rho}.\]
If $0$ is not an eigenvalue of $\rho$ then $\langle -, -\rangle_\rho$ is an inner product and $\|A\|_\rho$ a norm.
\end{defn}
So, by (\ref{Cauchy-Schwartz}), we have
\[|\Tr[A^\ast B \rho]\,| \leq \sqrt{\Tr[A^\ast A \rho]\Tr[B^\ast B \rho]}.\]
It is not visually obvious that $\Delta_\rho$ satisfies the triangle inequality. However, because
\[\Delta_\rho A = \sqrt{\Tr[(A-\langle A\rangle_\rho)^2\rho]} = \|A-\langle A\rangle_\rho\|_\rho,\]
we see that $\Delta_\rho$ is the composition of the linear map $\ell(A) = A-\Tr[\rho A]I$ and a seminorm. So, we see that $\Delta_\rho$ satisfies the triangle inequality and hence
\begin{equation}\label{varPerturbation}
|\Delta_\rho A' - \Delta_\rho A| \leq \Delta_\rho (A'-A) \leq \|A'-A\|. 
\end{equation}

In the case that we have been exploring of $\psi$ being an eigenstate for $A'$, we have $\Delta_\psi A' = 0$. So, the variance of $A$ in this state is small which is another way of capturing the idea that $\psi$ is localized with respect to the spectral decomposition of $A$. 

We can also generalize this observation as follows.
Suppose that $\psi$ were in fact not an eigenstate of $A'$ but instead only localized with respect to the spectral decomposition of $A'$ in the sense of  $\Delta_\psi A'$ being small. Then for $\|A'-A\|$ small, it is the case that $\psi$ is also localized with respect to the spectral decomposition of $A$.

\section{The Uncertainty Principle}
\label{Uncertainty Principle}
The Uncertainty Principle represents the idea that non-commuting observables in quantum mechanics are not compatible in some of the ways that classical quantities are. The standard example of such non-commuting (infinite dimensional) observables are the position and momentum operators.
\cite{busch2007heisenberg} categorizes (and later presents ways of formalizing) the Uncertainty Principle in three forms:
\begin{quote}
The uncertainty principle is usually described, rather vaguely, as comprising one or more of the following no-go statements...:\\
(A) It is impossible to prepare states in which position and momentum are simultaneously arbitrarily well localized.\\
(B) It is impossible to measure simultaneously position and momentum.\\
(C) It is impossible to measure position without disturbing momentum, and vice versa.
\end{quote}
For some very readable surveys of the physics and inequalities (some of which we will discuss below) as generalizations of the Heisenberg Uncertainty relation for position and momentum:
\[\Delta x \Delta p \geq \hbar.\]
 see \cite{sen2014uncertainty,
englert2024uncertainty,
wehner2010entropic}.
Much research has gone into working out qualitative and quantitative forms of the uncertainty principle. 
For some expository literature exploring other ways of addressing joint measurement and other issues related to the uncertainty principle see
\cite{busch2007heisenberg,
kiukas2019complementary, busch2006complementarity}.
See
\cite{bullock2018measurement, busch2014measurement,
busch2010notion,
busch2007universal,
mayumi2024uncertainty,
elad2002generalized,
ghobber2011uncertainty,
zozor2014general}
as some examples of recent work in this subject.

As we have said in prior sections, one cannot jointly measure two (finite dimensional) observables that do not share eigenvectors. In a sense, that is the statement of (B).
However, one can form so-called non-sharp joint measurements as approximations of  the ``actual measurements'' that we defined in prior sections.
We will later discuss using nearby commuting observables for a form of non-sharp measurement.

For the time being we will particularly focus on (A). Namely, we will explore to what extent can a state have a particular value for two non-commuting observables.

In classical mechanics, you can have a thought experiment concerning a particle that  at a particular moment of time has an exact location $x_0$ and is moving with exact velocity $v_0$. From the viewpoint of classical mechanics, there is not a logical conflict with simultaneously considering the exact values of position, of momentum, or of many other quantities.

A role for probability in classical mechanics is to capture the (classical) uncertainty in a system due to impracticality and ignorance. For instance, in classical statistical mechanics, we can only prepare a system with certain properties in a statistical sense: we cannot practically arrange that the particles of a large system of particles have a certain exact average energy, let alone an exact prescribed set of positions and velocities. 
We could ``try our best'' then describe our confidence in what the actual values are using probability theory.
Classically thinking, when we arrange a system, its particles have some positions and velocities at any particular time but how certain we are of these values is imperfect due to the limitations of the design of our measuring devices. 

In terms of a simple practical example, if two objects are at rest then the only classical limitation to knowing the distance between these objects would be based on how finely spaced the ``tick marks'' are on the ruler with which we are measuring the distance.
Moreover, there is nothing conceptually contradictory about having an oracle, in the sense of computer science, which can number the elements of a system at some past time and list their exact positions and velocities to any desired accuracy.

In quantum mechanics, these things not true. Although there is room for the use of probability for ignorance, the use of probabilities in quantum mechanics is not explainable only using the standard conceptions of classical mechanical probability as discussed above.
In terms of the simplistic formalism and interpretation of quantum mechanics we have presented, if two observables $A,B$ do not have any shared eigenvectors then there is no speaking of a particle having a particular value of $A$ and a particular value of $B$ waiting to be discovered by careful measurement. Given any state, we only know that there is a probability distribution for the value we would measure if we were to measure any given observable.

\vspace{0.1in}

The following inequality is the Robertson uncertainty relation for mixed states: 
\begin{thm} Let $A, B \in B(\mathcal H)$ be observables on a (possibly infinite dimensional) Hilbert space $\mathcal H$ and $\rho$ a density operator on $\mathcal H$. Then
\[\Delta_\rho A \cdot \Delta_\rho B \geq \frac12|\langle i[A,B]\rangle_\rho|.\]
\end{thm}
\begin{proof}
The Cauchy-Schwartz inequality applies for $\langle -, - \rangle_\rho$ defined in Definition \ref{rhoNorm}, so
\[|\langle X Y \rangle_\rho|= |\langle X^\ast, Y \rangle_\rho| \leq \|X^\ast\|_\rho\|Y\|_\rho\]
for any $X, Y \in B(\mathcal H)$.
Let $A, B \in B(\mathcal H)$ be self-adjoint and $X = A-\Delta_\rho A$, $Y=B-\Delta_\rho B$. Then 
\begin{align*}
|\langle [A,B]\rangle_\rho| = |\langle [X,Y]\rangle_\rho| \leq |\langle XY\rangle_\rho|+|\langle YX\rangle_\rho|
\leq 2\|X\|_\rho\|Y\|_\rho = 2\Delta_\rho A\Delta_\rho B.
\end{align*}
\end{proof}
\begin{remark}
This inequality has been extended for unbounded operators and with a covariance term to the so-called Robertson-Schr\"odinger inequality:
\[(\Delta_\rho A \Delta_\rho B)^2 \geq \left(\frac12\langle i[A,B]\rangle_\rho\right)^2+\left(\frac12\langle A\circ B\rangle_\rho-\langle A\rangle_\rho\langle B\rangle_\rho\right)^2,\]
where $A\circ B = AB+BA$ is the symmetric product of $A$ and $B$.
See \cite{gutierrez2021robertson, sen2014uncertainty}
and Section 2 of \cite{folland1997uncertainty} for more about the history and mathematical content of this inequality. 
Note that \cite{gutierrez2021robertson} proves that this inequality is always an equality for observables $A, B \in M_2(\C)$ and $\rho=P_\psi$.
\end{remark}

Suppose that $\rho$ is a state whose range belongs to a single eigenspace of one of the operators $A, B$. Without loss of generality, suppose that $R(\rho)\subset R(E_{\lambda}(A))$ so 
\[A\rho = \rho A = \lambda \rho.\] 
Then there is an equality in the Robertson uncertainty relation due to
\[\langle A\rangle_\rho = \Tr[A\rho] = \lambda\Tr[\rho]=\lambda,\]
\[\Delta_\psi A = \Tr[A^2\rho]-\Tr[A\rho]^2=0,\] 
and 
\[\left\langle[A,B]\right\rangle_\rho=\Tr[B\rho A-BA\rho]=0.\]
There is an equality in the Robertson-Schr\"odinger inequality as well because
\[\langle A\circ B \rangle_\rho=\Tr[B\rho A + BA\rho]=2\lambda\Tr[B\rho]=2\langle A\rangle_\rho\langle B\rangle_\rho.\]
So, we see that the uncertainty relation does not tell us anything when the state belongs to an eigenspace of one of the observables. This in particular holds if the state is a pure eigenstate of one of the observables. 

For the Robertson inequality to be of use we will need the observable $i[A,B]$ to not have zero expected value in the state $\psi$. For instance, if $A, B \in M_d(\C)$ and $\psi$ is an eigenvector of $i[A,B]$ with eigenvalue $\lambda$ having maximal absolute value then
\[|\langle i[A,B] \rangle_\psi| = |\lambda| = \|[A,B]\|.\]
This means that for this state,
\[\Delta_\psi A \cdot\Delta_\psi B \geq \frac12|\|[A,B]\|.\]
We then see that the largest lower bound that can be provided by the uncertainty inequality is closely related to the norm of the commutator $\|[A,B]\|$. 
However, as we said above, we are interested in the question of how localized a state can be with respect to the spectral decomposition of $A$ and of $B$.

Even if $\|[A,B]\|$ is large, it is still the case that $\langle [A,B]\rangle_\psi$ is to equal zero if $\psi$ is an eigenvector of one of the observables.                    
Without directly referencing the eigenvectors of $A$ and $B$, we know that $i[A,B]$ is a self-adjoint matrix with zero trace. This means that it has non-negative and non-positive eigenvalues. Since the numerical range $\{\langle\psi, i[A,B]\psi\rangle : \psi \in H, \|\psi\|=1\}$ is a convex set, we see that $\langle i[A,B]\rangle_\psi=0$ for some state $\psi$. So, by this analysis there are states for which the uncertainty relation is trivial.

If $A, B \in B(\mathcal H)$ are self-adjoint with $\mathcal H$ infinite dimensional then it is not necessarily the case that $A, B, i[A,B]$ have eigenvectors. Let $\psi_n$ be a sequence of approximate eigenstates for the element $\lambda \in \sigma(A)$: $\|\psi_n\|=1$, $\epsilon_n=A\psi_n - \lambda \psi_n\in \mathcal H$, $\|\epsilon_n\| \to 0$. Then
\begin{align*}
|\langle i[A,B]\rangle_{\psi_n}| &= |\langle A\psi_n, B\psi_n\rangle-\langle \psi_n, BA\psi_n\rangle| 
=|\langle \epsilon_n, B\psi_n\rangle-\langle \psi_n, B\epsilon_n\rangle| \\
&\leq 2\|B\|\|\epsilon_n\| \to 0.
\end{align*}
This means that there are states for which the uncertainty principle is approximately trivial simply from the fact that $B$ is a bounded operator and $A$ has approximate eigenvectors.

\begin{example}
Before continuing to discuss more limitations of the Robertson-Schr\"odinger inequality, we will see what it says for the position and momentum observables.
The case where the uncertainty inequality implies that the observables cannot be simultaneously measured is as in the case of the unbounded observables: the position operator $X: f(x) \mapsto xf(x)$ and the momentum operator $P: f(x) \mapsto -i\hbar \frac{d}{dx}f(x)$ which satisfy
\[[X,P] = i\hbar I,\]
where the products and equality are interpreted in terms of unbounded linear operators which are defined on a dense (non-closed) subspace of $\mathcal H=L^2(\R)$.

The Heisenberg Uncertainty Principle
\[\Delta X \Delta P \geq \frac12\hbar\]
then is a consequence that relies on the fact that $\langle \hbar I \rangle_\psi = \hbar$ which is unlike the cases we discussed above since the expected value of the commutator in those cases can made arbitrarily small for certain states. 

Note that this type of result does not conflict with what we showed for bounded observables because it is not possible for the commutator of bounded observables to equal a multiple of the identity. In fact, this is not even possible for elements of a Banach algebra (\cite{rudin1991functional}).

The unbounded self-adjoint operators $X$ and $P$ are also related through the Fourier transform. This fact can also be used to prove the Heisenberg Uncertainty relation and other uncertainty inequalities.
For very approachable surveys of the uncertainty principle for a function and its Fourier transform see
\cite{busch1985note,
folland1997uncertainty}.
\end{example}

We now return to discuss an intrinsic limitation of the Robertson-Schr\"odinger inequality.
As we discussed, there are cases where the inequality is approximately $0 \geq 0$ for states approximately localized with respect to $A$. One might consider rewriting the inequality as
\[\Delta_\psi B \geq \frac1{\Delta_\psi A}\sqrt{\left(\frac12\langle i[A,B]\rangle_\psi\right)^2+\left(\frac12\langle A\circ B\rangle_\psi-\langle A\rangle_\psi\langle B\rangle_\psi\right)^2}\]
in order to potentially obtain a positive lower bound from an inequality approximately of the form $\Delta_\psi B \geq 0/0$.

Suppose that $\psi$ is localized with respect to the spectrum of $A$ in the sense that there is a set $S\subset\R$ with diameter $\diam(S)$ such that $E_{S}(A)\psi = \psi$.
A careful inspection indicates that we can replace $B$ on the right-hand side of the inequality with $E_{S}(A)BE_{S}(A)$ without changing the inequality. 

What this tells us is that only a submatrix of $B$ is actually being used to provide a lower bound for $\Delta_\psi B$.  
This also explains why the inequality is approximately trivial when $\psi$ is localized with respect to the spectral decomposition of $A$ since for any $x \in S$, 
\begin{align*}
\|[A,E_{S}(A)BE_{S}(A)]\|&=\|[AE_{S}(A)-xE_{S}(A),E_{S}(A)BE_{S}(A)]\| \\
&\leq 2\|B\|\|(A-xI)E_{S}(A)\| \leq 2\|B\|\operatorname{diam}(S)
\end{align*}
can be made small by simply having $\operatorname{diam}(S)$ small, regardless of whether $B$ actually almost commutes with $A$ in any way. Loosely speaking, being a bounded operator is enough for $B$ to locally almost commute with a (possibly unbounded) self-adjoint operator $A$ in the sense that $\|[A,E_{S}(A)BE_{S}(A)]\|$ is small. So, we would only expect the Robertson-Schr\"odinger inequality to be useful if the state $\psi$ is not localized with respect to either of the observables.

\begin{example}\label{far[A,B]}
In particular, consider $A = \diag(1, \dots, 2d)$ with eigenvectors $e_j, j=1,\dots, 2d$ and the self-adjoint unitary matrix $B$ which interchanges the pairs $e_j$, $e_{j+d}$ for each of $j = 1, \dots, d$.
We see that $A, B$ both break into direct sums on the invariant subspaces $\operatorname{span}(e_j, e_{j+d})$, $j = 1, \dots, d$ on which the matrices act as
\[A \sim \bp j & 0\\0& j+d\ep,\;\; B \sim \bp 0 & 1\\1& 0\ep, \;\; [A,B] \sim \bp 0 & -d\\ d & 0 \ep.\]
In particular, $\|[A,B]\| = \frac12\|A\|=d$ and $\|B\|=1$ so these matrices do not have a small commutator.

However, if $\psi$ is a state so that $E_S(A)\psi=\psi$ and the length $|S| < d$ then $E_S(A)BE_S(A)$ equals $0$. This means that the Robertson-Schr\"odinger inequality tells us nothing about $\langle B \rangle_\psi$. This is rather non-ideal because in this example $\psi$ can be supported on an interval of almost half the length of the spectrum of $A$ (so that it can be greatly or even poorly localized) but we obtain no information about how $\psi$ is spread out with respect to $B$ because the norm of $[A,B]$ comes from $B$ interchanging far apart eigenspaces of $A$.
\end{example}

\section{Entropy Uncertainty Inequalities}

So, to obtain a quantitative version of the uncertainty principle that works well for finite dimensional almost commuting observables we will need to search elsewhere. 
Fortunately there are other formulations of the uncertainty principle. 

In particular, Deutsch (\cite{deutsch1983uncertainty}) in 1982 observed several weaknesses of the Robertson's uncertainty inequality, including that the right-hand side can vanish for some $\psi$ even though $\psi$ is not localized with respect to $A$ or with respect to $B$. 
Deutsch demands that the uncertainty relation's lower bound only vanishes if $A, B$ share a common eigenstate (or joint approximate eigenstates in the case that $A, B \in B(\mathcal H)$ for $\mathcal H$ infinite dimensional), which is one of the deficiencies that we observed in the prior section. He also details physics-based arguments for why one should be interested in a measure of uncertainty that only makes use of the eigenspaces of $A$ and of $B$.

Entropy uncertainty relations had existed prior to this (\cite{bialynicki1975uncertainty}). Deutsch in \cite{deutsch1983uncertainty} proposed using the Shannon entropy to construct a measure of the uncertainty for non-commuting observables and reviewed some fundamental relevant properties. We will not state the formalism precisely until later since it was refined by later research. 

Partovi in 1984 (\cite{partovi1983entropic}) expanded upon Deutsch's approach in a way that accommodated for infinite dimensional systems. 
Partovi described the measurement of an observable $A$ as corresponding to some collection of spectral projections $P_1, \dots, P_m$ of $A$ such that $\sum_j P_j = I$. 
If the observable has only a few eigenvalues then we can imagine that the projections are simply all the projections onto the eigenspaces of $A$ such as in Example \ref{jointMeasurementRefinement} or refinements of the spectral projections as in Example \ref{jointMeasurementRefinement}. 

More generally (and more pragmatically), Partovi discusses that we should view a measurement as a partition of the possible values of $A$ into ``bins'' with the assumption that the measurement device cannot tell the difference between elements of the same bin. This could correspond to choosing $P_j = E_{I_j}(A)$, where $I_j\subset \R$ are some disjoint intervals of positive length covering $\sigma(A)$. 
For instance, we naturally are familiar with this sort of measurement from the use of a ruler since it has a certain length scale (such as 1/16th of an inch or 1 millimeter) and we can measure the length of an object with this ``device'' only to the finite accuracy it permits, even if the to-be-measured value can take values that are much more finely distributed (or are continuous).

More work on related entropies using the Shannon entries was done (for example \cite{kraus1987complementary, maassen1988generalized}).
In 2002, Krishna and Parthasarathy (\cite{krishna2002entropic}) improved the estimate for a generalization of this entropy measure of uncertainty discussed by Partovi. 
We describe the definitions and state one of the main results.
\begin{defn}
Consider the notion of two non-commuting measurements consisting of self-adjoint operators $(X_1, \dots, X_m)$ and $(Y_1, \dots, Y_n)$ which satisfy $0 \leq X_i, Y_j \leq I$ and completeness: $\sum_i X_i = I, \sum_j Y_j = I$. 

If the state is $\rho$, then the probability of obtaining $X_j$ is $p_j=\Tr[X_j \rho]$ and the Shannon entropy of this probability distribution is \[H(X,\rho) = -\sum_{j=1}^m p_j \log_2(p_j).\]
\end{defn}
See, for instance \cite{ellerman2021new,
carlen2010trace}, for information about the Shannon entropy.
Note that the function $-x\log_2(x)$ on $[0,1]$ is positive in the interior and equal to zero at the end-points. 
The Shannon entropy has the property that it is zero if and only if the probability distribution $(p_j)_j$ is localized at a single point and is maximized with value $\log_2(m)=-\sum_{j=1}^m \frac1m \log_2\left(\frac1m\right)$ if and only if all the probabilities are the same: $p_j = 1/m$. 
This means that the Shannon entropy is a measure of how spread out the probabilities are.

Krishna and Parthasarathy then proved
\[H(X,\rho)+H(Y,\rho) \geq -2\log_2(\max_{i,j}\|X_i^{1/2}Y_j^{1/2}\|).\]
In particular, if the measurements $X$, $Y$ are spectral projections $P_i$, $Q_j$ of observables $A$, $B$, respectively, then
\[H(A,\rho)+H(B,\rho) \geq -2\log_2(\max_{i,j}\|P_iQ_j\|).\]
Note that $\|P_iQ_j\|\leq 1$ so the right-hand side is always non-negative.

This is a sort of inequality that has the properties that we were looking for as a measure of uncertainty that can help us determine to what extent a state can be localized with respect to $A$ and with respect to $B$. 
This satisfies the condition required by Deutsch that the lower bound be $0$ if and only if $A$ and $B$ share an eigenvector since that will be an eigenvector of $P_iQ_j$ for some $i,j$. 
Note that we obtain a larger lower bound for $H(A,\rho)+H(B,\rho)$ when $\max_{i,j}\|P_iQ_j\|$ is small which happens when $R(P_i), R(Q_j)$ do not approximately coincide in any direction. 

Note that this result for projections is intimately connected to the version of uncertainty inequalities based on exploring different representations of a vector.
For a survey of such related results, see
\cite{ricaud2014survey}.

\section{Commutator Uncertainty Relations}

Now that we have surveyed some fundamental research into entropic uncertainty relations, let us see how this inequality can tell us something about non-commuting operators. Deutsch's paradigm involved obtaining a measure of uncertainty that did not involve the actual values of the eigenvalues of $A, B$, except for repetitions. Consequently, a lot of information about $A$ and $B$ are being disregarded in forming this inequality. 

Since we are interested in the application of almost commuting and nearly commuting matrices to quantum mechanical observables, we would need to make use of some of the information captured in $[A,B]$. There are two ways of doing this. The first way is to simply try to derive some uncertainty principles directly from properties of $[A,B]$ and the second is to use properties of $[A,B]$ to obtain non-trivial upper bounds for the norms $\|P_iQ_j\|$. 

For the first way, \cite{mayumi2024uncertainty} obtained estimates:
\begin{thm} Let $A, B \in M_d(\C)$ be self-adjoint and $\rho \in M_d(\C)$ a density matrix with eigenvalues in decreasing order $\lambda_M, \dots, \lambda_{sm}, \lambda_m$. Then
\[ \Delta_\rho A \Delta_\rho B \geq C_\rho \|[A,B]\|_\rho,\]
where $C_\rho$ is a constant depending only of $\rho$. This inequality holds for $C_\rho = 2\lambda_m/\sqrt{\lambda_M}$ and $C_{\rho}=\sqrt{\lambda_m\lambda_{sm}/(\lambda_m+\lambda_{sm})}$.
\end{thm}
Though interesting, this still does not meet our needs because the estimate necessarily satisfies $C_\rho \leq 2\sqrt{\lambda_m}\leq 2/\sqrt{d}$ in these examples which gets worse the larger that the matrices become regardless of any properties of $[A,B]$.

The rest of this section is devoted to proving an inequality for the product of spectral projections in terms of the smallest singular value of $[A,B]$. In the next section we will prove a semi-converse of this inequality and of the Robertson inequality for almost commuting matrices. Note that there is similarity between some of these inequalities and the uncertainty inequalities in chapters 12 and 13 of the treatise \cite{busch2016quantum}.

\begin{example}
Suppose that we have two observables $A, B\in M_d(\C)$. Let $P_i, Q_j$ be the eigenprojections of $A, B$, respectively. If $v$ is any unit vector then
\[1=\|v\|^2 = \sum_i \|P_i v\|^2 \leq d\max_i\|P_iv\|^2.\]
A similar inequality holds for $Q_j$ so
\[\max_i\|P_iQ_j\| \geq \frac1{\sqrt{d}}, \;\; \max_j\|P_iQ_j\| \geq \frac1{\sqrt{d}}.\]
This provides a trivial lower bound that does not make use of any properties of the projections except completeness.
Likewise, for any $\ell > 0$ and any (non-zero) spectral projection $Q$ of $B$, there is an interval $S$ of length $\ell$ so that
\[\|E_{S}(A)Q\| \geq {\left\lceil\frac{2\|A\|}{\ell}\right\rceil}^{-1/2}\sim \sqrt{\frac{\ell}{2\|A\|}}\]
for $\ell/\|A\|$ small. This is another trivial bound.

In order to ask whether or not one can have a state localized with respect to $A$ and  $B$ then we would need some condition on $A, B$ so that 
\[\|E_{S_A}(A)E_{S_B}(B)\|\]
is not large even though the lengths $|S_A|, |S_B|$ are small.
\end{example}

We prove the following bound showing that if there is a pure state that is localized with respect to $A, B$ then $[A,B]$ almost vanishes on this state.  
\begin{lemma}
Suppose that $A, B \in B(\mathcal H)$ are bounded self-adjoint operators and $\psi \in \mathcal H$ is a unit vector.
Let $S_A \subset [-\|A\|, \|A\|], S_B \subset [-\|B\|, \|B\|]$ be intervals. Then
\begin{equation}\label{[A,B]ineq}
\|[A,B]\psi\| \leq \|A\||S_B| +\|B\||S_A|+ 4\|A\|\|B\|\left(\|(1-E_{S_A}(A))\psi\| +\|(1-E_{S_B}(B))\psi\|\right).
\end{equation}
\end{lemma}
\begin{proof}
Let $\psi\in \mathcal H$ be a unit vector and define $\epsilon_A = \psi-E_{S_A}(A)\psi$, $\epsilon_B = \psi-E_{S_B}(B)\psi$. If $c_A, c_B$ are the midpoints of the intervals $S_A, S_B$, respectively, then
\[\|(A-c_AI)E_{S_A}(A)\| \leq \frac12|S_A|, \;\; \|(B-c_BI)E_{S_B}(B)\| \leq \frac12|S_B|.\]
So,
\begin{align*} 
\|[A,B]\psi\| &= \|[A-c_AI, B-c_BI]\psi\|\\
&\leq \|A-c_AI\|\|(B-c_BI)\psi\|+\|B-c_BI\|\|(A-c_AI)\psi\|.
\end{align*}
We now estimate 
\begin{align*} 
\|(A-c_AI)\psi\| & = \|(A-c_AI)(\epsilon_A+E_{S_A}(A)\psi)\| \\
&\leq \frac12|S_A| + \|(A-c_AI)\|\|\epsilon_A\|.
\end{align*}
A similar estimate holds for $\|(B-c_BI)\psi\|$, so
\begin{align*} 
\|[A,B]\psi\| &\leq \|[A-c_AI, B-c_BI]\psi\|\\
&\leq \frac12\|A-c_AI\||S_B| + \|(A-c_AI)\|\|B-c_BI\|\|\epsilon_B\|\\
&\;\;\;\,+\frac12\|B-c_BI\||S_A| + \|(A-c_AI)\|\|B-c_BI\|\|\epsilon_A\|.
\end{align*}
Due to the requirement on $S_A, S_B$, we see that $|c_A|\leq \|A\|, |c_B|\leq \|B\|$. So, (\ref{[A,B]ineq}) follows since $\|A-c_AI\|\leq 2\|A\|$, $\|B-c_BI\|\leq 2\|B\|$.
\end{proof}
So, using this inequality, we obtain the following uncertainty principle for bounded operators:
\begin{prop}\label{[A,B]UP}
Suppose that $A, B \in B(\mathcal H)$ are bounded self-adjoint operators such that $[A,B]$ is invertible.
Let $S_A \subset [-\|A\|, \|A\|], S_B \subset [-\|B\|, \|B\|]$ be intervals. 
Let $\sigma = \inf_{\|\phi\|=1} \|[A,B]\phi\| = \|[A,B]^{-1}\|^{-1}$. Then for any pure state $\psi$,
\[\sigma \leq \|A\||S_B| +\|B\||S_A|+ 4\|A\|\|B\|\left(\|(1-E_{S_A}(A))\psi\| +\|(1-E_{S_B}(B))\psi\|\right).\]
\end{prop}
\begin{remark}
There are two types of terms in this uncertainty relation. The first type are those representing the length of the intervals that we want to localize the state within. The second type are error terms representing how much that the state is within these intervals.

So, if $\psi$ belongs to the range of $E_{S_A}(A)$ and the range of $E_{S_B}(B)$, then we would have
\[\sigma \leq \|A\||S_B| +\|B\||S_A|.\]
So, we see that it is possible to localize $\psi$ within these two intervals with respect to $A, B$ only if the intervals are not too small with respect to the smallest singular value of $[A,B]$.
\end{remark}
\begin{example}
Let $A, B$ be as in Example \ref{far[A,B]}.
Recall that the Robertson-Schr\"odinger uncertainty relation can easily be trivial for these operators. 

However, $\sigma = d$, $\|A\|=2d$, and $\|B\|=1$ so Proposition \ref{[A,B]UP} implies that
\[d \leq 2d|S_B| +|S_A|+ 8d\left(\|(1-E_{S_A}(A))\psi\| +\|(1-E_{S_B}(B))\psi\|\right)\]
or
\[1 \leq 2|S_B| +\frac{|S_A|}{d}+ 8\left(\|(1-E_{S_A}(A))\psi\| +\|(1-E_{S_B}(B))\psi\|\right).\]
Note that the term $\frac1{2d}|S_A|$ is a measure of the ratio of $[-\|A\|, \|A\|]=[-d,d]$ contained in $S_A$.
We see that this inequality is non-trivial.
\end{example}

If we wish to use the entropy uncertainty principle to capture uncertainty, then we can apply our result to obtain an upper bound on the product of projections:
\begin{thm}\label{upperProjBound}
Suppose that $A, B \in B(\mathcal H)$ are bounded self-adjoint operators such that $[A,B]$ is invertible.
Let $S_A \subset [-\|A\|, \|A\|], S_B \subset [-\|B\|, \|B\|]$ be intervals. 
Let $\sigma = \inf_{\|\phi\|=1} \|[A,B]\phi\| = \|[A,B]^{-1}\|^{-1}$. If $\frac{|S_B|}{\|B\|} +\frac{|S_A|}{\|A\|} < \frac\sigma{\|A\|\|B\|}$ then
\[\|E_{S_A}(A)E_{S_B}(B)\| 
\leq \sqrt{1-\left(\frac\sigma{4\|A\|\|B\|}-\frac{|S_B|}{4\|B\|} -\frac{|S_A|}{4\|A\|}\right)^2}.\]
\end{thm}
\begin{proof}
We apply Proposition \ref{[A,B]UP} to any unit vector $\psi$ in the range of $E_{S_B}(B)$ to obtain
\[\sigma \leq \|A\||S_B| +\|B\||S_A|+ 4\|A\|\|B\|\sqrt{1-\|E_{S_A}(A)\psi\|^2}.\]
Now, choosing a sequence $\psi_n$ of such vectors so that $\|E_{S_A}(A)\psi_n\| \to \|E_{S_A}(A)E_{S_B}(B)\|$ implies that
\[\sigma \leq \|A\||S_B| +\|B\||S_A|+ 4\|A\|\|B\|\sqrt{1-\|E_{S_A}(A)E_{S_B}(B)\|^2}.\]
This implies the result.
\end{proof}
\begin{remark}
Note that the norm of the commutator satisfies the inequality  
\[\|[A,B]\| \leq 2\|A\|\|B\|\]
which can be an equality.
So, our result provides an uncertainty principle for observables $A, B$ with $|[A,B]|$ bounded below by a constant that is large relative to the maximum possible norm of $[A,B]$.
\end{remark}

\section{Almost Commuting Observables}

The notion of almost commuting operators associated with observables being near actually commuting observables is discussed and used in a 1929 paper by von Neumann, translation provided in \cite{von2010proof}. A specific passage in the beginning of the article states:
\begin{quote}
Still, it is obviously factually correct that in macroscopic measurements the coordinates and momenta are measured simultaneously – indeed, the idea is that that becomes possible through the inaccuracy of the macroscopic measurement, which is so great that we need not fear a conflict with the uncertainty relations.\\
...\\
We believe that the following interpretation is the correct one: in a macroscopic measurement of coordinate and momentum (or two other quantities that cannot be measured simultaneously according to quantum mechanics), really two physical quantities are measured simultaneously and exactly, which however are not exactly coordinate and momentum. They are, for example, the orientations of two pointers or the locations of two spots on photographic plates– and nothing keeps us from measuring these simultaneously and with arbitrary accuracy, only their relation to the really interesting physical quantities ($q_k$ and $p_k$) is somewhat loose, namely the uncertainty of this coupling required by the laws of nature corresponds to the uncertainty relation[.]
\end{quote}

What von Neumann is suggesting here would be an example of an ``unsharp'' measurement of the non-commuting observables by commuting observables. If these commuting observables are nearby then we would see based on the results of the prior sections that these nearby observables behave very similarly to the original observables.

This analysis of an aspect of the measurement problem presumes that such nearby commuting self-adjoint observables exist and that the only limitation to them existing is the size (in some sense) of the commutator.
Ogata's theorem, which we discuss in a later section, confirms a mathematical formulation of the statement that macroscopic observables are nearby commuting observables with error going to zero as the uncertainty obstruction goes to zero due to the increasing size of the system.

An interesting counter-factual twist in the story might have been if von Neumann's physical argument was correct without Ogata's theorem being true. This certainly could be the case for certain observables of macroscopic objects defined under other assumptions. 
In such a scenario, it would be interesting if the error of measurement of these commuting observables did not go to zero as the uncertainty obstruction vanishes, but instead the error of such a measurement was numerically much smaller than would be detected macroscopically.

However, even with knowing Ogata's theorem, there may be limitations of its applicability due to our lack of knowledge of how close the exactly commuting observables $Y_{i, N}$ can be chosen to the given macroscopic observables $T_N(\sigma_i)$. This case has much in common with the speculation of a world where Ogata's theorem did not hold.
In particular, based on the non-constructive proof in \cite{ogata2013approximating}, it is conceivable that Ogata's theorem might only be non-trivial for $N$ much larger than what is seen in any physical application. It is conceivable then that reality may reject our description of macroscopic observables by Ogata's theorem not being capable of providing suitable estimates. (However, it may still allow von Neumann's intuitive argument to be realized using a different mathematical formalism.) 

Our extension of Ogata's theorem which is the subject of this thesis shows that the estimates in Ogata's theorem are indeed useful for $d = 2$ and so the speculative musings of the previous paragraph are defeated in this case.

\section{Uncertainty Relations and Almost Commuting \\Observables}

We now discuss how the existence of nearby commuting matrices for almost commuting matrices can provide a way of obtain reverse variants of the uncertainty inequalities that we explored in a prior section. Compare with the uncertainty inequalities in chapters 12 and 13 of \cite{busch2016quantum}.

We make crucial use of the optimal asymptotic estimate for Lin's theorem as proved by Kachkovsky and Safarov (\cite{kachkovskiy2016distance}) to obtain reverse bounds with similar estimates as that of the uncertainty relations.
Hence, another take-away from this section is the benefit of obtaining optimal explicit estimates for nearby commuting matrices.

We first start off with a result for the Robertson-Schr\"odinger inequality for two almost commuting matrices. But first, we state a direct consequence of (\ref{varPerturbation}).
\begin{lemma}\label{varLemma}
Let $A, B, A', B' \in B(\mathcal H)$ be self-adjoint operators with $A', B'$ commuting. 
Let $\rho$ be a state such that $\Delta_\rho A'=\Delta_\rho B'=0$.
Then
\[\Delta_\rho A \Delta_\rho B \leq \|A'-A\|\|B'-B\|.\]
\end{lemma}

Recall that \cite{kachkovskiy2016distance} proved that for compact self-adjoint operators $A, B \in K(\mathcal H)$, there exist commuting self-adjoint $A', B' \in B(\mathcal H)$ with discrete joint spectrum such that
\begin{equation}\label{KS-ineq}
\|A'-A\|, \|B'-B\| \leq C_{KS}\|[A,B]\|^{1/2}.
\end{equation}
They proved much more, but what we have stated here provides a generality that avoids mention of index obstructions.
With this in mind, we obtain:
\begin{thm}\label{varTheorem}
Let $A, B \in K(\mathcal H)$ be compact self-adjoint operators. There exists an orthonormal basis of pure states $\psi$ that satisfy
\[\Delta_\psi A\cdot \Delta_\psi B \leq C_{KS}^2\|[A,B]\|.\]
\end{thm}
\begin{proof}
The inequality follows from  (\ref{varLemma}) and (\ref{KS-ineq}) if $\psi$ is a joint eigenvector of $A'$, $B'$. 

Because $N'=A'+iB'$ is a normal operator with discrete spectrum, $N$ has a spectral decomposition of the same form as that of a matrix: $N' = \sum_{\lambda \in \sigma(N')}\lambda E_{\{\lambda\}}(N)$ where $E_{\{\lambda\}}(N)$ projects onto the $\lambda$-eigenspace of $N$. We can choose an orthonormal basis $\beta_\lambda$ for each of these eigenspaces. Then $\bigcup_\lambda \beta_\lambda$ is the desired basis. 
\end{proof}
\begin{remark}
Announced after the main results of this thesis, Lin presented a non-constructive argument in \cite{lin2024almost} that any finite collection of almost commuting operators which has a type of approximate joint spectrum that is nearby a type of approximate joint essential spectrum is nearly commuting. 
This was done for the sake of proving a result about approximate joint measurement of almost commuting operators however it did not provide the result in the desired generality.

After that, Lin presented what appears to be a simplification \cite{lin2024existence} of the original argument which avoids the use of nearby commuting matrices and directly shows the existence of certain pure states with which one can approximately simultaneously measure infinite dimensional almost commuting operators.

Our Theorem \ref{varTheorem}, which is likely not a surprising result, provides a version of Lin's simultaneous measurement result which is based in \cite{kachkovskiy2016distance}'s work so our estimate is constructive and provides an asymptotic estimate.
\end{remark}

We now will seek to provide an opposite direction inequality for the uncertainty inequality in terms of products of spectral projections.
We first state the following inequality:
\begin{lemma}\label{OverlapLemma}
Let $A, B, A', B' \in B(\mathcal H)$ be self-adjoint operators with $A', B'$ commuting. 
Suppose that $(\lambda_A, \lambda_B) \in \R^2$ belongs to the joint spectrum of $A', B'$ and $S_A = [\lambda_A-r_A, \lambda_A+r_A]$, $S_B = [\lambda_B-r_B, \lambda_B+r_B]$ for $r_A, r_B > 0$. 

Then
\[\|E_{S_A}(A)E_{S_B}(B)\| \geq 1 - \left(\frac2{|S_A|}\|A'-A\|+\frac2{|S_B|}\|B'-B\|\right).\]
\end{lemma}
\begin{proof}
For any neighborhood $S_A'=(\lambda_A-\epsilon,\lambda_A+\epsilon)$, $S_B'=(\lambda_B-\epsilon,\lambda_B+\epsilon)$ we have
\begin{align}\label{overlap}
\|E_{S_A'}(A')E_{S_B'}(B')\|=1.
\end{align}
We choose $0 < \epsilon < \min(r_A, r_B)$ arbitrary. Note that $S_A'\subset S_A$, $S_B'\subset S_B$.

By the Davis-Khan Theorem,
\[\|E_{S_A'}(A')E_{S_A}(A)-E_{S_A'}(A')\| \leq \frac1{r_A-\epsilon}\|A'-A\|,\] 
\[\|E_{S_B'}(B')-E_{S_B}(B)E_{S_B'}(B')\|\leq \frac1{r_B-\epsilon}\|B'-B\|.\]
So, using the inequality 
\[\|XY-\tilde{X}\tilde{Y}\| \leq \|X\|\|Y-\tilde{Y}\|+\|X-\tilde{X}\|\|\tilde{Y}\|,\]
we see that
\begin{align*}
\|E_{S_A'}(A')&E_{S_B'}(B') - E_{S_A'}(A')E_{S_A}(A)E_{S_B}(B)E_{S_B'}(B')\|\\
&\leq \frac1{r_A-\epsilon}\|A'-A\|+\frac1{r_B-\epsilon}\|B'-B\|.
\end{align*}
Therefore,
\begin{align*}
\|E_{S_A}&(A)E_{S_B}(B)\| \geq \|E_{S_A'}(A')E_{S_A}(A)E_{S_B}(B)E_{S_B'}(B')\|\\
&\geq \|E_{S_A'}(A')E_{S_B'}(B')\| - \left(\frac1{r_A-\epsilon}\|A'-A\|+\frac1{r_B-\epsilon}\|B'-B\|\right).
\end{align*}
Using (\ref{overlap}) and taking $\epsilon \to 0$ provides the result.
\end{proof}
\begin{remark}
This lemma has the following implication.
Suppose that $A, B$ are almost commuting observables. Then the existence of nearby commuting observables implies that many spectral projections of $A, B$ have a large product given that the lengths of the intervals $S_A, S_B$ are comparable to $\|A'-A\|, \|B'-B\|$, respectively.
\end{remark}

We now have
\begin{thm}\label{lowerProjBound}
Let $A, B \in K(\mathcal H)$ be compact self-adjoint operators.
Then there exist intervals $S_A, S_B$ of any given non-zero lengths such that
\[\|E_{S_A}(A)E_{S_B}(B)\| \geq 1 - 4C_{KS}\left(\frac{\|[A,B]\|}{|S_A||S_B|}\right)^{1/2}.\]
\end{thm}
\begin{proof}
For $C, D \in K(\mathcal H)$ compact self-adjoint operators, there exist commuting self-adjoint $C', D' \in B(\mathcal H)$ such that
\[\|C'-C\|, \|D'-D\| \leq C_{KS}\|[C,D]\|^{1/2}.\]
Then from Lemma \ref{OverlapLemma}, we obtain intervals $S_C, S_D$ of any given prescribed lengths such that
\[\|E_{S_C}(C)E_{S_D}(D)\| \geq 1 - 2C_{KS}\left(\frac1{|S_C|}+\frac1{|S_D|}\right)\|[C,D]\|^{1/2}.\]

We will derive the desired result by scaling this inequality.
Suppose that the prescribed lengths from the statement of the theorem for $S_A, S_B$ are $s_A, s_B$, respectively.
We cannot directly use the AM-GM inequality $2\sqrt{s_A^{-1}s_B^{-1}}\leq s_A^{-1}+s_B^{-1}$ to derive the desired result because the inequality is ``in the wrong direction''. If $s_A=s_B$ then we could derive the desired result because that is the condition for equality in the AM-GM inequality.

So, define $C = \sqrt{\frac{s_B}{s_A}}A$, $D = \sqrt{\frac{s_A}{s_B}}B$ and define the scaled intervals $S_C=\sqrt{\frac{s_B}{s_A}}S_A$, $S_D=\sqrt{\frac{s_A}{s_B}}S_B$ having lengths $|S_C| = \sqrt{\frac{s_B}{s_A}}|S_A| = \sqrt{s_As_B}$, $|S_D|=\sqrt{\frac{s_A}{s_B}}|S_B| = \sqrt{s_As_B}$. Then $[C,D] = [A,B]$ and $E_{S_C}(C)=E_{S_A}(A)$, $E_{S_B}(B)=E_{S_D}(D)$. 
We see that
\begin{align*}
\|E_{S_A}(A)E_{S_B}(B)\|&=\|E_{S_C}(C)E_{S_D}(D)\| \geq 1 - 2C_{KS}\left(\frac1{|S_C|}+\frac1{|S_D|}\right)\|[C,D]\|^{1/2}\\
&= 1 - 4C_{KS}\left(\frac1{\sqrt{s_As_B}}\right)\|[A,B]\|^{1/2}.
\end{align*}
This is the desired result.
\end{proof}
\begin{remark}
Note that due to the optimal exponent of $1/2$ for Lin's theorem, we obtain a lower bound that is a function of $\frac{\|[A,B]\|}{|S_A||S_B|}$. This provides a similar expression as Robinson's uncertainty principle.
\end{remark}

So, from Theorem \ref{varTheorem}, we see that for two almost commuting compact observables that a reverse Robertson uncertainty relation holds with the expectation of the commutator replaced with the norm of the commutator. This shows that the Robertson inequality (and hence the Robertson-Schr\"odinger inequality) is asymptotically sharp for two finite compact almost commuting observables since we can always find a sequence of pure states $\psi_n$ so that $\langle\psi, i[A,B]\psi \rangle \to\|[A,B]\|$.

From Theorem \ref{upperProjBound} we obtained an upper bound for the norm of products of projections in terms of the lengths of the spectral sets and in terms of a lower bound for $|[A,B]|$. This is an uncertainty principle which in conjunction with an entropy uncertainty principle can provide information about the spread of a state $\psi$ with respect to the spectral decompositions of two non-commuting observables.

From Theorem \ref{lowerProjBound}, we obtained a lower bound for the norm of products of some projections in terms of the lengths of the spectral sets and in terms of an upper bound for $|[A,B]|$. This provides an estimate that is in a sense an attempt at a converse of the inequality in Theorem \ref{upperProjBound}, showing that it can be close to being sharp for finite-dimensional almost commuting observables.

This completes our discussion of the relationship between uncertainty relations and almost commuting operators for two observables. We now move to discuss specific observables relevant to the main results of this thesis.

\section{Macroscopic Observables}

We will expand upon the idea of composite systems in the case that the systems are composed of $N$ many identical subsystems $\C^d$ and consider the observable corresponding to measuring the same observable $A$ on each of the subsystems. This leads to the observable on the composite system of 
\[A \otimes I_d \otimes\cdots\otimes I_d+I_d \otimes A \otimes \cdots\otimes I_d+\cdots+I_d \otimes \cdots\otimes I_d \otimes A	\]
which can be expressed as
\[\sum_{k=0}^{N-1}I_d^{\otimes (N-1-k)}\otimes A \otimes I_d^{\otimes k}.\]
This is a mathematical formulation of the ``macroscopic measurements'' as discussed by von Neumann. See Section II B. of \cite{poulin2005macroscopic} for more about this. 

In Appendix D of \cite{ogata2013approximating}, Ogata provides a generalization of Ogata's theorem for translation invariant local interactions for a quantum spin system. (See \cite{parkinson2010introduction} for an introduction to the topic of spin systems.)
Different generalizations are also possible but we will focus on this particular mathematical representation in line with our cursory review of quantum mechanical observables.

As a specific example, consider the non-commuting spin-1/2 observables:
\[S_x = \frac{\hbar}{2}\bp 0&1\\1&0\ep, \;\; S_y = \frac{\hbar}{2}\bp 0&-i\\i&0\ep, S_z = \frac{\hbar}{2}\bp 1&0\\0&-1\ep,\]
where $\hbar \approx  1.05459 \times 10^{-34}$Js is a very small constant.
(Note that we will later use a different convention for the normalized Pauli matrices in Chapter \ref{4.RepTheory} to be notationally consistent with the standard representations of $su(2)$.)

The macroscopic observable of $N$ copies of $S_z$ is
\[\sum_{k=0}^{N-1}I_d^{\otimes (N-1-k)}\otimes S_z \otimes I_d^{\otimes k}\]
and has eigenvalues $0, \pm\frac{\hbar}{2}, \dots, \pm \frac{N}2\hbar.$ We see that the larger the system is, the more eigenvalues there are and the larger that the absolute values of the eigenvalues can get. 
This of course makes sense since when we measure a macroscopic observable, we will simultaneously be interacting with many small systems which can together produce a large measurement.

The eigenvalues of the system as $N\to \infty$ are quantized because the additive group generated by $\sigma(S_z)=\{\pm\hbar/2\}$ in $\R$ is the discrete set $\frac\hbar2\Z$.
Moreover, although the macroscopic observable has  $2N-1$ distinct eigenvalues, the size of the matrix is $2^N$. This means that the typical eigenvalue of this macroscopic observable has an extremely large multiplicity compared to its norm  $N\hbar/2$.

If we alternatively are viewing the same macroscopic observable from the perspective of an outside observer then we would see a large discrete set of eigenvalues (which almost appear continuous due to the extremely small size of $\hbar$). Likewise, if the macroscopic observable has a norm that is approximately $1$ then we would know that $N \approx 2/\hbar$ is extremely large.

If we wanted to view the macroscopic observable for $S_z$ in a different way, we could normalize the operator by dividing by its norm to obtain the self-adjoint matrix 
\[\frac1{N\hbar/2}\sum_{k=0}^{N-1}I_d^{\otimes (N-1-k)}\otimes S_z \otimes I_d^{\otimes k} = \frac1N\sum_{k=0}^{N-1}I_d^{\otimes (N-1-k)}\otimes \frac2\hbar S_z \otimes I_d^{\otimes k},\]
where 
\[\frac2\hbar S_z = \bp 1&0\\0&-1\ep\]
is a normalized form of $S_z$. This perspective emphasises the role of the properties of $S_z$ and not so much the number of particles $N$ considered. 

It is this form of a macroscopic observable that we will interest ourselves in:
\begin{defn}
Define $T_N:M_d(\C)\to M_{d^N}(\C)$  by:
\[T_N(A) = \frac{1}{N}\sum_{k=0}^{N-1}I_d^{\otimes (N-1-k)}\otimes A \otimes I_d^{\otimes k}.\]
\end{defn}

We now list some properties of $T_N$.
When $A$ is diagonal, we see that \begin{align}\label{T_N spectrum}
\sigma(T_N(A)) = \frac{1}{N}\sum_{k=0}^{N-1}\sigma(A).
\end{align}
Thus, the spectrum of $T_N(A)$ is a discrete approximation of the convex hull of $\sigma(A)$. $T_N$ also satisfies 
\[T_N(A^\ast) = T_N(A)^\ast, \;\; T_N(A^T) = T_N(A)^T, \;\; T_N(UAU^{\ast}) = U^{\otimes N}T_N(A)U^{\ast\otimes N},\]
where $U$ is unitary.
There is additionally a symmetry due to permuting the tensor product factors. Note that $T_N$ is not multiplicative.

Additionally, if $A$ is normal (resp. self-adjoint) then $T_N(A)$ is normal (resp. self-adjoint).
Because of Equation (\ref{T_N spectrum}), $\|T_N(A)\| = \|A\|$ when $A$ is normal and in general $\|T_N(A)\| \leq \|A\|$ by definition. Applying
\[\|A\| \leq \|\Re(A)\| + \|\Im(A)\| \leq 2\|A\|\] to $T_N(A)$, we see that
\[\|A\| \leq \|\Re(A)\| + \|\Im(A)\| = \|\Re(T_N(A))\| + \|\Im(T_N(A))\| \leq 2\|T_N(A)\|\]
so
\begin{align}\label{norm-equiv}
\frac{1}{2}\|A\|\leq \|T_N(A)\| \leq \|A\|.
\end{align} 
Because 
\[\left[I_d^{\otimes (N-1-j)}\otimes A \otimes I_d^{\otimes j},I_d^{\otimes (N-1-k)}\otimes B \otimes I_d^{\otimes k}\right] = \left\{ \begin{array}{ll} I_d^{\otimes (N-1-k)}\otimes [A,B] \otimes I_d^{\otimes k}, & j = k\\
0, & j \neq k\end{array}\right.,\]
we have the commutator identity \begin{align}\label{T_Ncommutator}
[\,T_N(A),T_N(B)\,] = \frac{1}{N}T_N(\,[A,B]\,).
\end{align}

So, given any bounded collection of matrices in $M_d(\C)$, applying $T_N$ provides sequences of almost commuting matrices as $N \to \infty$. 
Two almost commuting self-adjoint matrices are nearby commuting self-adjoint matrices by Lin's theorem. The analogous statement is not true for more than two almost commuting matrices as discussed previously.

However, Ogata's theorem (Theorem \ref{OgataThm}) provides an extension of Lin's theorem in this special case of arbitrarily many macroscopic observables.
\begin{thm}\label{OgataThm}
For $A_1, \dots, A_k \in M_d(\C)$ self-adjoint, there are commuting self-adjoint matrices $Y_{i,N}\in M_{d^N}(\C)$ so that $\|T_N(A_i)-Y_{i,N}\|\to 0$ as $N\to\infty$.  
\end{thm} 
\begin{remark}
Note that the statement of Ogata's theorem in \cite{ogata2013approximating} is for $N = 2n + 1$. However, because \[T_{N+1}(A) = \frac{N}{N+1}  T_N(A)\otimes I_d + \frac{1}{N+1}I_d^{\otimes(N-1)}\otimes A,\] having shown the existence of nearby commuting matrices for $N$ odd, it follows for $N+1$ by choosing \[Y_{i,N+1} = \frac{N}{N+1}Y_{i,N}\otimes I_d.\]
This gives us the formulation we stated above.
\end{remark}

Because the $H_{i,N}=T_N(A_i)$ satisfy $\|[H_{i,N}, H_{j,N}]\| \leq Const. N^{-1}$, the optimal estimate for Lin's theorem in \cite{kachkovskiy2016distance} implies that if $k = 2$, there are nearby commuting self-adjoint matrices within a distance of $Const. N^{-1/2}$.
Based on the proof of Ogata's theorem in \cite{ogata2013approximating} which guarantees that $\varepsilon = o(1)$ as $N \to \infty$, we cannot infer if this or a similar estimate holds for more than two matrices.

In line with von Neumann's motivation for the almost-nearly commuting matrices problem, Ogata's Theorem has had applications to the theory of quantum statistical mechanics as explored by various authors (\cite{goldstein2015thermal, goldstein2017macroscopic, tasaki2016typicality}).
As an example, \cite{halpern1512microcanonical, halpern2016microcanonical} apply Ogata's theorem to construct what the authors of those papers call an approximate microcanonical subspace. Due to the nonconstructive proof of Ogata's theorem, objects constructing using Ogata's theorem are also not constructive, as observed in Remark 7.1 of \cite{khanian2020quantum}. One consequence of this is that one cannot know if the result of Ogata's theorem is non-trivial for reasonably sized systems.

\chapter{Requisite $su(2)$ Representation Theory}
\label{4.RepTheory}

Here we review some of the standard properties of representations of the Lie algebra $su(2)$ as well as some further properties of these representations that will be useful later. The standard material can be found in \cite{hall2015lie} or \cite{hayashi2017group}. All Lie algebra representations discussed will be assumed to be skew-Hermitian, coming from unitary representations of $SU(2)$. All direct sums are orthogonal. 

\section{Irreducible Representations of $su(2)$}

Consider the Pauli spin matrices (with eigenvalues $\pm1/2$) with the convention that $\sigma_3$ is diagonal with increasing eigenvalues:
\begin{align}\label{pSpin}
\sigma_1= \frac12\bp 0 & 1\\1&0  \ep,\;\; \sigma_2 = \frac12\bp 0 & i\\ -i &0 \ep,\;\; \sigma_3=\frac12\bp -1&0\\0&1 \ep.
\end{align}
These matrices span, with real coefficients, the trace-free self-adjoint matrices in $M_2(\C)$. The Pauli spin matrices satisfy the commutation relations
\[[\sigma_i, \sigma_j] = i\sum_k \epsilon_{ijk}\sigma_k,\]
where 
\[\epsilon_{ijk} =
\left\{\begin{array}{ll} \sgn (i\, j\, k), & i, j, k \mbox{ are distinct} \\
0, & \mbox{otherwise}\end{array}\right..
\]
Note also that the $\sigma_i$ anticommute:
\[\sigma_1\sigma_2+\sigma_2\sigma_1 = \sigma_2\sigma_3+\sigma_3\sigma_2 = \sigma_3\sigma_1+\sigma_1\sigma_3 = 0.\]

An arbitrary element of $su(2)$ can be represented as $i$ multiplied by the self-adjoint $c_1\sigma_1 + c_2\sigma_2 + c_3\sigma_3$ for $c_i\in\R$. This is the so-called defining representation of $su(2)$. By removing a factor of $i$, any representation $S$ of $su(2)$ is equivalent to a linear map $\tilde{S}$ defined on the $\C$-span of $\sigma_1, \sigma_2, \sigma_3$ with the same commutation relations
\[\left[\tilde S(\sigma_i), \tilde S(\sigma_j)\right] = i\sum_k \epsilon_{ijk}\tilde S(\sigma_k).\]
So, we identify any representation $S$ of $su(2)$ with its linear extension linear $\tilde{S}$.

Up to unitary equivalence, there is a unique irreducible representation of $su(2)$ of each dimension.  
For $\lam$ a non-negative integer or half-integer, the unique irreducible representation $S^\lambda$ on $\C^{2\lam+1}$ can be explicitly expressed as follows. 

Let $\sigma_+ = \sigma_1 + i\sigma_2$ and $\sigma_- = \sigma_+^\ast$. Note that
\begin{align}
\nonumber
\sigma_+= \bp 0 & 0\\1&0  \ep,\;\; \sigma_- = \bp 0 & 1\\ 0 &0 \ep.
\end{align}
Let \[d_{\lam,m} = \sqrt{\lam(\lam+1)-m(m+1)}=\sqrt{(\lam-m)(\lam+m+1)}, \;\; -\lam \leq m < \lam.\] The condition that $\lam$ and $m$ are both integers or both half-integers will be expressed as $\lam - m \in \Z$. Then 
\[S^{\lam}(\sigma_3) = 
\bp 
-\lam& & & & \\
 &-\lam+1& & & \\
 & & -\lam + 2 & &\\
 & & &\ddots& \\ 
 & & & & \lam \\
\ep, \;\; S^\lam(\sigma_+) =  
\bp 
0  & & & &\\
d_{\lam,-\lam}&0& & &\\
 &d_{\lam, -\lam+1}&0& &\\
 & &\ddots& \ddots&\\
 & & &d_{\lam,\lam-1}&0
\ep\]
and $S^{\lambda}(\sigma_-) = S^{\lambda}(\sigma_+)^\ast$. Then extend $S^\lam$ to $su(2)$ by linearity. 

In particular, if $v_{-\lam}, \dots, v_{\lam}$ are the standard basis vectors for $\C^{2\lam + 1}$, then 
\[S^{\lam}(\sigma_3)v_m = mv_m, \;\; S^{\lam}(\sigma_+)v_{m} = d_{\lam,m}v_{m+1}, \;\; S^{\lam}(\sigma_-)v_{m} = d_{\lam,m-1}v_{m-1}.\]
The trivial representation $S^{0}$ on $\C^1$ is given by $S(\sigma_i) = 0$.
The first nontrivial irreducible representation is the $2(1/2)+1=2$ dimensional representation $S^{1/2}$, the ``defining representation'', given by $S^{1/2}(\sigma_i) = \sigma_i$.

It is important to note that in representation theory $\lam$ is often called the ``weight'' of the representation $S^\lam$. However due to our usage of the term ``weight'' in Definition \ref{weightdef}, we will instead always refer to $d_{\lam, i}$ as the weights of the weighted shift matrix $S^\lam(\sigma_+)$ and will not refer to $\lam$ as a ``weight''. To distinguish between these two usages, we will use the common physics terminology that $S^\lam$ is ``the irreducible spin-$\lam$ representation'' if necessary. 

We now proceed to discuss some of the properties of the weights $d_{\lam, i}$ of the representation $S^\lam$. Note that
\begin{align}
\label{lamnorm}
\|S^{\lam}(\sigma_i)\|=\lam, \;
\|S^{\lam}(\sigma_{\pm})\| 
\leq 2\lam.
\end{align}
In particular, we see that for $i\neq j$, $\frac1\lam S^{\lam}(\sigma_i), \frac1\lam S^{\lam}(\sigma_j)$ are almost commuting with
\begin{align}\label{repAlmostCommuting}
\left\|\left[\frac1\lam S^{\lam}(\sigma_i), \frac1\lam S^{\lam}(\sigma_j) \right]\right\| =  \frac{1}{\lam}.
\end{align}

We state some estimates concerning the weights $d_{\lam, i}$ in the following lemma.
In particular, $(i)$ below provides a refinement of the bound of $\|S^\lam(\sigma_{\pm})\| = \max_i d_{\lam, i}$ in Equation (\ref{lamnorm}).
\begin{lemma}\label{d-ineq}
Suppose that $|i|\leq \mu \leq \lam$ are such that $\lam-i, \mu-i \in \Z$.
\begin{enumerate}[label=(\roman*)]
\item We have 
\begin{align*}
d_{\mu,i}\leq d_{\lam,i}\leq  \lam + \frac12 \leq 2\lam.
\end{align*}

\item If $\lam - |i| \leq M$ then
\begin{align*}
d_{\lam, i}\leq \sqrt{2\lam(M+1)}.
\end{align*}

\item
If $|\lam - \mu| \leq L$ then
\begin{align*}
d_{\lam, i}-d_{\mu,i}\leq \sqrt{2\lam L}.
\end{align*}

\item If $|\lam-\mu| \leq L$ and $l > 0$ is given, then at least one of  \[d_{\lam, i} - d_{\mu, i} \leq \sqrt{\lam}\frac{2L}{\sqrt{l}},\] \[\d_{\lam, i} \leq \sqrt{2\lam(l+1)}\]
hold. Consequently,
\begin{align}\label{d-Gbound}
d_{\lam, i} &- d_{\mu, i} + C \max(d_{\lam, i}, d_{\mu, i})\nonumber
\\ &\leq \max\left(\sqrt{\lam}\frac{2L}{\sqrt{l}}+C(\lam + 1/2), \sqrt{2\lam L}+C\sqrt{2\lam(l+1)}\right)
\end{align}

\item If $|\lam - \mu| \leq L$ then
\begin{align*}
d_{\lam,i}^2-d_{\mu,i}^2 \leq 2\lam L.
\end{align*}

\item $\|\, [S^\lam(\sigma_+)^\ast, S^\lam(\sigma_+)] \,\| = 2\lam.$

\end{enumerate}
\end{lemma}
\begin{remark}
For a fixed $\lam$, the graph of $d_{\lam, i}$ as a function of $i$ are points on a semicircle with center $-1/2$ and radius about $\lam+1/2$. See Illustration \ref{All}. 
\begin{figure}[htp]  
    \centering
    \includegraphics[width=8cm]{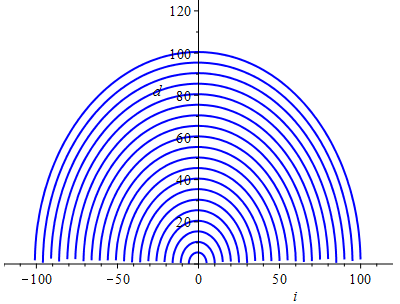}
    \caption{\dark \label{All} For each $\lam$, the points of $(i, d_{\lam, i})$ for $i = -\lam, -\lam+1, \dots, \lam-1$ all lie on a single semicircle. This Illustration depicts the weights $d_{\lam, i}$ for $\lam = 5, 10, \dots, 100$. }
\end{figure}
The maximum value of $d_{\lam,i}$ is asymptotically $\lam$, however it is always bounded by $2\lam$. This is $(i)$.
When $|i|$ is close to $\lam$, $d_{\lam, i}$ is small. This is $(ii)$. In other words, near the boundary of the circle, the weights are comparable to a smaller power of $\lam$. In particular, if $i = -\lam$ or $i = \lam-1$, $d_{\lam,i} = \sqrt{2\lam}$. 

When $\mu$ is close to $\lam$ then $d_{\lam, i}-d_{\mu,i}$ is small compared to $\lam$. However, if we put a separation of $M$ between $|i|$ and $\lam$ then this difference can be made smaller since it corresponds to taking the difference between values of consecutive semicircles away from the edges of the semicircles. This is the Claim in the proof. See Illustration \ref{Diff}.
\begin{figure}[htp]  
    \centering
    \includegraphics[width=8cm]{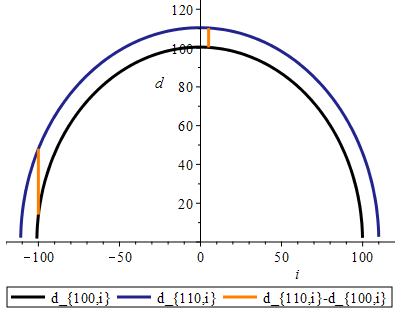}
    \caption{\label{Diff}\dark Illustration of the difference $d_{\lam, i}-d_{\mu,i}$ when $i$ is close to $-\mu$ and when $i$ is much smaller than $\mu$. Note that this difference is the \emph{vertical} distance between the arcs, not the radial distance.}
\end{figure}

As stated above, when $i$ corresponds to a point away from the boundary of the semicircle, one obtains an improved estimate for the differences of weights. When $i$ corresponds to a point near the boundary of the semicircle, one obtain an improved estimate for the size of the weight. This is $(iv)$.
As above, all notions of ``small'' or ``close'' should be interpreted in terms of the size of $\lam$. In particular $\sqrt{2\lam L}$ is much smaller than $\lam$ when $L$ is much smaller than $\lam$.

The similarity between $(ii)$ and $(iii)$ is due to the fact that $d_{\lam,i} \leq d_{\lam,\mu} + d_{\mu,i}$ because $d_{\lam,i}^2 = d_{\lam,\mu}^2 + d_{\mu,i}^2$. So, a bound for $d_{\lam, \mu}$ gives a bound for the difference $d_{\lam,i} - d_{\mu,i}$. This can be seen in the proof. 
Also, the pervasive ``$+1$'' is due to the small asymmetry of the terms $d_{\lam, i}$ with respect to $i \mapsto -i$. 
\end{remark}
\begin{proof}
\begin{enumerate}[label=(\roman*)]
\item  The first inequality follows since 
\[d_{\lam, i}^2 = \left(\lam(\lam+1)+\frac14\right) - \left(i(i+1)+\frac14\right) = \left(\lam+\frac12\right)^2-\left(i+\frac12\right)^2.\]
So, one obtains $\max_i d_{\lam, i} \leq \lam + 1/2$ with equality when $\lam$ is a half-integer.

\item If $0 \leq i < \lam$ then
\[d_{\lam,i} = \sqrt{(\lam-|i|)(\lam+|i|+1)}\leq \sqrt{M(2\lam)}.\]
If instead $-\lam \leq i < 0$ then $i = -|i|$ so \[d_{\lam,i} = \sqrt{(\lam+|i|)(\lam-|i|+1)}\leq \sqrt{2\lam(M+1)}.\]

\item 
If $\lam = \mu$ then the stated inequality is trivial so suppose that $\mu < \lam$. We calculate
\begin{align} 
\nonumber
d_{\lam, i}&= \sqrt{\lam(\lam+1)-i(i+1)} \leq \sqrt{\lam(\lam+1)-\mu(\mu+1)} + \sqrt{\mu(\mu+1)-i(i+1)}\\
&= \sqrt{(\lam-\mu)(\lam+\mu+1)} + d_{\mu,i} \leq \sqrt{L(2\lam)} + d_{\mu,i}.\nonumber
\end{align}
So, we obtain the desired inequality.

\item 
Given the Claim below, choose $M = l$. If $\lam - |i|\leq l$ then we obtain the second inequality by $(ii)$ above. If $\lam - |i| > l$ then we obtain the first inequality by the Claim below. To obtain Equation (\ref{d-Gbound}), we apply the same case analysis along with the unconditional bounds in $(i)$ and $(iii)$.
So, we only need to show:

\noindent \underline{Claim}:
Suppose  $|\lam - \mu| \leq L$. If  $\lam - |i| > M$ then
\begin{align}
\nonumber
d_{\lam, i}-d_{\mu,i}\leq \sqrt{\lam}\frac{ 2L}{\sqrt{M}}.
\end{align}

\noindent \underline{Proof of Claim}:
As before, suppose $\mu < \lam$. We calculate
\[d_{\lam, i} - d_{\mu, i} = \frac{d^2_{\lam, i} - d^2_{\mu, i}}{d_{\lam, i}+d_{\mu, i}}
= \frac{\lam(\lam+1)-\mu(\mu+1)}{d_{\lam, i}+d_{\mu, i}}
\leq \frac{(\lam-\mu)(\lam+\mu+1)}{d_{\lam, i}}.\]
Suppose $\lam-|i| > M$. If $i \geq 0$ then 
\[d_{\lam,i}  = \sqrt{(\lam+|i|+1)(\lam-|i|)} > \sqrt{\lam M}.\]
If $i < 0$ then
\[d_{\lam,i}  = \sqrt{(\lam-|i|+1)(\lam+|i|)} > \sqrt{M\lam}.\]

So,
\[d_{\lam, i} - d_{\mu, i} < \frac{L(2\lam) }{\sqrt{\lam M}} = \frac{2\sqrt{\lam} L}{\sqrt{M}}.\]

\item We have
\begin{align*}
d_{\lam,i}^2-d_{\mu,i}^2 &= \lambda(\lambda+1)-i(i+1) - \mu(\mu+1)+i(i+1) \\
&= (\lam+\mu+1)(\lam-\mu).
\end{align*}
If $\lam = \mu$ then $(\lam+\mu+1)(\lam-\mu) = 0 < 2\lam$. If $\mu < \lam$ then $\lam+\mu+1 \leq 2\lam$.

\item For $-\lam \leq i  \leq  \lam-2$, 
\[|d_{\lam, i+1}^2-d_{\lam, i}^2| = |(i+1)(i+2)-i(i+1)| = 2|i+1| \leq 2\lam.\]
Also, 
\[d_{\lam, -\lam}^2 = d_{\lam, \lam-1}^2 = 2\lam.\]
So, \[\|\,[S^\lam(\sigma_+)^\ast, S^\lam(\sigma_+)]\,\| = \max\left(d_{\lam, -\lam}^2,\, \max_{-\lam \leq i \leq \lam-2}|d_{\lam, i+1}^2-d_{\lam, i}^2|,\, d_{\lam, \lam-1}^2\right) = 2\lam.\]

\end{enumerate}
\end{proof}

\section{Multiplicities of Reducible Representations of $su(2)$}

We now recall some general properties of the tensor products of the irreducible representations of $su(2)$.
The reason we are interested in this is that if we have two representations $S_1$ on $\C^{n_1}$ and $S_2$ on $\C^{n_2}$, then their tensor product representation is expressed as
\[S_1\otimes S_2(\sigma_i) = S_1(\sigma_i)\otimes I_{n_2} + I_{n_1}\otimes S_2(\sigma_i).\] So, we can view $T_N(\sigma_i)$ in the statement of Ogata's theorem as the scaled matrix tensor product $\displaystyle\frac{1}{N}(S^{1/2})^{\otimes N}(\sigma_i)$. 
From this perspective, understanding how to break down this tensor product representation into irreducible representations will give us a handle on some of the underlying structure of $T_N(\sigma_i)$. 

Suppose that $\lambda_1 \leq \lambda_2$. Then the tensor product representation satisfies
\[S^{\lambda_2}\otimes S^{\lambda_1} \cong S^{\lambda_2-\lambda_1}\oplus S^{\lambda_2-\lambda_1+1}\oplus \cdots \oplus S^{\lambda_2 + \lambda_1}.\]
This means that there is a unitary matrix $U$ such that for all $i$, \[U^\ast\left(S^{\lambda_2}\otimes S^{\lambda_1}(\sigma_i)\right)U = S^{\lambda_2-\lambda_1}(\sigma_i)\oplus S^{\lambda_2-\lambda_1+1}(\sigma_i)\oplus \cdots \oplus S^{\lambda_2 + \lambda_1}(\sigma_i).\]
The unitary matrix can be expressed in terms of Clebsch-Gordan coefficients. These coefficients can be chosen to be real. Algorithms for the calculation of such coefficients have been well-studied. See for instance \cite{alex2011numerical}. 

The repeated tensor product of representations can be gotten by using this result along with standard manipulations of tensor products. In particular,
\begin{align*}
(S^{1/2})^{\otimes3} &\cong S^{1/2}\otimes(S^{1/2}\otimes S^{1/2}) \cong S^{1/2}\otimes(S^0 \oplus S^1) \cong S^{1/2}\otimes S^0 \oplus S^{1/2}\otimes S^{1} \\
&\cong S^{1/2}\oplus S^{1/2} \oplus S^{3/2} \cong 2S^{1/2}\oplus S^{3/2}.
\end{align*}

So, we see that $S^{1/2}$ has multiplicity $2$ and $S^{3/2}$ has multiplicity $1$ in the decomposition of the tensor representation into irreducible representations. By similar calculations, the representation $S=(S^{\lambda})^{\otimes N}$ can be calculated explicitly in terms of Clebsch-Gordan coefficients for any value of $N$. With that as a given, we focus on the distribution of multiplicities that occur when we write such a tensor representation as a direct sum of irreducible representations for general $N$.

Recall that the eigenvalues of $S^{\lambda}(\sigma_3)$ are $-\lambda, \dots, \lambda$.
By analyzing this, we obtain the following standard property that is used in the proof of the tensor product property given above. (See Theorem C.1 of \cite{hall2015lie}.)
Observe that $(S^{1/2})^{\otimes N}$ is a direct sum of irreducible representations $S^{\lambda}$ where all the $\lambda$ are integers if $N$ is even and all the $\lambda$ are half-integers if $N$ is odd. In particular, the eigenvalues of $(S^{1/2})^{\otimes N}(\sigma_3)$ will be integers if $N$ is even and will be half-integers if $N$ is odd. 
\begin{lemma}
Suppose that $S = n_0S^{0}\oplus n_{1/2}S^{1/2}\oplus \cdots n_{k}S^{k}$ is a representation of $su(2)$. Then the multiplicity of the eigenvalue $m$ of $S(\sigma_3)$ is $\sum_{i \geq 0} n_{|m|+i}$, where the sum is over integral $i$.

Conversely, if the eigenvalue $m$ of $S(\sigma_3)$ has multiplicity $k_m=k_{|m|}$ then the representation multiplicities $n_j$ can be reconstructed as $n_{m} = k_{|m|} - k_{|m|+1}$.
\end{lemma}
\begin{proof}
For the first statement, the eigenvalues of $S^{\lam}(\sigma_3)$ are $-\lambda, \dots, \lambda$. So, $S^{\lam}(\sigma_3)$ has an eigenvalue $m$ if $|m| \leq \lam$ and $\lam - m$ is an integer. Therefore, there is a non-negative integer $i$ such that $\lam = |m| + i$. Because such eigenvalues appear with multiplicity one, the first result then follows.

The converse follows directly from the first part.
\end{proof}

A simple way to express the multiplicities of eigenvalues is to identify the representation $S^{\lam}$ with the polynomial $x^{-\lam} + x^{-\lam + 1} + \cdots + x^{\lam-1} + x^{\lam}$ in the variables $x^{1/2}, x^{-1/2}$. The coefficient of the $x^m$ term is the multiplicity of the eigenvalue $m$ of $S^{\lam}(\sigma_3)$. When performing the direct sum of representations, this corresponds to adding the respective polynomials. The correspondence remains valid because the multiplicities and coefficients both add.
Likewise, the product of the polynomial corresponding to irreducible representations corresponds to tensor products of the irreducible representations. To see this consider the case that $j_1 \leq j_2$:
\begin{align*}
(&x^{-j_1}+x^{-j_1+1} + \cdots x^{j_1-1}+x^{j_1})(x^{-j_2}+x^{-j_2+1} + \cdots x^{j_2-1}+x^{j_2})\\
&\;= (x^{-j_1-j_2}+\cdots+x^{j_1-j_2})+(x^{-j_1-j_2+1}+\cdots+x^{j_1-j_2+1})+(x^{-j_1-j_2+2}+\cdots+x^{j_1-j_2+2})\\
&\;\;\;\;\;\;\;+\cdots +(x^{-j_1+j_2}+\cdots+x^{j_1+j_2})\\
&\;=x^{-j_1-j_2}+2x^{-j_1-j_2+1}+\cdots+(2j_1+1)x^{j_1-j_2}+(2j_1+1)x^{j_1-j_2+1}\\
&\;\;\;\;\;\;\;+\cdots+(2j_1+1)x^{j_2-j_1-1}+(2j_1+1)x^{j_2-j_1}+\cdots+2x^{j_1+j_2-1}+x^{j_1+j_2}\\
&\;=(x^{-j_1-j_2}+\cdots+x^{j_1+j_2})+(x^{-j_1-j_2+1}+\cdots+x^{j_1+j_2-1})+\cdots+(x^{j_1-j_2}+\cdots+x^{-j_1+j_2}).
\end{align*}
Hence, by the distributive property of multiplication and tensor products, the algebraic identification holds for all such polynomials.
This provides a method to easily calculate the multiplicities of the representations for computer algebra systems and also a simple closed form expression for $(S^{1/2})^{\otimes N}$. 

In particular, taking powers of $x^{-1/2} + x^{1/2}$ and using the binomial formula gives the following result. We interpret $\binom{N}{s}$ to be zero if $s$ is not an integer in $[0, N]$ and summations of the form $\sum_{k=a}^b$ where $b-a \in \Z$ to be the sum over $k=a, a+1 \dots, b$.
\begin{lemma}
For $0\leq \lam = N/2, N/2-1, \dots,$ the multiplicity of $S^\lam$ in $(S^{1/2})^{\otimes N}$ is \[\binom{N}{\lam+N/2} - \binom{N}{\lam+1+N/2}.\]
\end{lemma}
\begin{proof}
We calculate
\begin{align*}
(x^{-1/2} + x^{1/2})^{N} = \sum_{k=0}^N\binom{N}{k}x^{-\frac{N-k}{2}}x^{\frac{k}{2}} = \sum_{k=0}^N\binom{N}{k}x^{k - N/2} = \sum_{m=-N/2}^{N/2}\binom{N}{m+N/2}x^m.
\end{align*}
So, the multiplicity of the $S^\lam$ representation is
$\binom{N}{\lam+N/2} - \binom{N}{\lam+1+N/2}$. 
\end{proof}

Using the previous result, we can then investigate the behavior of the multiplicities. A graph of the multiplicities for $N = 1000$ is depicted in Illustration \ref{5}.
\begin{figure}[htp]  
    \centering
    \includegraphics[width=7cm]{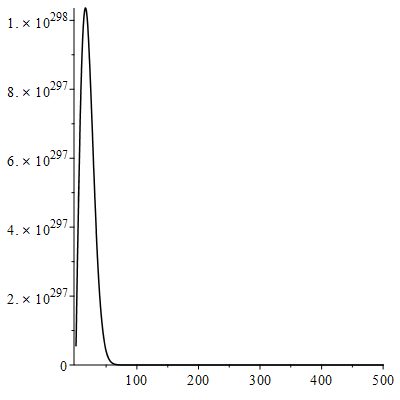}
    \caption{\label{5}\dark Illustration of multiplicities of $S^{\lam}$ for $\lam \in \Z$ of $(S^{1/2})^{\otimes N}$ for $N = 1000$. For this value of $N$, $\sqrt{N}/2 \approx 15.8$.}
\end{figure}

In particular, the multiplicities are increasing until the inflection point of the binomial distribution then afterward it decreases. Although numerical explorations suggest a rapid decrease of the multiplicities, since we are only investigating the operator norm, our method will only involve using that the coefficients strictly decrease after $O(\sqrt{N})$. A further discussion of properties of differences of binomial coefficients can also be found in \cite{shan1990gaps}, which influenced the statement of the following.

\begin{lemma}\label{1/2mult}
The multiplicity of $S^\lam$ in $(S^{1/2})^{\otimes N}$ is zero if $2\lam$ has a different parity than $N$. For $2\lam$ having the same parity as $N$, the multiplicity $n_\lambda$ of $S^\lam$ satisfies 
\[ 
\left\{\begin{array}{ll}
n_\lam < n_{\lam+1}, & \lam < \lam_\ast \\
n_\lam = n_{\lam+1}, & \lam = \lam_\ast \\
n_\lam > n_{\lam+1}, & \lam > \lam_\ast \\
\end{array}\right.,
\]
where
\[\lam_\ast =  \frac{\sqrt{N+2}}{2}-1 \leq \frac12N^{1/2}.\]
\end{lemma}
\begin{proof}
We use 
\[\binom{n}{k+1} = \frac{n!}{(k+1)!(n-k-1)!} = \binom{n}{k}\cdot\frac{n-k}{k+1}.\]
Therefore, 
\begin{align*}
    \left(\binom{n}{k}-\binom{n}{k+1}\right) &-\left(\binom{n}{k+1}-\binom{n}{k+2}\right) =  \binom{n}{k} - 2\binom{n}{k}\frac{n-k}{k+1} + \binom{n}{k+1}\frac{n-k-1}{k+2} \\
    &= \binom{n}{k}\cdot\left(1-  2\frac{n-k}{k+1}+\frac{n-k-1}{k+2}\cdot\frac{n-k}{k+1}\right)\\
    &= \binom{n}{k}\cdot\frac{(k+2)(k+1)-2(n-k)(k+2)+(n-k-1)(n-k)}{(k+2)(k+1)}\\
    &=\binom{n}{k}\cdot\frac{4k^2+(8-4n)k+(n^2-5n+2)}{(k+1)(k+2)}.
\end{align*}
Finding the (potentially irrational) values of $k$ such that this expression equals zero, we obtain
\[k = \frac{1}{2}n-1 \pm \frac{1}{8}\sqrt{(4n-8)^2-16(n^2-5n+2)} = \frac{1}{2}n-1 \pm \frac{1}{8}\sqrt{16n+32}.\]

By the previous lemma, the difference of coefficient multiplicities is 
 \[n_{\lam+1}-n_{\lam}=\left(\binom{N}{\lam+1+N/2} - \binom{N}{\lam+2+N/2}\right)-\left(\binom{N}{\lam+N/2} - \binom{N}{\lam+1+N/2}\right).\]
 Compared to the calculations above, we have $n=N$ and $k = \lambda + N/2$. So, the multiplicities begin decreasing after $\lam_\ast= \displaystyle\frac{\sqrt{N+2}}{2}-1$ as stated in the statement of the lemma.
\end{proof}

\chapter{The Gradual Exchange Lemma}
\label{5.GradualExchangeLemma}

A key component for the construction in later chapters will be I. D. Berg's Gradual Exchange Lemma, sometimes referred to as ``Berg's technique''. 
This method has been used in various arguments to prove results for matrices and also normal and nilpotent operators on a separable Hilbert space (\cite{berg1978index, davidson1984berg, marcoux1991distance, marcoux1996quasidiagonality, herrero1981unitary}). Also, in addition to the proof provided by Loring in \cite{loring1988k}, Loring remarked that Davidson knew how to use Berg's gradual exchange (by an argument similar to that found in \cite{davidson1984berg}) to provide a construction of nearby commuting matrices for the modified version of Voiculescu's almost commuting unitaries: $U_n \oplus U_n^\ast, V_n \oplus V_n$.

The lemma has appeared in different forms. A nice paper containing  reflections on the different uses and generalizations (with many diagrams) is Loring's \cite{loring1991berg}.  The argument we present below is a simple modification of Berg's original argument, although recast in terms of perturbing matrix blocks instead of a basis. It is similar to the argument in Lemma 2.1 of \cite{loring1991berg}.
Comparing this with the version stated in \cite{davidson1984berg}, one sees that the main difference is that the perturbation is real and the constant of the second term of the estimate is $\pi/2$ instead of the usual $\pi$ because we only require that $w_{N_0+1}' = -v_{N_0+1}$ instead of $w_{N_0+1}' = v_{N_0+1}$.

\section{Weighted Shift Operators and Weighted Shift \\
Diagrams}

We first give a definition of weighted shift operators.
\begin{defn}\label{weightdef}
Suppose that an orthonormal basis $v_1, \dots, v_n$ is given. We call a linear operator $A$ diagonal with respect to this basis, expressed as $\diag(a_1, \dots, a_n) = \diag(a_i)$, if $Av_i = a_iv_i$. 

We call a linear operator $S$ a weighted shift operator with respect to this basis, expressed as $\ws(c_1, \dots, c_{n-1}) = \ws(c_i)$, if $Sv_i = c_iv_{i+1}$. 
We can express the action of $S$ as:
\begin{align}\label{Sarrows}
S:v_1\overset{c_1}{\rightarrow}v_2\overset{c_2}{\rightarrow}\cdots\overset{c_{n-2}}{\rightarrow}v_{n-1}\overset{c_{n-1}}{\rightarrow}v_n\rightarrow0.
\end{align}
If the basis is not mentioned, the basis is assumed to be the ``standard basis''.

By multiplying the basis vectors by phases, we can choose each $c_i$ to be non-negative. This is discussed in more detail in Example \ref{ws-basisChange}. At this point it need only be said that if all the weights are real, then the phases can be chosen to be $\pm1$. \end{defn}
\begin{defn}
Suppose that $S = \ws(c_1, \dots, c_{n-1}) = \ws(c_i)$ is a weighted shift operator with respect to the basis $v_1, \dots, v_n$.
We refer to the lines spanned by the vectors $v_{k}, v_{k+1}, \dots, v_n$ as the ``orbit'' of $v_k$ under $S$.
We may refer to the vectors $v_k, \dots, v_n$ belonging to the orbit of $v_k$ under $S$.

If all the weights $c_k, \dots, c_{n-1}$ are non-zero, this coincides with the lines: $\spn(v_k)$, $\spn(S v_k)$, $\spn(S^2 v_k)$, $\dots$, $\spn(S^{n-k}v_k)$. In this case, we could call the weighted shift ``irreducible''. 

Note that this definition of orbit digresses from a typical notion of ``orbit'' from Dynamical Systems (such as in \cite{barriera2013dynamical, coudene2013ergodic, nillsen2010randomness}) if the weighted shift is not irreducible. In particular, our definition of orbit more closely aligns with what \cite{barriera2013dynamical} calls a ``forward-invariant set''. 

In particular, the fact that we have called $v_k, \dots, v_n$ \emph{the} orbit of $v_k$ indicates a choice made when writing $S = \ws(c_1, \dots, c_{n-1})$ as it may be possible to decompose $v_{k}, \dots, v_{n}$ as the disjoint union of orbits of irreducible weighted shift operators.
Generically, the weights $c_i$ may all be non-zero so that this definition coincides with the standard notion of orbit.
\end{defn}

We now describe the diagrams in Illustration \ref{WeightedShiftDiagrams}.
\begin{figure}[htp]     \centering
    \includegraphics[width=12cm]{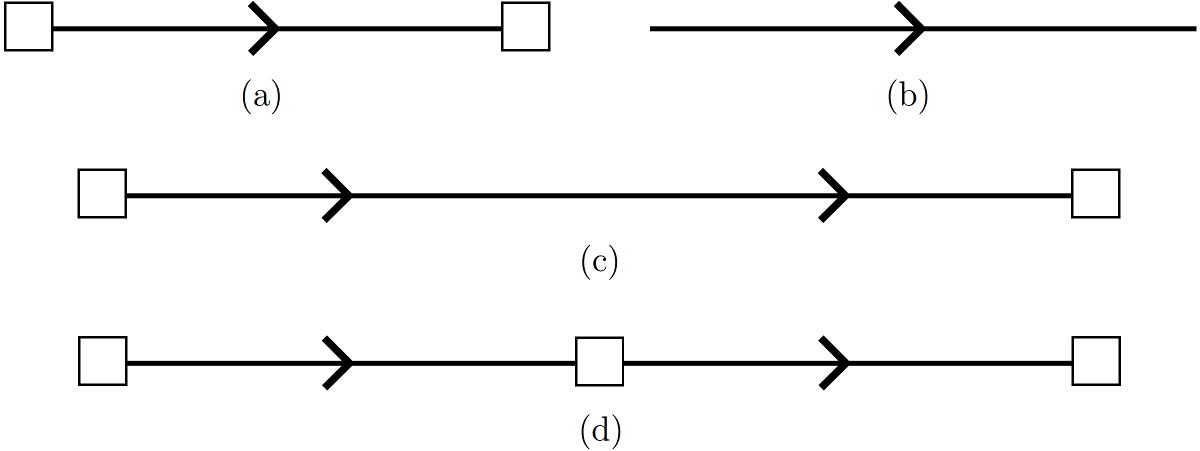}
    \caption{\label{WeightedShiftDiagrams}\dark
    Illustration of several weighted shift diagrams.}
\end{figure}
 Illustration \ref{WeightedShiftDiagrams}(a) is an illustration of the weighted shift matrix $S = \ws(c_1, \dots, c_{n-1})$ with respect to the orthonormal vectors $v_1, \dots, v_n$. It can be thought of as a graphical illustration of Equation (\ref{Sarrows}). Moving from left to right along the horizontal line segment corresponds to increasing the index of the vectors $v_i$. The vector $v_1$ is depicted by the square on the left and $v_n$ is depicted by the square on the right. For the purposes of this paper, we illustrate $v_1$ and $v_n$ in the diagram while suppressing explicit depictions of $v_2, \dots, v_{n-1}$. Note that the values of the weights and the size of $n$, while being important, are not illustrated in the diagram either.

Strictly speaking, this weighted shift diagram is a continuous illustration of a discrete system, similar to previous diagrams using Berg's interchange method. See \cite{loring1991berg} for diagrams that are discrete, which involve drawing a point for each $v_i$ and involve a ``$\cdots$'' in numerous places for complicated diagrams. Similar discrete diagrams sometimes appear in illustrations of the irreducible representations of $su(2)$ and other contexts. For instance, see Figure 8.1 of \cite{woit2017quantum}, Figures 4.1 and 9.4 of \cite{hall2015lie}, or quivers as in \cite{gruson2018journey}. Figure 8.1 of \cite{woit2017quantum} illustrates the weighted shift matrix $S^{\lam}(\sigma_+)$ and the diagonal matrix $S^\lam(\sigma_3)$ in the same diagram. 

Illustration \ref{WeightedShiftDiagrams}(b) is an illustration of the same weighted shift matrix as \ref{WeightedShiftDiagrams}(a) on a subset $v_{i_0}, \dots, v_{i_1}$ where $1 < i_0 < i_1 < n$. This can be expressed as
\begin{align}\label{Sarrows2}
S:v_{i_0}\overset{c_{i_0}}{\rightarrow}v_{i_0+1}\overset{c_{i_0+1}}{\rightarrow}\cdots\overset{c_{i_1-2}}{\rightarrow}v_{i_1-1}\overset{c_{i_1-1}}{\rightarrow}v_{i_1},
\end{align}
where Equation (\ref{Sarrows2}) only indicates the action of $S$ on the relevant vectors and is silent on whether we are viewing $S$ acting as  a weighted shift starting at $v_{i_0}$ and whether its orbit ends with $v_{i_1}$ or what $Sv_{i_1}$ is. The weighted shift diagram in Illustration \ref{WeightedShiftDiagrams}(b) does not include an initial square, indicating that we are not viewing $v_{i_0}$ as initiating a complete orbit (but a sub-orbit). It also does not end in a square, indicating that the orbit of $v_{i_0}$ is not being viewed as ending with $v_{i_1}$.

 Illustration \ref{WeightedShiftDiagrams}(c) is an illustration of $S = \ws(c_1, \dots, c_{n-1})$ where $c_{\tilde n} = 0$ for some $1 < \tilde{n} < n-1$. So, $S v_{\tilde n} = 0v_{\tilde n+1}$.
 Note that the diagram itself gives no indication that the $\tilde n$-th weight is zero.
 Also note that the second arrow in the diagram has no additional meaning and is added for aesthetic reasons related to Illustration \ref{WeightedShiftDiagrams}(d).  Because of our terminology we view $v_1, \dots, v_n$ as the $S$-orbit of $v_1$ and the illustration reflects this with only having the squares for the first and last vectors.
    
 Illustration \ref{WeightedShiftDiagrams}(d) is an illustration of the same operator $S$ as in \ref{WeightedShiftDiagrams}(c), except that we now view $c_{\tilde{n}} = 0$ as breaking $S$ into two weighted shift operators $\ws(c_1, \dots, c_{\tilde{n}-1})$ with respect to the vectors $v_1, \dots, v_{\tilde{n}}$ and $\ws(c_{\tilde{n}+1}, \dots, c_{n-1})$ with respect to the vectors $v_{\tilde n+1}, \dots, v_n$. With this choice of perspective, we view $v_1, \dots, v_{\tilde{n}}$ as the $S$-orbit of $v_1$.
 
 The distinction between (c) and (d) is based on the decision to view $S$ as a single weighted shift matrix
 \[S: v_1 \overset{c_1}{\rightarrow} \cdots \overset{c_{\tilde n -1}}{\rightarrow} v_{\tilde n} \overset{0}{\rightarrow} v_{\tilde n + 1} \overset{c_{\tilde n+1}}{\rightarrow}\cdots \overset{c_{n-1}}{\rightarrow}v_n\rightarrow0\]
 or as a weighted shift on two invariant subspaces:
 \[S: v_1 \overset{c_1}{\rightarrow} \cdots\overset{c_{\tilde n -1}}{\rightarrow} v_{\tilde n} \rightarrow 0, \;\;\; S:  v_{\tilde n + 1} \overset{c_{\tilde n+1}}{\rightarrow}\cdots \overset{c_{n-1}}{\rightarrow}v_n\rightarrow0.\]

 \section{A Gradual Exchange Lemma}
 We now present our first formulation of the gradual exchange lemma in terms of vectors. Later, we will formulate this in terms of weighed shift operators in a way that keeps track of the norm of the self-commutator.
\begin{lemma}\label{pre-gel}
Let $\{v_k, w_k\}_{k=1,\dots,N_0+1}$ be a collection of orthonormal vectors in a Hilbert space $\mathcal H$ and $S$ be a linear operator on $\mathcal H$ such that for $k = 1, \dots, N_0$, $Sv_k = a_kv_{k+1}$, $Sw_k = b_kw_{k+1}$ for some constants $a_k, b_k$. 

Then there is a linear operator $S'$ such that $S'^{N_0}v_1$ is a multiple of $w_{N_0+1}$, $S'^{N_0}w_1$ is a multiple of $v_{N_0+1}$, and 
\[\|S'-S\| \leq \max_{k\in[1, N_0]}\left(\frac12|a_k-b_k| + \frac{\pi}{2N_0}\max(|a_k|, |b_k|)\right).\]

Moreover, there are rotated orthonormal vectors $v_k', w_k'$ with
$\spn(\{v_k', w_k'\})$ equalling $\spn(\{v_k, w_k\})$ for $k = 1, \dots, N_0+1$, 
$S'v_k' = \frac{a_k+b_k}2v'_{k+1}$ and $S'w_k' = \frac{a_k+b_k}2w'_{k+1}$ for $k = 1, \dots, N_0$, 
and  $v_1' = v_1, v_{N_0+1}' = w_{N_0+1}, w_1' = w_1, w_{N_0+1}' = -v_{N_0+1}$. Also, $S'-S$ is supported on and has range in $\spn\left(\bigcup_{k=1}^{N_0+1}\{v_k, w_k\}\right)$. 
\end{lemma}
\begin{proof}
We can restrict $S$ to $V = \spn(v_1, w_1, \dots, v_{N_0+1}, w_{N_0+1})$ and will leave $S$ alone on $V^\perp$. We will identify $V$ with $\C^{2(N_0+1)}$.

Let the standard basis vectors $e_i$ of $\C^{2(N_0+1)}$ be identified with a basis of $V$ by $v_k \sim e_{2k-1}, w_k \sim e_{2k}$. We can then write $S$ as a matrix of the form
\[S = \begin{pmatrix} 0&&&&\ast\\
C_1&0&&&\ast\\
&C_2&\ddots&&\ast\\
&&\ddots&0&\ast\\
&&&C_{N_0}&\ast\\
&&&&\ast\\
\end{pmatrix},\]
where $C_k = \diag(a_k, b_k) \in M_2(\C)$ and the column of $\ast$'s depicts the action of $S$ on the subspace $\spn(v_{N_0+1}, w_{N_0+1})\oplus V^\perp$. The rows correspond to the subspaces $\spn(v_{1}, w_{1})$, $\dots$, $\spn(v_{N_0+1}, w_{N_0+1})$, $V^\perp$.

Let $0_2$ be the zero vector in $\C^2$ and $0_2^{\oplus \ell}$ denote the $\ell$-fold direct sum of $0_2$.
So, the basis vectors $v_k, w_k$ can be identified with direct sums of vectors in $\C^2$ by padding the standard basis vectors $\bp 1\\0\ep$, $\bp 0\\1\ep$ in $\C^2$ with $2N_0$ zeros appropriately:
\[v_k = 0_{2}^{\oplus (k-1)}\oplus \bp 1\\0\ep \oplus 0_{2}^{\oplus (N_0+1-k)}, \;\; w_k = 0_{2}^{\oplus (k-1)}\oplus \bp 0\\1\ep \oplus 0_{2}^{\oplus (N_0+1-k)}.\]
So, the results of repeatedly multiplying  $v_1$ and $w_1$ by $S$ correspond to the action of the matrix product $C_k \cdots C_1$ on the standard basis vectors in $\C^2$.

Since the product $C_k \cdots C_1$ is diagonal, the main idea of the proof is that if we introduce a small rotation into the terms $C_i$ then we can eventually have the product $C_k \cdots C_1$ be of the form $\bp 0 & \ast \\ \ast & 0\ep$ which would be what is required to interchange the orbits.

Let $R_\theta$ be the rotation matrix $\bp \cos \theta & -\sin\theta \\ \sin\theta & \cos\theta\ep$. Note that  $\|C_k - (a_k+b_k) I_2/2\| \leq |a_k - b_k|/2$. 
Let $S'$ act as the block weighted shift operator on $V$ with weights $C_k' = \frac{a_k+b_k}2 R_{\pi/2N_0}$ and equal to $S$ on $V^\perp$.

Then using $C_{N_0}'\cdots C_1' = \left(\prod_{k=1}^{N_0}\frac{a_k+b_k}2\right) R_{\pi/2} = \left(\prod_{k=1}^{N_0}\frac{a_k+b_k}2\right)\bp 0 & -1 \\ 1 & 0\ep$ we see that $S'$ satisfies the primary conditions of the lemma with 
\begin{align*}
\|S' - S\| &= \max_k \|C_k' - C_k\| \leq \max_k\left(\left\|C_k-\frac{a_k+b_k}2I_2\right\|+\left\|\frac{a_k+b_k}2I_2-C_k'\right\|\right)\\
&\leq\max_k\left(\frac{|a_k-b_k|}2+ \frac{|a_k|+|b_k|}2\left|1-e^{i\pi/2N_0}\right|\right)  \\
&\leq \max_k\left(\frac12|a_k-b_k|
+ \frac{\pi }{2N_0}\max(|a_k|, |b_k|)\right).
\end{align*}

Further, because $R_\theta$ is a real orthogonal matrix, we can define 
\[v_k' = 0_{2}^{\oplus (k-1)}\oplus R_{(k-1)\pi/2N_0}\bp 1\\0\ep \oplus 0_{2}^{\oplus (N_0+1-k)},\] \[w_k' = 0_{2}^{\oplus (k-1)}\oplus R_{(k-1)\pi/2N_0}\bp 0\\1\ep \oplus 0_{2}^{\oplus (N_0+1-k)}\]
to have the required properties from the second part of the statement of the lemma.
\end{proof}
\begin{remark}
Note that in \cite{berg1978index}, there is a phase factor close to $1$ that appears as well to remove the $-1$ term in $R_{\pi/2}$ so that $w_{N_0+1}' = v_{N_0+1}$.
This is unnecessary for our purposes. 

Moreover, because our change of basis: $v_k, w_k \to v_k', w_k'$ is performed by a real orthogonal matrix, this will provide additional structure for the matrices that we later obtain for Ogata's theorem. So, our modification of the construction is preferred. 
\end{remark}

\section{The Gradual Exchange Lemma for Almost Normal Weighted Shift Matrices}

We will now express the gradual exchange lemma in terms of direct sums of weighted shift operators. Because we will be interested in applying the gradual exchange lemma to direct sums of almost normal weighted shift operators, we will want the perturbation using the gradual exchange lemma to not change the norm of the self-commutator much. See the next chapter for more about this. The only thing that we need here is to state that if $S = \ws(c_1, \dots, c_{n-1})$ on $\C^n$ then the norm of the self-commutator of $S$ can be expressed as
\[\|\,[S^\ast, S]\,\| = \max\left(|c_1|^2, |c_{n-1}|^2, \max_{i \in [1, n-2]}||c_{i+1}|^2-|c_i|^2|\right).\]

The following is what will be referred to as the gradual exchange lemma.
\begin{lemma}\label{GELws}
Let $S_1 = \ws(a_i)$ with respect to an orthonormal basis $v_i$ of $\C^{n_1}$ and $S_2 = \ws(b_i)$ with respect to an orthonormal basis $w_i$ of $\C^{n_2}$. Assume that $a_i, b_i \geq 0$.
Let $i_0< i_1$ be indices in $[1, \min(n_1, n_2)]\cap \N$ satisfying $\#[i_0, i_1]\cap \N \geq N_0 + 1$.

Then there are $S_1', S_2'$ and orthonormal vectors $v_{i_0}', \dots, v_{i_1}', w_{i_0}', \dots, w_{i_1}'\in \C^{n_1}\oplus \C^{n_2}$
with the following properties:
\begin{enumerate}[label=(\roman*)]
\item \[v_{i_0}' = v_{i_0}\oplus 0, v_{i_1}' = 0\oplus w_{i_1}, w_{i_0}' = 0\oplus w_{i_0}, w_{i_1}' = -v_{i_1}\oplus 0,\] and $\spn(v_i\oplus 0,0\oplus w_i) = \spn(v_i',w_i')$ for $i = i_0, \dots, i_1$.
\item
$S_1' = \ws(a_i')$ with respect to \[v_1\oplus0, \dots, v_{i_0-1}\oplus0, v_{i_0}', \dots, v_{i_1}', 0\oplus w_{i_1+1}, \dots, 0\oplus w_{n_2}\] and $S_2'=\ws(b_i')$ with respect to \[0\oplus w_1, \dots, 0\oplus w_{i_0-1}, w_{i_0}', \dots, w_{i_1}', -v_{i_1+1}\oplus0, \dots, -v_{n_1}\oplus0.\] 
\item  For $i \leq i_0$, $a_i'=a_i$ and $b_i' =  b_i$. For $i \in (i_0, i_1)$,  the $a_i$ and $b_i$ are convex combinations of the $a_i, b_i$.
For $i \geq i_1$, $a_i' = b_i$ and $b_i' = a_i$.
\item The perturbation $S'-S$ has support and range in $\displaystyle\spn\left(\bigcup_{i\in[i_0, i_1]}\{v_i\oplus0, 0\oplus w_i\}\right)$.
\item If  $S = S_1\oplus S_2$ and $S'=S_1'\oplus S_2'$ then
\[\|S'-S\| \leq \max_{i\in[i_0, i_1)}\left(|a_i - b_i| + \frac{\pi}{2N_0}\max(|a_i|, |b_i|)\right)\]
and
\[
\|\,[S'^\ast, S']\,\|\leq \|\,[S^\ast, S]\,\|+\frac1{N_0}\max_{i\in(i_0, i_1]}||b_i|^2-|a_i|^2|.
\]
\end{enumerate}
\end{lemma}
\begin{proof}
We apply Lemma \ref{pre-gel} to the at least $N_0+1$ vectors $v_i\oplus 0$ and $0\oplus w_i$ from the statement of this lemma for $i = i_0, \dots, i_1$.
This provides what we will call $\tilde S$ expressed as the direct sum of $\tilde S_1$ and $\tilde S_2$ as follows. 

This provides vectors which we call $v_{i_0}', \dots, v_{i_1}'$ with the properties that $\tilde S$ acts as \[\tilde Sv_{i_0}' = \frac{a_{i_0}+b_{i_0}}{2}v_{i_0+1}', \dots, \tilde Sv_{i_1-2}' = \frac{a_{i_1-2}+b_{i_1-2}}{2}v_{i_1-1}',  \tilde Sv_{i_1-1}' = \frac{a_{i_1-1}+b_{i_1-1}}{2}v_{i_1}',\]
\[\tilde S (v_1\oplus 0) = S (v_1\oplus 0) = a_1(v_2\oplus 0), \dots, \tilde S (v_{i_0-1}\oplus 0) = S (v_{i_0-1}\oplus 0) = a_{i_0-1}(v_{i_0}\oplus 0),\]
and
\[\tilde S (0\oplus w_{i_1}) = S (0\oplus w_{i_1}) = b_{i_1}(0\oplus w_{i_1+1}), \dots, \tilde S( 0\oplus w_{n_2-1}) = b_{n_2-1}(0\oplus w_{n_2}), \tilde S(0\oplus w_{n_2})=0. \]
Because $v_{i_0}' = v_{i_0}\oplus 0$ and $v_{i_1}' = 0\oplus w_{i_1}$, we have
\[\tilde S_1 = \ws\left(a_1, \dots, a_{i_0-1}, \frac{a_{i_0}+b_{i_0}}{2}, \dots, \frac{a_{i_1-1}+b_{i_1-1}}{2}, b_{i_1}, b_{i_1+1},\dots, b_{n_2-1}\right),\]
with respect to the orthonormal 
\[v_1\oplus0, \dots, v_{i_0-1}\oplus0, v_{i_0}', \dots, v_{i_1-1}',  v_{i_1}', 0\oplus w_{i_1+1}, \dots, 0\oplus w_{n_2-1}, 0\oplus w_{n_2}.\]

The lemma also provides vectors which we call $w_{i_0}', \dots, w_{i_1}'$ with the properties that $\tilde S$ acts as \[\tilde Sw_{i_0}' = \frac{a_{i_0}+b_{i_0}}{2}w_{i_0+1}', \dots, \tilde Sw_{i_1-2}' = \frac{a_{i_1-2}+b_{i_1-2}}{2}w_{i_1-1}',  \tilde Sw_{i_1-1}' = \frac{a_{i_1-1}+b_{i_1-1}}{2}w_{i_1}',\]
\[\tilde S (0\oplus w_1) = S (0\oplus w_1) = b_1(0\oplus w_2), \dots, \tilde S (0\oplus w_{i_0-1}) = S (0\oplus w_{i_0-1}) = b_{i_0-1}(0\oplus w_{i_0}),\]
and
\[\tilde S ( v_{i_1}\oplus0) = S ( v_{i_1}\oplus0) = a_{i_1}( v_{i_1+1}\oplus0), \dots, \tilde S( v_{n_1-1}\oplus0) = a_{n_1-1}(v_{n_1}\oplus0), \tilde S( v_{n_1}\oplus0)=0. \]
Because $w_{i_0}' = 0\oplus w_{i_0}$ and $w_{i_1}' = -v_{i_1}\oplus0$, we have \[\tilde S_2 = \ws\left(b_1, \dots, b_{i_0-1}, \frac{a_{i_0}+b_{i_0}}{2}, \dots, \frac{a_{i_1-1}+b_{i_1-1}}{2}, -a_{i_1}, a_{i_1+1},\dots, a_{n_1-1}\right)\]
with respect to the orthonormal  
\[0\oplus w_1, \dots, 0\oplus w_{i_0-1}, w_{i_0}', \dots, w_{i_1-1}', w_{i_1}', v_{i_1+1}\oplus0, \dots, v_{n_1-1}\oplus0, v_{n_1}\oplus0.\] 

By changing the basis of this second mixed list of vectors through introducing and propagating a negative sign to the vectors after $w_{i_1}'$, we see that we can express $\tilde S_2$ unchanged as a weighted shift matrix with all non-negative weights:
\[\tilde S_2 = \ws\left(b_1, \dots, b_{i_0-1}, \frac{a_{i_0}+b_{i_0}}{2}, \dots, \frac{a_{i_1-1}+b_{i_1-1}}{2}, a_{i_1}, a_{i_1+1},\dots, a_{n_1-1}\right)\]
with respect to 
\[0\oplus w_1, \dots, 0\oplus w_{i_0-1}, w_{i_0}', \dots, w_{i_1-1}', w_{i_1}', -v_{i_1+1}\oplus0, \dots, -v_{n_1-1}\oplus0, -v_{n_1}\oplus0.\] 

We will now alter the weights of $\tilde S_1$ and $\tilde S_2$ so that the weights change more gradually while interchanging orbits. This will provide the operators $S_1'$ and $S_2'$.
Note that $i_1 - i_0 \geq N_0$. Define 
\[t_i = \left\{\begin{array}{ll}
0 & i < i_0 \\
\frac{i-i_0}{i_1-i_0} & i_0 \leq i \leq i_1 \\
1 & i > i_1
.\end{array}\right..\] So the $t_i$ satisfy $0 \leq t_i \leq 1$, $t_{i_0} = 0$, $t_{i_1}= 1$, and $|t_{i+1}-t_i| \leq 1/N_0$.

Define $a_i', b_i'$ to be non-negative satisfying
\[|a_i'|^2 = (1-t_i)|a_i|^2 + t_i|b_i|^2 = |a_i|^2 + t_i(|b_i|^2-|a_i|^2),\]
\[|b_i'|^2 = t_i|a_i|^2 + (1-t_i)|b_i|^2 = |b_i|^2 + t_i(|a_i|^2-|b_i|^2).\]
Now, change the weights of $\tilde S_1$ and $\tilde S_2$ to be $a_i'$ and $b_i'$ to obtain $S_1'$ and $S_2'$, respectively.

We now verify the statements of the lemma. (i) and (ii) are clear from our discussion of $\tilde S_1$  and $\tilde S_2$ in the beginning of the proof. 

Because $0\leq t_i \leq 1$, we have that $|a_i'|^2$ and $|b_i'|^2$ are each convex combinations of $|a_i'|^2$ and $|b_i'|^2$. Because $a_i, a_i', b_i, b_i'$ are all non-negative, we have that $a_i'$ and $b_i'$ belong to the interval $[\min(a_i, b_i), \max(a_i, b_i)]$ for $i$ in $[i_0, i_1]$. This and the above comments about $t_i$ show (iii). 
 
Because $S' - S = (S'-\tilde S) + (\tilde S - S)$, we see that (iv) holds as well by construction.
 
Because the $a_i', b_i'$ are convex combinations of the $a_i, b_i$, they are then within a distance of $|b_i-a_i|/2$ from $(a_i+b_i)/2$.
So, \[\|S'-\tilde S\| \leq \frac12\max_{i \in [i_0, i_1)}|b_i-a_i|.\] This then provides the estimate for $\|S'-S\|$. 

We now obtain the other estimate of (v). For a sequence $c_i$, let $\Delta$ denote the forward difference operator: $\Delta c_i = c_{i+1} - c_i$.
Notice that
\begin{align*}
\Delta|a'|^2_i &= \Delta|a|^2_i + t_{i+1}(|b_{i+1}|^2-|a_{i+1}|^2) - t_i(|b_i|^2-|a_i|^2)
\\
&= \Delta|a|^2_i +t_{i}(\Delta|b|^2_i-\Delta|a|^2_i) + (t_{i+1}-t_i)(|b_{i+1}|^2-|a_{i+1}|^2).
\end{align*}

So, 
\begin{align*}
\|\,[S_1'^\ast, S_1']\,\| &= \max\left(|a_1'|^2, |a_{n
_2}'|^2, \max_i\Delta|a'|^2_i\right)
\\&\leq \max\left(\|\,[S^\ast, S]\,\|, \|\,[S^\ast, S]\,\|  + \max_{i}\left((t_{i+1}-t_i)||b_{i+1}|^2-|a_{i+1}|^2|\right)\right)\\
&\leq \|\,[S^\ast, S]\,\|+\frac1{N_0}\max_{i\in(i_0, i_1]}||b_i|^2-|a_i|^2|,
\end{align*}
because $t_{i+1} = t_i$ unless $i_0\leq i< i_1$.
Interchanging the roles of $a_i$ and $b_i$ provides
\begin{align*}
\|\,[S_2'^\ast, S_2']\,\| \leq \|\,[S^\ast, S]\,\|+\frac1{N_0}\max_{i\in(i_0, i_1]}||b_i|^2-|a_i|^2|.
\end{align*}
This then provides the second inequality in the statement of the lemma.
\end{proof}
\begin{remark}
Note that we need not propagate the negative signs to the vectors $v_i$ for $i > i_1$ in our construction. What this amounts to is having a single negative weight $-b_{i_1}'$ for $S_2'$.

Note that when applying the gradual exchange lemma repeatedly on orthogonal subspaces, one can  apply the lemma as stated. This is done in detail for a simple case in Example \ref{GEL_Arrows}. Alternatively, one can apply the construction from the lemma without propagating negative signs as mentioned above, given that no weights that the lemma is applied to are ever negative.

With this modification, one may then propagate negative signs once after all the applications of the gradual exchange method to avoid relabeling or keeping track of which vectors inherit negative signs due to repeated applications along a single orbit. This difficulty comes up in the construction in Remark \ref{noNegatives} and is avoided due to this alternative in Example \ref{gep-Example}.
\end{remark}
\begin{remark}\label{weightinterchange}
The result of applying the gradual exchange lemma can be seen as perturbing
\begin{align*}
S: v_1\oplus0\overset{a_1}{\rightarrow} \cdots\overset{a_{i_0-1}}{\rightarrow} v_{i_0}\oplus0\overset{a_{i_0}}{\rightarrow} \cdots\overset{a_{i_1-1}}{\rightarrow} &v_{i_1}\oplus0\overset{a_{i_1}}{\rightarrow} v_{i_1+1}\oplus0\overset{a_{i_1+1}}{\rightarrow} \cdots\overset{a_{n_1-1}}{\rightarrow} v_{n_1}\oplus0\rightarrow0\\
S: 0\oplus w_1\overset{b_1}{\rightarrow} \cdots\overset{b_{i_0-1}}{\rightarrow} 0\oplus w_{i_0}\overset{b_{i_0}}{\rightarrow} \cdots\overset{b_{i_1-1}}{\rightarrow} &0\oplus w_{i_1}\overset{b_{i_1}}{\rightarrow} 0\oplus w_{i_1+1}\overset{b_{i_1+1}}{\rightarrow} \cdots\overset{b_{n_1-1}}{\rightarrow} 0\oplus w_{n_1}\rightarrow0
\end{align*}
to
\begin{align*}
S': v_1\oplus0 \overset{a_{1}}{\rightarrow}\cdots\overset{a_{i_0-1}}{\rightarrow} &v_{i_0}'\overset{a_{i_0}'}{\rightarrow} \cdots\overset{a_{i_1-1}'}{\rightarrow}  v_{i_1}'\overset{a_{i_1}'}{\rightarrow} 0\oplus w_{i_1+1}\overset{b_{i_1+1}}{\rightarrow} \cdots \overset{b_{n_2-1}}{\rightarrow}0\oplus w_{n_2}\rightarrow0\\
S': 0\oplus w_1 \overset{b_{1}}{\rightarrow}\cdots\overset{b_{i_0-1}}{\rightarrow}  &w_{i_0}'\overset{b_{i_0}'}{\rightarrow} \cdots\overset{b_{i_1-1}'}{\rightarrow} w_{i_1}'\overset{b_{i_1}'}{\rightarrow} -v_{i_1+1}\oplus0\overset{a_{i_1+1}}{\rightarrow} \cdots \overset{a_{n_1-1}}{\rightarrow} -v_{n_1}\oplus0\rightarrow0
\end{align*}
with the properties specified in the statement of the lemma. 

This is illustrated in the weighted shift diagram of Illustration \ref{GEL_Illustration}. 
\begin{figure}[htp]     \centering
    \includegraphics[width=12cm]{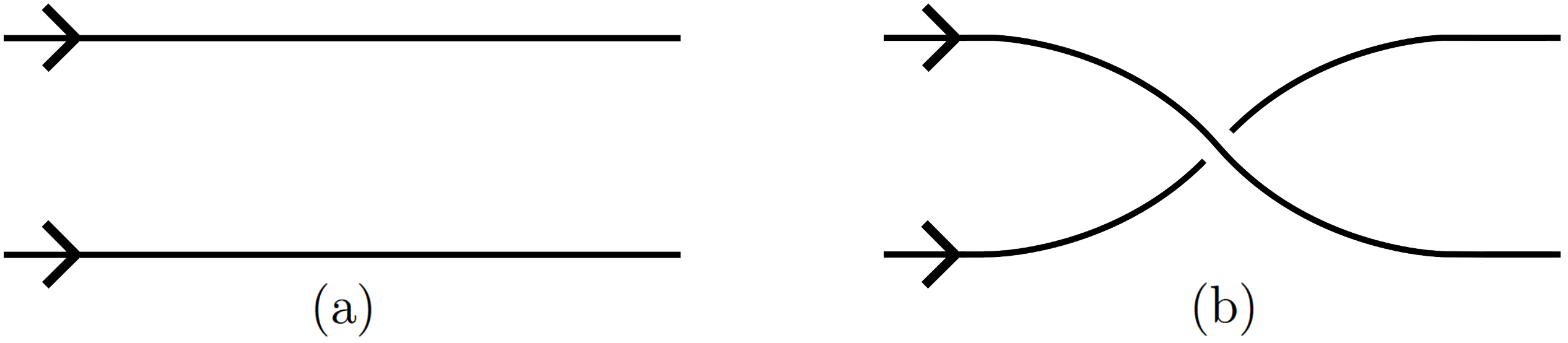}
    \caption{\label{GEL_Illustration}\dark
    Illustration of applying the gradual exchange lemma to two weighted shift diagrams.}
\end{figure}
Illustration \ref{GEL_Illustration}(a) depicts weighted shift diagrams for $S$ and Illustration \ref{GEL_Illustration}(b) depicts $S'$. The main focus of the diagrams is to illustrate that the  orbits of $S'$ begin with some vectors initially in the orbit of $S_1$ and $S_2$ then eventually are in the orbit of $S_2$ and $S_1$, respectively.
\end{remark}

\begin{remark}
This result applies to two weighted shifts whenever we have intervals of indices of length $i_1-i_0$ for each of the weighted shifts on which we apply the gradual exchange. The first index of these intervals need not be the same. We see this by simply relabeling the indices so that the first ``$v$'' vector is $v_{i_a}$, the first ``$w$'' vector is $w_{i_b}$, and the interval over which we apply the gradual exchange lemma begins with the same index $i_0$. 

Then the modification to $S$ would be as follows: 
\begin{align*}
S: &\;v_{i_a}\oplus0\overset{a_{i_a}}{\rightarrow} \cdots\overset{a_{i_0-1}}{\rightarrow} v_{i_0}\oplus0\overset{a_{i_0}}{\rightarrow} \cdots\\
S: &\;0\oplus w_{i_b}\overset{b_{i_b}}{\rightarrow} \cdots\overset{b_{i_0-1}}{\rightarrow} 0\oplus w_{i_0}\overset{b_{i_0}}{\rightarrow} \cdots.
\end{align*}
With this modification, $S$ and $S'$ are analogous to that of Remark \ref{weightinterchange}.
\end{remark}

\begin{example}\label{GEL_Arrows}
Because of
(iv) and (v) in the gradual exchange lemma, we can apply this lemma repeatedly to some direct sum of weighted shift operators without an increase in the norm of the perturbation or the self-commutator as long as no vectors are repeated in the different applications of the gradual exchange lemma.
As an example, consider $S_1 = \ws(a_i)$, $S_2 = \ws(b_i)$, $S_3 = \ws(c_i)$ with respect to $e_1, \dots, e_{2k}$ for $i = 1, \dots, 2k:=2(N_0+1)$. 
The action of $S = S_1\oplus S_2\oplus S_3$ is expressed in Illustration \ref{GEL_Arrows1}(a).

\begin{figure}[htp]     \centering
    \includegraphics[width=15cm]{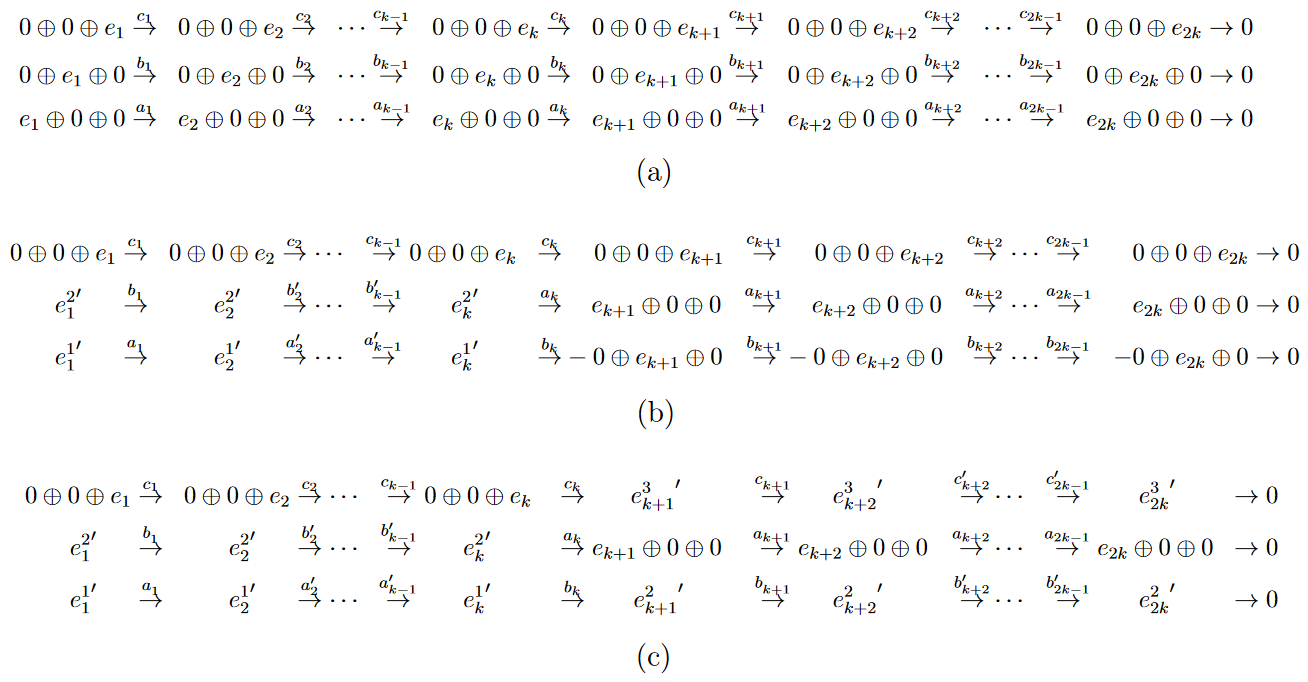}
    \caption{\label{GEL_Arrows1}\dark
    (a) is an illustration for Example \ref{GEL_Arrows} of the vectors and weights of $S_1\oplus S_2\oplus S_3$ on three invariant subspaces on which $S$ acts as a weighted shift with weights written above the arrows. (b) is an illustration of how the vectors and weights changed when applying the gradual exchange lemma to $S_2, S_1$ over the first $k$ vectors. (c) is an illustration of how the vectors and weights changed when applying the gradual exchange lemma to the result of the previous application over the orbits of $S_3, S_2$ over the latter $k$ vectors.}
\end{figure}
\begin{figure}[htp]     \centering
    \includegraphics[width=14cm]{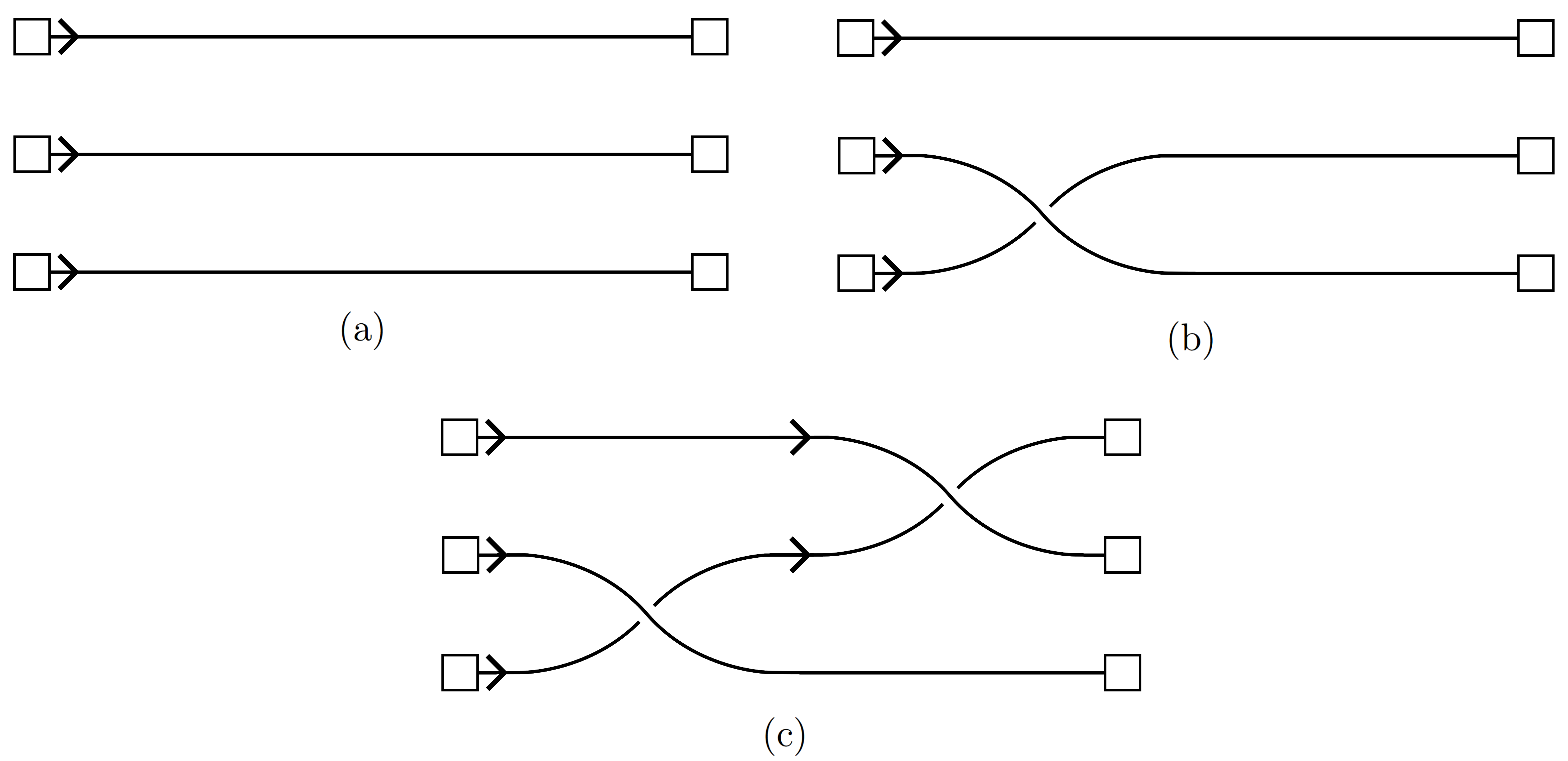}
    \caption{\label{GEL_Arrows2}\dark Illustrations of the weighted shifts from Illustration \ref{GEL_Arrows1} using weighted shift diagrams.
    }
\end{figure}

We now apply the gradual exchange lemma to $S_2, S_1$ over the vectors corresponding to $i = 1, \dots, k$ and to $S_3, S_2$ over the vectors corresponding to $i = k+1, \dots, 2k$. Note that the order of $S_2, S_1$ vs. $S_1, S_2$ is important inasmuch as it indicates which orbit's vectors inherit negative signs after the interchange. The second orbit listed inherits the negative signs.

We first apply the gradual exchange lemma to the orbits of $S_2$ and $S_1$ over the first interval to obtain Illustration \ref{GEL_Arrows2}(b). This provides vectors
\[{e_1^2}'=0\oplus e_1\oplus0, {e_2^2}', \dots, {e_{k-1}^2}', {e_k^2}'= e_k\oplus0\oplus0\]
and
\[{e_1^1}'=e_1\oplus0\oplus0, {e_2^1}', \dots, {e_{k-1}^1}', {e_k^1}'=-0\oplus e_k\oplus0\] and weights \[b_1'=b_1, b_2' \dots, b_{k-1}', b_k'=a_k, \;\;\; a_1'=a_1, a_2' \dots, a_{k-1}', a_k'=b_k.\]

Then we apply the gradual exchange lemma to the orbits of $S_3, S_2$ over the second interval to obtain Illustration \ref{GEL_Arrows2}(c). This provides vectors
\[{e_{k+1}^3}'=0\oplus0\oplus e_{k+1}, {e_{k+2}^3}', \dots, {e_{2k-1}^3}', {e_{2k}^3}'=-0\oplus e_{2k}\oplus0\]
and
\[{e_{k+1}^2}'=-0\oplus e_{k+1}\oplus0, {e_{k+2}^2}', \dots, {e_{2k-1}^2}', {e_{2k}^2}'=-0\oplus0\oplus e_{2k}\] 
and weights \[c_{k+1}'=c_{k+1}, c_{k+2}' \dots, c_{2k-1}', c_{2k}'=b_{2k}=0, \;\;\; b_{k+1}'=b_{k+1}, b_{k+2}' \dots, b_{2k-1}', b_{2k}'=b_{2k}=0.\]

We refer to the operator in  Illustration \ref{GEL_Arrows2}(c)  gotten by applying the gradual exchange lemma twice as $S'$. Notice that the perturbations in each application of Lemma \ref{GELws} are supported on and have range in orthogonal subspaces in accordance with Lemma \ref{GELws}(iv). Recall that $\#[1,k] = k = N_0+1$. So, the estimate for $\|S'-S\|$ is gotten as the maximum of the estimates from the two applications:
\begin{align*}\|S'-S\| &\leq \max\left(\max_{i\in[1, k)}\left(|a_i - b_i| + \frac{\pi}{2N_0}\max(|a_i|, |b_i|)\right),\right.\\
&\;\;\;\;\;\;\;\;\;\;\;\;\;\;\;\left.\max_{i\in[k+1, 2k)}\left(|b_i - c_i| + \frac{\pi}{2N_0}\max(|b_i|, |c_i|)\right)\right)\\
&\leq \max_{i}\left(\max(|a_i - b_i|, |b_i-c_i|) + \frac{\pi}{2N_0}\max(|a_i|, |b_i|, |c_i|)\right).
\end{align*}

The estimate for the self-commutator of $S'$ is not based on analyzing a perturbation of $S$ but instead the weights of $S'$. We then see that because each application of the gradual exchange lemma leaves the first and last weight in each orbit unchanged, the difference of the squares of weights in an orbit are those of one of the isolated applications of the gradual exchange lemma. So, we see that the norm of the self-commutator due to repeated applications is the maximum of the separate estimates:
\begin{align*}
\|\,[S'^\ast, S']\,\|&\leq \max\left(\|\,[S^\ast, S]\,\|+\frac1{N_0}\max_{i\in(1, k]}||b_i|^2-|a_i|^2|,\right.\\
&\;\;\;\;\;\;\;\;\;\;\;\;\;\;\;\left.\|\,[S^\ast, S]\,\|+\frac1{N_0} \max_{i\in(k+1, 2k]}||c_i|^2-|b_i|^2|\right)\\
&\leq \|\,[S^\ast, S]\,\|+\frac1{N_0}\max_{i}\max\left(||b_i|^2-|a_i|^2|, ||c_i|^2-|b_i|^2|\right).
\end{align*}

\end{example}

\chapter{Almost Normal Weighted Shift Matrices}
\label{6.BergThm}

Recall that the optimal upper bound by Kachkovskiy and Safarov in \cite{kachkovskiy2016distance} for how nearby an almost normal matrix $S$ is to a normal matrix $N$ is:
\[\|N-S\| \leq C_{KS}\,\|\,[S^\ast ,S]\,\|^{\alpha}\]
with $\alpha = 1/2$. It is not possible for such an estimate to hold with a different value of $\alpha$ without restrictions on the norm of $S$ for scaling reasons. A scaling-invariant form of this inequality obtained for $\|S\|=1$ would give
\[\|N-S\| \leq C_\alpha\|S\|^{1-2\alpha}\,\|\,[S^\ast ,S]\,\|^{\alpha}.\]
The main result of this chapter is Theorem \ref{BergResult} which contains an estimate of this type for $\alpha = 1/3$ for a weighted shift matrix $S$ with also the special property that $N$ can be chosen to be real when $S$ is real. 

Loring and S{\o}rensen in \cite{loring2016almost} showed the following structured Lin's theorem: if two almost commuting real self-adjoint matrices are real then they are nearby two actually commuting real self-adjoint matrices. They also showed that a real almost normal matrix is nearby a real normal matrix as well. However, these proofs are not constructive and do not provide any estimates.

In this chapter we present a refined version of Berg's constructive result in \cite{berg1975approximation} of Lin's theorem for an almost normal weighted shift matrix. Berg's construction when framed in terms of obtaining a result of this form would provide $\alpha =1/4$ due to the effect of small weights in some of the inequalities as described later. By refining the construction and estimates we obtain this result for $\alpha = 1/3$. Our modification of the construction and calculations also provide much smaller numerical constants with a structured result.

See \cite{berg1975approximation} for the details of Berg's original argument and also \cite{loring1991berg} for a discussion and illustration of how Berg's formulation of the gradual exchange concept is applied in this construction. 
Before we can say much more about the rest of this chapter, we make some definitions.

\section{Almost Normal Bilateral Weighted Shift Operators}
\begin{defn}
Given an orthonormal basis $v_1, \dots, v_n \in \C^n$, we define the bilateral weighted shift operators $T = \bws(c_1, \dots, c_n)$ to be the linear operator on $\C^n$ which satisfies $Tv_i = c_i v_{i+1}$. We use the convention that the vectors $v_k$ and weights $c_k$ are indexed cyclically.
We can express the action of $T$ as:
\[T:v_1 \overset{c_1}{\rightarrow}v_2\overset{c_2}{\rightarrow}\cdots\overset{c_{n-2}}{\rightarrow}v_{n-1}\overset{c_{n-1}}{\rightarrow}v_n\overset{c_n}{\rightarrow}v_{n+1}=v_1.\]

If $T =\bws(c_1,\dots, c_{n-1}, 0)= \ws(c_1,\dots, c_{n-1})$ then we say that $T$ is a (unilateral) weighted shift. In the previous chapter we expressed this as
\[T:v_1 \overset{c_1}{\rightarrow}\cdots\overset{c_{n-1}}{\rightarrow}v_n\rightarrow0,\] but expressed as a bilateral weighted shift this is:
\[T:v_1 \overset{c_1}{\rightarrow}\cdots\overset{c_{n-1}}{\rightarrow}v_n\overset{0}{\rightarrow}v_1.\]
\end{defn}
\begin{example}\label{ws-basisChange}
Let $T = \bws(c_1, \dots, c_n)$.
In the basis $\beta =  (v_1, \dots, v_n)$, $T$ is expressed as the matrix
\[[T]_\beta = 
\bp
0&&&&c_n\\
c_1&0&&&\\
&c_2&0&&\\
&&\ddots&\ddots&\\
&&&c_{n-1}&0
\ep.  \]
If we view $v_1, \dots, v_n$ as the standard basis, then we can think of $T$ as being this matrix. Otherwise, we can think of $T$ being unitarily equivalent to this matrix.
A simple example of this is that $T$ is unitarily equivalent to the matrix obtained by cyclically permuting the weights of $T$.

Consider the following change of basis obtained by multiplying the vectors $v_n$ by the phases $\omega_n \in \C$ with $|\omega_n|=1$.
Using the basis $\beta_\omega = (\omega_1v_1, \omega_2v_2\dots, \omega_nv_n)$, $T$ is seen to be unitarily equivalent to
\[[T]_{\beta_\omega} = 
\bp
0&&&&\frac{\omega_n}{\omega_1}c_n\\
\frac{\omega_1}{\omega_2}c_1&0&&&\\
&\frac{\omega_2}{\omega_3}c_2&0&&\\
&&\ddots&\ddots&\\
&&&\frac{\omega_{n-1}}{\omega_n}c_{n-1}&0
\ep.  \]

In particular, if we choose $\omega_1 = 1$ and define $\omega_k$ recursively by
\[\omega_{k+1} = \left\{ \begin{array}{ll} 
\frac{c_k}{|c_k|}\omega_k, & c_k \neq 0\\
1, & c_k =0\\
\end{array} \right.\]
then we obtain
\[[T]_{\beta_\omega} = 
\bp
0&&&&\omega c_n\\
|c_1|&0&&&\\
&|c_2|&0&&\\
&&\ddots&\ddots&\\
&&&|c_{n-1}|&0
\ep, \]
where if $c$ is the product of the $c_k\neq 0$ for $k < n$ then $\omega = c/|c|$.

We make a few observations. If one of the weights is zero, as in the case of a unilateral weighted shift, then all the weights can be made non-negative in this manner. If all the $c_k$ are real then $\omega_k=\pm1$, so upon conjugation by a diagonal matrix with diagonal entries $\pm1$ the weights  can be made all positive except perhaps the last. We can make all the weights positive exactly when the product of all the $c_k$ is positive.
\end{example}

\begin{example}
We use the same notation as in the previous example.
The self-commutator $[T^\ast, T] = T^\ast T-TT^\ast$ has a matrix representation of
\[[T^\ast T-TT^\ast]_\beta  = 
\bp
|c_1|^2-|c_n|^2&&&&\\
&|c_2|^2-|c_1|^2&&&\\
&&|c_3|^2-|c_2|^2&&\\
&&&\ddots&\\
&&&&|c_n|^2-|c_{n-1}|^2
\ep  \]
so that $\|[T^\ast, T]\| = \max_k||c_{k+1}|^2-|c_k|^2|$. 
So, $T$ is normal if all the $c_k$ have the same absolute value and $T$ is almost normal if the $|c_k|^2$ change slowly.

In particular, if $T$ is a unilateral weighted shift operator then
\[\|[T^\ast, T]\| = \max\left(|c_1|^2, |c_{n-1}|^2, \max_{1\leq k \leq n-2}||c_{k+1}|^2-|c_k|^2|\right)\] and $T$ is normal only if $T = 0$ identically.

A standard example of an almost normal unilateral weighed shift matrix is used in \cite{davidson1985almost} where the weights of $S$ start near zero, slowly increase to one, then decrease back to zero. We see that such a matrix is nearby a normal matrix by Lin's theorem. However, any nearby normal matrix cannot be a bilateral weighted shift matrix in the same basis since all the weights would need to have the same absolute value. It also cannot be a bilateral weighted shift matrix in any other basis since then all the singular values of the normal matrix should be the same, which is not a possible property of a small perturbation of $S$. We will show in this chapter that an almost normal weighted shift is nearby a direct sum of normal weighted shift matrices in some bases.

\end{example}

We now complete our introduction to this chapter by discussing the results that we obtain.
Lemma \ref{gel-forNearbyNormal} and Lemma \ref{HelpingBerg} can be seen as an adaption of Berg's original argument. There are two main differences. First, our implementation of the ``gradual exchange'' idea in Lemma \ref{gel-forNearbyNormal} has a simpler definition, has a tighter estimate, and does not involve complex numbers at the expense of having the negative sign in $\eta_{k_0} = -v_0$. 

The second difference is that Berg expressed his estimates in terms of $\max_k||c_{k+1}|-|c_k||$. The motivation for this is based in the characterization of a normal operator as one that satisfies $\|Nv\| = \|N^\ast v\|$ for all vectors $v$. We showed above that the norm of the self-commutator $[S^\ast , S]$ equals $\max_k||c_{k+1}|^2-|c_k|^2|$. Although Berg's construction produces an estimate of the form
\[\|N-S\| \leq Const.\sqrt{\max_k||c_{k+1}|-|c_k||}\]
for $\|S\|=1$ and $\max_k||c_{k+1}|-|c_k||$ small enough, this result produces an estimate in terms of the self-commutator having exponent $\alpha = 1/4$ due to
\[\max_k||c_{k+1}|-|c_k|| \leq \sqrt{\max_k||c_{k+1}|^2-|c_k|^2|}.\]
Because \[|c_{k+1}|^2-|c_k|^2 = (|c_{k+1}|-|c_k|)(|c_{k+1}|+|c_k|),\]
the inequality above is asymptotically sharp when the difference $||c_{k+1}|-|c_k||$ has a similar size as the sum $|c_{k+1}|+|c_k|$. This can happen when, for instance, $|c_{k+1}|$ is much larger than $|c_{k}|$.

Then in Theorem \ref{BergResult} we present a version of a condition of Theorem 2 of \cite{berg1975approximation} that does not require the operator to have norm $1$ or have any requirement on the size of the self-commutator.
This includes a result with exponent $\alpha = 1/3$ 
and also an estimate with $\alpha = 1/2$ with a scaling-invariant factor that is large when there are weights of the matrix that are much smaller than the norm. 

\section{Modification of Berg's Construction}
We now proceed to the results of this chapter. 
The proof of \cite{berg1975approximation} was formulated in terms of a recursive algorithm. We isolate this part as the following lemma so that the entire proof in Lemma \ref{HelpingBerg} is expressed as a single step. The modification of Berg's construction here can be seen as applying the gradual exchange lemma to two portions of $S$. 

\begin{lemma}\label{gel-forNearbyNormal}
Suppose that $S$ is a linear map on $\C^n$ such that there are orthonormal vectors $v_i, w_j$ in $\C^n$ with $Sv_i=bv_{i+1}, Sw_j = bw_{j+1}$ for $i = 0, \dots, k_0$ and $j = 0, \dots, k_0+1$.

Let $\alpha_k = \cos\left(\frac{\pi k}{2k_0}\right)$ and $\beta_k = \sin\left(\frac{\pi k}{2k_0}\right)$ and define \[\xi_k = \alpha_kv_k+\beta_kw_k, \;\;\; \eta_k = -\beta_kv_k+\alpha_kw_k\] for $k = 0, \dots, k_0$.
Note that $(\xi_k, \eta_k)$ is gotten by rotating $(v_k, w_k)$ in a two dimensional subspace by $\pi k/2k_0$ so that $(\xi_0, \eta_0) = (v_0, w_0)$ and $(\xi_{k_0}, \eta_{k_0}) = (w_{k_0}, -v_{k_0})$.
Let $S'$ be the linear operator that satisfies
\[S'\xi_k= a\xi_{k+1}, S'\eta_k = b\eta_{k+1}, 0 \leq k \leq k_0-1,\]
\[S'\xi_{k_0}= S'w_{k_0} = aw_{k_0+1},  S'\eta_{k_0}= -S'v_{k_0} = -Sv_{k_0},\]
and equals $S$ on the orthogonal complement of the span of the $v_k$ and $w_k$, $k = 1, \dots, k_0$. 

Then 
\[\|S'-S\| \leq |b-a| + |b|\frac{\pi}{2k_0}.\]
\end{lemma}
\begin{proof}
Let $\mathcal U_k = \spn(v_k, w_k)$ for $k = 0, \dots, k_0$. Notice that both $(v_k, w_k)$ and $(\xi_k, \eta_k)$ form orthonormal bases for $\mathcal U_k$.  We first claim that \[\|S' - S\| = \max_k\|(S'-S)P_{\mathcal U_k}\|.\] 
Notice that $S'-S$ is only non-zero on the span on the $\mathcal U_k$. Also, $S'-S$ maps $\mathcal U_{k}$ into $\mathcal U_{k+1}$ for $0 \leq k \leq k_0-1$ and $S'-S$ maps $\mathcal U_{k_0} $ into the span of $w_{k_0+1}$ as seen below. So, the restrictions $(S'-S)P_{\mathcal U_k}$ have orthogonal ranges which is enough to prove this claim.

We now continue with calculating $\|(S'-S)P_{\mathcal U_k}\|$ for $k = k_0$:
\[(S'-S)\xi_{k_0}= (a-b)w_{k_0+1}\]
\[(S'-S)\eta_{k_0}=0.\]
So, $\|(S'-S)P_{\mathcal U_{k_0}}\| \leq |b-a|$.
This also shows that $S'-S$ maps $\mathcal U_{k_0} $ into the span of $w_{k_0+1}$ as referenced above.

Recall the real orthogonal rotation matrix
\[R_\theta = \bp \cos(\theta) & -\sin(\theta) \\ \sin(\theta) & \cos(\theta)\ep\] which satisfies $R_{\theta}R_{\varphi} = R_{\theta+\varphi}$ and has eigenvalues $e^{\pm\pi i\theta}$.
Notice that if $(e_1, e_2)$ is the standard basis of $\C^2$ then the coordinates of  $\xi_k$ and  $\eta_k$ with respect to $(v_k, w_k)$ are exactly those of $R_{\pi k/2k_0}e_1$ and $R_{\pi k/2k_0}e_2$, respectively. 

We now consider the case when $0 \leq k \leq k_0-1$. We will represent $S$ and $S'$ on $\mathcal U_k$ with the matrices 
$[SP_{\mathcal U_{k_0}}], [S'P_{\mathcal U_{k_0}}]$ 
with respect to the bases $(\xi_k, \eta_k)$ of $\mathcal U_k$ and $(v_{k+1}, w_{k+1})$ of $\mathcal U_{k+1}$. 
We obtain
\[[SP_{\mathcal U_{k}}] = \bp b\alpha_k& -b\beta_k\\b\beta_k &b\alpha_k \ep = bR_{\pi k/2k_0}\]
and
\[[S'P_{\mathcal U_{k}}] = \bp a\alpha_{k+1}& -b\beta_{k+1}\\ a\beta_{k+1} &b\alpha_{k+1} \ep = \bp (a-b)\alpha_{k+1}& 0\\ (a-b)\beta_{k+1} &0 \ep + bR_{\pi(k+1)/2k_0}.\]
So,
\begin{align*}
\|(S'-S)P_{\mathcal U_k}\| &\leq |b|\|R_{\pi k/2k_0}-R_{\pi(k+1)/2k_0}\|+ |b-a| = |b|\|I-R_{\pi/2k_0}\|+ |b-a| \\
&= |b||1-e^{\pi i/2k_0}| + |b-a| \leq \frac{\pi}{2k_0}|b| + |b-a|.
\end{align*}

\end{proof}

We now move to our modification of the main construction from \cite{berg1975approximation}. Note that an explicit construction is not provided there for the first step of the following lemma so we provide it for completeness. We also express our estimate in terms of $\|S\|$ because it will allow us to optimize the constant $C_\alpha$ later.
\begin{lemma} \label{HelpingBerg}
Suppose that $S \in M_n(\C)$ is a bilateral weighted shift matrix with weights $c_1, \dots, c_{n}$. Let $M \geq 4$ be an even integer. If
\[\|\,[S^\ast, S]\,\| < \frac{1}{M^{3}}\]
then there is a normal matrix $N$ such that  
\[\|N-S\| < \left(\|S\|\frac{\pi M}{M-2}+2\right)\frac{1}{M}.\]

Additionally, $N$ is a direct sum of weighted shift unitary matrices in another basis with $\|N\| \leq \|S\|$. In particular, the weights in all the direct sums are between $\min_k|c_k|$ and $\max_k|c_k|=\|S\|$. 

Also, if $S$ is real then $N$ is real and the basis in which $N$ is a direct sum of real normal weighted shift matrices is obtained using a real orthogonal matrix. 

The same conclusion holds if instead of the commutator estimate above we have that all the weights $c_k$ satisfy $|c_k|\geq \sigma$ and we have the commutator estimate\[\|\,[S^\ast, S]\,\| < \frac{2\sigma}{M^{2}}.\]
\end{lemma}
\begin{proof}
The proof proceeds in four steps. Before step 1, we provide some inequalities used in the proof. In the first step we show that we can group the basis vectors into blocks that roughly correspond to level sets of the $|c_k|$. In the second step, we lay out how to perturb $S$ on certain pairs of basis vectors to obtain $N$. In the third step, we verify the norm inequality for $\|N-S\|$. In the fourth step we verify that $N$ is normal.

As with the weights $c_k$, all intervals of indices that we construct will be cyclically indexed by integers. 
Because all such intervals will be proper subsets of the set of all indices, 
it makes sense to use the terminology of ``first'' and ``last'' entry of such an interval to refer to the left-most and the right-most element due to the orientation of increasing the indices cyclically.

We first perform some estimates. We know that
\[||c_{k+1}|^2-|c_k|^2| < \frac1{M^3}.\]
This implies that 
\[\sqrt{||c_{k_1}|^2 - |c_{k_2}|^2|} < 
\sqrt{\frac{|k_1-k_2|}{M^3}}.\]
We relate this to an estimate for the differences of the absolute values of the weights. For $x, y \in \C$, 
\[||x|-|y||=\sqrt{(|x|-|y|)^2} \leq \sqrt{||x|-|y||(|x|+|y|)} =\sqrt{||x|^2-|y|^2|}.\]
Using this, we see that
\[||c_{k_1}| - |c_{k_2}||\leq \sqrt{||c_{k_1}|^2 - |c_{k_2}|^2|} < 
\sqrt{\frac{|k_1-k_2|}{M^3}}.\]

If we had the alternative restriction that $|c_k|\geq \sigma$ and $||c_{k+1}|^2-|c_k|^2| < 2\sigma/M^2$ then we would obtain the estimate: \[||c_{k_1}|-|c_{k_2}|| = \frac{||c_{k_1}|^2-|c_{k_2}|^2|}{|c_{k_1}|+|c_{k_2}|}\leq \frac{|k_1-k_2|\max_k ||c_{k+1}|^2-|c_k|^2|}{2\sigma}< \frac{|k_1-k_2|}{M^2}.\]       
So,
in either case
we have
\begin{align}\label{diffEst}
|k_1-k_2|\leq M \Rightarrow ||c_{k_1}|-|c_{k_2}||< \frac1M.
\end{align}

\underline{Step 1}: We now begin with the construction. Dividing $n$ by $M$ with remainder gives $q, d \in \N_0$ with $n = qM+d$ and $0 \leq d < M$. 

We first address the case where $n \leq 2M$. 
Because the distance is calculated cyclically, we see that the distance from $\max_j |c_j|$ to $\min_j |c_j|$ is less than $1/M$ by Equation (\ref{diffEst}).
We then change $c_k$ radially in $\C$ so that they all have the absolute value equal to $\frac{1}{2}(\max_j |c_j|+\min_j |c_j|)$. This provides a normal matrix $N$ with the desired properties and
\[\|N-S\| < \frac1{2M}.\]

We now assume that $n > 2M$.
Choose an integer $\tilde k$ so that $ |c_{\tilde k}| = \|S\|$. 
Then partition the sequence $1, \dots, n$ into the intervals $I_j'$ for $j = 0, \dots, q$ of consecutive integers as follows. 
We require all intervals to contain $M$ integers except the interval that contains $\tilde k$ which will contain $M+d$ integers. 
We will choose this particular interval so that there are $d$ integers to the left of $\tilde k$ and $M-1$ integers to its right. 
We relabel the basis vectors $e_k$ if necessary by cycling the indices (by at most $d$) so that $I_1'$ begins with $e_1$ to avoid any interval containing both $e_1$ and $e_n$ due to the shifting of the intervals when we included the additional $d$ indices in the interval containing $\tilde k$.

Because we assume that $n > 2M$, we then have that $I_j'$ are at least two consecutive disjoint intervals. Let $s_r = \|S\|-r/M$ for $0\leq r \leq r_0:=\lceil M\|S\|\rceil$. If $M\|S\|$ is an integer then $s_{r_0} = 0$. If it is not an integer, then $s_{r_0} < 0$. This provides a list of real numbers $s_r$:
\[\|S\|=s_0 > s_1 > \dots, s_{r_0-1} >0\geq s_{r_0}\]
spaced by $1/M$. For $c\in [0, \|S\|]$, we define the function $\sround(c):= \min\{s_r: s_r \geq c\}$ that ``rounds up'' to a nearby value of $s_r$. We know then that $\sround(c)=s_r$ for some $r$ and $\sround(c)-1/M< c \leq \sround(c)$.
We will replace all the weights in an interval $I_j'$ with a single absolute value $s_r$ now.

Let $A_j = \{|c_k|: k \in I_j'\}$. By Equation (\ref{diffEst}), we see that $\diam A_j < 1/M$ as follows. This is clearly true for the intervals $I_j'$ containing $M$ integers but also for the potentially longer interval since the index $\tilde k$ of a weight with maximum absolute value is less than $M$ away from the other integers in the interval.

So, we define $a_j = \sround(\min A_j)$. Then $a_j-1/M < \min A_j \leq a_j$ so that $A_j \cap (a_j-1/M, a_j] \neq \emptyset$ and $A_j \subset (a_j-1/M, a_j+1/M)$. 
In particular, $||c_k| - a_j| < 1/M$ for all $k \in I_j'$. Note that when $\max A_j = \|S\|$, because there is a distance of less than $M$ from a place where this maximum can take place this shows that $\min A_j > \|S\|-1/M$ so $a_j = s_0=\|S\|$.  Note also that $a_j \geq 0$ and the situation where $a_j = 0$ is only possible when both $s_{r_0} = 0$ and some weight in $I_j'$ equals zero.

Let $\underline{a}$ denote the smallest of the $a_j$. Choose a value $\underline{j}$ of $j$ so that $a_j=\underline{a}$ and then choose a value $\underline{k}$ of $k$ so that $\underline{k}$ lies in $I'_{\underline{j}}$. 
Using the change of basis like that indicated in Example \ref{ws-basisChange}, we see that $T$ is unitarily equivalent to a matrix with $c_k \geq 0$ except possibly $c_{\underline{k}}$. Each $k \neq \underline{k}$ lies in an interval $I_{j}'$ and we replace $c_k$ with $a_j$. 
We change $c_{\underline{k}}$ radially in $\C$ to have the absolute value equal to $a_{\underline{j}}$. 
Let $S_1$ denote this perturbation of $S$ so that
\[\|S_1-S\| < \frac1M.\]

If there is only one distinct value of $a_j$ then $S_1$ is normal and we are done. We will now assume that there are multiple distinct values of $a_j$.

We now show that consecutive weights $a_j$ are either equal or differ by at most $1/M$. 
Without loss of generality, suppose that $a_j < a_{j+1}$. Then there is a $k_j \in I_j'$ such that $|c_{k_j}| = \min A_j$. Because the intervals $I_{j}'$ and $I_{j+1}'$ are consecutive, there is an index $k_{j+1}$ of $I_{j+1}'$ that is within $M$ of $k_j$. So,
\[\min A_{j+1} \leq |c_{k_{j+1}}| \leq |c_{k_j}|+||c_{k_j}|-|c_{k_{j+1}}|| < \min A_j + \frac1M.\]
So, $a_{j+1} \leq a_j + 1/M$, which is what we wanted to show.

So, we now merge consecutive intervals $I_j'$ of the same weight $a_j$ to obtain reindexed intervals $I_j$ for $j = 1, \dots, q_0\leq q$ where the reindexed weights $a_j$ of the perturbed weighted shift matrix satisfy $a_{j+1} = a_j \pm1/M$.

\vspace{0.1in}

\underline{Step 2}:
We now need to determine how we will apply Lemma \ref{gel-forNearbyNormal}. The non-negative weights of $S_1$ are  spaced by $1/M$: $s_{n_0}> s_{n_0}-1/M > \dots > s_{n_1}$. 
The only weight that is potentially not non-negative is a single weight of minimal absolute value $s_{n_1}=\underline{a}$. 
Now, for a non-negative weight $b$, let $J_b$ be the level set  $\bigcup_{j: a_j \geq b} I_j.$ 
Then $J_b$ is the union of maximal sequences of consecutive intervals, each of the form $I_{j_0}, I_{j_0+1}, \dots, I_{j_1}$. We refer to $I_{j_0} \cup \cdots \cup I_{j_1}$ as a ``connected component'' of $J_b$ in analogy to how every open set in the unit circle is the disjoint union of countably many (connected) open arcs.

Define the integer $k_0 =
(M-2)/2$. So, each interval $I_j$ contains at least $2(k_0+1)$ integers.
Consider a connected component of $J_b $ with weight $b > \underline{a}$. 
Suppose that the connected component is formed by $I_{j_0}, \dots, I_{j_1}$. We will apply 
Lemma \ref{gel-forNearbyNormal} to obtain a perturbation of $S_1$ on the span of the first $k_0+1$ vectors of $I_{j_0}$, the last $k_0+1$ vectors of $I_{j_1}$, and the first vector of $I_{j_1+1}$. Namely, write $I_{j_0} = \{i_0, \dots, i_0'\}$ and $I_{j_1} = \{i_1, \dots, i_1'\}$ and observe that
$\#[i_0,i_0'] \geq M$ and $\#[i_1,i_1'] \geq M$. We define $v_i = e_{i_0+i}$  for $i = 0, \dots, k_0$ and $w_j= e_{i_1'-k_0+j}$  for $j = 0, \dots, k_0+1$.  Notice that $w_{k_0+1}$ is the first vector of $I_{j_1+1}$. We now apply Lemma \ref{gel-forNearbyNormal} to $S$ with $b$ and $a = b-1/M$. 

We do this for all such connected components of all such $J_b$ with $b > \underline{a}$. We claim that this provides the desired normal matrix $N$. 

\underline{Step 3}: 
We will obtain the estimate for  $\|N-S\|$. 
Because we perturb $S_1$ on orthogonal subspaces using Lemma \ref{gel-forNearbyNormal}, we see that 
\[\|N-S\|\leq \|N-S_1\| + \|S_1-S\| < \frac2M + \|S\|\frac{\pi}{M-2}.\]

\underline{Step 4}:
Because it is clear that $N$ satisfies the other conditions, we now prove that $N$ is normal as a direct sum of normal bilateral weighted shift operators. In order for each summand to be normal, it is necessary that each of these weighted shifts all have weights that have the same absolute value.

Consider a weight $b$.
There are three cases to consider. Compare the arguments for these three cases to Remark \ref{CaseRemark} which contains illustrations for them. 

We first consider a connected component of $J_b$ for $b > \underline{a}$ composed of $I_{j_0}, \dots, I_{j_1}$ which corresponds to cases 1 and 2 below. 
The value of $k_0$ was chosen so that $k_0+1 = \frac12M\leq \frac12\#I_{j_0}, \frac12\#I_{j_1}$.
Case 1 corresponds to when $I_{j_0}=I_{j_1}$ so that the $v_k$ and $w_k$ for $k=0,\dots, k_0$ are orthogonal vectors of the same interval $I_{j_0}$. 
Case 2 is when the intervals in question are distinct.
We now introduce some statements that apply for these first two cases.

Using the notation in Step 2, we can define the vectors $v_i$ and $w_j$ by rewriting the vectors $e_{i_0}, \dots, e_{i_1'}$ as
\begin{align}\label{vectors}
v_0, \dots, v_{k_0}, e_{i_0+k_0+1}, \dots, e_{i_1'-k_0-1}, w_0, \dots, w_{k_0}
\end{align}
and having $w_{k_0+1} = e_{i_1'+1}$. Because this connected component of $J_b$ is not all of the indices, we see that $e_{i_1'+1}$ belongs to $I_{j_1+1}$ and is thus orthogonal to the other vectors listed above. 
Note that by construction the vector  $e_{i_1'+1}$ is not included in any other application of Lemma \ref{gel-forNearbyNormal} because it is the first vector of an interval that cannot be the first interval of a level set $J_b$ for any $b$.

The span of these vectors equals the span of these two groups of vectors:
\[e_{i_0+k_0+1}, \dots, e_{i_1'-k_0-1}, \eta_0, \dots, \eta_{k_0}\]
\[\xi_0, \dots, \xi_{k_0}.\]
Recall that $w_0 = \eta_0$ and $N\eta_{k_0} = -Sv_{k_0}=-be_{i_0+k_0+1}$.

Note that if this connected component has exactly $M$ indices (which can happen only in Case 1 below) then $i_0+k_0+1=i_1'-k_0$ so $e_{i_0+k_0+1}=w_0=\eta_0$ and $e_{i_1'-k_0-1}=v_{k_0}=-\eta_{k_0}$.
So, to avoid redundancies, it is best to think of the $e$ and $\eta$ list of vectors as just
\[\eta_0, \dots, \eta_{k_0}\]
where then $N\eta_{k_0}=-b\eta_0$.

Now, we know that $N$ acts on the second grouping of vectors as:
\begin{align}\label{Case1Teleport}
N: \xi_0 \overset{b-\frac1M}{\rightarrow}\xi_1\overset{b-\frac1M}{\rightarrow}\xi_2\overset{b-\frac1M}{\rightarrow}\cdots\overset{b-\frac1M}{\rightarrow} \xi_{k_0}\overset{b-\frac1M}{\rightarrow} w_{k_0+1}=e_{i_1'+1}.
\end{align}
The second grouping of vectors will be put together with vectors from $J_{b-1/M}$. We will now use this information directly for the first two cases.

\underline{Case 1}: In this first case, the connected component will not contain any interval $I_j$ of a higher weight $a_j$. We have the vectors in Equation (\ref{vectors}).
The vectors $v_k, w_k$ all correspond to vectors in $I_{j_0}$. With $M=2(k_0+1)$, we have at least this many indices in $I_{j_0}$: $i_0, \dots, i_0'$.

$N$ acts on the first grouping of vectors as a bilateral weighted shift with weights having absolute value $b$:
\begin{align}\label{Case1Loop}
N: e_{i_0+k_0+1}&\overset{b}{\rightarrow}e_{i_0+k_0+2}\overset{b}{\rightarrow} \cdots\overset{b}{\rightarrow} e_{i_1'-k_0-1}\overset{b}{\rightarrow}e_{i_1'-k_0}=\eta_0\overset{b}{\rightarrow}\eta_1\overset{b}{\rightarrow} \cdots\nonumber\\
&\overset{b}{\rightarrow}\eta_{k_0-1}\overset{b}{\rightarrow}\eta_{k_0}= -e_{i_0+k_0}\overset{-b}{\rightarrow}e_{i_0+k_0+1}.
\end{align}
So, the first grouping of vectors spans an invariant subspace of $N$ on which $N$ is normal. When $I_{j}$ has only $M$ indices, one should think of the above orbit of $N$ as
\[N:\eta_0\overset{b}{\rightarrow}\cdots\overset{b}{\rightarrow}\eta_{k_0}\overset{-b}{\rightarrow}\eta_0.\]

\underline{Case 2}: In this case, the connected component of $J_b$ will contain some intervals $I_j$ of higher weights $a_j > b$ and we also require that $b > \underline{a}$. We have the vectors in Equation (\ref{vectors}) with at least $M$ vectors between $v_{k_0}$ and $w_{0}$ coming from $J_{b+1/M}$. 
We decompose the middle block of vectors in (\ref{vectors}):
\begin{align*}
 e_{i_0+k_0+1}, \dots, e_{i_1'-k_0-1}
\end{align*}
as
\begin{align*}
&e_{j_1}, \dots, e_{j_2}, \overline{e}_{j_1^+}, \dots, \overline{e}_{j_2^+}, e_{j_3}, \dots, e_{j_4}, \overline{e}_{j_3^+}, \dots, \overline{e}_{j_4^+}, \cdots,\\ &e_{j_m}, \dots, e_{j_{m+1}}, \overline{e}_{j_m^+}, \dots, \overline{e}_{j_{m+1}^+}, e_{j_{m+2}}, \dots, e_{j_{m+3}}  
\end{align*}
where the block $\overline{e}_{j_r^+}, \dots, \overline{e}_{j_{r+1}^+}$ corresponds to each of the connected components of $J_{b+1/M}$ within the component of $J_b$ on which we are focusing. The remaining blocks of the form $e_{j_r}, \dots, e_{j_{r+1}}$ belong to $J_b$. Note that the first and/or last block of this form may be empty.

Based on Case 1 for $b+1/M$ or the (recursive) application of Case 2 for $b+1/M$, we obtain the passed-down vectors $\overline\xi_{k_r}, \dots, \overline\xi_{k_{r+1}}$ within the span of the block $\overline{e}_{j_r^+}, \dots, \overline{e}_{j_{r+1}^+}$ such that $\overline\xi_{k_{r}}=\overline{e}_{j_r^+}$ and by Equation (\ref{Case1Teleport}),
\begin{align}\label{Case2PartialLoop}
N:\overline{e}_{j_r^+}=\overline\xi_{k_r}\overset{b}{\rightarrow}\overline\xi_{k_r+1}\overset{b}{\rightarrow}\cdots\overset{b}{\rightarrow}\overline\xi_{k_{r+1}}\overset{b}{\rightarrow}e_{j_{r+2}}.
\end{align}

Now, for this case we will use the $\overline\xi$ and the $\eta$ vectors to make a closed orbit with the $e_{j}$ vectors of this block. The $\xi$ vectors will be passed down for use for $J_{b-1/M}$. So, putting together Equations (\ref{Case1Loop}) and (\ref{Case2PartialLoop}) we see that

\begin{align*}
\eta_0, \dots, \eta_{k_0}, 
\;&e_{j_1}, \dots, e_{j_2}, \overline\xi_{k_1}, \dots, \overline\xi_{k_2}, e_{j_3}, \dots, e_{j_4}, \overline\xi_{k_3}, \dots, \overline\xi_{k_4}, \cdots,\\ &e_{j_m}, \dots, e_{j_{m+1}}, \overline\xi_{k_m}, \dots, \overline\xi_{k_{m+1}}, e_{j_{m+2}}, \dots, e_{j_{m+3}} 
\end{align*}
form an invariant subspace for $N$ on which $N$ is a bilateral weighted shift with weights $\pm b$:

\begin{align*}
N: \eta_0 \overset{b}{\rightarrow} 
\cdots \overset{b}{\rightarrow} \eta_{k_0} \overset{-b}{\rightarrow}
\;&e_{j_1} \overset{b}{\rightarrow} 
\cdots \overset{b}{\rightarrow} e_{j_2} \overset{b}{\rightarrow}
\overline{e}_{j_1^+}=
\overline\xi_{k_1}\overset{b}{\rightarrow}  
\cdots \overset{b}{\rightarrow} \overline\xi_{k_2} \overset{b}{\rightarrow}\\
&e_{j_3}\overset{b}{\rightarrow} 
\cdots \overset{b}{\rightarrow} e_{j_4}\overset{b}{\rightarrow}\overline{e}_{j_3^+}= \overline\xi_{k_3} \overset{b}{\rightarrow} 
\cdots\overset{b}{\rightarrow} \overline\xi_{k_4}\overset{b}{\rightarrow} 
\cdots\overset{b}{\rightarrow}\\ &e_{j_m}\overset{b}{\rightarrow} 
\cdots\overset{b}{\rightarrow} e_{j_{m+1}}\overset{b}{\rightarrow}\overline{e}_{j_{m}^+}= \overline\xi_{k_m}\overset{b}{\rightarrow} 
\cdots\overset{b}{\rightarrow} \overline\xi_{k_{m+1}}\overset{b}{\rightarrow}\\
&e_{j_{m+2}}\overset{b}{\rightarrow}
\cdots\overset{b}{\rightarrow} e_{j_{m+3}}\overset{b}{\rightarrow} w_0 = \eta_{0} 
\end{align*}

Note that if one of the blocks of $e$ vectors is empty then the corresponding vectors would just be skipped in showing the orbit of $N$. For instance, if the first $e$ block is empty then we would instead have $\eta_{k_0} \overset{-b}{\rightarrow}
 \overline\xi_{k_1}$.

\underline{Case 3}: In this last case, $b = \underline{a}$. Focus on the intervals $I_j$ such that $a_j = \underline{a}$. The complement of the union of these intervals is $J_{\underline a + 1/M}$. Consider a connected component of $J_{\underline a + 1/M}$ as in Case 2. Consider the interval(s) $I_{\underline j_1}$ and $I_{\underline j_2}$ with weight $\underline a$ that are immediately before and after this connected component. When $J_{\underline a + 1/M}$ has one connected component, it is the case that $\underline j_1 = \underline j_2$ as in Illustration \ref{wsSystem1}. Illustration \ref{wsSystemExtended} illustrates a more general case.

Let $e_{i^\ell_0}, \dots, e_{i^\ell_1}$ be the vectors corresponding to $I_{\underline j_1}$ and $e_{i^r_0}, \dots, e_{i^r_1}$ be the vectors corresponding to $I_{\underline j_2}$. We can express the action of $N$ on these basis vectors as
\begin{align*}
N:&\;e_{i^\ell_0}\overset{\underline{a}}{\rightarrow} \cdots\overset{\underline{a}}{\rightarrow} e_{i^\ell_1}\overset{\underline{a}}{\rightarrow} e_{i^\ell_1+1}\\
N:&\;e_{i^r_0}\overset{\underline{a}}{\rightarrow} \cdots\overset{\underline{a}}{\rightarrow} e_{i^r_1}\overset{\underline{a}}{\rightarrow}e_{i^r_1+1}
\end{align*}
generically. It is possible that a single one of these weights is not positive but instead just has absolute value equal to $\underline{a}$.

We proceed in a way similar to Case 2 except that we do not change any of the vectors of the lowest weight because the original operator that we started with was a bilateral shift.
Based on Case 1 or the application of Case 2 for $\underline a+1/M$, we obtain vectors $\overline\xi_0, \dots, \overline\xi_{k_0}$ such that $\overline\xi_0=e_{i^\ell_1+1}$ and by Equation (\ref{Case1Teleport}),
\begin{align}
\nonumber
N:e_{i^\ell_1+1}=\overline\xi_0 \overset{\underline a}{\rightarrow}\overline\xi_1\overset{\underline a}{\rightarrow} \cdots\overset{\underline a}{\rightarrow}\overline\xi_{k_0}\overset{\underline a}{\rightarrow} e_{i^r_0}.
\end{align}

This shows that by including the vectors $\overline \xi_k$ that were passed down as follows:
\[e_{i^\ell_0}, \dots, e_{i^\ell_1}, \overline\xi_0, \dots, \overline\xi_{k_0}, e_{i^r_0}, \dots, e_{i^r_1}\]
then $N$ maps each vector in the list to the next multiplied by $\pm \underline{a}$ except perhaps the last vector as its image might be orthogonal the span of the vectors listed here.

However, once we have included all the vectors that were passed down from the connected components of $J_{\underline a + 1/M}$ we see that this provides a subspace on which $N$ acts as a bilateral weighted shift with weights having absolute value $\underline{a}$.

This completes the verification and also the proof of this lemma.

\end{proof}
\begin{remark}\label{CaseRemark}
\begin{figure}[htp]  
    \centering
    \includegraphics[width=14cm]{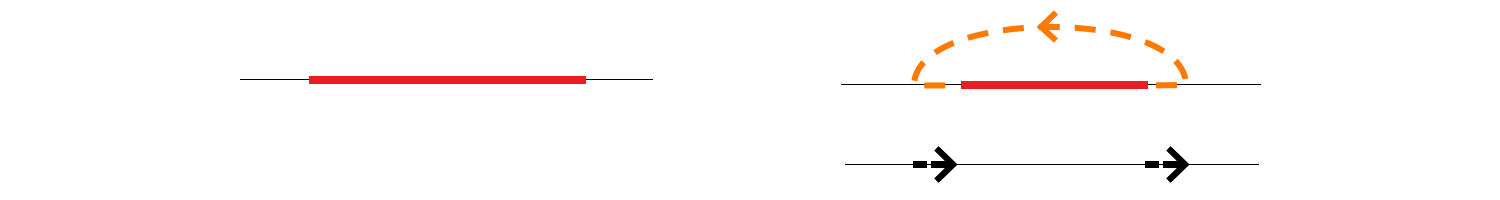}
    \caption{\dark \label{wsSimpleCase1}\dark 
    Illustration of Case 1 in the proof of Lemma \ref{HelpingBerg}.}
    \end{figure}

In this remark, we discuss Illustrations \ref{wsSimpleCase1}, \ref{wsSimpleCase2}, and \ref{wsSimpleCase3} as illustrations of the constructions in cases 1, 2, and 3, respectively, in the proof of Lemma \ref{HelpingBerg}.

\underline{Case 1}:
The red line on the left side corresponds to the vectors $e_i$ that correspond to a connected component of the interval $J_b$. One should think of $e_{i}$ as a point on this red line that moves from the left-most part of the red line to its right-most point as $i$ increases from $i_0$ to $i_1'$. The reason that we have singled out these specific basis vectors with a red line is that they have weight $b$ for $S'$. The thin black lines starting before and continuing after the red line segment correspond to basis vectors $e_i$ for $i < i_0$ and $i > i_1'$, respectively, and will have potentially different weights because they do not belong to this connected component of $J_b$.

The right side of this illustration illustrates $N$ acting on the vectors $\eta_k$ and $e_i$ and the $\xi_k$. The orbit of $N$ in Equation (\ref{Case1Loop}) is illustrated in the top right side of this illustration. The red line corresponds to $e_{i_0+k_0+1}, \dots, e_{i_1'-k_0-1}$ and the orange loop corresponds to the action of $N$ on the $\eta_k$. The weights of $N$ on this orbit are the same as the weights of the red line, namely $b$.

The action of $N$ on the $\xi_k$ is illustrated in the line diagram on the bottom right of this illustration. 
The vectors $\xi_k$ belong to the span of the vectors $e_i$ that correspond to the beginning and ending portions of the red line that vertically line up with the two arrows in the diagram. 
Because $\xi_0=e_{i_0}$ and $\xi_{k_0}=e_{i_1'}$, we view the action of $N$ on the $\xi_k$ as a perturbation of $S$ with the orbit of $N$ to starting at $e_{i_0}$ and ``teleporting'' to $e_{i_1'}$  with the $\xi_k$ being orthogonal to the span of the $e_i$ that correspond to the red line above it (the vectors that are not equal to a $v_k$ or $w_k$). 
The positioning of this diagram below the other diagram on the right side is to illustrate that the weight of $N$ on the $\xi_k$ is $b-1/M < b$. This will be ``passed down'' to constructions in cases 2 and 3.

\begin{figure}[htp]  
    \centering
    \includegraphics[width=14cm]{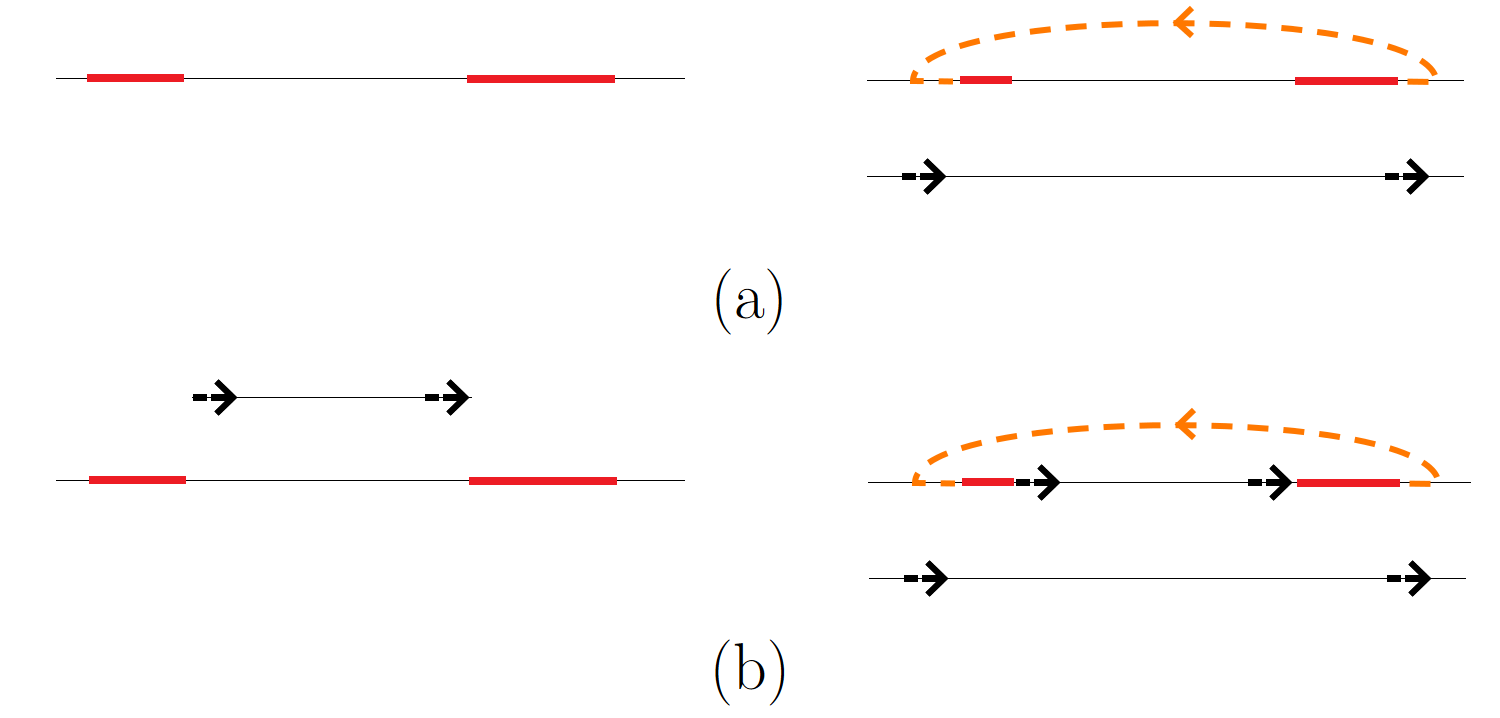}
    \caption{\dark \label{wsSimpleCase2} Illustration of Case 2 in the proof of Lemma \ref{HelpingBerg}.}
\end{figure}

\underline{Case 2}: Illustration \ref{wsSimpleCase2}(a) is an illustration similar to that of Illustration \ref{wsSimpleCase2} with the exception that there is a gap in the red line because the connected component of $J_b$ has vectors that have weight higher than $b$. The main difference here is that the top diagram on the right side of Illustration \ref{wsSimpleCase2}(a) does not represent an invariant subspace of $N$ because the right-most point of the left subset of the red line indicates that $N$ will map that vector to the black line, which is outside the orbit that we are considering.

The resolution of the fact that we do not obtain an invariant subspace in (a) is to include two arrows composed of some $\xi_k$ with weight $b$ originating from $J_{b+1/M}$. The left side of (b) shows that we are including this so that on the right side of (b) will have a closed orbit. The bottom two arrows on the right side of (b) will have weight $b-1/M$ and will be passed down to the construction for $J_{b-1/M}$. 

Note that (b) illustrates the case where the portion of $J_{b+1/M}$ in the connected component of $J_{n}$ on which we are focused is made of only one connected component. For an example where the relevant portion of $J_{b+1/M}$ contains two connected components, see the second-to-the-bottom line in Illustration \ref{wsSystem1}(a) and Illustration \ref{wsSystem1}(b).

\underline{Case 3}: 
\begin{figure}[htp]  
    \centering
    \includegraphics[width=14cm]{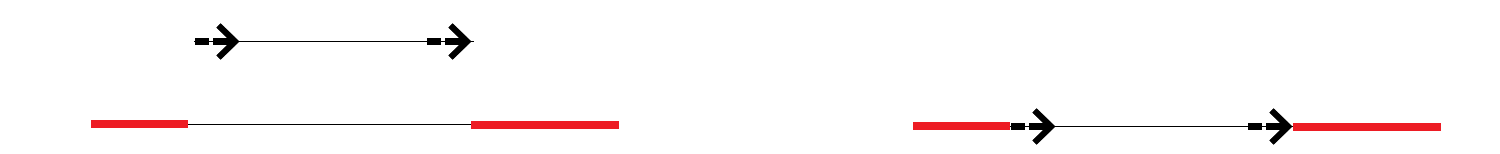}
    \caption{\dark \label{wsSimpleCase3} Illustration of Case 3 in the proof of Lemma \ref{HelpingBerg}. \dark}
\end{figure}
Case 3 does not have any change to the basis vectors in red. The only issues that can arise is when the there are gaps in the lowest weight intervals $I_j$ due to there being weights greater than $\underline{a}$. However, the $\xi_k$ that are passed down removes this difficulty. This is depicted in the illustration in that the passed down arrows with weights $(\underline{a}+1/M)-1/M=\underline{a}$. 

Note that in this illustration the red line on the right is not begun or ended by a black line. This indicates that the red line is a single segment (viewed cyclically) because it contains $e_1$ and $e_n$. The bottom row of Illustration \ref{wsSystemExtended} illustrates a slightly more general scenario of having $J_{\underline{a}+1/M}$ with two connected components so that there are two lowest weight intervals $I_j$.
\end{remark}

\begin{example}
We now provide two visual examples of the construction of the normal matrix $N$ in Lemma  \ref{HelpingBerg}. Illustration \ref{wsSystem1} provides an illustration of such an example, starting with the weights perturbed as described in the proof of the lemma in (a) and showing the constructed $N$ in (b) using the diagrams described in Remark \ref{CaseRemark}. 
\begin{figure}[htp]  
    \centering
    \includegraphics[width=14cm]{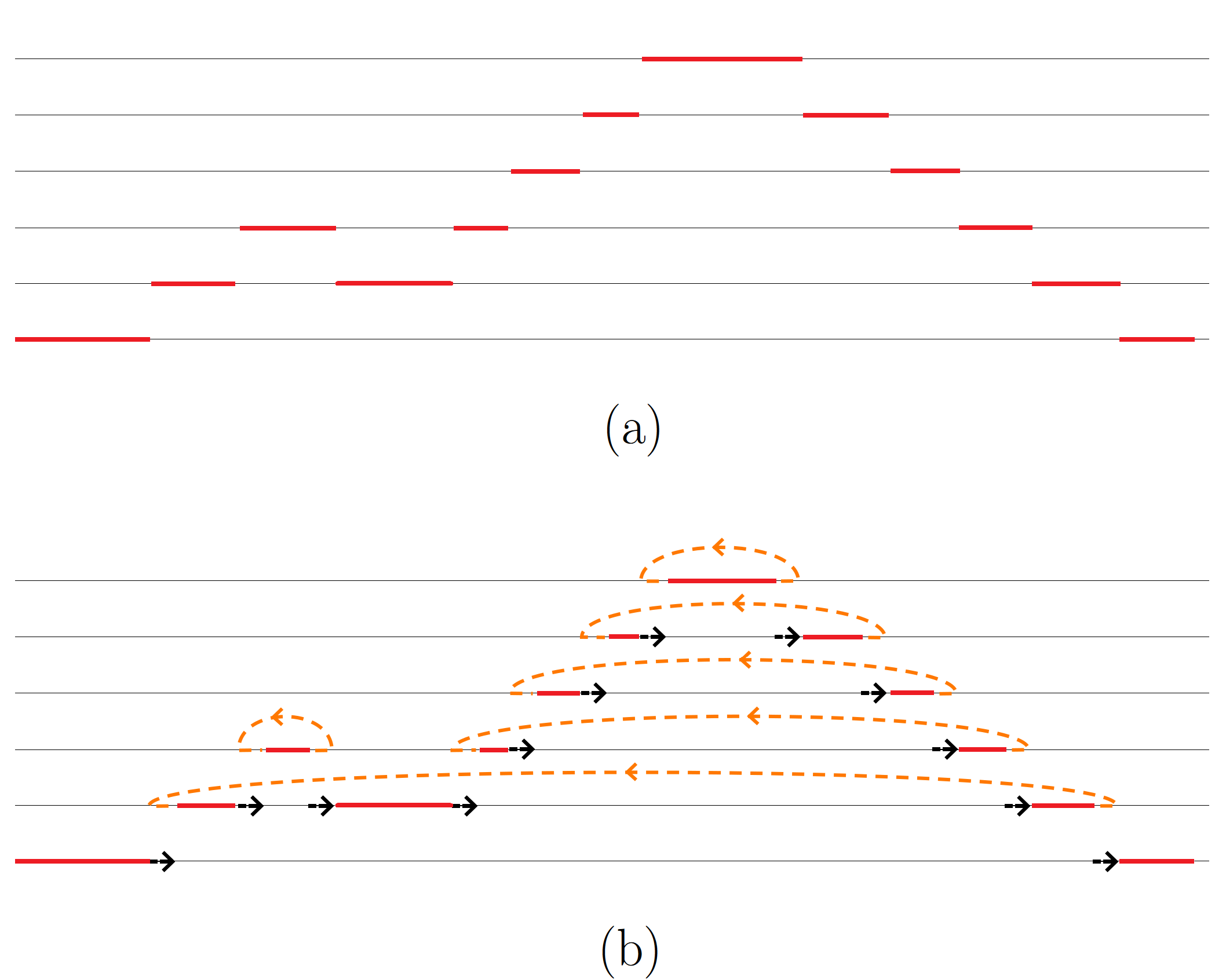}
    \caption{\dark \label{wsSystem1} Illustration of construction in Lemma  \ref{HelpingBerg}.  \dark}
\end{figure}

Illustration \ref{wsSystemExtended} provides a more general example of the construction where $J_{\underline{a}+1/M}$ has two connected components.
\begin{figure}[htp]  
    \centering
    \includegraphics[width=14cm]{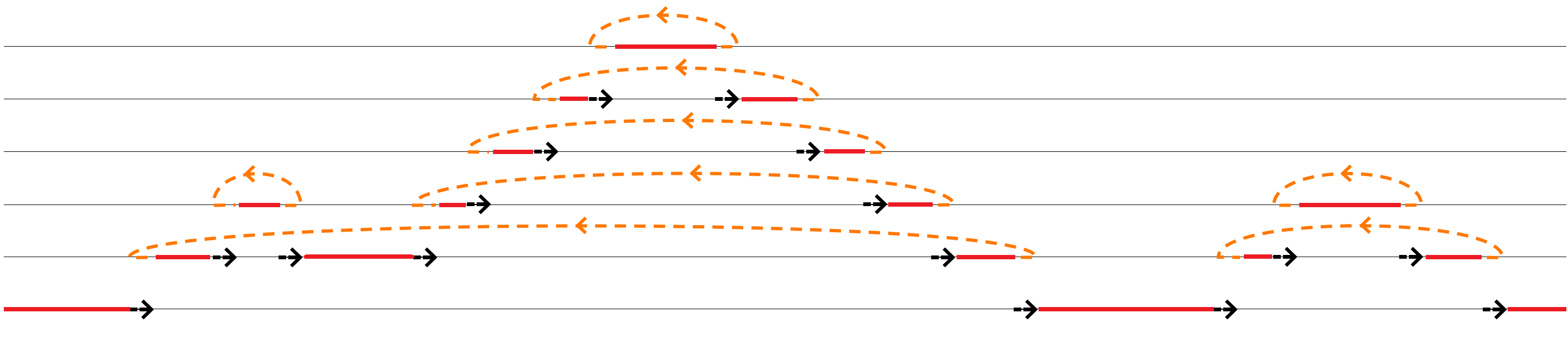}
    \caption{\dark \label{wsSystemExtended}  \dark
     Illustration of construction in Lemma  \ref{HelpingBerg}.}
\end{figure}
\end{example}

\begin{remark}\label{noNegatives}
The constructed normal matrix $N$ is a direct sum of bilateral weights shifts with weights $\pm b$ except the lowest weighted shift which may have a complex phase if the weights of $S$ were complex. 
However, it is possible to change the construction so that the bilateral shift with lowest absolute value weights is the only summand with a non-positive weight. 
Further, if $S$ is a unilateral shift, then all the weights of the summands of $N$ can be made non-negative even though only the lowest weight summand is a unilateral shift.

We presently have no need for this modification so for us such a modification would be purely aesthetic, but we discuss it nonetheless.
We modify the construction to minimize the number of negative signs left after our applications of Lemma \ref{gel-forNearbyNormal}. This same effect is accomplished by Berg's original construction due to the use of complex phases even if $S$ is real, but we opt for a different approach so that we obtain the structured result that $N$ is real if $S$ is.

One way to modify the proof is as follows.
First note that we will either use Lemma \ref{gel-forNearbyNormal} as stated or a modified form of Lemma \ref{gel-forNearbyNormal} that has a different definition of the rotated basis: $\tilde\xi_k=\alpha_kv_k - \beta_kw_k$ and $\tilde\eta_k=\beta_kv_k+\alpha_kw_k$ so that the $\tilde\eta_k$ satisfy $\tilde\eta_0 = w_0, \tilde\eta_{k_0}= v_{k_0}$ and the $\tilde\xi_k$ have ``the negative sign'': $\tilde\xi_0 = v_0, \tilde\xi_{k_0}=-w_{k_0}.$
We will apply one of the versions of the lemma so that the number of weights with a negative sign in invariant orbit of $N$ for Case 1 or Case 2 is even. This way, a simple change of variables in this invariant subspace for the $e_i, \eta_k$ orbit will result in all the weights being positive. Ultimately, the choice of which version of Lemma \ref{gel-forNearbyNormal} to apply will affect the choice for smaller weights due to the signs of the weights of the passed down vectors $\xi_k$.

Note that we can determine which passed down vectors will carry down a negative sign by noting that Case 1 always passes down a negative sign and Case 2 always passes down one (modulo two) negative sign more than the sum of the negative signs passed down to it.

We repeat this process where each $J_b$ for $b>\underline{a}$ will pass down some negative signs, at most one from each of its connected components. We then come to Case 3. This is the only place where we cannot remove the negative sign if $\underline{a}>0$. 

If $\underline{a}$ is close to zero, then we can replace it with zero with a small additional error. This is possible if $S$ is a unilateral weighted shift. If $\underline{a}$ is far away from zero then we might not be able to  remove a last remaining negative sign of the lowest weight bilateral shift with this method even with a perturbation.
\end{remark}

\section{A Nearby Normal for an Almost Normal Bilateral Weighted Shift Matrix}
We return to Lin's theorem for a weighted shift matrix.
Reformulating the previous lemma, we obtain the following theorem. This first inequality is inherent to Theorem 2 of \cite{berg1975approximation} with $C = 100$ and exponent $1/4$. Additionally, this result applies to not just unilateral shifts and we have the two additional properties of $N$ stated at the end of the statement of the theorem. The ability to choose $N$ real is an improvement on the construction of Berg's original proof as well as the greatly reduced constant. We also obtain a second construction in a more specific case that provides the optimal exponent.
\begin{thm}\label{BergResult}
Suppose that $S \in M_n(\C)$ is a bilateral weighted shift matrix. Then there is a normal matrix $N$ such that
\begin{align}\label{BergIneq1}
\|N-S\| \leq C_{\alpha} \|S\|^{1-2\alpha}\|[S^\ast,  S]\|^{\alpha}
\end{align}
for $\alpha=1/3$ and $C_{1/3} < 5.3308.$
Further, $N$ is equivalent to a direct sum of bilateral weighted shift operators, $\|N\|\leq \|S\|$, and if $S$ is real then $N$ is real.

If the weights of $S$ all have absolute value at least $\sigma$ then $N$ can be chosen with the above properties but the alternate estimate\begin{align}\label{BergIneq2}\|N-S\| \leq 4.8573\sqrt{\frac{\|S\|\,}{\sigma}} \|[S^\ast, S]\|^{1/2}.\end{align}
\end{thm}
\begin{remark}
Note that Equation  (\ref{BergIneq1}) is asymptotically weaker than the optimal upper estimate $\|N-S\|\leq C_{1/2}\|[S^\ast, S]\|^{1/2}$  by using \[\|[S^\ast, S]\|^{1/2}=\|[S^\ast, S]\|^{1/2-1/3}\|[S^\ast, S]\|^{1/3} \leq (2\|S\|^2)^{1/6}\|[S^\ast, S]\|^{1/3}.\]  
Equation  (\ref{BergIneq2}) is also weaker than the optimal upper estimate since $\|S\|\geq \sigma$. However, when $\sigma/\|S\|$ is not too small Equation  (\ref{BergIneq2}) can be of great use due to the small constant. 

The proof of the optimal estimate in \cite{kachkovskiy2016distance} does not provide a value of $C_{1/2}$, however it appears from the proof that it will be much larger than $C_{1/3}$ given above. For this reason, Equation  (\ref{BergIneq1}) will still be of use in addition to the simplicity of the construction of $N$ and the additional structure of $N$.

In our application to Ogata's theorem, we will have almost normal (unilateral) weighted shifts and hence will not be able to procure a usable lower bound for the absolute values of the weights. So, Equation  (\ref{BergIneq1}) with $\alpha = 1/3$ will be of use to us in later chapters.
\end{remark}

\begin{proof}
Assume that $\|S\|\leq s$.  Let $x = \|[S^\ast, S]\|$. 
Note that if $M\geq 4$ is an even integer, then when $x < M^{-3}$ the normal matrix constructed in Lemma \ref{HelpingBerg} satisfies the properties therein. 

Let $M_0\geq4$ be a real number. Consider the case that $x \leq (M_0+2)^{-3}$ so that $M_0+2 \leq x^{-1/3}$ and define \[M = 2\left(\left\lceil \frac{x^{-1/3}}2 \right\rceil -1\right)\] so that $M$ is an even integer that satisfies
\[M < 2\left( \frac{x^{-1/3}}{2}\right)=  x^{-1/3},\] hence $x<M^{-3}$. Also, \[M \geq 2\left( \frac{x^{-1/3}}2 -1\right) \geq 2\left( \frac{M_0+2}{2} -1\right) = M_0.\]

Apply Lemma \ref{HelpingBerg} to obtain a normal matrix $N$ with the properties from that lemma. Because $t \mapsto t/(t-2) = 1 + 2/(t-2)$ for $t > 2$ is decreasing and $M\geq M_0$, we have
\[\frac{M}{M-2}\leq \frac{M_0}{M_0-2}.\] 
Since
\[\frac M2\geq\frac{x^{-1/3}}{2} -1 \geq \frac{x^{-1/3}}{2}-\frac{x^{-1/3}}{M_0+2} = \frac{M_0}{2(M_0+2)}x^{-1/3},\]
we have
\begin{align}\label{HelpingEstimate}
\|N-S\| < \left(s\frac{\pi M}{M-2}+2\right)\frac{1}{M} \leq \left(s\frac{\pi M_0}{M_0-2}+2\right)\frac{M_0+2}{M_0}x^{1/3}.\end{align}
We have obtained an estimate when $x \leq (M_0+2)^{-3}$. 

If $x > (M_0+2)^{-3}$ then  we can choose $N = 0$ so that
\[\|N-S\| = \|S\| \leq s \leq s(M_0+2)x^{1/3}.\]
So, putting this case together with Equation (\ref{HelpingEstimate}) we have some normal matrix $N$ such that
\[\|N-S\| \leq \max\left(s(M_0+2), \left(s\frac{\pi M_0}{M_0-2}+2\right)\frac{M_0+2}{M_0}  \right)x^{1/3} = f(s,M_0)x^{1/3}.\]

In general, when $\|S\| > 0$ apply this result to the rescaled $\tilde S = \frac r{\|S\|}S$ with norm $r$ to obtain a normal $\tilde N$. With $N = \frac{\|S\|}r\tilde N$, we have
\[\|N-S\| = \frac{\|S\|}r\|\tilde N - \tilde S\| \leq 
\frac{\|S\|}rf(r, M_0)\|[\tilde S^\ast, \tilde S]\|^{1/3} = \frac{f(r, M_0)}{r^{1/3}}\|S\|^{1/3}\|[S^\ast,  S]\|^{1/3}.\]

So, we want to choose $r$ and $M_0$ to minimize
\[\frac{f(r, M_0)}{r^{1/3}}= \max\left(r^{2/3}(M_0+2), \left(r^{2/3}\frac{\pi M_0}{M_0-2}+\frac2{r^{1/3}}\right)\frac{M_0+2}{M_0}  \right).\]
We choose $M_0=15.937$ and $r = 0.162$ to obtain the $f(r, M_0)r^{-1/3}< 5.3308$.

\vspace{0.05in}

We obtain the second result as follows. Let $x = \|[S^\ast, S]\|/2\sigma$ and define $M_0$ as above. We now change the definition of $M$ to instead have an exponent $\alpha = 1/2$. We assume that $x \leq (M_0+2)^{-2}$ so that $M_0+2 \leq x^{-1/2}$ and define \[M = 2\left(\left\lceil \frac{x^{-1/2}}2 \right\rceil -1\right)\] analogous to what is done above.
Then \[\|[S^\ast, S]\| = 2\sigma x < 2\sigma M^{-2}\]
and $M \geq M_0$ as before.

As  before, \[\|N-S\| \leq \max\left(s(M_0+2), \left(s\frac{\pi M_0}{M_0-2}+2\right)\frac{M_0+2}{M_0}  \right)x^{1/2} = f(s,M_0)\left(\frac{\|[S^\ast, S]\|}{2\sigma}\right)^{1/2}.\]
If we perform the same change of variables $\tilde S = \frac{r}{\|S\|}S$ then the weights of $\tilde S$ have absolute value at least $\tilde \sigma = \frac{r\sigma}{\|S\|}$. So, as before we obtain normals $\tilde N$ and $N$ so that
\[\|N-S\| = \frac{\|S\|}r\|\tilde N - \tilde S\| \leq 
\frac{\|S\|}r\frac{f(r, M_0)}{(2\tilde \sigma)^{1/2}}\|[\tilde S^\ast, \tilde S]\|^{1/2} = \frac{f(r, M_0)}{(2r)^{1/2}}\left(\frac{\|S\|}{\sigma}\right)^{1/2}\|[S^\ast,  S]\|^{1/2}.\]
Choosing $M = 10.762$ and $r = 0.2897$ provides the estimate.

\end{proof}

\chapter{The Gradual Exchange Process}
\label{7.GradualExchangeProcess}

We begin this chapter by motivating the construction in Lemma \ref{proto-gep}. The proceeding lemmas: Lemmas \ref{proto-gep2} and \ref{gep} are generalizations of this lemma that we will need for the main result of the paper. 

\section{Motivating Examples}
Recall that several of the counter-examples of almost
commuting matrices that are not nearly commuting have the same structure: a diagonal matrix $A$ and a weighted shift matrix $S$, where there is a lower bound on the absolute value of the weights of $S$ over a long span of the spectrum of $A$. Consider the following example, which is essentially Example 2.1 of Hastings and Loring's \cite{hastings2010almost}.
\begin{example}\label{repEx}
Let $A = \frac{1}{\lam}S^{\lam}(\sigma_3)$ and $S = \frac{1}{\lam}S^{\lam}(\sigma_+)$. Recall that by Equation (\ref{repAlmostCommuting}), $A$ and $S$ are almost commuting. 
Note that $A$ is self-adjoint and $S$ is almost normal.
Using an invariant called the Bott index, \cite{hastings2010almost} shows that there are no nearby commuting matrices $A', S'$ with $A'$ self-adjoint and $S'$ normal.

Written in matrix form, these are
\[A = 
\bp 
-1& & & & \\
 &\displaystyle-1+\frac1\lam& & & \\
 & & \displaystyle-1 + \frac2\lam & &\\
 & & &\ddots& \\ 
 & & & & 1 \\
\ep, \;\; S=  
\bp 
0  & & & &\\
\displaystyle\frac{d_{\lam,-\lam}}\lam&0& & &\\
 &\displaystyle\frac{d_{\lam, -\lam+1}}\lam&0& &\\
 & &\ddots& \ddots&\\
 & & &\displaystyle\frac{d_{\lam,\lam-1}}\lam&0
\ep.\]
Note that for $|i| \leq \lam/2-1$,
\[\displaystyle\frac{d_{\lam, i}}{\lam} = \sqrt{\frac{\lam(\lam+1)}{\lam^2} - \frac{i(i+1)}{\lam^2} } \geq \sqrt{1 - \frac14} = \frac{\sqrt3}2.\]
This sort of lower bound on the weights of $S$ is a crucial part of why $A$ and $S$ are not nearly commuting as we illustrate using the following construction.
\end{example}

\begin{example}\label{gepMotivation}
Suppose that $A = \diag(a_1, a_2, \dots, a_n)$ and $S = \ws(c_1, \dots, c_{n-1})$ where $a_1 < \dots < a_n$ and $c_i \geq 0$.
\begin{figure}[htp]     \centering
    \includegraphics[width=14cm]{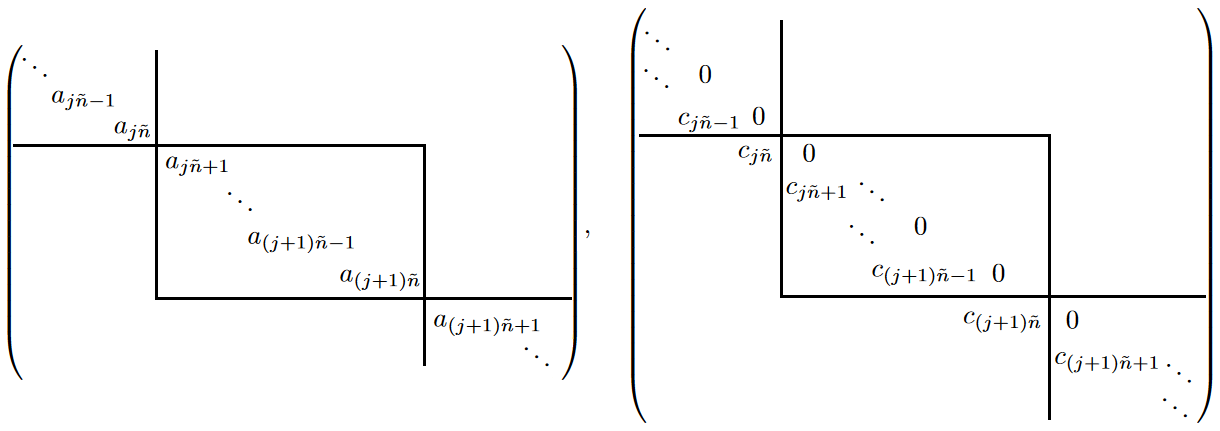}
    \caption{\label{AS-matrices}\dark
    Illustration of $A$ and $S$, respectively.}
\end{figure}
We will suppose further that $A$ and $S$ are nearly commuting: $\max_i |a_{i+1}-a_i||c_i|$ is small. So, if the $a_i$ are close then the $c_i$ are not required to be too small.

For the sake of the example, suppose that for some $\tilde{n} \ll n$ that divides $n$, it is true that 
$c_{\tilde{n}}, c_{2\tilde{n}}, \dots, c_{n-\tilde{n}}$ are no greater than some constant $D > 0$. Then
define $S'$ to be the linear operator where the weights $c_{j\tilde{n}}$ for $0< j < n/\tilde{n}$ of $S$ are replaced with zero. Then 
\[\|S'-S\| \leq D.\]

Let $J_j = [j\tilde{n}+1,  (j+1)\tilde{n}]$. Now, for $j \geq 0$, the subspaces $\mathcal W_j = \spn_{i \in J_j} e_i$ are invariant under $S'$. 
So, let $A'$ be an operator that is a multiple $a_j'$ of the identity when restricted to $\mathcal W_j$. If $F_j$ is the projection onto $\mathcal W_j$ then $A' = \sum_j a_j' F_j$. We choose then $a_j' = (a_{j\tilde{n}+1} + a_{(j+1)\tilde{n}})/2$ so that 
\[\|A' - A\| \leq \frac12\max_j \diam \{a_i: i  \in J_j\}.\]
We then see that if $D$ is small and the eigenvalues $a_i$ do not vary much for $a_i \in J_j$ then $A$ and $S$ are nearly commuting. The second condition can be restated as the property that the orbits of $S'$ do not span long stretches of the spectrum of $A$.

Expressed in matrix form, this construction replaces the almost commuting matrices $A$ and $S$ given in Illustration \ref{AS-matrices} with the commuting matrices $A'$ and $S'$ given in Illustration \ref{A'S'-matrices}, respectively.
\begin{figure}[htp]     \centering
    \includegraphics[width=14cm]{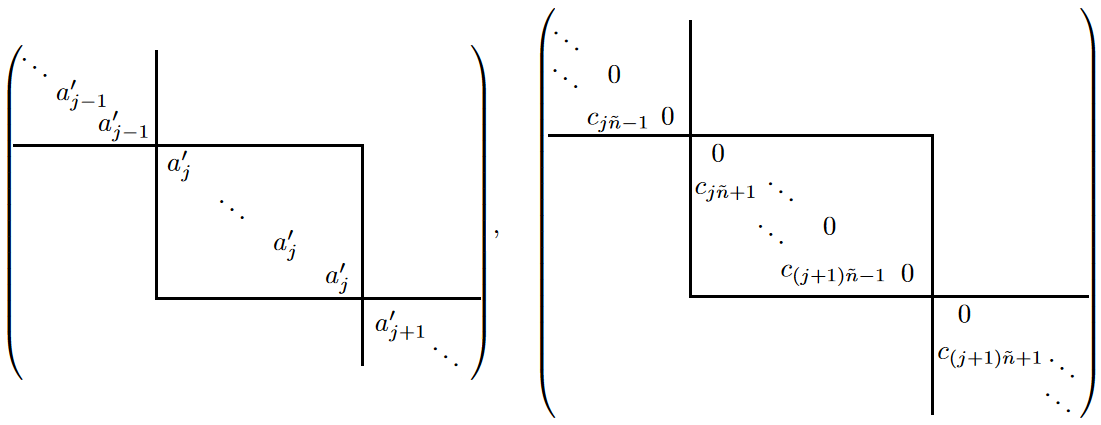}
    \caption{\label{A'S'-matrices}\dark
    Illustration of $A'$ and $S'$, respectively.}
\end{figure}

\end{example}

Example \ref{repEx} and the argument in Example \ref{gepMotivation} are illustrated in Illustration \ref{WSheatmap_GEP_Motivation_SingleWS}. The weights $c_i$ for $i = 1, \dots, n=100$ in Illustration \ref{WSheatmap_GEP_Motivation_SingleWS}(b) are $(1+\sin(i/\sqrt{n}))/2$.
\begin{figure}[htp]     \centering
    \includegraphics[width=12cm]{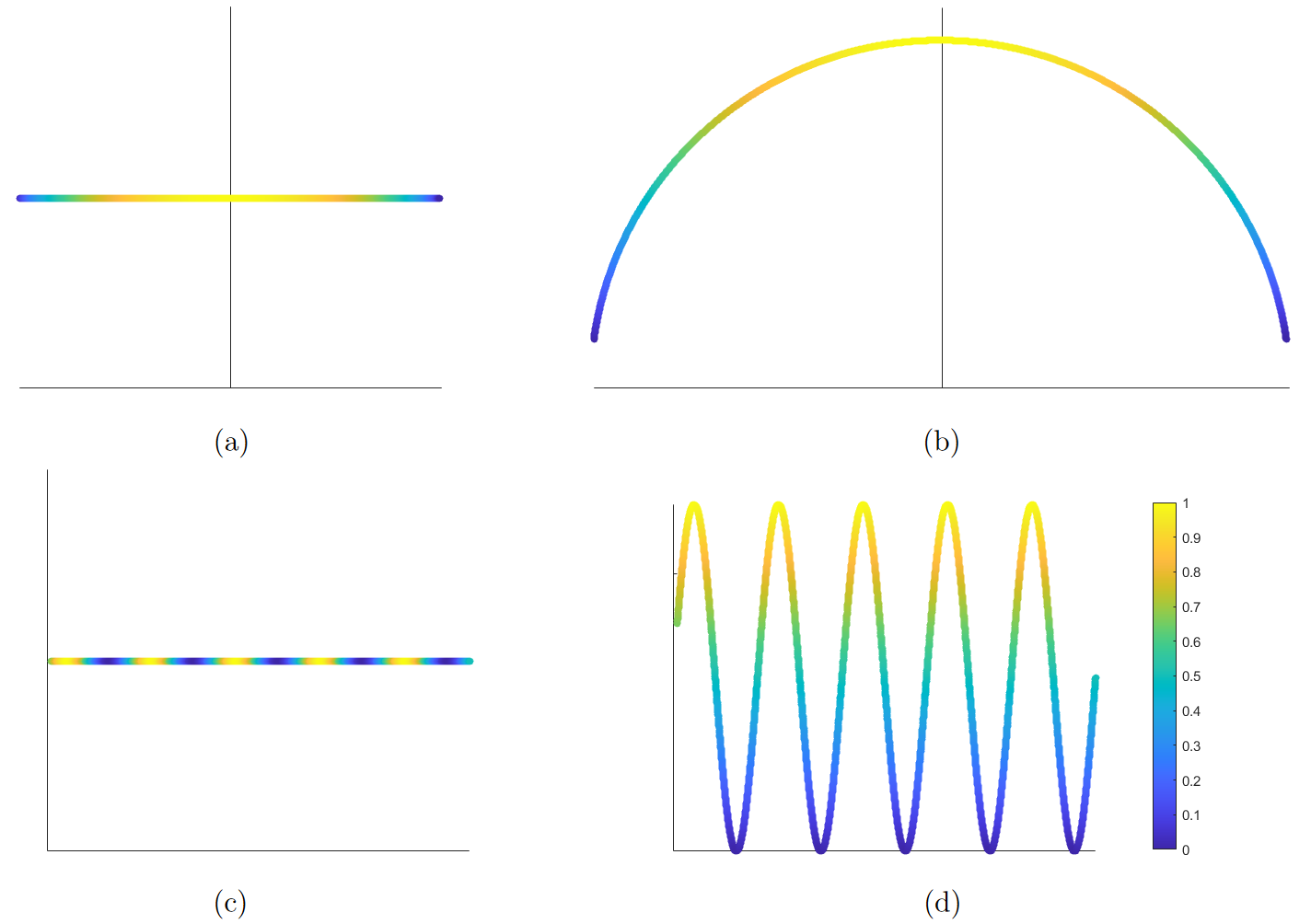}
    \caption{\label{WSheatmap_GEP_Motivation_SingleWS}\dark
    Example \ref{repEx} is illustrated in row (a) and (b). An example similar to Example \ref{gepMotivation} is illustrated in row (c) and (d). The graphs (a) and (c) on the left are weighted shift diagrams (without the arrow) colored according to the values of the weights at that point in the spectrum. The graphs (b) and (d) on the right illustrate these weights as graphs.}
\end{figure}
The construction of $S'$ from $S$ can be illustrated in weighted shift diagrams as in Illustration \ref{WeightedShift_Decomposition}.

\begin{figure}[htp]     \centering
    \includegraphics[width=12cm]{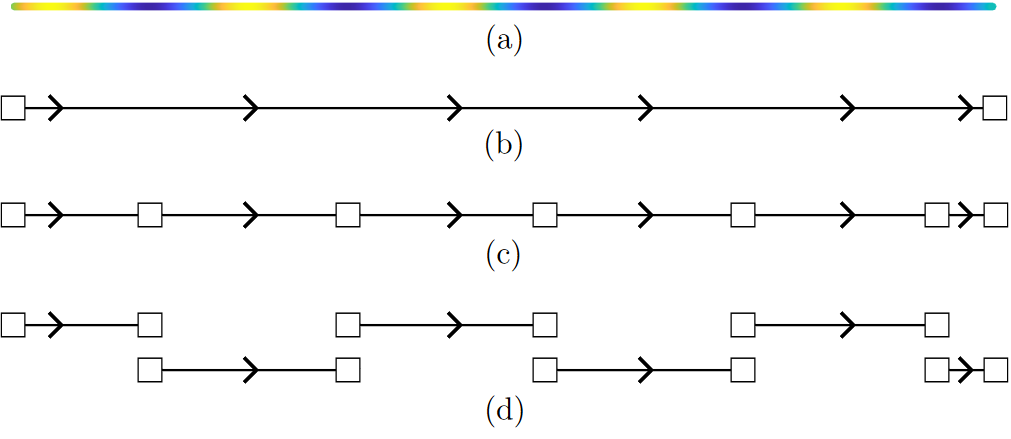}
    \caption{\label{WeightedShift_Decomposition}\dark
    The weighted shift diagram (without the arrow) from Illustration \ref{WSheatmap_GEP_Motivation_SingleWS}(c) is illustrated in (a). 
    The rest of this illustration depicts how the construction in Example \ref{gepMotivation} changes the weighted shift matrix to obtain $S'$. A (standard) weighted shift diagram is seen in (b).
    In (c), we break the single orbit in (b) into smaller orbits by replacing  some of the small weights with zero. In (d), we then view this broken weighted shift diagram as the direct sum of numerous weighted shift diagrams with smaller orbits.}
\end{figure}

\section{Gradual Exchange Process -- Basic Case}

To state the problem that we address in this chapter, suppose that $A$ and $S$ are given by block matrices:
{
\begin{align}\label{blockMatrices}
{\setlength{\arraycolsep}{0pt}A = 
\begin{pmatrix} 
\alpha_1 I_{k_1}&&&&\\
&\alpha_2 I_{k_2}&&&\\
&&\ddots&&\\
&&&\alpha_{n-1}I_{k_{n-1}}&\\&&&&\alpha_n I_{k_n}\end{pmatrix},\;\;\;\;\;} 
{\setlength{\arraycolsep}{2pt}
S = 
\begin{pmatrix} 0&&&&\\C_1&0&&&\\&C_2&\ddots&&\\&&\ddots&0&\\&&&C_{n-1}&0\end{pmatrix},}
\end{align}
}
where the $\alpha_i$ are distinct and each $C_i \in M_{k_{i+1}, k_i}(\C)$ is ``diagonal'', with its only non-zero entries being those with the same row and column number.
We are trying to construct nearby commuting matrices $A'$ and $S'$. We also want to perturb $S'$ to $S''$ that is additionally normal.

If many of the blocks $C_i$ had only small entries (and hence has small operator norm), then we could apply the exact argument from Example \ref{gepMotivation}. In the case that the $C_i$ typically have small and large diagonal entries, we will develop a method to use the small diagonal entries to break $S$ into a direct sum (in a rotated basis) of weighted shifts for which the arguments in Example \ref{gepMotivation} apply. However, the estimates will depend on the distribution of values.

\begin{example}\label{gep-Example}

We now illustrate this mechanism for constructing projections analogous to those from Example \ref{gepMotivation} in an example when the $C_i$ all are square matrices. This is the case addressed by Lemma \ref{proto-gep}. 
\begin{figure}[htp]     \centering
    \includegraphics[width=14cm]{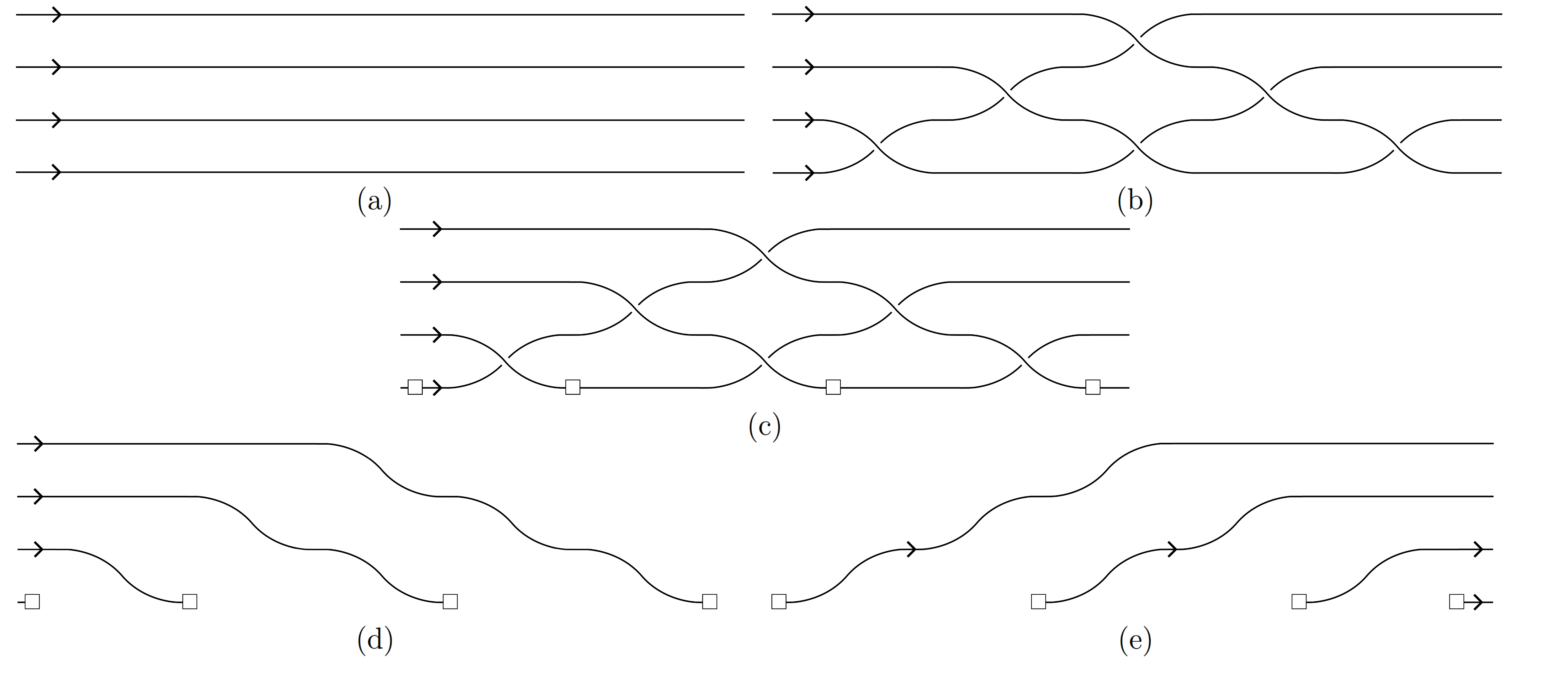}
    \caption{\label{GEP_Motivation1Rescaled}\dark
    Illustration of the gradual exchange process with $m=4$ for Example \ref{gep-Example}.}
\end{figure} 
In this example we focus on constructing only a single projection.
Let $A$ and $S$ be of the form of (\ref{blockMatrices}) where the $\alpha_i$ are strictly increasing real numbers and the identical matrix-valued weights $C_i$ of $S$ are
\[{\setlength{\arraycolsep}{1.5pt}
C_i = \bp 
c^4 & & & \\
& c^3 & & \\
& & c^2 & \\
& & & c^1  \\
\ep}\]
for some $c^r\geq 0$ and $c^1 \leq D$.  Note that the index $r$ is a superscript so that when the blocks $C_i$ are not identical as in Lemma \ref{proto-gep} then the diagonal entries of $C_i$ can be written with the similar notation: $c_i^r$. 

We now define weighted shift matrices $S_r = \ws(c^r)$ and will construct certain projections $F$ and $F^c$ for the direct sum of the $S_r$. We will later explain how $S$ can be seen as the direct sum of the $S_r$.

We describe the diagrams in Illustration \ref{GEP_Motivation1Rescaled}.
Illustration \ref{GEP_Motivation1Rescaled}(a) is weighted shift diagram for $S = \bigoplus_{r=1}^m S_r$ with $m = 4$ in the direct sum basis. In the diagram, the weighted shift diagram for $S_1$ is on the bottom of (a) and $S_4$ is illustrated on the top. Only a portion of the orbits is shown. For this example, we will apply our method within this window and outside of this window $S$ will not be changed.
Illustration \ref{GEP_Motivation1Rescaled}(b) is an illustration of how we will apply the gradual exchange lemma. For the following discussion, please see Illustration \ref{GEP_Size4_Columns} below for a description of what the ``columns'' are.

We first apply the gradual exchange lemma to $S_2, S_1$ over the span of $N_0+1$ vectors. We will have $N_0+1$ vectors in each orbit corresponding to where applications of the gradual exchange lemma occur in a column.

Later in the basis, we apply the gradual exchange lemma to $S_3, S_2$ over $N_0+1$ vectors. This is the second column of application(s) of the gradual exchange lemma. 
Later in the basis, we simultaneously (in the same column) apply the gradual exchange lemma to $S_4, S_3 $ and to $S_2, S_1$ in parallel. 

This is the end of the first stage. What we have done so far has changed the orbit of $S_1$ so that it ends up in the orbit of $S_m$ and the orbit of $S_m$ has finally been lowered to $S_{m-1}$. 
After the first stage, we continue to lower the orbits. We apply the gradual exchange lemma to $S_3, S_2$ in the next column. Then we apply it to $S_2, S_1$ in the last column.  

\begin{figure}[htp]     \centering
    \includegraphics[width=10cm]{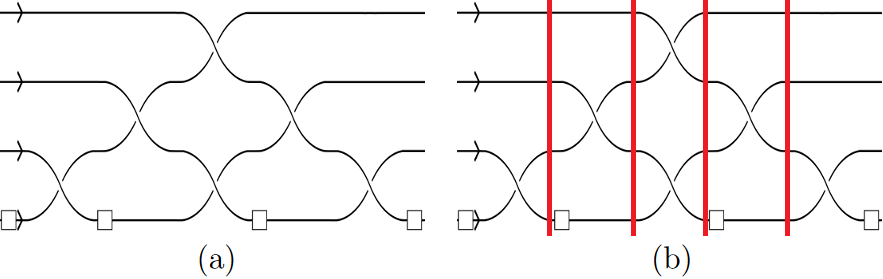}
    \caption{\label{GEP_Size4_Columns}\dark
    Part (a) of this Illustration is equivalent to Illustration \ref{GEP_Motivation1Rescaled}(c). Part (b) of this Illustration has four vertical red bars inserted. The ``first column'' refers to the portion of the diagram to the left of the first bar. The second column refers to the portion between the first and second bars and so on.}
\end{figure}
In more generality (see Illustration \ref{GEP_Motivation2Rescaled2}), the first stage has $m-1$ columns and the second stage has $m-2$ columns. In each column, the gradual exchange lemma is applied to pair(s) of weights shift operators in parallel. 
In the proof of Lemma \ref{proto-gep},  the column in which we apply the gradual exchange lemma is spanned by $\mathcal V_{3 + (N_0+1)(j-1)}$, $\dots$, $\mathcal V_{2 + (N_0+1)j}$.
\begin{figure}[htp]     \centering
    \includegraphics[width=14cm]{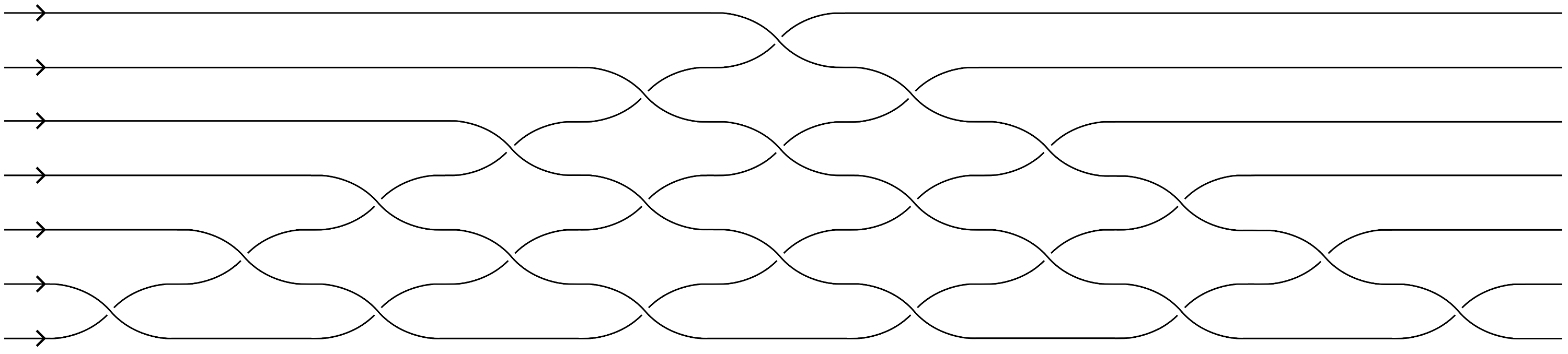}
    \caption{\label{GEP_Motivation2Rescaled2}\dark
    Illustration of the applications of the gradual exchange lemma as a part of the gradual exchange process for $m=7$ weighted shift matrices. \\
    When $1\leq j \leq m-1$,  we apply the gradual exchange lemma to the the pairs of weighted shift operators: $S_{j+1-e}, S_{j-e}$ in the $j$th column for all even $e \in [0, j)$.
    When $m \leq j \leq 2m-3$,  we apply the gradual exchange lemma to the pairs of weighted shift operators: $S_{2m-j-1-e}, S_{2m-j-2-e}$ in the $j$th column for all even $e \in [0, 2m-j-2)$.
    }
\end{figure}

These applications of the gradual exchange lemma give the linear operator $\tilde{S}$, which is a direct sum of weighted shift matrices $\tilde{S}_r$ in a rotated basis. 
By following the orbit of the first basis vector of each of the direct summands $\tilde{S}_r$, we see that each of these orbits eventually lie in the orbit of $S_1$. 
The particular weaving done with the applications of the gradual exchange lemma was for this reason since we will assume that the weights of $S_1$ are small.

Choose a vector belonging to the portion of the orbit of $\tilde{S}_r$ that is in the orbit of $S_1$. 
We then break the orbit of $\tilde{S}_r$ into two orbits by replacing the weight of that vector with zero. This can be done by a perturbation of size at most $D = \|S_1\|$. We do this for each $r$ to  obtain $S'$.
Illustration \ref{GEP_Motivation1Rescaled}(c) is an illustration of this.

Illustration \ref{GEP_Motivation1Rescaled}(d) is an illustration of the orbit of the initial basis vector of each $S_r$ under $S'$. $S'$ acts as the direct sum of some weighted shift matrices which terminate within the window of \ref{GEP_Motivation1Rescaled}(a) that we began with. We define the projection $F$ to have range equaling the portion of the orbits illustrated in Illustration \ref{GEP_Motivation1Rescaled}(d) that are within the window illustrated. 

Additionally, there are other weighted shift operators that form part of $S'$ that are illustrated in Illustration \ref{GEP_Motivation1Rescaled}(e). The orbits of these operators begin within the window of \ref{GEP_Motivation1Rescaled}(a) that we began with and exit the window. The portion of the orbits illustrated in this window span the range of a projection that we call $F^c$. 

Observe a few key properties of $S', F,$ and $F^c$. First, $F$ is an invariant subspace of $S'$. When restricting $S'$ to $F$, we see that $S'$ has the structure of the direct sum of weighted shift operators. Notice that in the subspace corresponding to the window of the weighted shift diagrams, the range of $F^c$ is the orthogonal complement of the range of $F$. Moreover, although $F^c$ is not an invariant subspace of the entire domain of $S$, the image of  $R(F^c)$ under $S'$ is orthogonal to $R(F)$ and belongs to the span of $R(F)$ and the basis vectors of the weighted shift diagram that lie outside the window to the right. 
These properties will allow us to construct invariant subspaces when we apply the construction illustrated in this example later when forming various projections $F_i$ and $F_i^c$ for all windows as in Lemma \ref{gep}.

The estimates obtained will depend on the weights. The weights of consecutive $S_r$ contribute to the estimate through the gradual exchange lemma and the weights of $S_1$ contribute to the value of $\|S'-\tilde S\|$ when we break the orbits of $\tilde{S}_r$.
A key property of applying the gradual exchange lemma is that because the applications of the gradual exchange lemma are only applied to $S$ on orthogonal subspaces, the norms of perturbations do not add. Similarly, because the vectors in the orbit of $S_1$ whose weights of $\tilde{S}$ that we changed to zero were not affected by our application of the gradual exchange lemma, the norms of the perturbations of breaking up the orbits will not add either. We now estimate $\|S'-S\|$.

In the first column of applications of the gradual exchange lemma, we applied this lemma to $S_2, S_1$ incurring a perturbation of norm at most $|c^2-c^1| + \frac{\pi}{2N_0}\max(c^1, c^2)$. Next we applied the gradual exchange lemma to $S_3, S_2$, incurring an independent perturbation of norm at most $|c^3-c^2| + \frac{\pi}{2N_0}\max(c^2, c^3)$. Then we applied the gradual exchange lemma to $S_2, S_1$ and also $S_4, S_3$ in parallel, incurring independent perturbations of norm at most $|c^2-c^1| + \frac{\pi}{2N_0}\max(c^1, c^2)$ and $|c^4-c^3| + \frac{\pi}{2N_0}\max(c^3, c^4)$, respectively. Continuing this analysis, we observe that by applying the gradual exchange lemma in our construction of $\tilde{S}$ incurred a perturbation of norm at most 
\[G = \max_{1 \leq r\leq 3}\left(|c^{r+1}-c^r| + \frac{\pi}{2N_0}\max(c^{r+1}, c^r)\right).\]
Changing some of the weights to zero incurred an independent perturbation of norm $c^1 \leq D$. So,
\[\|S'-S\| \leq \max(G, D).\]

\vspace{0.1in}

We now return to the identification of $S$ as this direct sum of weighted shift matrices. We then describe the construction of $F$ in terms of basis vectors. 
Let the subspaces corresponding to the blocks $C_i$ be $\mathcal V_1, \dots, \mathcal V_n$. Write the standard basis vectors of $\C^{4n}$ as $e_1^4, e_1^3, e_1^2, e_1^1$, $\dots$, $e_n^4, e_n^3, e_n^2, e_n^1$ so that the subspace $\mathcal V_i$ is spanned by $e_i^4, e_i^3, e_i^2, e_i^1$.

Let $A_r = \diag(\alpha_1, \dots, \alpha_n)$ and $S_r = \ws(c^r, \dots, c^r)$ for $r = 1, \dots, 4$. 
By grouping the standard basis vectors of $\C^{4n}$ as $e_1^r, e_2^r, \dots, e_n^r$, 
we can express $A$ and $S$ as $A = \bigoplus_{r=1}^4 A_r$ and $S = \bigoplus_{r=1}^4 S_r$. 
In particular, the span of $e_1^r, e_2^r, \dots, e_n^r$ is invariant under $A$ and $S$ with
$Ae_i^r = a_i e_i^r$ and $S e_i^r = c^re_{i+1}^r$. This is the orbit of $e^r_1$ under $S$.

So, the formulation of $A$ and $S$ as block matrices of the form of Equation (\ref{blockMatrices}) with the same size is equivalent to expressing $A$ as a direct sum of the identical diagonal matrices $A_r$ and expressing $S$ as a direct sum of the weighted shift matrices $S_r$ by rearranging the direct sum basis. In the block matrix perspective,  
$e^r_i$ can be expressed as $0_4^{\oplus(i-1)}\oplus e_r \oplus 0_4^{\oplus(n-i)}$, where $0_4^{\oplus k}$ is the $k$-fold direct sum of the zero vector $0_4$ in $\C^4$.

After this set-up, we now state the required properties of $S', F,$ and $F^c$ as in the statement of Lemma \ref{proto-gep}.
Let $a, b \in \sigma(A)$ and $\alpha_1 \leq a < b \leq \alpha_n$. This specifies the window in which we focus.

We will require that the projections $F\leq E_{[a,b)}(A)$ and $F^c = E_{[a,b]}(A)-F$ satisfy $E_{\{a\}}(A) \leq F $, $R(F)$ is invariant under $S'$, and $S'$ maps $R(F^c)$ into $R(F^c) + R(E_{b+}(A))$, where $b+$ is the eigenvalue of $A$ that equals $\min \sigma(A) \cap (b, \infty)$, if it exists. If $\sigma(A) \cap (b, \infty)=\emptyset$, then $R(F^c)$ will just be an invariant subspace. These are conditions that we will use in Lemma 
\ref{proto-gep}.

\vspace{0.05in}

We will now describe the vectors spanning $F$. Note that our description of these vectors, some of which are obtained by many applications of the gradual exchange lemma, will not mention how negative signs are propagated in the sort of detail seen in Example \ref{GEL_Arrows}. We will instead use the statement of the gradual exchange lemma that we proved which takes care of the propagated negative signs after each application. Keeping track of the negative signs is not necessary to state what $F$ is, however it is necessary if we wanted to have an explicit description of the basis with respect to $S'$ breaks into a direct sum of weighted shift matrices with positive weights in order to apply Berg's construction in Theorem \ref{BergResult}.

So, we begin. The vectors 
\[e^4_1, e^3_1, e^2_1, e^1_1\] correspond to the first block because they form a basis for $\mathcal V_1$. Each $e^r_1$ corresponds to a point on each of the four orbits lying on a vertical line on the far left of Illustration \ref{GEP_Size4_Columns}(b) to the left of the box at the bottom of this first column. Because we require $\mathcal V_1 \subset R(F)$, we include these vectors in our collection of spanning vectors of $R(F)$. For the sake of not perturbing the weights $c_1^r$ on the boundaries of this window, we need the subspace $\mathcal V_2$ to also be included:
\[e^4_2, e^3_2, e^2_2, e^1_2\]
since $Se^r_1 = c^re^r_2$.

When we continue our list of vectors, we drop the last vector to obtain
\[e^4_3, e^3_3, e^2_3, 0.\]
Now, these three vectors will also form a part of the basis for $R(F)$. Although 0 does not contribute to the span, we leave it there as a placeholder.
Because $Se^1_2 = c^1e^1_3$ and because we will set one of the weights $c^1$ equal to zero so that $S'e^1_2 = 0$, our dropping $e^1_3$ corresponds to a perturbation of $S$ of norm  $c^1 \leq D$ only on the the orbit of $S_1$. The box in the first column of \ref{GEP_Size4_Columns}(b) reflects that although $e^1_2$ belongs to the orbit of $S_1$ we made a weight equal to zero so that now $e^1_3$ is excluded from the orbit of $S'$.

We now apply the gradual exchange lemma to obtain orthonormal vectors $e_k'^2, e_k'^1$, orthogonal to all other vectors that we list, so that $e_{2+1}'^2=e_{2+1}^2$, $e_{2+(N_0+1)}'^2 = e_{2+(N_0+1)}^1$, $e_{2+1}'^1 = e_{2+1}^1$, $e_{2+(N_0+1)}'^1 = -e_{2+(N_0+1)}^2$.
Our list of vectors continues with (the first line is what we have listed above)
\[e^4_3, e^3_3, e'^2_3, 0,\]
\[e^4_4, e^3_4, e'^2_4, 0,\]
\[\dots\]
\[e^4_{2+(N_0+1)}, e^3_{2+(N_0+1)}, e'^2_{2+(N_0+1)}, 0,\]
which is
\[e^4_{2+(N_0+1)}, e^3_{2+(N_0+1)}, 0, e^1_{2+(N_0+1)}.\]
This application of the gradual exchange lemma happens in the first column of \ref{GEP_Size4_Columns}(b).

We then apply the gradual exchange lemma to obtain vectors $e_k'^3, e_k'^2$ so that $e_{3+(N_0+1)}'^3$ $=e_{3+(N_0+1)}^3$, $e_{2+2(N_0+1)}'^3 = e_{2+2(N_0+1)}^2$, $e_{3+(N_0+1)}'^2 = e_{3+(N_0+1)}^2$, $e_{2+2(N_0+1)}'^2 = -e_{2+2(N_0+1)}^3$.
Our list of vectors continues as follows. Note that we drop the lowest weight vector as well in the third step.
\[e^4_{3+(N_0+1)}, e'^3_{3+(N_0+1)}, 0, e^1_{3+(N_0+1)},\]
\[e^4_{4+(N_0+1)}, e'^3_{4+(N_0+1)}, 0, e^1_{4+(N_0+1)},\]
\[e^4_{5+(N_0+1)}, e'^3_{5+(N_0+1)}, 0, 0,\]
\[\dots\]
\[e^4_{2+2(N_0+1)}, e'^3_{2+2(N_0+1)}, 0, 0,\]
which is
\[e^4_{2+2(N_0+1)}, 0, e^2_{2+2(N_0+1)}, 0.\]
This application of the gradual exchange lemma happens in the second column of \ref{GEP_Size4_Columns}(b). The dropping a vector in the orbit of $S_1$ corresponds to the box in the second column of \ref{GEP_Size4_Columns}(b).

Now that there are not any consecutive non-zero vectors in our list of vectors, we apply the gradual exchange lemma twice to ``lower'' all the non-zero vectors. 
Now, we obtain vectors  $e_k'^2, e_k'^1$ so that $e_{3+2(N_0+1)}'^2=e_{3+2(N_0+1)}^2$, $e_{2+3(N_0+1)}'^2 = e_{2+3(N_0+1)}^1$, $e_{3+2(N_0+1)}'^1 = e_{3+2(N_0+1)}^1$, $e_{2+3(N_0+1)}'^1 = -e_{2+3(N_0+1)}^2$ as well as vectors
$e_k'^4, e_k'^3$ so that $e_{3+2(N_0+1)}'^4=e_{3+2(N_0+1)}^4$, $e_{2+3(N_0+1)}'^4 = e_{2+3(N_0+1)}^3$, $e_{3+2(N_0+1)}'^3 = e_{3+2(N_0+1)}^3$, $e_{2+3(N_0+1)}'^3 = -e_{2+3(N_0+1)}^4$.

Our list of vectors continues with
\[e'^4_{3+2(N_0+1)}, 0, e'^2_{3+2(N_0+1)}, 0,\]
\[e'^4_{4+2(N_0+1)}, 0, e'^2_{4+2(N_0+1)}, 0,\]
\[\dots\]
\[e'^4_{2+3(N_0+1)}, 0, e'^2_{2+3(N_0+1)}, 0,\] which is
\[0, e^3_{2+3(N_0+1)}, 0, e^1_{2+3(N_0+1)}.\]
These two applications of the gradual exchange lemma happen in the third column of \ref{GEP_Size4_Columns}(b).

Then we apply the gradual exchange lemma to obtain vectors  $e_k'^2, e_k'^3$ with the expected properties so that our list of vectors continues with
\[0, e'^3_{3+3(N_0+1)}, 0, e'^1_{3+3(N_0+1)},\]
\[0, e'^3_{4+3(N_0+1)}, 0, e'^1_{4+3(N_0+1)},\]
\[0, e'^3_{5+3(N_0+1)}, 0, 0,\]
\[\dots\]
\[0, e'^3_{2+4(N_0+1)}, 0, 0,\]
which is
\[0, 0, e^2_{2+4(N_0+1)}, 0.\]
This application of the gradual exchange lemma happens in the fourth column of \ref{GEP_Size4_Columns}(b). 
The dropping a vector in the orbit of $S_1$ corresponds to the box in the fourth column of \ref{GEP_Size4_Columns}(b).

Then we apply the gradual exchange lemma again to continue our list as
\[0, 0, e'^2_{3+4(N_0+1)}, 0.\]
\[\dots\]
\[0, 0, e'^2_{2+5(N_0+1)}, 0,\]
which is
\[0, 0, 0, e^1_{2+5(N_0+1)}.\]
We finally drop the last vector to obtain
\[0, 0, 0, 0\]
in the next block.
This corresponds to the box in the last column of \ref{GEP_Size4_Columns}(b).
We also  include another 
\[0, 0, 0, 0\]
for the last block
because the dropping of the vector corresponds to setting a weight to zero and we want to not change the first or last weights to facilitate calculating the change to the norm of the self-commutator by allowing us to restrict to each window.
This completes the construction of $F$ using $5(N_0+1)+4$ blocks.

Because $m = 4$, the constant $5$ (the number of columns) is the $2m-3$ that appears in the statement of Lemma \ref{proto-gep}. The $m-1$ comes from the first stage, consisting of the first three columns and $m-2$ comes from the second stage, consisting of the last two columns. 

If we follow the orbits of the vectors that were dropped, we obtain a basis for $R(F^c)$.  We will refer these vectors forming the orbits of $S'$ and the basis of $R(F)$ and $R(F^c)$ by $v_i^r$. \end{example}

Illustration \ref{GEP_Motivation2Rescaled2} is an illustration of the method for $m = 7$ and Illustration \ref{GEP_Motivation2a} illustrates breaking of the diagram into orbits that terminate and begin in this window in the construction of $F$ and $F^c$. 

\begin{figure}[htp]     \centering
    \includegraphics[width=14cm]{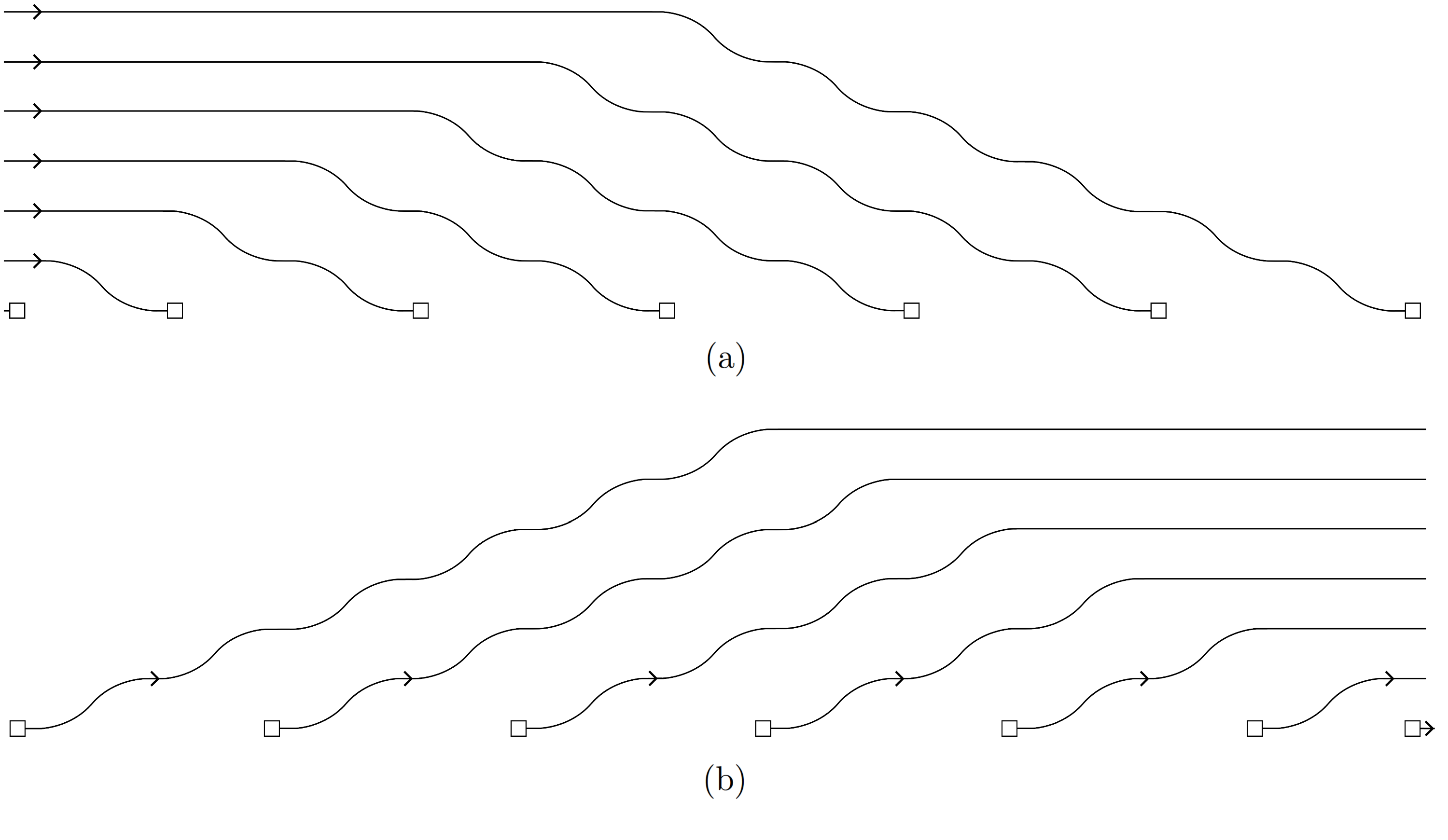}
    \caption{\label{GEP_Motivation2a}\dark
    Illustration of the decomposed weighted shifts as a continuation of   Illustration \ref{GEP_Motivation2Rescaled2}. Compare to Illustration \ref{GEP_Motivation1Rescaled}.
    The portion seen in this window of the orbits of the weighted shifts in (a) form a basis of the range of the projection $F$ and those of (b) in this window form a basis of the range of the projection $F^c$.}
\end{figure}

\vspace{0.1in}

The next three lemmas should be thought of as composing a single lemma but are stated independently to make the construction clearer.  Along the way we include more examples to illustrate the ideas of the proofs. The following is the gradual exchange process for constant-sized but not identical blocks.
\begin{lemma} \label{proto-gep}
Let $A_r=\diag(\alpha_i), S_r = \ws(c_i^r)$ with respect to some orthonormal basis of $M_{n}(\C)$ for $r = 1, \dots, m$. Suppose that the $\alpha_i$ are real and strictly increasing.  Define $A = \bigoplus_r A_r, S = \bigoplus_r S_r$. 

Let $a, b \in \sigma(A)$ with $a < b$. Let $N_0\geq 2$ be a natural number such that \\$n_0= \#\sigma(A)\cap [a,b]\geq (2m-3)(N_0+1)+4$ and $n_0 \geq 3$ in the case that $m = 1$.

Then there is a projection $F$ such that $E_{\{a\}}(A) \leq F \leq E_{[a,b)}(A)$ and a perturbation $S'$ of $S$ with $S'-S$ having support and range in $E_{(a,b]}(A)$ such that $S'$ is a direct sum of weighted shift matrices in a different eigenbasis of $A$, $F$ is an invariant subspace for $S'$, and 
\[\|S'-S\| \leq  \max(G_{[a,b]},\, D_{[a,b]})\]
\[\|\,[S'^\ast, S']\,\| \leq \max\left(\|[S^\ast, S]\|+T_{[a,b]},\, D_{[a,b]}^2\right)\]
where
\begin{align}
G_{[a,b]}&=\max_{1 \leq r \leq m-1}\max_{a \leq \alpha_i\leq b} \left(||c_{i}^{r+1}|-|c_i^r| | + \frac{\pi}{2N_0}\max(|c_i^r|,|c_i^{r+1}|)\right),
\nonumber
\\
D_{[a,b]}&=\max_{a\leq \alpha_i \leq b}|c_i^1|,
\\
T_{[a,b]}&=\frac1{N_0}\max_{1 \leq r \leq m-1}\max_{a \leq \alpha_i\leq b}||c_i^{r+1}|^2-|c_i^r|^2|.\nonumber
\end{align}
Additionally, define $F^c = E_{[a,b]}(A)-F$. Then $S'$ maps $R(F^c)$ into $R(F^c) + R(E_{\{b+\}}(A))$, where $b+ = \min \sigma(A) \cap (b,\infty)$ if $\sigma(A) \cap (b,\infty)\neq\emptyset$ or $b+=b$ otherwise.

If the $c_i^r$ are all real then there is an orthonormal basis of vectors $v_i^r$ that are real linear combinations of the given basis vectors such that $F$ and $F^c$ are each the span of a collection of these vectors and $S'$ is a direct sum of weighted shift matrices with real weights in this basis. The $v_i^r$ are also eigenvectors of $A$.

Note that if $m = 1$ then we use the convention that $G_{[a,b]}=T_{[a,b]}=0$.
\end{lemma}
\begin{remark}
We briefly explain the variable names. The term $G_{[a,b]}$ is the maximal error accrued due to an application of the gradual exchange lemma. The term $D_{[a,b]}$ bounds the weights that are set to zero and hence allow us to ``drop'' vectors from the range of $F$. The term $T_{[a,b]}$ is an additional ``term'' of the norm of the self-commutator that takes into account the interchange of orbits.

Define
\begin{align}
\varepsilon_{[a,b]} &= \max_{1\leq r\leq m-1}\max_{a \leq \alpha_i \leq b}||c_{i}^{r+1}|-|c_i^r||,
\nonumber
\\
R_{[a,b]}&=\frac{\pi}{2N_0}\max_{r}\max_{a \leq \alpha_i \leq b}|c_i^r|,
\nonumber
\end{align}
where
$\varepsilon_{[a,b]}$ is the maximal error due to the small difference in weights inherit in $S$  and $R_{[a,b]}$ is the maximal rotational error from proof of the gradual exchange lemma. It follows that $G_{[a,b]} \leq\varepsilon_{[a,b]}+ R_{[a,b]}$, although this inequality may be strict.
\end{remark}
\begin{proof}
We re-index the $\alpha_i$ in $[a,b]$ and choose
$n_0$ so that $\alpha_1 = a$ and $\alpha_{n_0}=b$.  Without loss of generality, we can assume that $c_i^r \geq 0$ by a change of basis as indicated in Example \ref{ws-basisChange}. Note that this change of basis is done only by multiplying the basis vectors by phases, so it does not affect the structure of $A$ and $S$ as direct sums of diagonal matrices and weighted shift matrices, respectively. The phases are $\pm1$ when the $c_i^r$ were real.

We first consider the trivial case of $m = 1$. With the relabeling given above, $S=\ws(c_{n_\ast}, \dots, c_{n^\ast})$ for $n_\ast\leq 1$ and $n^\ast \geq n_0 \geq 3$.
 We define $S'$ to equal $S$ except $c_2$ is replaced with zero.
Define $F = E_{\{a_1, a_2\}}(A)$ and $F^c = E_{[a_3,a_{n_0}]}(A)$. 

So, \[\|S'-S\| = c_2\leq \max(c_1, \dots, c_{n_0})= D_{[a,b]}.\]
Also,
\begin{align*}
\|\,[S'^\ast, S']\,\| &= \max(c_{n_\ast}^2, |c_{n_\ast+1}^2-c_{n_\ast}^2|,\dots, |c_1^2-c_0^2|, |0^2-c_1^2|, |c_3^2-0^2|, |c_4^2-c_3^2|, \dots,\\
&\;\;\;\;\;\;\;\;\;\;\;\;\;\;\;\;\;|c _{n^\ast}^2-c_{n^\ast-1}^2|, c _{n^\ast}^2) \\
&\leq \max(\|[S^\ast, S]\|, c
_1^2, c_3^2)\leq \max(\|[S^\ast, S]\|, \, D_{[a,b]}^2).
\end{align*}
The rest of the lemma then follows for this case.
 
\vspace{0.05in}

We now do the case that $m \geq 2$.
Let $e_r$ be the standard basis vectors of $\C^m$ and $e^r_i = 0_m^{\oplus(i-1)}\oplus e_r \oplus 0_m^{\oplus(n-i)}$ so that $Se_i^r = c_i^re_{i+1}^r$ for $i < n$ and $e_i^1, \dots, e_i^m$ form a basis for $\mathcal{V}_i = R(E_{a_i}(A))$.  Note that $(2m-3)(N_0+1)+3 < n_0$.

We now group the subspaces $\mathcal V_i$ as follows. The first grouping will consist of $\mathcal V_1$, $\mathcal V_2$. The second grouping will consist of $2m-3$ subgroupings of the $N_0+1$ subspaces
$\mathcal V_{3+(N_0+1)(j-1)}$, $\dots$, $\mathcal V_{2+(N_0+1)j}$, $j = 1, \dots, 2m-3$. Let $\mathcal U_j = \bigoplus_{i=3+(N_0+1)(j-1)}^{2+(N_0+1)j}\mathcal V_i$. The third grouping is formed from the subspaces $\mathcal V_{3+(N_0+1)(2m-3)}, \dots, \mathcal V_{n_0}$. Note that the first and third groupings each consist of at least two of the subspaces $\mathcal V_i$.

For $j = 1, \dots, 2m-3$, we apply the gradual exchange lemma, Lemma \ref{GELws}, 
to pairs of weighted shift operators on the $N_0+1$ subspaces that compose  $\mathcal U_j$.
When $1\leq j \leq m-1$, the pairs of weighted shift operators that we apply the gradual exchange lemma to over the $N_0+1$ subspaces of $\mathcal U_j$ are $S_{j+1-e}, S_{j-e}$ for all even $e \in [0, j)$. 
When $m \leq j \leq 2m-3$, we apply the gradual exchange lemma over those latter $N_0+1$ subspaces of $\mathcal U_j$ to the operators $S_{2m-j-1-e}, S_{2m-j-2-e}$ for all even $e \in [0, 2m-j-2)$. 

Notice that the last pairs of operators in the first range are $S_{m-e}, S_{m-1-e}$ and the first pairs of operators in the second range are $S_{m-1-e}, S_{m-2-e}$. This means that if we have interchanged the orbits of some $S_{t}, S_{t-1}$ over $\mathcal U_{m-1}$ and $t-1>1$ then over $\mathcal U_{m}$ we will interchange of orbits of $S_{t-1}, S_{t-2}$. 
So, we will continue lowering the orbit of $S_t$ to $S_{t-1}$ then to $S_{t-2}$ across the value $j = m$. Because the indices $2m-j-2$ and $2m-j-1$ decrease by one for each increase of $j$ by one, we see that the orbit of $S_{t-2}$ will continue to be lowered if $t-2 > 1$. This will be useful later in the proof.

Let $\tilde{S}$ be the operator obtained from these modifications of $S$. Consider an orbit of $\tilde S$ while it is interchanging the orbits of two operators $S_t, S_{t-1}$ over the interval of indices $[i_0, i_1] = [3+(N_0+1)(j-1),2+(N_0+1)j]$. 
By Lemma \ref{GELws}(iii), when interchanging one orbit to the other, the weight at $i_0$ is the weight of $S$ corresponding to the former orbit and the weight at $i_1$ is the weight of $S$ corresponding to the latter orbit.
So, using the fact that the applications of the gradual exchange lemma are done independently over orthogonal subspaces, we see that with the arguments used in Example \ref{GEL_Arrows} that 
\begin{align}\label{GTildeEst}
\|\tilde{S} - S\| &\leq G_{[a,b]}
\\\label{TTildeEst}\|\,[\tilde S^\ast, \tilde S]\,\|&\leq \|[S^\ast, S]\|+T_{[a,b]}.
\end{align}

Now consider the orbit of $e_1^r$ under $\tilde{S}$. We know that $\tilde{S}$ is a direct sum of weighted shift operators whose orbits each start with a $e_1^r$. We claim that for each $r$, the orbit of $e_1^r$ under $\tilde S$ is eventually in the orbit of $S_1$. The following discussion is devoted to discussing this and finding particular weights $c_k^1$ in the orbit of $S_1$ that we will set equal to zero.

First, suppose that $r = 1$. In this case, we can just choose $k= 2$ just as in the case that $m=1$. The basis vector then belongs to the second subspace of the first grouping of subspaces. Suppose now that $2 \leq r \leq m$. Notice that the action of $S$ and $\tilde{S}$ on $e^r_1$ are identical on the $\mathcal U_j$ for $j < r-1$. 
When $j = r-1$, the gradual exchange lemma is applied to $ S_r, S_{r-1}$ over $\mathcal U_{r-1}$. 
So, the orbit of $e_1^r$ under $\tilde{S}$ moves from the orbit of $S_{r}$ to the orbit of $S_{r-1}$ by the beginning of $\mathcal U_{r}$. 
Then upon each application of the gradual exchange lemma, the orbit of $e_1^r$ under $\tilde{S}$ moves to $S_t$ with decreasing values of $t$. This clearly continues while both $j \leq m-1$ and the orbit is still not in the orbit of $S_1$. 

Observe that since $e^r_1$ begins to be lowered over $\mathcal U_{r-1}$ and $r-1$ orbits must be lowered, the orbit is finally lowered to the orbit of $S_1$ over $\mathcal U_{j_r}$ when $j_r = r-2+r-1 = 2r-3$.
Note that $S_m$ is the last orbit to begin to be lowered and, by construction, once it is lowered to $S_1$ over $\mathcal U_{2m-3}$, no more applications of the gradual exchange lemma are applied.
Note also that for all $r$ but $r=m$, the orbit of $e_1^r$ under $\tilde{S}$ will move back upward into the orbit of $S_t$ for some increasing values of $t$ as the result of the subsequent applications of the gradual exchange lemma.

In particular, if the orbit is moved from $S_2$ into $S_1$ over $\mathcal U_j$ then no application of the gradual exchange lemma is applied to $S_1$ over $\mathcal U_{j+1}$. 
More specifically, when $j$ is even, the gradual exchange lemma is not applied to $S_1$. 
So, for $j=2r-2$ with $2 \leq r\leq m-1$, we replace $c_i^1$ with zero for the second value of $i$ in $[3+(N_0+1)(j-1), 2+(N_0+1)j]$.
Denote this value of $i$ by $i_r=4+(N_0+1)(2r-3)$. So, we see that $e_1^r$ is annihilated by the $i_r$-th application of $\tilde{S}$ after this modification.

We extend this property to $r = 1, m$ by also replacing $c_2^1$ and $c_{3+(N_0+1)(2m-3)}^m$ with zero and defining with $i_1 = 2$ and $i_{m}=3+(N_0+1)(2m-3)$. 
So, all the $i_r$ are greater than $1$ and less than $n_0$. 

Let $S'$ be the operator gotten by applying these modifications to $\tilde{S}$. 
The estimate for $\|S'- S\|$ follows from  Equation (\ref{GTildeEst}) and the way that we set weights equal to zero that are bounded by $D_{[a,b]}$, just as in the case when $m = 1$.

Now, $S'$ is a direct sum of weighted shift operators in different $n$-dimensional orthogonal subspaces of $M_n(\C)^{\oplus m}$. Hence, we can obtain vectors ${v_i^r}$ due to the applications of the gradual exchange lemma with respect to the summands of $S'$ are weighted shift matrices.  So, $S'{v^r_i} = {c^r_i}'{v^r_{i+1}}$ for $i < n$ and $v^1_i, \dots, v^{m}_i$ form a basis for $\mathcal{V}_i = R(E_{a_i}(A))$, having the same span as $e^1_i,\dots, e^m_i$.
Define $F$ to be the span of \[v_i^r:\, 1 \leq r \leq m,\, 1\leq i\leq i_r.\] We see that $F$ is an orthogonal projection such that $R(F)$ is an invariant subspace for $S'$ and the other desired properties hold.
By this definition, we have that $F^c$ is the span of 
\[v_i^r:\, 1 \leq r \leq m,\, i_r < i \leq n_0.\]
We then see that $S'(R(F^c))$ is orthogonal to $R(F)$. So, because $S'$ maps $R(E_{[a,b]}(A))$ into $R(E_{[a,b+]}(A))$, the desired property of $F^c$ is obtained. 

When the $c_i^r$ are real, the desired properties follow from the use of real phases and the real coefficient properties from Lemma \ref{GELws}(i).

We now justify the estimate of the self-commutator of $S'$. Observe that replacing weights $d_{k}$ for $k$ in the index set $\mathcal I$ of a weighted shift matrix $T=\ws(d_i)$ with zero to create a weighted shift $T'$ will produce the estimate
\[\|[T'^\ast, T']\| \leq \max\left(\|[T^\ast, T]\|,\, \max_{k \in \mathcal I}\max(|d_{k-1}|^2, |d_{k+1}|^2)\right) \]
by the argument used in the case where $m=1$. When going from $\tilde S$ to $S'$ we are doing exactly this for the weighted shift operator summands of $\tilde S$.
By construction, the weights before and after the weight set to zero are weights of $S_1$ and hence are bounded by $D_{[a,b]}$. By this argument and Equation (\ref{TTildeEst}), we obtain the desired estimate for $\|[S'^\ast, S']\|$. 

When considering the support and range of $S'-S$, we see that the perturbations due to the gradual exchange lemma have support and range in the $\mathcal U_j$:
\begin{align}\label{GEParrow}
\mathcal V_{3+(N_0+1)(j-1)}\overset{}{\rightarrow}\mathcal V_{4+(N_0+1)(j-1)}\overset{\ast}{\rightarrow}\mathcal V_{5+(N_0+1)(j-1)}\overset{}{\rightarrow}\cdots\overset{}{\rightarrow} \mathcal V_{2+(N_0+1)j}\overset{}{\rightarrow}
\end{align}
where we have illustrated the action of either $S$ or $S'$ using the arrows between subspaces. Because $\mathcal V_1, \mathcal V_{n_0}$ are not included in the $\mathcal U_j$, 
the range and support of the perturbation $\tilde S - S$  is within the range of $E_{(a,b)}(A).$

The $\ast$ in (\ref{GEParrow}) indicates where the weights in the orbit of $S_1$ may be potentially set to zero.  The contribution to the perturbation $S'-S$ of setting the weight equal to zero within the $\mathcal U_j$ then has support and range in $E_{(a,b)}(A)$ as well.

Likewise, consider where the first and last weight is set equal to zero outside the $\mathcal U_j$ as indicated by the $\ast$'s:
\begin{align}\label{GEParrow2}
\mathcal V_1 \overset{}{\rightarrow} \mathcal V_2 \overset{\ast}{\rightarrow} \mathcal U_1\overset{}{\rightarrow}\cdots \overset{}{\rightarrow} \mathcal U_{2m-3}\overset{}{\rightarrow}\mathcal V_{3+(N_0+1)(2m-3)}\overset{\ast}{\rightarrow} \mathcal V_{4+(N_0+1)(2m-3)}\overset{}{\rightarrow}
\end{align}
We see that because $n_0 \geq 4+(N_0+1)(2m-3)$ that the support of $S'-S$ is in the range of $E_{(a,b)}(A)$ and the range of $S'-S$ is in the range of $E_{(a,b]}(A)$.
So, in total, the support and range of $S'-S$ is as stated in the lemma.
\end{proof}

\section{Gradual Exchange Process -- General Case}
Now, we illustrate the following result concerning when the blocks $C_i$  are not all  the same size. This is equivalent to the statement of the previous lemma when the matrices $A_r$ have spectrum growing in $r$. The idea is that if $\sigma(A_r)$ for some $r$ does not contain the entire spectrum of $A = \bigoplus_r A_r$ in the interval that we are looking at then $S_r$ already has an invariant subspace that we just include.

For example, suppose that $2 < n_1 < n_2 = \cdots = n_5$ and consider  $A_r =$ $\diag(\alpha_1, \dots, \alpha_{n_r})$ where $\alpha_1 < \cdots < \alpha_{n_1} < \cdots < \alpha_{n_2}=\cdots=\alpha_{n_5}$ and $S_r = \ws(c_i^r)$ in $M_{n_r}(\C)$. Then $\sigma(A_1) = [\alpha_1, \alpha_{n_1}] \cap \sigma(A)\subsetneq \sigma(A_2) = \cdots = \sigma(A_5)=\sigma(A)$. In this example, $r_0 = 2$ as defined  in the lemma below. Illustration \ref{GEP_Motivation3} illustrates the method that is used in the following lemma. 
\begin{figure}[htp]     \centering
    \includegraphics[width=14cm]{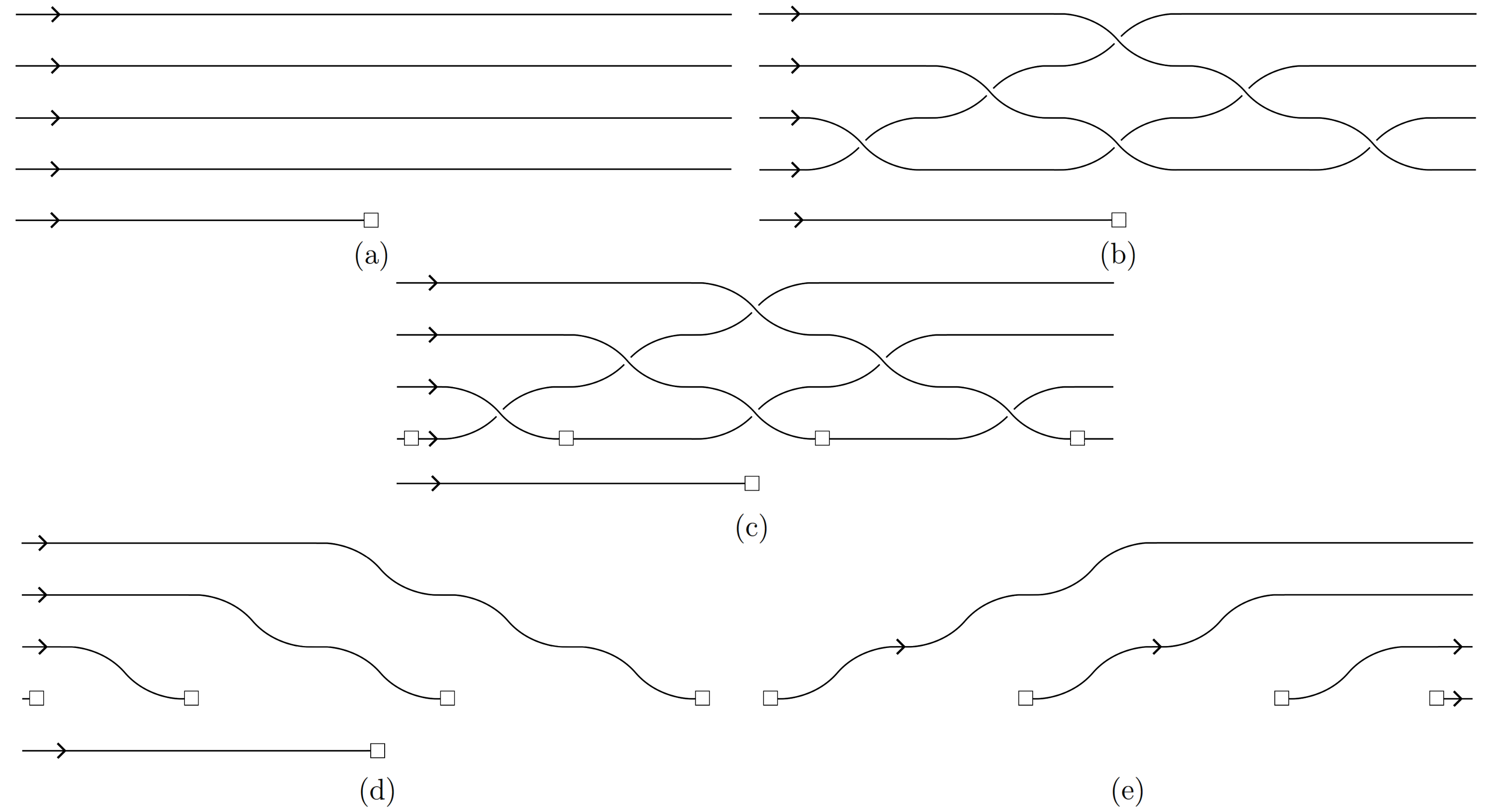}
    \caption{\label{GEP_Motivation3}\dark
    Illustration of appending another weighted shift operator whose orbit does not span the window to that of Example \ref{gep-Example}. Compare to Illustration \ref{GEP_Motivation1Rescaled}.}
\end{figure}

Although more general forms of this lemma can be imagined, we only state what we will find useful in the next chapter. Note that if  $\underline{i}_r$ and $\overline{i}_r$ are constant, this lemma follows from the previous lemma. 
\begin{lemma} \label{proto-gep2}
Let $A_r=\diag(\alpha_i), S_r = \ws(c_i^r)$ with respect to some orthonormal basis of $M_{n_r}(\C)$ for $r = 1, \dots, m$ and $i= \underline{i}_r, \dots, \overline{i}_r$, where $[\underline{i}_r, \overline{i}_r]\subset [\underline{i}_{r+1}, \overline{i}_{r+1}]$.
Suppose that the $\alpha_i$ are real and strictly increasing.  Define $A = \bigoplus_r A_r, S = \bigoplus_r S_r$. 

Let $a, b \in \R$ with $a < b$. Let $N_0\geq 2$ be a natural number such that 
\[\#\sigma(A)\cap [a,b]\geq \max(3, (2m-3)(N_0+1)+4).\]
Let $\mathscr R_{I} = \{r: \sigma(A)\cap I\subset \sigma(A_r)\}$. Consequently, $\mathscr R_{[a,b]}$ is empty or equal to $r_0, r_0+1, \dots, m$ for some $r_0 \geq 1$. Let $a^\sigma = \min \sigma(A) \cap [a,b]$ and $b^\sigma = \max \sigma(A) \cap [a,b]$.

Then there is a projection $F$ such that $E_{\{a^\sigma\}}(A) \leq F \leq E_{[a^\sigma,b^\sigma)}(A)$ and a perturbation $S'$ of $S$ with $S'-S$ having support and range in $E_{(a^\sigma,b^\sigma]}(A)$ such that $S'$ is a direct sum of weighted shift matrices in a different eigenbasis of $A$, $F$ is an invariant subspace for $S'$, and 
\begin{align}\|S'-S\| &\leq  \max(G_{[a,b]}, D_{[a,b]}), \nonumber \\
\|\,[S'^\ast, S']\,\| &\leq \max\left(\|[S^\ast, S]\|+T_{[a,b]},\, D_{[a,b]}^2\right),\nonumber
\end{align}
where
\begin{align}
G_I &= \max_{\substack{r<m\\ r\in\mathscr R_{I} }}\max_{\alpha_i \in I} \left(||c_{i}^{r+1}|-|c_i^r| | + \frac{\pi}{2N_0}\max(|c_i^r|,|c_i^{r+1}|)\right),\nonumber
\\
D_{I}&=\max_{\alpha_i \in I}|c_i^{r_0}|\nonumber,\\
T_{I}&=\frac1{N_0}\max_{\substack{r<m\\ r\in\mathscr R_{I} }}\max_{\alpha_i \in I}||c_i^{r+1}|^2-|c_i^r|^2|.\nonumber
\end{align}
If $\mathscr R_{[a,b]}$ is empty then $S' = S$.
Additionally, define $F^c = E_{[a,b]}(A)-F$. Then $S'$ maps $R(F^c)$ into $R(F^c) + R(E_{\{b+\}}(A))$, where $b+ = \min \sigma(A) \cap (b,\infty)$ if $\sigma(A) \cap (b,\infty)\neq\emptyset$ or $b+=b^\sigma$ otherwise.

If the $c_i^r$ are all real then there is an orthonormal basis of vectors $v_i^r$ that are real linear combinations of the given basis vectors such that $F$ and $F^c$ are each the span of a collection of these vectors and $S'$ is a direct sum of weighted shift matrices with real weights in this basis. The $v_i^r$ are also eigenvectors of $A$.
\end{lemma}
\begin{proof}
Note that $n_r = \overline{i}_r - \underline{i}_r + 1$.

Let $G$ be the projection in $\mathcal M=\bigoplus_{r}M_{n_r}(\C)$ onto \[\mathcal G= \bigoplus_{r < r_0}0^{\oplus n_r}\oplus \bigoplus_{r \geq r_0}M_{n_r}(\C).\] Note that $\mathcal G$ is clearly an invariant subspace of $A$ and $S$. 
Now, we apply Lemma \ref{proto-gep} to $A_r, S_r$ for $r = r_0,\dots, m$ over $[a^\sigma,b^\sigma]$. This provides an operator $S'$ and projection $F$ on $\mathcal G$ with the desired properties with the exception that $F$ contains the projection onto $R(E_{\{a\}}(A))\cap \mathcal G$ and the estimate we have for $S'$ is 
\[\|(S'-S)G\| \leq \max(G_{[a^\sigma,b^\sigma]}, D_{[a^\sigma,b^\sigma]})\]
for $G_I$ and $D_I$ in the statement of the lemma.

We will identify $S'$, $F$, and $F^c$ with the operators on $\mathcal M$ that are gotten by trivially extending them to be zero on $\mathcal M \ominus \mathcal G$. However, the operator $S_{\mathcal M}$ and projections $F_{\mathcal M}$ and $F_{\mathcal M}^c$ that we construct for the first part of the statement of this lemma will in general be non-trivial extensions.

If $r_0 = 1$, then $G = I$ so the proof is complete.  So, suppose that $r_0 > 1$. 
Define 
\begin{align}\label{F_Mdef}
F_{\mathcal M} = F + \sum_{\substack{r < r_0\\ a^\sigma \in \sigma(A_r)}}E_{[a,b]}(A_r),\; F_{\mathcal M}^c = F^c + \sum_{\substack{r < r_0\\ a^\sigma \not\in \sigma(A_r)}}E_{[a,b]}(A_r).
\end{align}
Note that $F_{\mathcal M} - F$ and $F_{\mathcal M}^c - F^c$ are both projections into $R(I-G)$.

Recall the following basic property of $A_r = \diag(\alpha_i)$ and $S_r=\ws(c^r_i)$. If $v_i \in R(E_{\alpha_i}(A_r))$ then $S_r^k v_i \in R(E_{\alpha_{i+k}}(A_r))$. The following statements about $E_{[a,b]}(A_r)$ are then straightforward consequences of the assumptions on the $A_r$.
For $r < r_0$, there is an $\alpha \in [a,b]\cap \sigma(A)$ such that $\alpha \not \in \sigma(A_r)$. Because $\sigma(A_r) = \{\alpha_i: i \in [\underline{i}_r, \overline{i}_r]\}$ and $\sigma(A) = \{\alpha_i: i \in [\underline{i}_m, \overline{i}_m]\}$, it is not possible that $\sigma(A_r)$ contains both $a^\sigma$ and $b^\sigma$.

For each $r < r_0$ such that $a^\sigma \in \sigma(A_r)$, since $b^\sigma \not \in \sigma(A_r)$, we see that there is a $b_r \in [a^\sigma,b^\sigma)$ such that $[a,b] \cap \sigma(A_r) = [a^\sigma,b_r]$. Consequently, $R(E_{[a,b]}(A_r)) = R(E_{[a,b_r]}(A_r))$ is invariant under $S_r$.  
Likewise, consider $r < r_0$ such that $a^\sigma \not\in \sigma(A_r)$. If $[a,b]\cap\sigma(A_r)= \emptyset$, then $E_{[a,b]}(A_r) = 0$. Otherwise, there is an $a_r \in (a^\sigma, b^\sigma]$ such that $[a,b]\cap \sigma(A_r) \subset [a_r, b^\sigma]$. 
So, we see that $R(E_{[a,b]}(A_r))=R(E_{[a_r, b^\sigma]}(A_r))$ is mapped into $R(E_{[a_r, b+]}(A_r))$ by $S_r$.  
So, we obtain $E_{\{a^\sigma\}}(A) \leq F_{\mathcal M} \leq E_{[a^\sigma, b^\sigma)}(A)$ and  $F_{\mathcal M}^c = E_{[a^\sigma, b^\sigma]}(A)-F_{\mathcal M}$. 

We now extend $S'$ from $\mathcal G$ to $S'_{\mathcal M} = S'G + S(1-G)$ on $\mathcal M$. We then have $\|S'_{\mathcal M} - S\| = \|(S'-S)G\|$ with the above estimate. The estimate for the self-commutator of $S'$ holds similarly.  By the discussion above, $F_{\mathcal M}$ is invariant under $S'_{\mathcal M}$.

Therefore the desired property for $F^c_{\mathcal M}$ follows from that of $F^c$ from Lemma \ref{proto-gep} and each summand $E_{[a,b]}(A_r)$ in the definition of $F^c_{\mathcal M}$.
 
\end{proof}

\begin{remark}
We can instead assume that the spectrum of $A$ lies on a nice simple curve homeomorphic to an interval in 
$\R$. For instance, instead of increasing real numbers on a line, the $\alpha_i$ could be complex numbers on the unit circle with increasing argument. In this case, $A$ would be unitary and the $S_r$ could be either unilateral or bilateral weighted shifts. There are other generalizations possible.
\end{remark}

We now give an example of the construction of the following lemma.
\begin{figure}[htp]     \centering
    \includegraphics[width=14cm]{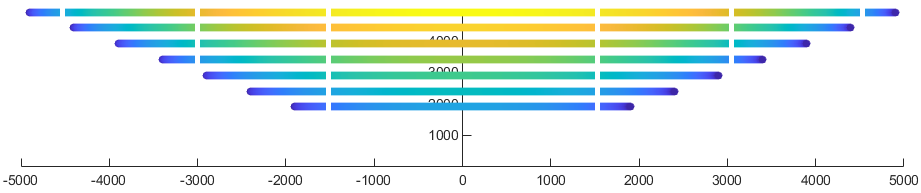}
    \caption{\label{WSheatmap_GEP_Motivation4}\dark
    Illustration of the weights of $S$ in Example \ref{S''Motivation}. The vertical gaps in the graph are shown to illustrate the windows in which we apply the gradual exchange method.}
\end{figure}
\begin{figure}[htp]     \centering
    \includegraphics[width=14cm]{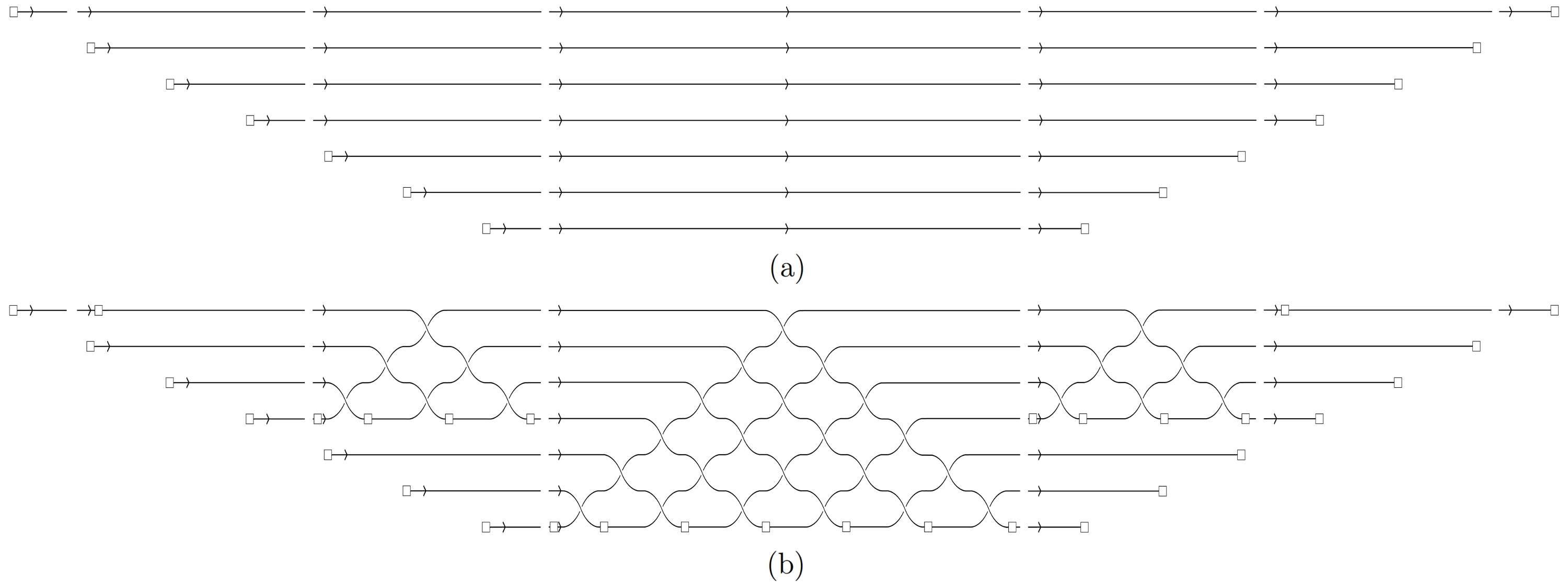}
    \caption{\label{GEP_Motivation4a}\dark
    Illustration of the applications of the gradual exchange lemma during the construction of $S'$ in Example \ref{S''Motivation}.}
\end{figure}
\begin{example}\label{S''Motivation}
Here we illustrate the construction of $S'$ and the $F_i, F_i^c$. Consider 
\[A = \frac1{4900}\left(S^{1900}\oplus S^{2400} \oplus \cdots \oplus S^{4900}(\sigma_3)\right), S = \frac1{4900}\left(S^{1900}\oplus S^{2400} \oplus \cdots \oplus S^{4900}(\sigma_+)\right).\]
A weighted shift diagram for $S$ is provided in Illustration \ref{WSheatmap_GEP_Motivation4}. Note that the vertical gaps in the orbits are included to illustrate the windows that we deal with using the prior lemma and not that the orbits terminate.

\begin{figure}[htp]     \centering
    \includegraphics[width=14cm]{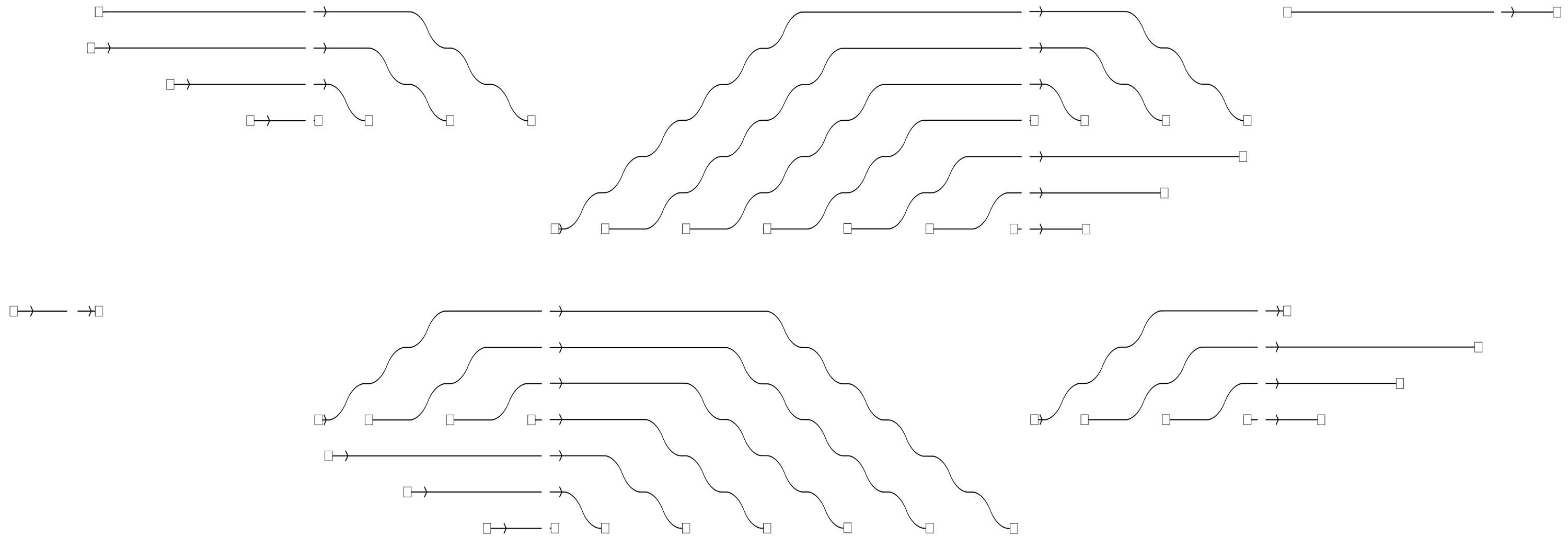}
    \caption{\label{GEP_Motivation4b}\dark
    Illustration of decomposed weighted shift operators of $S'$ in Example \ref{S''Motivation}.}
\end{figure}
\begin{figure}[htp]     \centering
    \includegraphics[width=14cm]{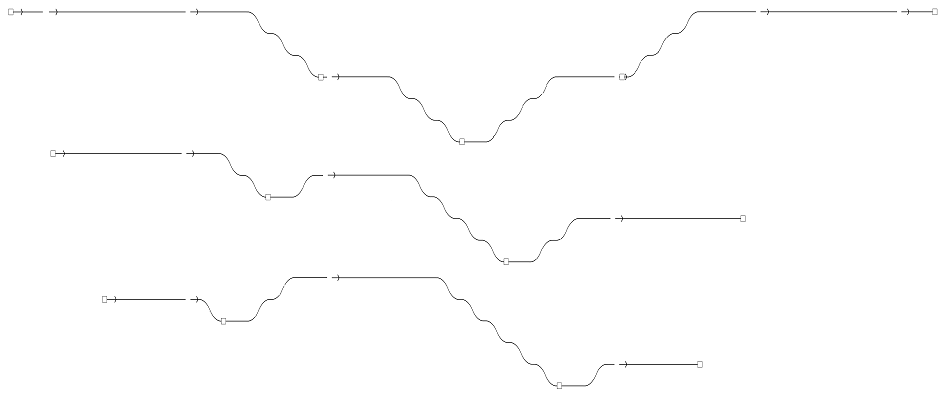}
    \caption{\label{GEP_Motivation4bOrbits}\dark
    Illustration of three orbits of $\tilde{S}$ in Example \ref{WSheatmap_GEP_Motivation4}(b).}
\end{figure}
\begin{figure}[htp]     \centering
    \includegraphics[width=14cm]{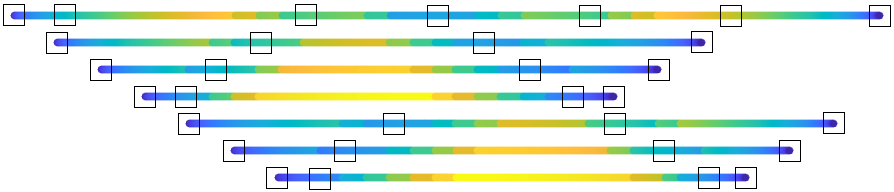}
    \caption{\label{WSheatmap_GEP_MixDecomp}\dark
    Illustration of decomposed weighted shift operators of $S'$ in Example \ref{S''Motivation}. Marked with boxes are the weights that are dropped in the construction detailed above.\\ Note that generating the colors was done using a different version of the gradual exchange lemma that does not continuously change the values of weights between orbits. }
\end{figure}

\begin{figure}[htp] \centering
    \includegraphics[width=14cm]{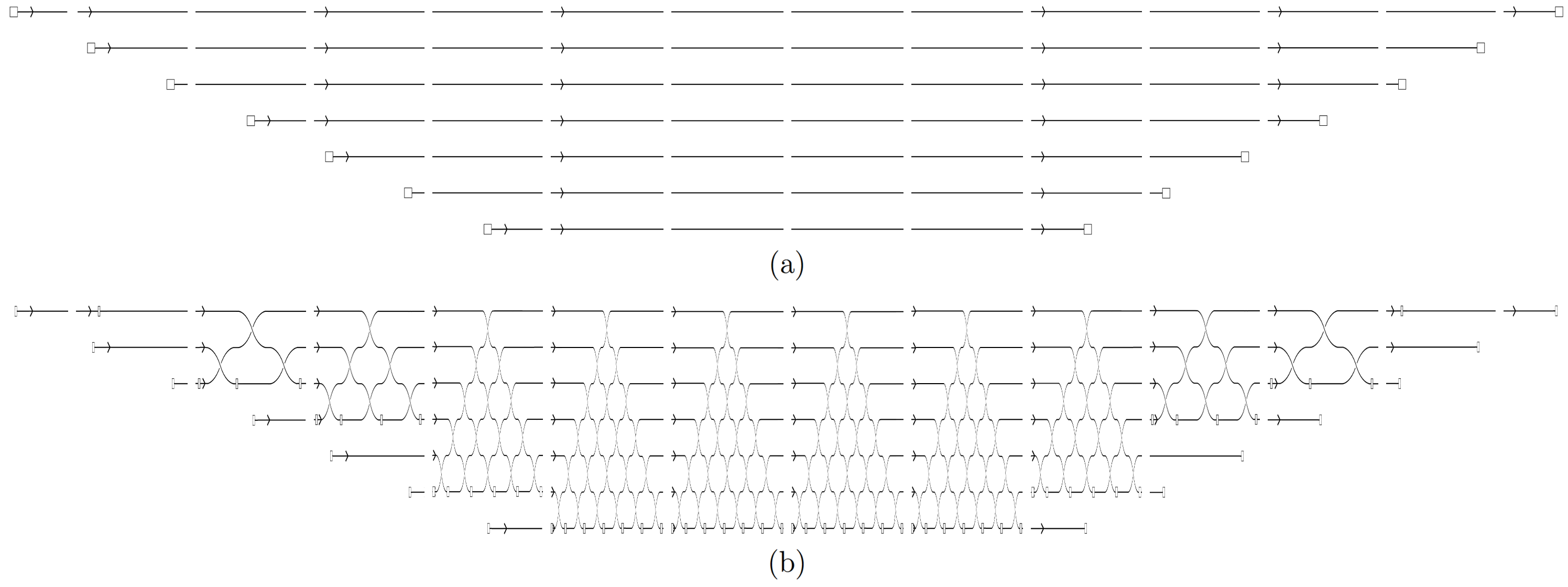}
    \caption{\label{GEP_Motivation_MoreOsc}\dark
    Illustration of applying gradual exchange process with a smaller window size.}
\end{figure}
\begin{figure}[htp] \centering
    \includegraphics[width=14cm]{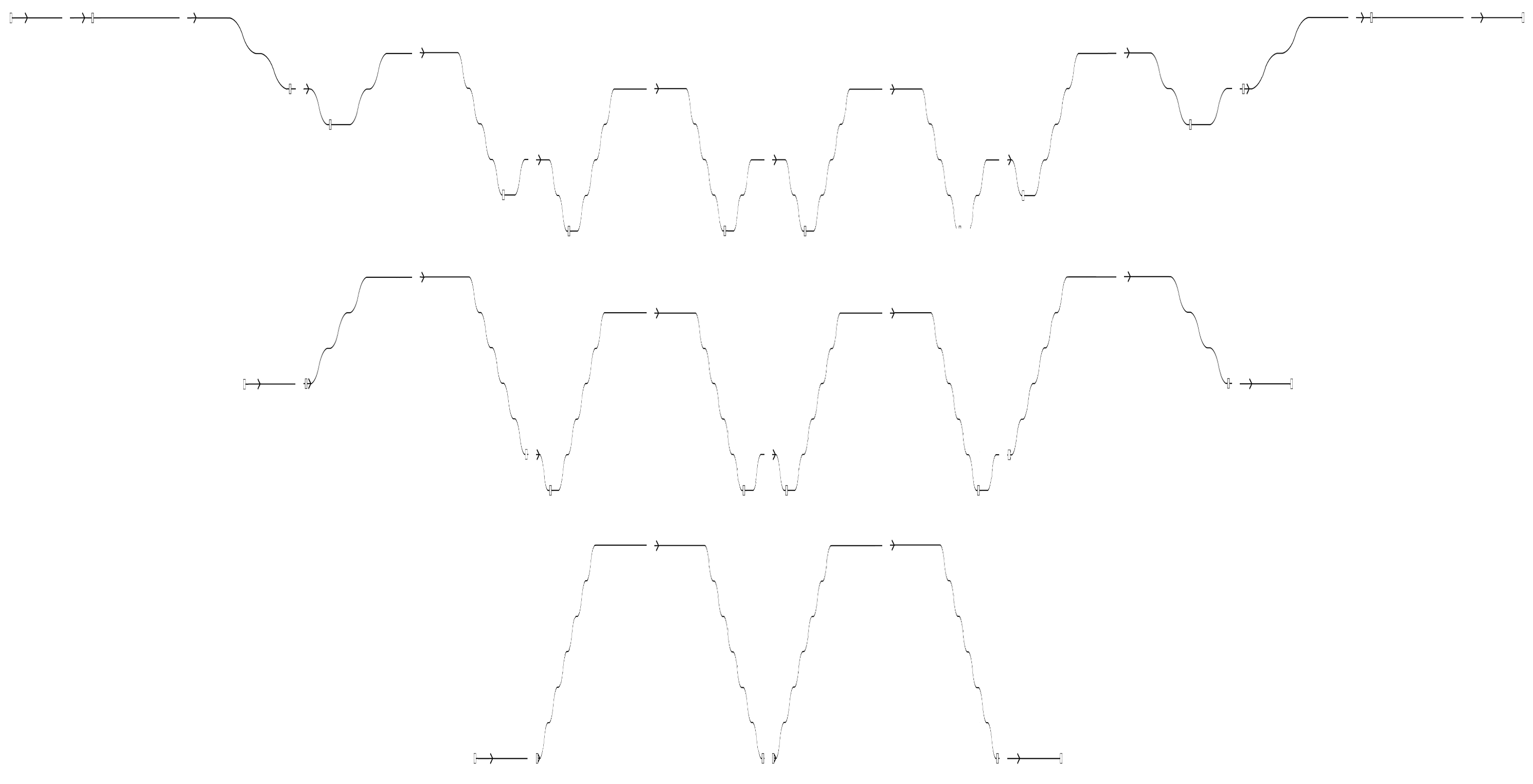}
    \caption{\label{GEP_Motivation_MoreOscOrbits}\dark Illustration of several of the orbits of $\tilde{S}$.}
\end{figure}
\begin{figure}[htp] \centering
    \includegraphics[width=14cm]{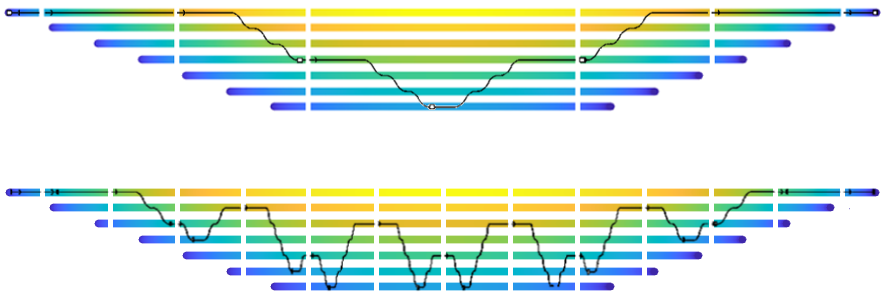}
    \caption{\label{GEP_MotivationOrbitsColor}\dark Illustration of several of the orbits of $\tilde{S}$ for different window sizes superimposed upon the colorbar graph of the values of the weights.}
\end{figure}
Illustration \ref{GEP_Motivation4a}(a) is a depiction of $S$. 
Illustration \ref{GEP_Motivation4a}(b) depicts the gradual exchange process that we developed earlier in each window. For each window, we explored earlier that $S'$ is a direct sum (in a rotated basis) of weighted shift operators whose orbits are broken in the window. Because we do this in each window, we can piece together these orbits.
Illustrations \ref{GEP_Motivation4b}(a) and \ref{GEP_Motivation4b}(b) illustrate these orbits. We then use these orbits to construct projections $E_j$ so that $A'$ has spectral projections $E_j$. Because each orbit belongs to at most two consecutive windows, $A'$ will be approximately equal to $A$ if the window length is small. 

For each orbit, we construct a nearby normal using Theorem \ref{BergResult}. Then putting these normals together gives $S''$.

\end{example}

We repeat the notation from the previous lemma in the statement of the next lemma. This result completes the construction of nearby commuting matrices using the gradual exchange process. The use of projections to construct nearby commuting matrices is motivated by the constructions in \cite{hastings2009making, davidson1985almost}.
\begin{lemma}\label{gep}
Let $A_r=\diag(\alpha_i), S_r = \ws(c_i^r)$ with respect to some orthonormal basis of $M_{n_r}(\C)$ for $r = 1, \dots, m$ and $i= \underline{i}_r, \dots, \overline{i}_r$, where $[\underline{i}_r, \overline{i}_r]\subset [\underline{i}_{r+1}, \overline{i}_{r+1}]$. 
Suppose that the $\alpha_i$ are real and strictly increasing.  Define $A = \bigoplus_r A_r, S = \bigoplus_r S_r$.  Let $\mathscr R_{I} = \{r: \sigma(A)\cap I\subset \sigma(A_r)\}$. Consequently, $\mathscr R_{I}$ is empty or equal to $r_0, r_0+1, \dots, m$ for some $r_0 = r_0(I) \geq 1$ which may depend on $I$. 

Let $a_k \in \R$, $a_1 < a_2 < \dots < a_{n_0}$, $I_k= [a_k, a_{k+1})$ for $k+1 < n_0$ and $I_{n_0-1} = [a_{n_0-1}, a_{n_0}]$, satisfying $\sigma(A) \subset \bigcup_k I_k$.
Let $m_k = m+1-r_0(I_k) \leq m$ and let $N_{I_k} = N_k$ be natural numbers such that
\begin{align}\label{minspectrum}
\#\sigma(A)\cap I_k\geq \max\left(3, (2m_k-3)(N_k+1)+4\right).
\end{align}
Let
\begin{align}
G_I &= \max_{\substack{r<m\\ r\in\mathscr R_{I} }}\max_{\alpha_i \in I} \left(||c_{i}^{r+1}|-|c_i^r|| + \frac{\pi}{2N_I}\max(|c_i^r|,|c_i^{r+1}|)\right)
\\
D_{I}&=\max_{\alpha_i \in I}|c_i^{r_0}|\\
T_{I}&=\frac1{N_I}\max_{\substack{r<m\\ r\in\mathscr R_{I} }}\max_{\alpha_i \in I}||c_i^{r+1}|^2-|c_i^r|^2|.
\end{align}

Then there is a self-adjoint matrix $A'$ commuting with a matrix $S'$ that is a direct sum of weighted shift matrices in an eigenbasis of $A'$ such that 
\begin{align}\|A' - A\| &\leq \max_k\diam I_k,\\
\|S'-S\| &\leq \max_k\max(G_{I_k}, D_{I_k}),\noindent
\\
\|\,[S'^\ast, S']\,\| &\leq \max_k \max\left(\|[S^\ast, S]\|+T_{I_k},\, D_{I_k}^2\right).
\end{align}
Moreover, there is a normal $S''$ that is a direct sum of weighted shift matrices in an eigenbasis of $A'$ such that
\begin{align}
 \|S''-S'\| \leq C_\alpha \|S\|^{1-2\alpha}\|\,[S'^\ast, S']\,\|^\alpha \label{S''ineq}
\end{align}
where $\alpha, C_\alpha > 0$ are constants such that a nearby normal matrix can be obtained by Theorem \ref{BergResult}.

If the $c_i^r$ are real then using $\alpha = 1/3, C_{1/3} = 5.3308$ allows $S''$ to be real.
Moreover, there is a real change of basis that makes $S''$ (and also $S'$) a direct sum of weighted shift matrices with real weights.
\end{lemma}
\begin{remark}
If we estimate
\begin{align}
\varepsilon_{I} &= \max_{\substack{r<m\\ r\in\mathscr R_{I} }}\max_{\alpha_i \in I}||c_{i}^{r+1}|-|c_i^r| |,\nonumber\\
R_{I}&=\frac{\pi}{2N}\max_{r \in \mathscr R_{I}}\max_{\alpha_i \in I}|c_i^r|\nonumber
\end{align}
separately then we obtain the bounds for $G_I$: $\max(\varepsilon_{I},R_{I}) \leq G_I \leq \varepsilon_{I} + R_{I}$.
\end{remark}
\begin{proof}
\underline{Construction of and estimates for $A'$ and $S'$}: Let $a^\sigma_k = \min \sigma(A)\cap [a_k, a_{k+1})$ and $b^\sigma_k = \max \sigma(A)\cap [a_k, a_{k+1})$. 
Let $F_k$ be the projection gotten by applying the construction in Lemma \ref{proto-gep2} for $[a^\sigma_k,b^\sigma_k]$, let $S_k'$ be the constructed perturbation of $S$, and $F_k^c = E_{[a^\sigma_k,b^\sigma_k]}(A) - F_k$. 
Note that $E_{\{b^\sigma_k\}}(A)\leq F_k^c \leq E_{(a^\sigma_k, b^\sigma_k]}(A)$. 
Define \[S' = S+\sum_k (S'_k - S).\]
The definition that we give here for $S'$ is the same as applying all these perturbations from the previous lemma in each window separately. Because the perturbations $S'_k-S$ are supported on and have range in the orthogonal subspaces $R(E_{I_k}(A))$, we obtain the desired estimate for $\|S'-S\|$.

Consider the orthogonal projections $E_k$ defined to be the \[F_1, F_1^c + F_2, \dots, F_{k}^c+F_{k+1}, \dots, F_{n_0-1}^c+F_{n_0}, F_{n_0}^c.\]   
Because the $F_k$ are invariant under $S'$ and $S'$ maps $R(E_{[a_k,b_k]}(A))$ into $R(E_{[a_k,a^\sigma_{k+1}]}(A))$, 
we see that $S'$ maps $R(F_{k}^c)$ into $R(F_{k}^c)+R(F_{k+1})$. 
Hence, the projections $E_k$ commute with $S'$. 
Note that $E_k\leq E_{[a_{k-1}, a_{k+1}]}(A)$ if $a_0$ is defined to be $a_1$ and $a_{n_0+1}$ is defined to be $a_{n_0}$. So, letting $A' = \sum_k a_{k}E_k$, we see that $[S',A']=0$ and $\|A'-A\|\leq \max_k (a_{k+1}-a_k)$. 

\vspace{0.05in}

\underline{Construction of and estimates for $S''$}: We now take advantage of the structure of $S'$ through the operators $S'_k$, which were called  $S'_{\mathcal M}$ in the proof of Lemma \ref{proto-gep2}. Please recall the construction of what was called $S'$ in Lemma \ref{proto-gep}, in particular the statement about the support and range of $S'-S$ illustrated in Equations (\ref{GEParrow}) and (\ref{GEParrow2}). These contribute to the construction of each $S'_k$. 

We know that $S'$ is a direct sum of weighted shift operators. Because the construction of $S'_k$ in each window did not change the weights of the weighted shifts on the boundaries, we see that the differences of the squares of the $S'$ weights between windows are the same as those of $S$ between windows. Within windows, the differences of squares of $S'$ weights are bounded by the estimates for the self-commutator of the $S_k'$ in Lemma \ref{proto-gep2}. So, the desired estimate for the self-commutator of $S'$ holds.

Because $S'$ commutes with $A'$, we can view the orbits of $S'$ as lying within the eigenspaces of $A'$. We then apply Theorem \ref{BergResult} to each such weighted shift orbit to obtain $S''$. If the $c_i^r$ are real then the additional structure follows from that of Lemma \ref{proto-gep2}.

\end{proof}

\begin{remark}
We now discuss the utility of the estimates gotten in this construction. 

We first discuss the term $D_I$.
Under some mild conditions, we need the singular values $\min_{r}\min_{\alpha_i \in I}|c_i^r|$ to be small in order for there to exist structured nearby commuting matrices by a generalization of Voiculescu's argument in \cite{voiculescu1983asymptotically}. This suggests that the estimate of $D_I = \max_{\alpha_i\in I}|c_i^{r_0}|$ might be small for situations where we want to construct nearby commuting matrices. 

The construction in Lemma \ref{proto-gep} strictly speaking does not make use of the fact that all $|c_i^{r_0}|$ are small for $\alpha_i \in I$ since only $m$ weights are set equal to zero in the construction of the invariant subspace. A different choice of which weights to set equal to zero based on the particular problem at hand might be able to improve this estimate when the values of $|c_i^r|$ vary rapidly in $i$. However, if each $S_r$ is almost normal then we expect such variation to be controlled by the self-commutator of $S$. 

We now discuss the term $G_I$. This term is a consequence of the application of the gradual exchange lemma to consecutive weighted shift operators $S_t, S_{t-1}$. 
Based on the details of this construction, the term $G_I$ can be changed by reordering the weighted shift operators $S_{r_1}$, $S_{r_2}$ in the direct sum given that $A_{r_1}=A_{r_2}$. In our application to Ogata's theorem in the next chapter, the weights $c_i^r$ will be increasing in $r$ so the natural ordering based on the spin of the representations is optimal.

The only contribution to $G_I$ that depends explicitly on $A$ is the appearance of the $N_I$ in the term corresponding to $R_I$. In applications, we will choose the points $a_i$ first so that then $N_I$ is chosen to be as large as possible. 
There is a trade-off between how small the spacing of the $a_i$ can be and how large $N_I$ can be. The spacing of the $a_i$ may directly affect all the terms $\varepsilon_I, R_I, G_I, D_I$ while the size of $N_I$ only directly affects $R_I$.

Because we assume that $[A,S]$ is small, we know that \[|\alpha_{i+1}-\alpha_i||c_i^r| \leq \|[A,S]\|\]
is small. Assuming that the norm of $S$ on $E_{I}(A)$ is of order $1$, we know that $\max_{\alpha_i \in I} |c_i^r|$ is bounded and so $|\alpha_{i+1}-\alpha_i|$ is at most a constant multiple of $\|[A,S]\|$. 
So, we choose the $a_i$ so that $\diam I_k$ is much larger than the spacing of the eigenvalues of $A$ and hence $N_I$ is large. Exactly how large $N_I$ will be will depend on the situation, but we will want balance the size of the various components of the estimate to obtain the optimal result.

We now discuss the term $T_I$. The norm of the self-commutator of $S$, $\|[S^\ast, S]\|$, reflects the sizes of the differences of the squares of the absolute values of the weights of $S$ along individual orbits.  When applying the gradual exchange lemma, we then need to take into account that the weights of $S_{t}, S_{t-1}$ are blended together. The term $T_I$ reflects the size of the differences of the squares of the absolute values of the weights of $S$ between the consecutive orbits of $S_{t}, S_{t-1}$, reduced by the factor $N_I^{-1}$ due to how many vectors we have to smooth out the weights over. 
So, we expect that if the weights of the weights shifts $S_r$ do not vary much in $r$ then $T_I$ should not be too large. 
\end{remark}
\begin{remark}\label{refine}
As discussed previously, given any collection of $A_r, S_r$, we can refine the direct sum over all $r$ by partitioning the set of possible values of $r$ then apply this lemma to each partition of direct summands separately. 

An example of why one might want to do this is that it is easily possible that $m$ is comparable to  (or even larger than) $\# \sigma(A)$. In this case, $N_I$ cannot be large so the estimate of $R_I$ is not small. 
Conversely, making the refinements too sparse  conversely may increase the size of $\varepsilon_I$ and $T_I$.

For instance, take any non-trivial example of $A, S$ and repeatedly form direct sums with themselves. Having repeated summands only makes the estimate for $\|S'-S\|$ worse. This is because none of the estimates from the lemma change if the repeated summands are listed together in the lemma except that $N_I$ necessarily must decrease due to the increase of $m$.

This sort of difficulty is relevant for our application to Ogata's theorem.
In fact, it is on its face impossible to use this result without refinement for Ogata's theorem as in the next chapter due to the $N$-fold tensor product of $S^{1/2}$ being decomposed into many more than $N$ subrepresentations.
Our approach in the next chapter will be to refine the direct sum to then apply this lemma.
We also  obtain optimal results using the only freedom we have in this construction: the partition chosen and the windows $I_k$.
\end{remark}

\chapter{Main Theorem}
\label{8.MainTheorem}

We assume that $\lam_1 \leq  \dots \leq \lam_m$.
In Lemma \ref{Snearby} we will obtain nearby commuting self-adjoint matrices $A_i'$ for $A_i = \frac1N S^{\lam_1}\oplus\cdots\oplus S^{\lam_m}(\sigma_i)$. 

Let $A_r = \diag(i/N)$ for $-\lam_r \leq i \leq \lam_r$ and $S_r = \ws(d_{\lam_r, i}/N)$ for $-\lam_r \leq i < \lam_r$.  Then for $A = \bigoplus_r A_r$ and $S = \bigoplus_r S_r$, we have that $A_1 = \Re(S)$, $A_2 = \Im(S)$, and $A_3 = A$. The proof of Lemma \ref{Snearby} relies upon using the estimates in Lemma \ref{d-ineq} for the construction from Lemma \ref{gep}. We later optimize the result by choosing the lengths of the intervals $I_j$ optimally.

Dividing by $N$ here is referred to ``normalizing'' these operators.
For the moment we will focus only on the unnormalized weights $d_{\lam_r, i}$ and unnormalized spectrum. 
We assume that both $\lam_1$ and the maximum gap between the $\lam_r$ are not too small but also not too large. 
See Illustration \ref{smallremoved}.
\begin{figure}[htp]  
    \centering
    \includegraphics[width=8cm]{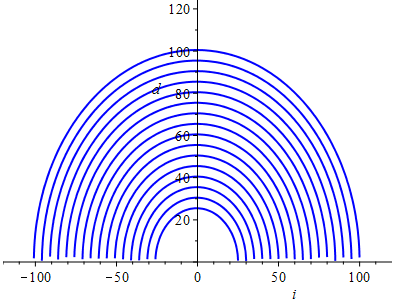}
    \caption{\label{smallremoved}\dark
    Illustration of the weights $d_{\lam_r, i}$ for  $\lam_r=25, 30, 35, \dots, 100$.}
\end{figure}
For this discussion, and hence the proof of Ogata's theorem, the estimates obtained in Lemma \ref{d-ineq} for $d_{\lam, i}$ are central to the calculation of the estimates for the nearby commuting matrices and influence the use of words such as ``small'' and ``large''. 

When calculating the estimate for $D_{I}$, 
\begin{figure}[htp]  
    \centering
    \includegraphics[width=8cm]{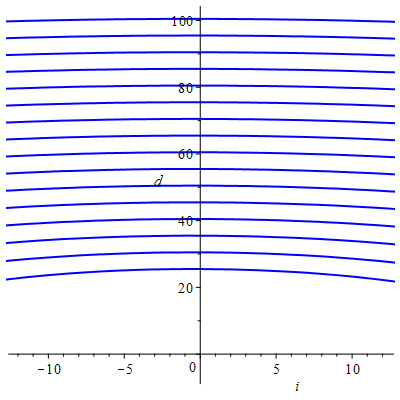}
    \caption{\label{near0}\dark
    Illustration of Illustration \ref{smallremoved} focused on a small unnormalized interval $I = [-12,12]$ near $0$.}
\end{figure}
one is concerned with the largest value of the weight of the representation $S^{\lam_{r_0}}(\sigma_+)$ in the interval $I$, where $r=r_0$ is the smallest index so that the spectrum of $S^{\lam_{r}}(\sigma_3)$ spans the interval $I$.
See Illustration \ref{near0} for an interval near $0$. In this example, $r_0 = 1$ and $D_I$ corresponds to the largest (unnormalized) weight of $S^{\lam_1}$, which is about $25$. 

In the proof of our extension of Ogata's theorem later in this paper, representations $S^\lam$ with small values of $\lam$ need to be dealt with separately due to the distribution of the multiplicities of the irreducible subrepresentations of the tensor representation. The reason that $\lam_{r+1}-\lam_r$ cannot be made very small and hence reduce the size of the $\varepsilon_I$ contribution to $G_I$ is also that it requires $m$ to be very large. 

As another example, consider the interval illustrated in Illustration \ref{nearmiddle} that is not near $0$ or the boundary of the spectrum of $S^{\lam_m}$.
\begin{figure}[htp]  
    \centering
    \includegraphics[width=8cm]{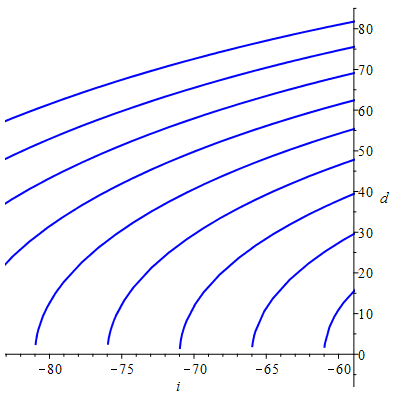}
    \caption{\label{nearmiddle}\dark
    Illustration of Illustration \ref{smallremoved} focused on a small unnormalized interval $I=[-82, -60]$ away from $0$.}
\end{figure} 
In this case, $S^{\lam_{r_0}}$ corresponds to the arc passing the vertical axis a little more than $60$. For each $r < r_0$, the spectrum of $S^{\lam_r}(\sigma_3)$ does not span the interval and for each $r \geq r_0$ the spectrum does span the interval.

Because the gradual exchange process will be applied for all $r \geq r_0$, the estimate for $D_I$ will involve the largest weight of $S^{\lam_0}(\sigma_+)$, which is slightly larger than $60$. 
For an interval in this position, it is important that the length of the interval not be too large since although the smallest weight of $S^{\lam_{r_0}}(\sigma_+)$ may be small, its largest weight may be large based on the growth of the weights within an orbit. 
The length of the interval and the spacing of the $\lam_r$ give an inequality of the form $|\lam_{r_0} - |i|| \leq M$ so that $D_I$ is controlled by Lemma \ref{d-ineq}$(ii)$.

In this illustration, the smallest weight of $S^{\lam_{r_0}}(\sigma_+)$ is about 25 and if the interval were extended to the right, the largest weight of $S^{\lam_{r_0}}(\sigma_+)$ would grow. If the interval were only extended to the left, then at some point $r_0$ would necessarily increase by multiples of $5$ which then increases the largest weight of $S^{\lam_{r_0}}(\sigma_+)$ to about $70$ and so on.
So, we see that the length of the interval $I$ cannot be too large. Alternatively, the length of $I$ cannot be too small since then the spectrum of the $S^{\lam_{r}}(\sigma_3)$ in that interval will be small. So, the $R_I$ contribution to $G_I$ will be large through $N_I$ being small. These estimates get larger the farther this interval is from $0$.

\section{Technical Lemmas}

We now proceed to constructing nearby commuting matrices with various parameters in the estimates.
\begin{lemma} \label{Snearby}
Let $S = \frac1NS^{\lam_1}\oplus \cdots \oplus S^{\lam_m}$ where $S^\lam$ is the irreducible $(2\lam+1)$-dimensional spin representation of $su(2)$ with
$0\leq \lam_{r+1}-\lam_r \leq L, \lam_m= \Lam$ and the $2\lam_r$ are all even or all odd. 
Let $l, \Delta > 0$ with $4 \leq N\Delta \leq 2\Lam$.

Then there are commuting self-adjoint matrices $A_i'$ such that
\begin{align}\|A_1'-S(\sigma_1)\|, \|A_2'-S(\sigma_2)\| &\leq \max(G,D)+C_\alpha\left(\frac{\Lam+1/2}{N}\right)^{1-2\alpha}\max(T^\alpha,D^{2\alpha}), \nonumber\\
\|A_3'-S(\sigma_3)\| &\leq c_\Delta, \end{align}
where
\begin{align}
c_\Delta &= \frac{2\Lam}{N\lfloor 2\Lam / N\Delta \rfloor} \leq \frac{2\Lam}{2\Lam \Delta^{-1}-N}\\
N_0 &= \left\lfloor \frac{N\Delta-5}{2m-3} \right\rfloor-1\geq \frac{N\Delta-5}{2(m-1)-1} -2 \\
T&=\left(2+\frac{2L}{N_0}\right)\frac{\Lam}{N^2}
\\
G=\frac1N\max\left( \sqrt{\Lam}\frac{2L}{\sqrt{l}}+ \frac{\pi}{2N_0}\right.&\left.\left(\Lam+1/2\right), \sqrt{2\Lam L}+\frac{\pi}{2N_0}\sqrt{2\Lam(l+1)}\right)
\end{align}
\begin{align}
D= \max\left(\sqrt{\frac{2\Lam}N\left(\frac{L+1}{N}+c_\Delta\right)},\frac{\lam_{1}+1/2}N
, \frac{c_{\Delta}}2+\frac{L+1/2}N\right),
\end{align}
$\alpha \in(0, 1/2], C_\alpha > 0$ are constants as in Theorem \ref{BergResult}, and $A_3'$ is real. Consequently, when using $\alpha = 1/3, C_{1/3} = 5.3308$, we have that $A_1', iA_2', A_3'$ are real.
\end{lemma}
\begin{proof}
We wish to apply Lemma \ref{gep} with 
\[
A_r  = \frac1NS^{\lam_r}(\sigma_3) = \diag\left(-\frac{\lam_r}N, \frac{-\lam_r+1}N, \dots, \frac{\lam_r}N\right)\]
\[S_r = \frac1NS^{\lam_r}(\sigma_+) = \ws\left(\frac{d_{\lam_r, -\lam_r}}N, \frac{d_{\lam_r, -\lam_r+1}}N, \dots, \frac{d_{\lam_r, \lam_r-1}}N\right)\]
so that
\[A_r = \diag\left(\frac{i}{N}\right), \; i = -\lam_r, -\lam_r+1, \dots, \lam_r\]
\[S_r = \ws\left(\frac{d_{\lam_r,i}}{N}\right), \; i = -\lam_r, -\lam_r+1, \dots, \lam_{r}-1.\]
Set $A = \bigoplus_r A_r$ and $S = \bigoplus_r S_r$ and $\alpha_i = i/N$, $c_i^r = d_{\lam_r,i}/N\geq0$ in accordance with the assumptions of Lemma \ref{gep}.
So, the estimates of $c_i^r$ and $c_i^{r+1}-c_i^r$ needed to apply Lemma \ref{gep} will be obtained from the inequalities for $d_{\lam_r, i}$ and $d_{\lam_{r+1},i}-d_{\lam_r, i}$ in Lemma \ref{d-ineq}.
We will then obtain nearby commuting $A', S''$ such that $A'$ is Hermitian and $S''$ is normal. We then set $A_1' = \Re(S''), A_2' = \Im(S''),$ and $A_3' = A'$.
 
We choose an increasing sequence of real numbers $a_i$ to satisfy the conditions of Lemma \ref{gep} with $a_{1} = -\Lam/N$ and $a_{n_0}=\Lam/N$ satisfying
\begin{align*}
a_{k+1}-a_k= 
c_\Delta,
\end{align*}
where 
\begin{align}
n_\Delta = \left\lfloor \frac{2\Lam/N}{\Delta} \right\rfloor, \, c_\Delta = \frac{2\Lam/N}{n
_\Delta}\geq \Delta,\nonumber
\end{align}
requiring $2\Lam/N \geq \Delta$ so $2\Lam \geq N\Delta$.
So, the intervals $I_k$ have the same length, which is at least $\Delta$ and is asymptotically equal to $\Delta$ as $N\Delta/\Lam \to 0$.
Note that
\begin{align}\label{num-size}
N\Delta -1\leq \#\sigma(A_r) \cap [a_k, a_{k+1})
\end{align}
and we require that $N\Delta -1 \geq 3$ so $N\Delta \geq 4$.

We now move to calculating the various estimates in Lemma \ref{gep}. 

\vspace{0.1in}

\noindent \underline{Estimating $D_{I_k}$}: There are two types of intervals $I=I_k$. If $n_\Delta$ is odd, then $I_{(n_\Delta+1)/2} = [-c_\Delta/2, c_\Delta/2]$. All other intervals are of the form  $[-b,-b+c_\Delta]$ or $[b-c_\Delta,b]$ for $b \geq c_\Delta$.

We first deal with the exceptional case. Recall that $\sigma(A_r)$ consists of $-\lam_r/N, \dots, \lam_r/N$. So, the sets $\sigma(A_r)$ are nested consecutive and symmetric intervals in $\frac1N\Z$. 
Recall that $r_0 = \min \mathscr R_{I}$ is the smallest $r$ so that $\sigma(A_r)$ contains $\sigma(A) \cap I$. We then bound
\[D_I\leq \max_i c_i^{r_0} \leq \frac{\lam_{r_0}+1/2}N\]
by Lemma \ref{d-ineq}$(i)$.
If $r_0 = 1$, then we obtain
\[D_I \leq \frac{\lam_{1}+1/2}N.\]
So, suppose that $r_0 > 1$. Because 
\[\frac{\lam_{r_0-1}}N< c_{\Delta}/2 \leq \frac{\lam_{r_0}}N\] and $\lam_{r_0} \leq \lam_{r_0-1}+L$, we see that $\lam_{r_0}\leq Nc_{\Delta}/2+L$. So,
\[D_I \leq \frac{Nc_{\Delta}/2+L+1/2}N=\frac{c_{\Delta}}2+\frac{L+1/2}N.\]

So, suppose that $I$ is not the central interval of the previous case. 
If $r_0 = 1$ we apply the same bound as before. So, suppose that $r_0 > 1$.
If $I = [-b, -b+c_\Delta]$ or $I = [b-c_\Delta,b]$ then \[\frac{\lam_{r_0-1}}N < b \leq \frac{\lam_{r_0}}N.\] Because $\lam_{r_0}\leq \lam_{r_0-1}+ L$, we obtain 
\[\lam_{r_0}-N|x| \leq L+Nc_\Delta, \; x \in I.\]
So, suppose $x=|i|/N \in I$ so that $i \in [-\lam_{r_0}, \lam_{r_0}]$.
Using $M= L+Nc_\Delta$ in  Lemma \ref{d-ineq}$(ii)$, we have
\[d_{\lam_{r_0}, i}\leq \sqrt{2\lam_{r_0}(M+1)} \leq \sqrt{2\Lam(L+Nc_\Delta+1)}\]
and hence
\[c_{i}^{r_0} \leq \frac1N\sqrt{2\Lam(L+Nc_\Delta+1)}.\]

Therefore, we obtain the bound from the statement of the lemma: $D_I \leq D$.

\vspace{0.1in}

\noindent \underline{Estimating $G_I$}:
Note that in order to apply Lemma \ref{gep}, we need $(2m-3)(N_0+1)+4 \leq \#\sigma(A)\cap I$, where we choose $N_k = N_0$ for all $k$. The definition of 
$N_0$ in the statement of the lemma was made to satisfy this inequality through Equation (\ref{num-size}).

Estimating $G_I$ involves estimating the sum of  $c^{r+1}_i-c^r_i$ and $\frac{\pi}{2N_0}\max(c^{r+1}_i,c^r_i)$. Using Lemma \ref{d-ineq}$(iv)$ and $\lam_{r+1} \leq \Lam$, we obtain the bound
\begin{align}\nonumber
|c^{r+1}_i-c^r_i| &+ \frac{\pi}{2N_0}\max(c^{r+1}_i,c^r_i)\leq G.
\end{align}

\vspace{0.1in}

\noindent \underline{Estimating Equation (\ref{S''ineq})}:
By Lemma \ref{d-ineq}$(vi)$, \[\|[S^\ast, S]\| \leq \frac{2\Lam}{N^2}.\] By Lemma \ref{d-ineq}$(v)$, for all the weights  \[|(c^{r+1}_i)^2-(c^r_i)^2| \leq \frac{2\Lam L}{N^2}.\]
By Lemma \ref{d-ineq}$(i)$, $\|S\| \leq (\Lam+1/2)/N$. Note that we require $\alpha \leq 1/2$ so that $1-2\alpha \geq 0$.
The desired estimate then follows from the estimates of $\|S'-S\|$ and $\|S''-S'\|$ from Lemma \ref{gep}.

When using $\alpha = 1/3, C_{1/3} = 5.3308$, we have $A'$ and $S''$ real so $\Re(S'')$  and $i\Im(S'')$ are as well. We now collect what we showed into the statement of the lemma.  
\end{proof}

\begin{example}\label{Ex1}
We assume that the constants in the statement of Lemma \ref{Snearby} satisfy the asymptotic estimates
\begin{align}\label{asymptBounds}
\lam_1\leq c_0N^{\gamma_0}, m-1 \leq c_1 N^{\gamma_1}, \underline{c_2} N^{\underline{\gamma_2}}\leq \lam_m \leq c_2 N^{\gamma_2}, \nonumber\\
L \leq c_3 N^{\gamma_3}, l = c_4 N^{\gamma_4}, \Delta = c_5 N^{-\gamma_5}.
\end{align}
We assume $N \geq N_\ast \geq 1$. Note that $N$ will be an integer, though $N_\ast$ is not assumed to be. Although we will prove more in this discussion, what we will use from it for Ogata's theorem is expressed in Lemma \ref{Ex2Lemma}.

We now explore some mild assumptions on the exponents to obtain nearby commuting matrices using Lemma \ref{Snearby}.
First, $\gamma_0,  \gamma_1, \gamma_2, \gamma_3, \gamma_5 > 0$.  
Because $\lam_1 \leq \lam_m$, we expect $\gamma_0 \leq \gamma_2$. 
Because $\lam_m-\lam_1 \leq (m-1)L$ and often $\lam_1 = o(\lam_m)$, we will often have $ \gamma_2 \leq \gamma_1 + \gamma_3$. For reasons explained below, we expect $\gamma_1 \leq \gamma_2$ as well. We will assume that $\gamma_1+\gamma_5\leq1$ so that $N_0$ can be large. To make the term coming from $\|S\|$ bounded by a constant, we will assume that $\gamma_2\leq 1$.

The constants $l$ and $\Delta$ are chosen, while the others are given. 
In particular, $l$ will be chosen so that the first and fourth term in the estimate of $G$ are equalized and negligible. 
Because the optimal value of $l$ is not a simple expression, we elect to choose $l$ after the estimate for $G$ is expressed in terms of the $c_i$, $\gamma_i$, and $N_\ast$. 

Choosing the optimal constant and exponent for $\Delta$ in this generality requires knowing more information about the relative sizes of the exponents in the definitions of $G$, $D$ and $T$. 
We make further assumptions about the exponents after having done as much simplification as possible.
The necessary condition $4 \leq N\Delta \leq 2\Lam$ becomes
\[4 \leq c_5N^{1-\gamma_5} \leq 2\underline{c_2}N^{\underline{\gamma_2}}\]
\[
4 \leq c_5N^{1-\gamma_5}, \;\; c_5 \leq 2\underline{c_2}N^{\underline{\gamma_2}+\gamma_5-1}.
\]
So, we further assume that $\gamma_5\leq 1$ and $\underline{\gamma_2}+\gamma_5\geq1$.

We first find the optimal exponent for $\max(G,D)+C_{\alpha, \Lam, N}\max(T^\alpha,D^{2\alpha})$. Note with $\alpha \leq 1/2$, we will use $\Lam \leq Const. N$ so that $C_{\alpha, \Lam, N}$ is bounded by a constant.
It should be noted that we will not consider the asymptotics of $c_\Delta$ for the matrix $A_3'$ during the optimization of the exponent because $D > c_\Delta/2$.

After finding the optimal exponent, we then bound all the terms by a constant factor multiplied by a single power of $N$. In particular, for $N \geq N_\ast$, all terms that are negligible will contribute to the constant factor in a way that depends on $N_\ast$ as follows. 
The primary inequality that will be used to choose optimal constant factors will be repeated applications of the following simple observation that if $a \geq b, N \geq N_\ast$ then
\[N^b = N^{b-a}N^a \leq N_\ast^{b-a}N^a\]
In particular, if $a \geq 0$ then \[1 \leq N_\ast^{-a}N^a.\]

\vspace{0.05in}

\noindent We now proceed to the calculations.\\ \underline{$c_\Delta$}:
Because $x\mapsto x/(ax-b)$ is decreasing as a function of $x > b/a$, we have
\begin{align}
c_\Delta &\leq \frac{2\lam_m}{2\lam_m\Delta^{-1}-N     }\leq \frac{2\underline{c_2}N^{\underline{\gamma_2}}}{\frac{2\underline{c_2}}{c_5}N^{\underline{\gamma_2}+\gamma_5}-N} \leq \frac{2\underline{c_2}N^{\underline{\gamma_2}}}{\frac{2\underline{c_2}}{c_5}N^{\underline{\gamma_2}+\gamma_5}-N_\ast^{1-\underline{\gamma_2}-\gamma_5}N^{\underline{\gamma_2}+\gamma_5}} \nonumber\\
&= c_5\left(\frac{2\underline{c_2}}{2\underline{c_2}-c_5N_\ast^{1-\underline{\gamma_2}-\gamma_5}}\right)N^{-\gamma_5} = d_\Delta N^{-\gamma_5},\label{dDelta}
\end{align}
where we assume that $d_\Delta > 0$ (or equivalently $c_5 < 2\underline{c_2}N_\ast^{\underline{\gamma_2}+\gamma_5-1}$).
Note that the upper bound for $c_\Delta$ through that of $d_\Delta$ is the only place in our calculations where we use the lower bound for $\lam_m$. This guarantees that $\lam_m$ is much larger than $N\Delta$ so that $c_\Delta$ is approximately equal to $\Delta = c_5N^{-\gamma_5}$.

\vspace{0.05in}

\noindent \underline{$N_0$}:
\begin{align}
N_0 &\geq  \frac{c_5N^{-\gamma_5+1}-5}{2c_1N^{\gamma_1}-1} -2 \geq \frac{c_5N^{-\gamma_5+1}-5N_\ast^{\gamma_5-1}N^{-\gamma_5+1}}{2c_1N^{\gamma_1}} -2N_\ast^{\gamma_1+\gamma_5-1}N^{-\gamma_1-\gamma_5+1}\nonumber\\
&= \left(\frac{c_5-5N_\ast^{\gamma_5-1}}{2c_1} -2N_\ast^{\gamma_1+\gamma_5-1}\right)N^{-\gamma_1-\gamma_5+1} = d_0 N^{-\gamma_1-\gamma_5+1},\label{d0}
\end{align}
where we used the assumption that $\gamma_1+\gamma_5 \leq 1$. We further assume that $d_0>0$ and $2c_1N^{\gamma_1}>1$.

\vspace{0.05in}

\noindent \underline{$T$}:
\begin{align*}
T&\leq\left(2+\frac{2c_3 N^{\gamma_3}}{d_0 N^{-\gamma_1-\gamma_5+1}}\right)\frac{c_2 N^{\gamma_2}}{N^2} = 2c_2N^{\gamma_2-2}+\frac{2c_2c_3}{d_0 }N^{\gamma_1+\gamma_2+\gamma_3+\gamma_5-3} \end{align*}

\vspace{0.05in}

\noindent \underline{$G$}:
\begin{align*}
G&\leq\frac1N\max\left( \sqrt{c_2N^{\gamma_2}}\frac{2c_3N^{\gamma_3}}{\sqrt{c_4}}N^{-\gamma_4/2}+ \frac{\pi}{2d_0 N^{-\gamma_1-\gamma_5+1}}\left(c_2N^{\gamma_2}+\frac12\right
),\right.\\
&\;\;\;\;\;\;\;\;\;\;\;\;\;\;\;\;\;\;\left.\sqrt{2c_2N^{\gamma_2}(c_3N^{\gamma_3})}+\frac{\pi}{2d_0 N^{-\gamma_1-\gamma_5+1}}\sqrt{2c_2N^{\gamma_2}(c_4N^{\gamma_4}+1)}\right)\\
&\leq \frac1N\max\left( 2c_3\sqrt{\frac{c_2}{c_4}}N^{\gamma_3+(\gamma_2-\gamma_4)/2}+ \frac{\pi}{2d_0 }N^{\gamma_1+\gamma_5-1}\left(c_2N^{\gamma_2}+\frac12N_\ast^{-\gamma_2}N^{\gamma_2}\right
),\right.\\
&\;\;\;\;\;\;\;\;\;\;\;\;\;\;\;\;\;\;\left.\sqrt{2c_2c_3}N^{(\gamma_2+\gamma_3)/2}+\frac{\pi}{2d_0} N^{\gamma_1+\gamma_5-1}\sqrt{2c_2c_4N^{\gamma_2+\gamma_4}+2c_2N_\ast^{-\gamma_4}N^{\gamma_2+\gamma_4}}\right)\\
&= \max\left( 2c_3\sqrt{\frac{c_2}{c_4}}N^{\gamma_3+(\gamma_2-\gamma_4)/2-1}+ \frac{\pi}{2d_0 }\left(c_2+\frac12N_\ast^{-\gamma_2}\right)N^{\gamma_1+\gamma_2+\gamma_5-2},\right.\\
&\;\;\;\;\;\;\;\;\;\;\;\;\;\;\;\;\;\;\left.\sqrt{2c_2c_3}N^{(\gamma_2+\gamma_3)/2-1}+\frac{\pi}{2d_0} \sqrt{2c_2c_4+2c_2N_\ast^{-\gamma_4}}N^{\gamma_1+\gamma_5+(\gamma_2+\gamma_4)/2-2}\right)
\end{align*}

With the choice of $\gamma_4 = -\gamma_1+\gamma_3-\gamma_5+1$, we equalize the exponents in the first and fourth terms, obtaining
\begin{align*}
G&\leq \max\left( 2c_3\sqrt{\frac{c_2}{c_4}}N^{(\gamma_1+\gamma_2+\gamma_3+\gamma_5-3)/2}+ \frac{\pi}{2d_0 }\left(c_2+\frac12N_\ast^{-\gamma_2}\right)N^{\gamma_1+\gamma_2+\gamma_5-2},\right.\\
&\;\;\;\;\;\;\;\;\;\;\;\;\;\;\;\;\;\;\left.\sqrt{2c_2c_3}N^{(\gamma_2+\gamma_3)/2-1}+\frac{\pi}{2d_0} \sqrt{2c_2c_4+2c_2N_\ast^{-\gamma_4}}N^{(\gamma_1+\gamma_2+\gamma_3+\gamma_5-3)/2}\right).
\end{align*}
Note that the first and fourth terms are not asymptotically larger than the third term because $\gamma_1+\gamma_5 \leq 1$. Later we will have a strict inequality so that these two terms become negligible as $N \to \infty$.

\vspace{0.05in}

\noindent \underline{$D$}:
\begin{align*}
D\leq \max&\left(\sqrt{\frac{2c_2N^{\gamma_2}}N\left(\frac{c_3 N^{\gamma_3}+1}{N}+d_\Delta N^{-\gamma_5}\right)},\frac1N\left(c_0N^{\gamma_0}+\frac12\right),\right. \\ &\;\;\;\;\;\;\;\;\;\;\;\;\;\left. \frac{d_\Delta }2N^{-\gamma_5}+\frac{c_3N^{\gamma_3}+1/2}N\right)\\
\leq \max&\left(\sqrt{2c_2c_3N^{\gamma_2+\gamma_3-2}+2c_2N^{\gamma_2-2}+2c_2d_\Delta N^{\gamma_2-\gamma_5-1}}, c_0N^{\gamma_0-1}+\frac12N^{-1},\right. \\ &\;\;\;\;\;\;\;\;\left. \frac{d_\Delta }2N^{-\gamma_5}+c_3N^{\gamma_3-1}+\frac12N^{-1}\right)
.
\end{align*}
Note that the first term in the bound for $D$ has three components, the first of which is asymptotically equal to the third term of $G$, considering the square root. 

\vspace{0.05in}

\noindent \underline{Optimal Asymptotics}:\\
Recall that $\alpha\leq 2\alpha \leq 1$. So, the slowest decaying term of $\max(G,D)+C_{\alpha, \Lam, N}\max(T^\alpha,D^{2\alpha})$ has exponent
\begin{align*}
-\gamma &= \max\left( \gamma_1+\gamma_2+\gamma_5-2, 2\alpha\left(\frac{\gamma_2+\gamma_3}{2}-1\right), 2\alpha\left(\frac{\gamma_2-\gamma_5-1}{2}\right), 2\alpha(\gamma_0-1),\right. \\
&\;\;\;\;\;\;\;\;\;\;\;\;\;\;\;\;\; \left.  -2\alpha\gamma_5,2\alpha(\gamma_3-1), \alpha(\gamma_2-2), \alpha(\gamma_1+\gamma_2+\gamma_3+\gamma_5-3)^{\tcw{|}}\right).
\end{align*}
So, $-\gamma$ is the largest of several exponents that, minimally, we wish to choose to be negative. We will then minimize  $-\gamma$. Note that its optimal value will depend on $\alpha$ as well as the appropriate choice of the $\gamma_i$.

We now impose additional assumptions on the exponents $\gamma_i$.
We further assume that we have $\gamma_2 = \gamma_1+\gamma_3$. So, we assume that $\gamma_3 \leq \gamma_2$. This corresponds to having a bound for the spacing $\lam_{r+1}-\lam_r$ that is asymptotically equal to the bound of the average spacing $(\lam_m - \lam_1)/(m-1)$ if additionally $m-1 \geq Const. N^{\gamma_1}$.

Substituting $\gamma_1 = \gamma_2-\gamma_3$, we obtain
\begin{align}\nonumber
-\gamma &= \max\left( 2\gamma_2-\gamma_3+\gamma_5-2, \alpha\left(\gamma_2+\gamma_3-2\right), \alpha\left(\gamma_2-\gamma_5-1\right), 2\alpha(\gamma_0-1),  -2\alpha\gamma_5,\right. \\
&\;\;\;\;\;\;\;\;\;\;\;\;\;\;\;\;\; \left.2\alpha(\gamma_3-1), \alpha(\gamma_2-2), \alpha(2\gamma_2+\gamma_5-3)\right).\label{gamma}
\end{align}
Note that the requirement $\gamma_1+\gamma_5\leq1$ becomes $\gamma_2-\gamma_3+\gamma_5\leq1$.

We now bound our estimates for $G, D, D^{2\alpha}, T^{\alpha}$ by a constant multiple of $N^{-\gamma}$. Note that by definition, if $a$ is an exponent such that $a \leq -\gamma$ then
\[N^a = N^{a+\gamma}N^{-\gamma} \leq N_0^{a+\gamma}N^{-\gamma}\]
since $a+\gamma \leq 0$.

So,
\begin{align}\label{Gestimate}
G&\leq \max\left( 2c_3\sqrt{\frac{c_2}{c_4}}N_\ast^{\frac{2\gamma_2+\gamma_5-3}2+\gamma}+ \frac{\pi}{2d_0 }\left(c_2+\frac12N_\ast^{-\gamma_2}\right)N_\ast^{2\gamma_2-\gamma_3+\gamma_5-2+\gamma},\right.\\
&\;\;\;\;\;\;\;\;\;\;\;\;\;\;\;\;\;\;\left.\sqrt{2c_2c_3}N_\ast^{\frac{\gamma_2+\gamma_3}2-1+\gamma}+\frac{\pi}{2d_0} \sqrt{2c_2c_4+2c_2N_\ast^{\gamma_2-2\gamma_3+\gamma_5-1}}N_\ast^{\frac{2\gamma_2+\gamma_5-3}2+\gamma}\right)N^{-\gamma},\nonumber
\end{align}
\begin{align}\label{Destimate}\nonumber
D\leq \max&\left(\sqrt{2c_2c_3N_\ast^{\gamma_2+\gamma_3-2+2\gamma}+2c_2N_\ast^{\gamma_2-2+2\gamma}+2c_2d_\Delta N_\ast^{\gamma_2-\gamma_5-1+2\gamma}}\right.,
\\ &\;\;\;\;\;\;\;\; c_0N_\ast^{\gamma_0-1+\gamma}+\frac12N_\ast^{-1+\gamma}, \\ &\;\;\;\;\;\;\;\;\left. \frac{d_\Delta }2N_\ast^{-\gamma_5+\gamma}+c_3N_\ast^{\gamma_3-1+\gamma}+\frac12N_\ast^{-1+\gamma}\right)N^{-\gamma},\nonumber
\end{align}
\begin{align}\label{Testimate}
T^\alpha&\leq \left(2c_2N_\ast^{\gamma_2-2+\frac\gamma\alpha}+\frac{2c_2c_3}{d_0 }N_\ast^{2\gamma_2+\gamma_5-3+\frac\gamma\alpha}\right)^\alpha N^{-\gamma}, \end{align}
\begin{align}\label{Dalphaestimate}\nonumber
D^{2\alpha}\leq \max&\left(\sqrt{2c_2c_3N_\ast^{\gamma_2+\gamma_3-2+\frac\gamma{\alpha}}+2c_2N_\ast^{\gamma_2-2+\frac\gamma{\alpha}}+2c_2d_\Delta N_\ast^{\gamma_2-\gamma_5-1+\frac\gamma{\alpha}}}\right.,
\\ &\;\;\;\;\;\;\;\; c_0N_\ast^{\gamma_0-1+\frac\gamma{2\alpha}}+\frac12N_\ast^{-1+\frac\gamma{2\alpha}}, \\ &\;\;\;\;\;\;\;\;\left. \frac{d_\Delta }2N_\ast^{-\gamma_5+\frac\gamma{2\alpha}}+c_3N_\ast^{\gamma_3-1+\frac\gamma{2\alpha}}+\frac12N_\ast^{-1+\frac\gamma{2\alpha}}\right)^{2\alpha}N^{-\gamma}.\nonumber
\end{align}

\end{example}

We write the result of the previous example as a lemma.
\begin{lemma}\label{Ex2Lemma}
Let $S = \frac1NS^{\lam_1}\oplus \cdots \oplus S^{\lam_m}$ where $S^\lam$ is the irreducible $(2\lam+1)$-dimensional spin representation of $su(2)$ with $0 \leq\lam_{r+1}-\lam_r \leq L$  and the $2\lam_r$ are all even or all odd. 

Suppose further that
\begin{align*}
\lam_1\leq c_0N^{\gamma_0}, m-1 \leq c_1 N^{\gamma_2-\gamma_3}, \underline{c_2} N^{\underline{\gamma_2}}\leq \lam_m \leq c_2 N^{\gamma_2}, \\
L \leq c_3 N^{\gamma_3}, l = c_4 N^{ -\gamma_2+2\gamma_3-\gamma_5+1}, \Delta = c_5 N^{-\gamma_5},
\end{align*}
where $\gamma_i, c_i, \underline{c_2} > 0, \gamma_0 < 1, \gamma_i, \underline{\gamma_2} \leq 1, \gamma_3 \leq \gamma_2$, $\underline{\gamma_2}+\gamma_5\geq1$, and $\gamma_2-\gamma_3+\gamma_5\leq 1$.
Suppose that the $c_i$ and $N_\ast$ satisfy the inequalities \begin{align}\label{ineq-requirements}
1 &< 2c_1N_{\ast}^{\gamma_2-\gamma_3}, \;\; 4c_1N_\ast^{\gamma_2-\gamma_3+\gamma_5-1}+5N_\ast^{\gamma_5-1} < c_5, \nonumber \\
4 &\leq c_5N_\ast^{1-\gamma_5}, \;\; c_5 < 2\underline{c_2}N_\ast^{\underline{\gamma_2}+\gamma_5-1}\nonumber
\end{align}

Let 
$C_{\alpha, \Lam, N} = C_\alpha\left(\frac{\lam_m+1/2}N\right)^{1-2\alpha}\leq Const.$, where $\alpha, C_\alpha$ are as in Theorem \ref{BergResult} with additionally $\alpha\leq 1/2$. Let $d_\Delta$ and $d_0$ be defined by Equations (\ref{dDelta}) and (\ref{d0}) and let $\gamma=\gamma(\alpha, \gamma_i)$ be defined by Equation (\ref{gamma}). 

Then we have the bounds for $G, D, T^\alpha, D^{2\alpha}$ from Lemma \ref{Snearby} of the form \\
$C(\alpha, c_i, \gamma_i, \underline{c_2}, \underline{\gamma_2}, N_\ast)N^{-\gamma}$ in Equations (\ref{Gestimate}), (\ref{Destimate}), (\ref{Testimate}), and (\ref{Dalphaestimate}) so that there are commuting self-adjoint matrices $A_i'$ such that
\begin{align}
\|A_1'-S(\sigma_1)\|, \|A_2'-S(\sigma_2)\| &\leq \max(G,D)+C_{\alpha, \Lam, N}\max(T^\alpha,D^{2\alpha}) \leq Const. N^{-\gamma}, \nonumber\\
\|A_3'-S(\sigma_3)\| &\leq d_\Delta N^{-\gamma_5}.\nonumber
\end{align}
Moreover, when using $\alpha = 1/3, C_{1/3} = 5.3308$, we have that $A_1', iA_2', A_3'$ are real.
\end{lemma}

\begin{example}
With the set-up of the previous example, suppose that we are interested in the optimal exponent and the constant obtained as $N_\ast\to \infty$ when $\alpha=1/3$.

For this example, we will assume that $c_1c_3 \geq c_2$. In the next lemma below, we treat the details of this constraint which approximately holds when $N_\ast$ is large, $\lam_1 = o(\lam_m)$, and $\lam_{r+1}-\lam_r$ is constant in $r$. Due to this assumption, we can easily remove the dependence of $c_1$ as follows:
The only 
occurrence of $c_1$ in our inequalities is in $G$ and $T$ through $d_0^{-1}$. 
We see that both $G$ and $T$ are decreased when $c_1$ is decreased, so we choose $c_1 = c_2/c_3$.

For this calculation, we assume that $\lam_m = N/2$ so that $\underline{c_2}=c_2 = 1/2, \underline{\gamma_2}=\gamma_2 = 1$. The condition $\gamma_2-\gamma_3+\gamma_5 \leq 1$ then becomes $\gamma_5 \leq \gamma_3$. We choose $\lam_1 = O(N^{1/2})$ by taking $\gamma_0 =1/2$.

For $\alpha = 1/3$, the optimal choices of $\gamma_3 = 4/7, \gamma_5 = 3/7$ give $\gamma = 1/7$. Then the exponents in Equation (\ref{gamma}) are
\[-\frac17, -\frac17, -\frac17, -\frac13, -\frac27, -\frac27, -\frac13, -\frac4{21}.\]
So, the slowest decaying terms have exponents $2\gamma_2 - \gamma_3 + \gamma_5 - 2, \alpha(\gamma_2 + \gamma_3 - 2), \alpha(\gamma_2 - \gamma_5 - 1)$ which equal $-1/7$.
We note that as $N_\ast \to \infty$, we obtain that $d_0 \sim c_5/(2c_1) = c_3c_5/2c_2, d_\Delta \sim c_5$.

So asymptotically,
\begin{align*}
G&\leq \frac{\pi c_2}{2d_0}N^{-1/7}+o(N^{-1/7})= \frac{\pi }{4c_3c_5}N^{-1/7}+o(N^{-1/7}),\\
T^{1/3}&=o(N^{-1/7}),\\
D \ll D^{2/3}&\leq (2c_2c_3+2c_2d_\Delta)^{1/3} N^{-1/7}+o(N^{-1/7}) = (c_3+c_5)^
{1/3} N^{-1/7}+o(N^{-1/7}).
\end{align*}
To approximately optimize our estimate of $\left(\frac{\pi }{4c_3c_5}+5.3308\left(\frac12\right)^{1/3}(c_3+c_5)^
{1/3}\right)N^{-1/7}$, we choose $c_3, c_5 = 0.95$. So, for $N$ large, there are nearby commuting matrices $A_i'$ satisfying the following inequalities 
\begin{align}\nonumber
\|A_1' - S(\sigma_1)\|, \|A_2' - S(\sigma_2)\|&\leq 6.111\, N^{-\frac17}\\
\|A_3' - S(\sigma_3)\|&\leq 0.951\, N^{-\frac37}\nonumber
\end{align}

This estimate  shows that we might as well assume that $N_\ast$ is at least $(2\cdot 6.111)^7 > 4.07\times 10^7$. This is because $\|S(\sigma_i)\|=\frac12$
so it is only when $N \geq (2\cdot 6.111)^7$ that the obtained estimate is better than trivially choosing $A_1' = A_2' = 0, A_3'=A_3$.  

\end{example}

We now prove the following lemma that is closer to what will be used for Ogata's theorem. This result is a modification of the previous example that holds for all $N$.
\begin{lemma}\label{bigLstepLemma}
Let $N \geq 1$, $\Lambda_0 \leq \frac12N^{1/2}+\frac32$, and $L = \lfloor 1.045\, N^{4/7} \rfloor$.
Let $S = \frac1NS^{\lam_1}\oplus \cdots \oplus S^{\lam_m}$ with  $\lam_1 \leq \Lambda_0+2L$, $\lam_m \leq N/2$, and $\lam_{r+1}-\lam_r = L$. 

Then there are commuting self-adjoint matrices $A_i'$ such that
\begin{align}\|A_1'-S(\sigma_1)\|, \|A_2'-S(\sigma_2)\| &\leq 6.286\,N^{-\frac17}, \nonumber\\
\|A_3'-S(\sigma_3)\| &\leq 1.083\,N^{-\frac37}\nonumber\end{align}
and $A_1', iA_2', A_3'$ are real.
\end{lemma}
\begin{proof}
Note that the variables $N_\ast, c_3,$ and $\underline{c_2}$ will be left undetermined until the end of the proof. We also at this point define $L = \lfloor c_3N^{4/7}\rfloor \leq c_3N^{\gamma_3}$ with $\gamma_3 = \frac47$. We will obtain estimates for three cases then choose the optimal values for these constants to obtain the result of the lemma.

\vspace{0.1in}

\noindent \underline{$\lam_m < \underline{c_2}N^{6/7}$}: \\
This case only relies the value of the variable $\underline{c_2}$. By Equation (\ref{lamnorm}) we have
\[ \|S(\sigma_i)\| < \frac{\underline{c_2}N^{6/7}}{N} = \underline{c_2}N^{-1/7}.\] So, we may safely choose $A_1' =  A_2' = 0$ and $A_3' = S(\sigma_3)$. 
The estimates in the statement of the lemma that we obtain are $\|A_3'-S(\sigma_3)\| = 0$ and for $i = 1, 2$,
\[\|A_i'-S(\sigma_i)\| = \|S(\sigma_i)\|  < \underline{c_2}N^{-1/7}.\]

\vspace{0.1in}

\noindent \underline{$N < N_\ast$}: \\
This case only relies the value of the variable $N_\ast$.

As in the previous case, we choose $A_1' = A_2' = 0$ and $A_3' = S(\sigma_3)$. Because $N^{1/7} < N_\ast^{1/7}$, we have
\[\|S(\sigma_i)\| \leq \frac{N/2}{N}< \frac12 N_\ast^{1/7}N^{-1/7}.\]

\noindent \underline{$N \geq N_\ast$, $\underline{c_2}N^{6/7} \leq \lam_m$}:
This is the only non-trivial case and it  relies on the values of $N_\ast, c_3,$ and $\underline{c_2}$. Due to our use of Lemma \ref{Ex2Lemma}, we will also have other constants 

We will apply Lemma \ref{Ex2Lemma} with exponents $\gamma_0 = \frac47, \underline{\gamma_2} = \frac67, \gamma_2 = 1, \gamma_3 = \frac47,  \gamma_5 = \frac37, \gamma=\frac17$ and with $c_2 = \frac12$. 

First note that 
\[\lam_0 \leq \Lambda_0+2L \leq 2c_3N^{\frac47}+\frac12N^{\frac12}+\frac32\leq \left(2c_3+\frac12N_\ast^{-\frac1{14}}+\frac32N_\ast^{-\frac47}\right)N^{\frac47}=c_0N^{\gamma_0}.\]
Also, because $\lam_{r+1}-\lam_r$ is constant, we see that 
\[m-1 = \frac{\lam_m - \lam_1}{L} \leq \frac{\lam_m}{c_3N^{\frac47}-1} \leq\frac{N}{2c_3N^{\frac47}-2N_\ast^{-\frac47}N^{\frac47}} =  \frac{1}{2c_3 - 2N_\ast^{-\frac47}}N^{\frac37} = c_1N^{\gamma_1}.\]
Observe that the exponent provided here is $\gamma_1 = \gamma_2-\gamma_3 = 1-\frac47 = \frac37.$

\vspace{0.05in}

\noindent \underline{Choice of constants}: So, at this point we only need to choose the values for $\underline{c_2}, c_3, c_4, c_5,$ and $N_\ast$ for the estimate. We choose the approximately optimal $c_3 = 1.045, c_4 = 18.65, c_5 = 1.082, \underline{c_2}=6.285,$ and $N_\ast = 4.962\times10^7$.
We then obtain the results of the lemma from all these cases, noting that the required conditions on the constants hold.
\end{proof}

\section{Proof of Main Results}

\begin{example}\label{OgataEx}
Using the following example, we will illustrate how we prove our extension of
Ogata's theorem (Theorem \ref{OgataTheorem}) over the next two theorems. Consider the 
scaled representation
\begin{align*}
S = \frac{1}{28}\left (2S^1\right.&\left.\oplus4S^2\oplus7S^3\oplus 8S^4\oplus7S^5\oplus6S^6\oplus4S^7\oplus4S^8\right.\\
&\left.\oplus4S^9\oplus3S^{10}\oplus3S^{11}\oplus2S^{12}\oplus S^{13}\oplus S^{14} \right)
\end{align*}
with multiplicities illustrated in Illustration \ref{PartitionOfReps}(a).
Recall that, just as in the next two results, the $1/28$ is a multiplicative factor while the constant $n_i$ of $n_iS^{\lam_i}$ indicates the multiplicity of $S^{\lam_i}$ in the (unscaled) representation $28S$.

\end{example}
\begin{figure}[htp]  
    \centering
    \includegraphics[width=14cm]{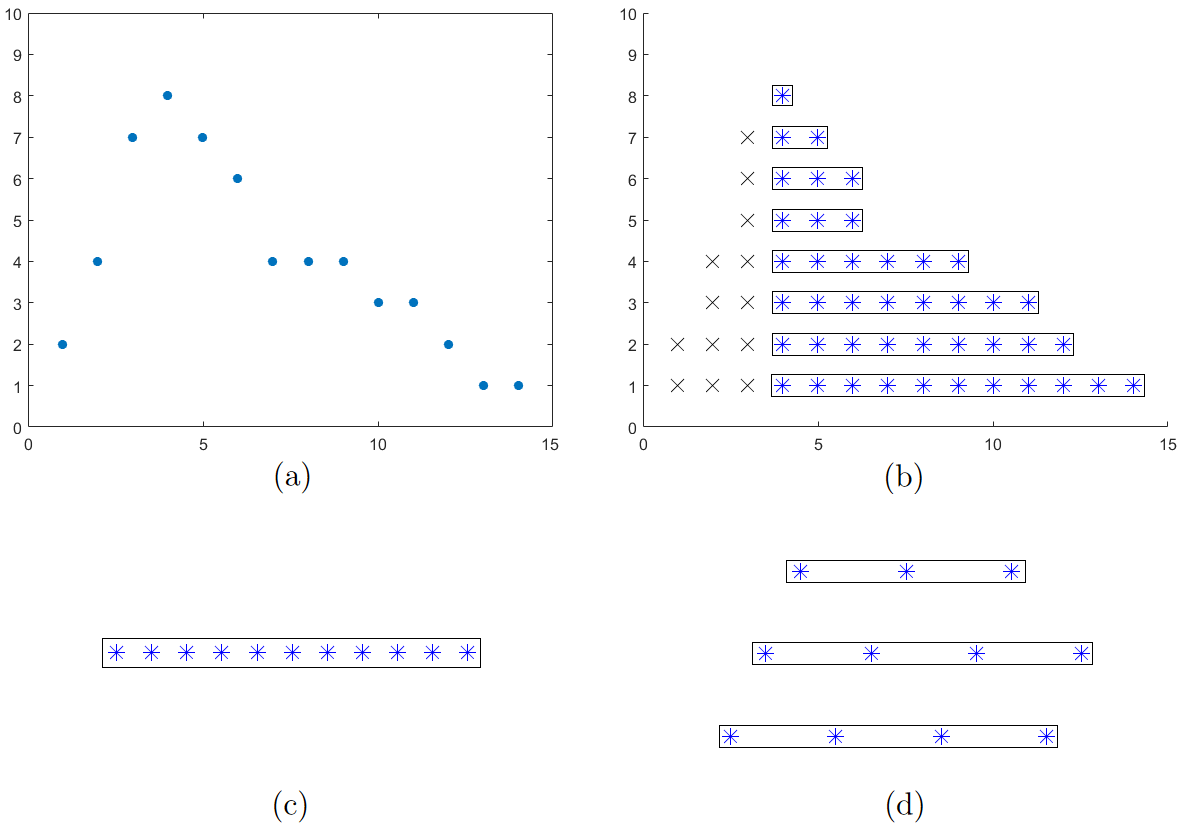}
    \caption{\dark \label{PartitionOfReps}  Illustration of irreducible representations for Example \ref{OgataEx}.}
\end{figure}

Illustration \ref{PartitionOfReps}(a) is a graph of the multiplicities of the irreducible representations in $S$. We construct the almost commuting matrices $A_i'$ nearby the $S(\sigma_i)$ as follows. 
We first partition the direct sum appropriately, which gives us subrepresentations acting on orthogonal invariant subspaces. For each of these subrepresentations we construct nearby commuting matrices.
Then the nearby commuting matrices $A_i'$ are formed by taking the direct sum of the commuting matrices formed in all the invariant subspaces. The distance $\|A_i'-S(\sigma_i)\|$ will be the maximal distance in each of the invariant subspaces corresponding to the partition.

We now discuss the partitions and how we construct their nearby commuting matrices. First, refine the representations illustrated in (b) into two subsets illustrated with $\times$'s and $\ast$'s. 
One such partition will correspond to the $\times$ irreducible representations.
Because the spins of the $\times$ representations are at most $3$, we will ``discard'' all of these by choosing trivial nearby commuting matrices as in the previous lemma. This provides an error of $\frac3{28}$.

We chose which representations were $\times$'s and $\ast$'s in such a way that the multiplicities of the $\ast$ irreducible representations were monotonically decreasing. We then can form a ``level set'' decomposition illustrated by some long and some short horizontal boxes that group the $\ast$ representations as in (b). 

A sample horizontal grouping of representations is given in (c).
Each such horizontal grouping of representations will be itself partitioned as follows. We choose a value of $L$, which is $3$ in this example. We partition each horizontal grouping of $\ast$ representations so that the spins in each partition increase by exactly $L$. These are illustrated in (d).

The way that this is described in the proof of the Theorem \ref{mainthm} is by choosing the arithmetic progression of spins $\mu_1, \mu_2, \dots, \mu_K$ where $\mu_{i+1}-\mu_i=L$ and $\mu_K$ is one of the last $L$ spins to the far right of the grouping in (c).
These provide the partitions of the $\ast$ representations for which we obtain nearby commuting matrices by Lemma \ref{bigLstepLemma}. 

Note that, strictly speaking, in order to apply Theorem \ref{mainthm}, we do not need the representations to be monotonically increasing in the sense that $n_i \geq n_{i+1}$ after some point. What is actually needed is that the multiplicities are monotonically decreasing with steps of size $L$: $n_i \geq n_{i+L}$.

\begin{thm}\label{mainthm}
Let $N \geq 1$ and $\lam_1, \dots, \lam_m$ be given with $\lam_{i_\ast} \leq \frac12\sqrt{N}+1$, $\lam_m \leq \frac12N$, and $\lam_{r+1}-\lam_r =1$. Define $L = \lfloor 1.045\, N^{4/7} \rfloor$.\\
Let $S = \frac1N\left( n_1S^{\lam_1}\oplus \cdots \oplus n_mS^{\lam_m}\right)$, where $n_{i} \geq n_{i+L}$ for $i \geq i_\ast$. 

Then there are commuting self-adjoint matrices $A_i'$ such that
\begin{align}\|A_1'-S(\sigma_1)\|, \|A_2'-S(\sigma_2)\| &\leq 6.286\,N^{-\frac17}, \nonumber\\
\|A_3'-S(\sigma_3)\| &\leq 1.083\,N^{-\frac37}\nonumber
\end{align}
and $A_1', iA_2', A_3'$ are real.

The same result applies if instead $\lam_{r+1}-\lam_r = 1/2$.
\end{thm}
\begin{proof}
We first relabel the  indices of the weights so that $i_\ast=1$ and the weights are $\lam_i$ for $i_0\leq i\leq m$ with $i_0 \leq 1$ being possibly negative. To avoid the trivial case, we can assume that $N \geq (2\cdot6.2)^7 \approx 4.5 \times 10^7$.

Because the differences $\lam_{r+1} - \lam_r$ are an integer, all the $\lam_r$ are integers or half-integers. 
 If we had instead $\lam_{r+1} - \lam_r = 1/2$ then we decompose $S$ into a direct sum of the representations with $\lam_r$ integers and $\lam_{r}'$ half-integers and apply the construction for each separately with $\lam_1 \leq \frac12\sqrt{N}+1, \lam_1' \leq \frac12\sqrt{N}+\frac32$. 
So, we assume that $\lam_{r+1}-\lam_r = 1$ and $\lam_1 \leq \Lambda_0$, where $\Lambda_0 = \frac12\sqrt{N}+\frac32$.

We now break the representation into subrepresentations as follows. If $\lam_m$ is an integer, let $Z$ be the set of integers. If $\lam_m$ is a half-integer, let $Z$ be the set of half-integers. Then the collection of all $\lam_r$ is equal to $[\lam_{i_0}, \lam_m] \cap Z$. We first partition $[\lam_{i_0}, \lam_m] \cap Z$ into
$[\lam_{i_0}, \Lambda_0+2L) \cap Z$ and $[\Lambda_0+2L, \lam_m] \cap Z$. 

For each $\mu_K \in (\lam_m-L, \lam_m] \cap Z$, we form a disjoint (with indices relabeled) arithmetic progression $\mu_1, \dots, \mu_K$, where $\Lambda_0+L< \mu_1 \leq \Lambda_0+2L$ and $\mu_{i+1}-\mu_i = L$. The set $[\Lambda_0+2L, \lam_m]\cap Z$ is thus contained in the union of these disjoint arithmetic progressions. 

We now focus on forming nearby commuting self-adjoint matrices for subrepresentations of the representation $N\cdot S$ corresponding to the arithmetic progressions and also to the representations not accounted for by one of the arithmetic progressions. Then the desired matrices $A_i'$ are formed from the appropriate direct sums.

Suppose $\lam_j\in [\lam_{i_0}, \Lambda_0+2L) \cap Z$ does not belong to one of the above constructed arithmetic progressions. Then 
\[\lam_j \leq \frac12\sqrt{N}+\frac32+2L \leq 5N^{\frac47},\]
hence \[\|S^{\lam_j}(\sigma_i)\| \leq \frac{5}{N}N^{\frac47} = 5\, N^{-\frac37}.\] So, on this summand we choose the component of $A_1'$ and of $A_2'$ to be zero and the component of $A_3'$ to be $\frac{1}{N}S^{\lam_j}(\sigma_3)$. This guarantees a contribution of at most $5N^{-3/7}$ to $\|A_1' - S(\sigma_1)\|$ and $\|A_2' - S(\sigma_2)\|$ on this summand and no contribution to  $\|A_3' - S(\sigma_3)\|$ on this summand.

Now, consider one of the above constructed arithmetic progression $\mu_1, \dots, \mu_K$. For simplicity of notation, let $n_{\mu}$ be the multiplicity of the representation $S^\mu$ in the representation $N\cdot S$. Then
\[\bigoplus_{i=1}^K n_{\mu_i}S^{\mu_i} = n_{\mu_K}\left(S^{\mu_1}\oplus \cdots \oplus S^{\mu_K}\right)\oplus \bigoplus_{r=2}^{K}(n_{\mu_{r-1}}-n_{\mu_r})(S^{\mu_1}\oplus \cdots\oplus S^{\mu_{r-1}}).\]
This is well-defined because $n_i - n_{i+L}\geq 0$ so $n_{\mu_{r-1}}- n_{\mu_r} \geq 0$.

So, we focus on obtaining nearby commuting matrices for the representation of the form $S^{\mu_1}\oplus \cdots\oplus S^{\mu_{r}}$.
Nearby commuting matrices are obtained  by applying Lemma \ref{bigLstepLemma} since $\mu_1 \leq \Lambda_0 + 2L, \mu_K \leq \lam_m\leq \frac12N, \mu_{i+1}-\mu_i=L$. 
So, we conclude the proof of the lemma by taking direct sums of the nearby commuting matrices obtained in each summand.
\end{proof}

We now prove Theorem \ref{OgataTheorem}, giving a constructive proof of Ogata's Theorem for $d=2$ with an explicit estimate and additional structure.
\begin{proof}[Proof of Theorem \ref{OgataTheorem}]
We begin with the first statement. Consider the representation $(S^{1/2})^{\otimes N}$ decomposed as a direct sum of irreducible representations as discussed in Chapter \ref{4.RepTheory}. We write
\[(S^{1/2})^{\otimes N} \cong n_{0}S^0 \oplus n_{1/2}S^{1/2}\oplus \cdots \oplus n_{N/2}S^{N/2}.\]
As discussed in Chapter \ref{4.RepTheory}, this decomposition as well as the unitary operator on $M_{2^N}(\C)$ that realizes this equivalence can be obtained  constructively. Moreover, we choose the unitary to be real.

Depending on whether $N$ is even or odd, the $n_\lam$ are only non-zero when the $\lam$ are all integers or are all half-integers, respectively. By Lemma \ref{1/2mult}, we know that $n_\lam \geq n_{\lam+1}$ for $\lam \geq \frac12\sqrt{N}$. So, we apply Theorem \ref{mainthm} with
\[T_N = \frac{1}{N}(S^{1/2})^{\otimes N} \cong \frac{1}{N}\left(n_{0}S^0 \oplus n_{1/2}S^{1/2}\oplus \cdots \oplus n_{N/2}S^{N/2}\right)\]
 to obtain commuting real self-adjoint matrices $Y_{1,N}, iY_{2,N}, Y_{3,N}$ that satisfy
\[\|T_N(\sigma_i)-Y_{i,N}\| \leq 6.286\,N^{-\frac17}\]
for $i = 1, 2$ and
\[\|T_N(\sigma_3)-Y_{3,N}\| \leq 1.083\,N^{-\frac37}.\]

To obtain the estimate for a more general operator, we proceed as discussed in Section \ref{Outline of Argument}. If $A$ is given by $c_1\sigma_1 + c_2\sigma_2 + c_3\sigma_3 + c_4I_2$ then define $Y_{N}(A) = c_1Y_{1, N} + c_2Y_{2, N} + c_3Y_{3, N} + c_4I_{2^N}$. Recall that $T_N(I_2) = I_{2^N}$ and by Equation (\ref{pauliNorm}),
\[\sqrt{|c_1|^2 + |c_2|^2 + |c_3|^2} \leq 2\|A\|.\]
 So, by the Cauchy-Schwartz inequality,
\begin{align*}
\|T_N(A) - Y_{N}(A)\| &\leq \sum_{i=1}^3 |c_i|\|T_N(\sigma_i)-Y_{i,N}\| \leq 6.286(|c_1|+|c_2|)N^{-\frac17} + 1.083|c_3|N^{-\frac37}\\
&\leq 2\sqrt{2(6.286^2) + 1.083^2N^{-4/7}}\|A\|N^{-\frac17} \leq  17.92\|A\|N^{-1/7}.
\end{align*}

Recall that by Equation (\ref{pSpin}),
\[\sigma_1= \frac12\bp 0 & 1\\1&0  \ep,\;\; \sigma_2 = \frac12\bp 0 & i\\ -i &0 \ep,\;\; \sigma_3=\frac12\bp -1&0\\0&1 \ep.\]
So, $\sigma_1, \sigma_3, I_2$ are symmetric self-adjoint $2\times 2$ matrices and $\sigma_2$ is an antisymmetric self-adjoint matrix. Because $Y_{1,N}, Y_{3,N}, Y_{N}(I_2)$ are real and self-adjoint, they are symmetric. Because $Y_{2, N}$ is imaginary and self-adjoint, it is antisymmetric. Therefore,
\[Y_{N}(A^\ast) = Y_{N}(\overline{c_1}\sigma_1 + \overline{c_2}\sigma_2 + \overline{c_3}\sigma_3 + \overline{c_4}I_2) = \overline{c_1}Y_{1,N} + \overline{c_2}Y_{2,N} + \overline{c_3}Y_{3,N} + \overline{c_4}Y_{I,N} = Y_{N}(A)^\ast\]
and
\[Y_{N}(A^T) = Y_{N}(c_1\sigma_1 -c_2\sigma_2 +c_3\sigma_3 +c_4I_2) = c_1Y_{1,N}-c_2Y_{2,N} + c_3Y_{3,N} + c_4Y_{I,N} = Y_{N}(A)^T.\]

The theorem then follow from these observations.
\end{proof}

\begin{remark}\label{useful}
For a 3 dimensional grid of $10^5$ particles along each axis, one sees that $N=10^{15}$ is a reasonable value of $N$ to apply our result to. We then would have the estimates
$\|T_N(\sigma_i) - Y_{i,N}\| \leq 0.046$ and for more general operators $\|T_N(A) - Y_{i,N}\| \leq 0.13\|A\|$.

For $N = (10^{7})^3$,
$\|T_N(\sigma_i) - Y_{i,N}\| \leq 0.0063$ and $\|T_N(A) - Y_{i,N}\| \leq 0.018\|A\|$.

For $N = (10^{10})^3$,
$\|T_N(\sigma_i) - Y_{i,N}\| \leq 0.00033$ and $\|T_N(A) - Y_{i,N}\| \leq 0.00093\|A\|$.
\end{remark}
\begin{remark}
Loring and S{\o}rensen in \cite{loring2016almost} extend Lin's theorem to respect real matrices. They show that two almost commuting real self-adjoint matrices are nearby commuting real self-adjoint matrices. We have shown that this result is true for Ogata's theorem for $d = 2$. 
The result that we found of the additional structure for $Y_{i,N}$ corresponds to what \cite{loring2015k} calls Class D in 2D (Section 5.2), which is the case of two real self-adjoint matrices and one imaginary self-adjoint matrix that are almost commuting and for which we want to find nearby commuting approximants with the same structure.
\end{remark}

It should be remarked that the suboptimal exponent $\alpha = 1/3$ was used because it provided the real structure of the $A_i'$ and a small explicit constant $C_{1/3}$. Using $\alpha=\frac12, \gamma_3 = \frac35, \gamma_5=\frac25, \gamma=\frac15$ and similar arguments as above, one can obtain the following result. Because $C_{1/2}\leq C_{KS}$ is undetermined we state this result with the best asymptotic decay that our method provides but without an explicit constant.
\begin{thm}\label{OptimalResult}
There is a linear map $Y_N:M_2(\C)\to M_{2^N}(\C)$ such that the $Y_{N}(A)$ commute for all $A\in M_2(\C)$,
\[Y_{N}(A^\ast)=Y_{N}(A)^\ast,\]
and
\[\|T_N(A)-Y_{N}(A)\| \leq Const. \|A\|\,N^{-1/5}.\]

Consequently, $Y_{N}$ preserves the property of being self-adjoint or skew-adjoint.
\end{thm}

	

\backmatter

\renewcommand{\biblistfont}{%
	\normalfont
	\normalsize
}

\phantomsection

\addcontentsline{toc}{chapter}{Bibliography}

\bibliography{thesis.bib}

	\end{document}